\documentclass[11pt]{article}
\usepackage{color}
\usepackage{graphicx}
\usepackage{amsmath}
\usepackage{amssymb}
\usepackage{mathrsfs}
\usepackage[numbers,sort&compress]{natbib}
\usepackage{amsthm}
\numberwithin{equation}{section}
\usepackage[dvipsnames, svgnames, x11names]{xcolor}
\usepackage{pgf}
\usepackage{CJK}
\usepackage{graphicx}
\usepackage{tikz}
\usepackage{stmaryrd}
\usepackage{graphicx}
\usepackage{hyperref}
\usepackage[latin1]{inputenc}
\usepackage[T1]{fontenc}
\usepackage[english]{babel}
\usepackage{listings}
\usepackage{xcolor,mathrsfs,url}
\usepackage{amssymb}
\usepackage{amsmath}
\usepackage{ifthen}
\newtheorem{theorem}{Theorem}
\newtheorem{lemma}{Lemma}
\newtheorem{corollary}{Corollary}
\newtheorem{Proposition}{Proposition}
\newtheorem{RHP}{RHP}

\usepackage {caption}
\usepackage{float}
\usepackage{subfigure}
\hypersetup{hypertex=true,
            colorlinks=true,
            linkcolor=blue,
            anchorcolor=blue,
            citecolor=red}

\DeclareMathOperator*{\res}{Res}

\begin{document}

\title{ Long-time asymptotic behavior of the  Novikov equation in  space-time  solitonic regions  }
\author{Yiling YANG$^1$\thanks{\ Email address: 19110180006@fudan.edu.cn } \  and  \  Engui FAN$^{1}$\thanks{\ Corresponding author and email address: faneg@fudan.edu.cn } }
\footnotetext[1]{ \  School of Mathematical Sciences  and Key Laboratory   for Nonlinear Science, Fudan   University, Shanghai 200433, P.R. China.}

\date{ }

\maketitle
\begin{abstract}
	\baselineskip=16pt

In this paper, we   study the long time asymptotic behavior for   the Cauchy problem    of the Novikov equation with $3\times 3$  matrix spectral problem
\begin{align}
	&u_{t}-u_{txx}+4 u_{x}=3uu_xu_{xx}+u^2u_{xxx}, \nonumber \\
	&u(x, 0)=u_{0}(x),\nonumber
\end{align}	
where $u_0(x)$    $u_0(x)\rightarrow \kappa>0, \ x\rightarrow \pm \infty$  and $u_0(x)-\kappa$ is
 assumed in the Schwarz space.
It is shown that
the solution of the  Cauchy problem  can be characterized   via a  Riemann-Hilbert  problem in a new scale $(y,t)$ with
$$y=x-\int_{x}^{\infty}\left( (u-u_{xx}+1)^{2/3} -1\right) ds.$$
 In different space-time solitonic regions of $\xi=y/t\in (-\infty,-1/8)\cup(1,+\infty) $ and $\xi \in(-1/8,1)$,  we  apply $\overline\partial$ steepest descent method  to
obtain  the different long time asymptotic expansions  of the solution $u(y,t)$.
 The corresponding  residual error order is $\mathcal{O}(t^{-1+\rho})$ and $\mathcal{O}(t^{-3/4})$ respectively
   from  a $\overline\partial$-equation.
Our result  implies that  soliton resolution   can
be characterized with  an $N(\Lambda)$-soliton whose parameters are modulated by
a sum of localized soliton-soliton
 interactions as one moves through the regions.

\end{abstract}

\baselineskip=17pt

\tableofcontents

\quad

\section {Introduction}

In this paper, we study the long time asymptotic behavior for the initial value problem  on a nonzero background  for the
Novikov equation
\begin{align}
	&u_{t}-u_{txx}+4 u_{x}=3uu_xu_{xx}+u^2u_{xxx},\label{Novikov1}\\
	&u(x, 0)=u_{0}(x), \quad x \in \mathbb{R},\  t>0,\\
	&u_{0}(x)\to\kappa>0,\ x\to\pm\infty,\label{Novikov3}
\end{align}	
where    $u=u(x,t)$  is a real-valued function of   $x$ and $t$. By introducing the momentum variable $m=u-u_{x x}$,  the Novikov equation (\ref{Novikov1}) can be rewritten as
\begin{align}
	&m_{t}+\left(m_xu+3mu_x\right)u=0,
\end{align}
or equivalently,
\begin{align}
	&(m^{2/3})_{t}+\left(u^2m^{2/3}\right)_x=0.\label{Novikov0}
\end{align}
The Novikov equation (\ref{Novikov1}) as a new integrable system was derived in the search for a classification of integrable
generalized Camassa-Holm equations of the form
\begin{align}
	(1-\partial_x^2)u_t=F(u,u_x,u_{xx},...)
\end{align}
possessing infinite hierarchies of higher symmetries.  A   scalar    Lax pair  involving  the third order derivative with respect to $x$
was given \cite{39,37}.  Further, by using the prolongation algebra method,
Hone and Wang proposed a $3\times 3$ matrix Lax pair and a bi-Hamiltonian structure for the Novikov equation (\ref{Novikov1}) \cite{HW29}. This Lax pair was used
to  construct   explicit  peakon solutions  on zero background, which  replicates  a feature characterizing
  the waves of great height-waves of largest amplitude that were exact solutions of the governing equations for water waves   \cite{HW29,CA1,CA2,CA3,TJF}.
   Hone et al. further  derived the
explicit formulas for multipeakon solutions of the Novikov equation (\ref{Novikov1})   \cite{HW28}.  By using  the  Hirota bilinear method,
   Matsuno  presented parametric representations of smooth multisoliton solutions
  for  the Novikov equation (\ref{Novikov1}) on a nonzero constant  background \cite{M36}.
He also demonstrated that a smooth soliton converged to a peakon in the limit where the constant background tended to $0$ while the velocity of the soliton is fixed.
Wu  et al.  obtained   N-soliton solutions to the Novikov equation  through Darboux transformations \cite{Wu6}.    Recently, Chang et al.
applied Pfaffian technique  to study multipeakons of the Novikov equation,
a link between the Novikov peakons and the finite Toda lattice of BKP type as well as  Hermite-Pade approximation to the Novikov peakon problem
\cite{CH2018, CH2022}.  Boutet de Monvel et al.  developed the inverse scattering theory  to
 the Novikov equation  (\ref{Novikov1}) with  nonzero constant background \cite{RHP},  where under a  simple  transformation
\begin{align}
		u(x, t)\rightarrow \kappa \tilde{u}(x-\kappa^2t, \kappa^2t)+\kappa,
\end{align}
the  Cauchy problem of Novikov equation (\ref{Novikov1})--(\ref{Novikov3}) reduces to the following  new Cauchy
problem on zero background
\begin{align}
	&(\tilde{m}^{2/3})_{t}+\left(\tilde{m}^{2/3}\left(u^{2}+2u\right)\right)_{x}=0, \  \tilde{m}=u-u_{x x}+1,\label{Novikov}\\
	&u(x,0)=u_0(x).\label{Novikov2}
\end{align}
where $u_0(x)$ satisfies the sign condition
\begin{equation*}
	u_0(x)- u_{0,xx}(x)+1>0.
\end{equation*}
Then there exists   a unique global solution $u(x, t)$  of the Novikov equation (\ref{Novikov1}),
 such that  $u(x,t)\to0$ as $x\to\pm\infty$ for all $t>0$  \cite{32}.

The Novikov equation (\ref{Novikov})  also  admits $3\times 3$  matrix spectral problem   as the Sasa-Satuma equation, Degasperis-Procesi  (DP) equation,  good Boussinesq equation,  three-wave equation
  \cite{Deift1982,Constantin1,Monvel3,Lenells1,RHP, Lenells2, Monvel2, Geng1,  Geng3, YF2}.
 However, the Novikov equation and  DP equation  possess  numerous common characteristics in their
  Riemann-Hilbert (RH)  problem  and   face   some    difficulties  \cite{RHP}.
One of the difficulties is the Lax pair associated with (\ref{Novikov}) has six  singularities at  $\varkappa_n=e^{\frac{n\pi i}{3}}$ for $n=1,...,6$.
  Boutet de Monvel et al.  developed  a   steepest descent approach to obtain the  long time asymptotic behavior for   the   DP
 equation \cite{Monvel2}.  However, to the best of our knowledge, the
the long time asymptotics  of the Novikov equation is still not presented yet.

\begin{figure}
\begin{center}
\begin{tikzpicture}
\draw[yellow!20, fill=yellow!20] (0,0)--(4,0)--(4,2)--(0, 2);
\draw[green!20, fill=green!20] (0,0 )--(4,2)--(0,2)--(0,0);
\draw[blue!20, fill=blue!20] (0,0 )--(-4,1.1)--(-4,2)--(0, 2)--(0,0);
\draw[yellow!20, fill=yellow!20] (0,0 )--(-4,0)--(-4,1.1)--(0,0);
\draw [ -> ] (-4.6,0)--(4.6,0);
\draw [ -> ](0,0)--(0,3.5);
\draw [red,thick  ](0,0 )--(4,2);
\draw [red,thick  ](0,0 )--(-4,1.1);
\node    at (0,-0.3)  {$0$};
\node    at (5,0)  {y};
\node    at (0,3.8 )  {t};
\node  [below]  at (1.3,1.8) {\footnotesize $0<\xi<1$};
\node  [below]  at (-1.5,1.8) {\footnotesize $-1/8<\xi<0 $};
\node  [below]  at (2.3,0.8) {\footnotesize $ \xi>1$};
\node  [below]  at (-3,0.6) {\footnotesize $ \xi<-1/8 $};
\node  [below]  at (-5,1.6) {\footnotesize $ \xi=-1/8 $};
\node  [below]  at (4.7,2.3) {\footnotesize $ \xi=1 $};
\end{tikzpicture}
\end{center}
\caption{\footnotesize  Asymptotic approximations    of the  Novikov equation  in  different space-time  solitonic regions,
where without stationary  phase points  in yellow region;  12  and  24  stationary   phase points in  green and  blue  regions   respectively.   }
\label{result1}
\end{figure}
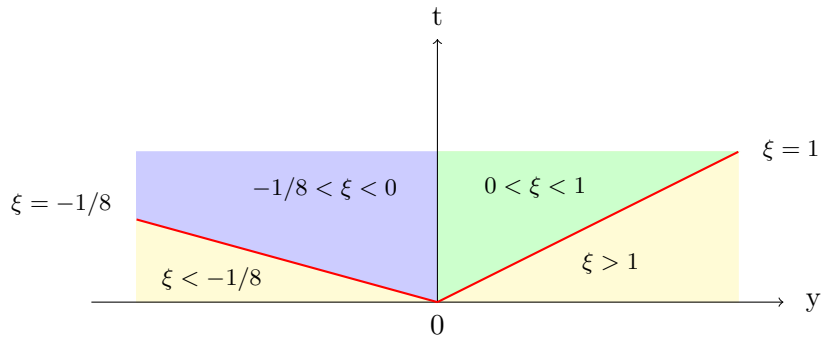

Motivated by their  work \cite{RHP,Monvel2},    we  carry out  the long-time analysis
for the Novikov equation in this paper.    We   obtain    different leading  order  asymptotic approximation  for the Novikov equation (\ref{Novikov})
in different  space-time  solitonic regions
(see Figure \ref{result1} and the detail results   are  given in the  Theorem \ref{last} in the section 8).  Our  main tool is  the   nonlinear $\bar{\partial}$ steepest approach introduced by McLaughlin-Miller,
  which was first  applied to analyze asymptotic of orthogonal polynomials \cite{MandM2006,MandM2008},
   later applied to analyze long time asymptotics of  integrable systems   \cite{DandMNLS,fNLS,Liu3,SandRNLS,YF1,YF3,YF4}.
The   first work on   the long-time behavior of nonlinear wave equations  with inverse scattering method
was  carried out  by Manakov   \cite{Manakov1974}.
Later,  this  method   was applied to other     integrable systems
 such as  KdV,  Landau-Lifshitz  and the reduced Maxwell-Bloch   system \cite{ZM1976,SPC,BRF,Foka}.
In 1993,
Deift and Zhou  developed a  nonlinear steepest descent method to rigorously obtain the long-time asymptotics behavior of the solution for the mKdV equation
 \cite{RN6}. Since then    this method
has been widely  applied on the   NLS equation, KdV equation, derivative NLS equation,
 Fokas-Lenells equation,  short-pulse equation and  Camassa-Holm equation,  etc. \cite{RN9,RN10,Grunert2009,MonvelCH,xu2015,xusp}.

Our  paper is arranged as follows.   In section \ref{sec2},   we recall  some main  results on
 the construction  process  of  RH  problem    \cite{RHP},  which will be used
 to analyze   long-time asymptotics  of the Novikov equation in our paper. Further  we demonstrate that  the norm of  reflection
  coefficient
  $\parallel r(z)\parallel_{L^\infty}<1$ to ensure the existence of solutions of  the RH problem $M (z)$.
  In section \ref{sec3},
a $T(z)$ function  is introduced to  normalize the RH problem
 $M (z)$ into  a new   RH problem  for  $M^{(1)}(z)$,  which  admits a regular discrete spectrum and  two  triangular  decompositions of the jump matrix
near $z=0$.
In section \ref{sec4},  by introducing a matrix-valued  function  $R(z)$  to make a continuous extension of  $M^{(1)}(z)$
into    a hybrid  $\bar{\partial}$-RH problem  for  $M^{(2)}(z)$.
  Further   we decompose  $M^{(2)}(z)$    into a
 pure RH   problem  for  $M^{R}(z)$ and a  pure $\bar{\partial}$ Problem for  $M^{(3)}(z)$.
 We find that  the contribution to the  $M^{R}(z)$ come from  three aspects  (
two aspects in the case  without  stationary phase points).
 The first   contribution comes from
 discrete spectrum, where  a  modified reflectionless RH problem $M^{(r)}(z)$   for the soliton components  is  solved   in Section \ref{sec6}.
 The second  contribution  in Section \ref{secpc}  comes from jump contours, which is   approximated by  a  local solvable  model $M^{lo}(z)$.
The third  contribution   is  from   a  pure jump RH problem near singularities $\varkappa_n$.
  Finally  the  residual  error function  is  is given  with  a small RH problem.
 In Section \ref{sec8},   we analyze  the $\bar{\partial}$-problem  for $M^{(3)}(z)$.
   Finally, in Section \ref{sec9},   based on  the result obtained above,   a relation formula
   is found
\begin{align}
 M(z) = M^{(3)}(z)E(z)M^{(r)}(z)R^{(2)}(z)^{-1}T(z)^{-1},\nonumber
\end{align}
from which   we then obtain the   long-time   asymptotic behavior  for the Novikov equation (\ref{Novikov}) via a reconstruction formula.

\section {Direct scattering  and basic   RH problem}\label{sec2}

To state our results precisely we introduce notation and function spaces  used in this paper.
If $I$ is an interval on the real line $\mathbb{R}$ and $X$ is a  Banach space, then $C^0(I,X)$ denotes the space of continuous functions on $I$ taking values in $X$. It is equipped with the norm
\begin{equation*}
	\|f\|_{C^{0}(I, X)}=\sup _{x \in I}\|f(x)\|_{X}.
\end{equation*}
Besides, we denote $C^0_B(X)$ as a  space of bounded continuous functions on $X$.

If the  entries  $f_1$ and $f_2$  are in space $X$,  then we call vector  $\vec{f}=(f_1,f_2)^T$  is in space $X$ with $\parallel \vec{f}\parallel_X\triangleq \parallel f_1\parallel_X+\parallel f_2\parallel_X$. Similarly, if every  entries of  matrix $A$ are in space $X$, then we call $A$ is also in space $X$.
We introduce  the normed spaces.
  A weighted $L^p(\mathbb{R})$ space is specified by
$$L^{p,s}(\mathbb{R})  =  \left\lbrace f(x)\in L^p(\mathbb{R}) | \hspace{0.1cm} |x|^sf(x)\in L^p(\mathbb{R}) \right\rbrace; $$
  A Sobolev space is defined by
$$W^{k,p}(\mathbb{R})  =  \left\lbrace f(x)\in L^p(\mathbb{R}) | \hspace{0.1cm} \partial^j f(x)\in L^p(\mathbb{R})  \text{ for }j=1,2,...,k \right\rbrace;$$
 A weighted Sobolev space  is defined by
 $$H^{k,s}(\mathbb{R})   =  \left\lbrace f(x)\in L^2(\mathbb{R}) | \hspace{0.1cm} (1+|x|^s)\partial^jf(x)\in L^2(\mathbb{R}),  \text{ for }j=1,...,k \right\rbrace.$$
 And the norm of $f(x)\in L^{p}(\mathbb{R})$ and $g(x)\in L^{p,s}(\mathbb{R})$ are  abbreviated to $\parallel f\parallel_{p}$,
  $\parallel g\parallel_{p,s}$, respectively.

\subsection{Spectral analysis}

\quad The Novikov equation (\ref{Novikov})      admits a Lax pair   \cite{HW29,RHP}
\begin{equation}
\breve{\Phi}_x (k) = \breve{X} \breve{\Phi}(k),\hspace{0.5cm}\breve{\Phi}_t(k) =\breve{T} \breve{\Phi}(k), \label{lax0}
\end{equation}
where
\begin{equation}
	\breve{X}=\left(\begin{array}{ccc}
		0&	k\tilde{m} & 1 \\
		0&0 & k\tilde{m}\\
		1&0&0
	\end{array}\right),\nonumber
\end{equation}
\begin{equation}
	\breve{T}=\left(\begin{array}{ccc}
		-(u+1)u_x+\frac{1}{3k^2}&	\frac{u_x}{k}-(u^2+2u)k\tilde{m} & u_x^2+1 \\
		\frac{u+1}{k}&-\frac{2}{3k^2} & -\frac{u_x}{k}-(u^2+2u)k\tilde{m}\\
		-u^2-2u & \frac{u+1}{k} &(u+1)u_x+\frac{1}{3k^2}
	\end{array}\right),
 \nonumber
\end{equation}
and  $k$ is a  spectral parameter.
In order to control the large $k$ behavior of the solutions of above Lax pair, we introduce a new transformation
\begin{align}
	\Phi (z) \triangleq P^{-1}(z)D^{-1}(x,t)\breve{\Phi}(k),
\end{align}
where  $z$ is a new uniformization variable to  avoid multi-value  of   eigenvalue $k$ and
\begin{align}
	&D(x,t)\triangleq\text{diag}\{q(x,t),q^{-1}(x,t),1\},\quad q=q(x,t)\triangleq \tilde{m}^{1/3}(x,t),\nonumber\\
	&P(z)\triangleq\left(\begin{array}{ccc}
		\lambda_1^2(z)&	\lambda_2^2(z) & \lambda_3^2(z) \\
		k (z) &k(z) & k(z)\\
			\lambda_1(z)&	\lambda_2(z) & \lambda_3(z)
	\end{array}\right),\label{P}\\
&P^{-1}(z)\triangleq\left(\begin{array}{ccc}
	\frac{1}{3\lambda_1^2(z)-1}&	0 & 0 \\
	0&\frac{1}{3\lambda_2^2(z)-1} & 0\\
	0&	0 & \frac{1}{3\lambda_3^2(z)-1}
\end{array}\right)\left(\begin{array}{ccc}
1&	\frac{z}{\lambda_1(z)} & \lambda_1(z) \\
1&\frac{z}{\lambda_2(z)} & \lambda_2(z)\\
1&	\frac{z}{\lambda_3(z)} & \lambda_3(z)
\end{array}\right)\label{P-1}.
\end{align}
Here $\lambda_j(z)$, $j=1,2,3 $   are  three  roots   of  the following  algebraic equation
\begin{align}
	\lambda^3(z)-\lambda (z)=k^2(z),
\end{align}
where
\begin{align}
	&k(z)^2 =\frac{1}{3\sqrt{3}}\left(z^3+\frac{1}{z^3} \right), \quad\lambda_j(z) =\frac{1}{\sqrt{3}}\left(\omega^j z+\frac{1}{\omega^j z} \right),\quad \omega=e^{\frac{2i\pi}{3}}.\label{lambda}
\end{align}
Obviously, $k(\kappa_n)=0$, for $\kappa_n=e^{\frac{i\pi}{6}+\frac{i\pi(n-1)}{3}}$, $n=1,...,6$.
By this transformation, $\Phi$ admits a new Lax pair as
\begin{align}
	& \Phi_x-q^2\Lambda(z)\Phi=U\Phi,\\
	&\Phi_t+\left[(u^2+2u)q^2\Lambda(z)-A(z) \right] \Phi=V\Phi,\label{laxphi}
\end{align}
where
\begin{align}
	&\Lambda(z)\triangleq\text{diag}\{\lambda_1(z),\lambda_2(z),\lambda_3(z)\},  \ \ A(z)\triangleq\frac{1}{3k^2}+\Lambda(z)^{-1},\nonumber\\
	&U\triangleq U_1U_2, \ \ V\triangleq U_1(V_1+V_2\Lambda),\label{laxU}\\
	&U_1\triangleq\text{diag} \left\{\frac{1}{3\lambda_1^2-1},\frac{1}{3\lambda_2^2-1},\frac{1}{3\lambda_3^2-1}\right\}, \nonumber\\
	&U_2\triangleq\left(\begin{array}{ccc}
		c_2\lambda_1&	c_1(\lambda_1\lambda_2-\lambda_2^2)+c_2\lambda_2 & c_1(\lambda_1\lambda_3-\lambda_3^2)+c_2\lambda_3 \\
		c_1(\lambda_1\lambda_2-\lambda_1^2)+c_2\lambda_1& c_2\lambda_2 & c_1(\lambda_3\lambda_2-\lambda_3^2)+c_2\lambda_3\\
		c_1(\lambda_1\lambda_3-\lambda_3^2)+c_2\lambda_1&	c_1(\lambda_3\lambda_2-\lambda_2^2)+c_2\lambda_2 & c_2\lambda_3
	\end{array}\right), \nonumber
\end{align}
with  $c_1=\frac{q_x}{q}$ and $c_2=q^{-2}-q^2$. $V_1$ has same  form of $U_1$ with $c_1$ and $c_2$ replaced by $c_3=-(u^2+2u)\frac{q_x}{q}$ and $c_4=(u^2+2u)q^2+\frac{u_x^2+1}{q^2}-1$ , respectively.
And we  denote $c_5=\frac{u_x}{q}$, $c_6=(u+1)q-1$, then $V_2$ can be written as
\begin{align}
	V_2\triangleq[V_2^{(jl)}]_{3\times3},\quad V_2^{(jl)}\triangleq c_5\left(\frac{1}{\lambda_l}-\frac{1}{\lambda_j} \right) +c_6\left( \frac{\lambda_j}{\lambda_l}+\frac{\lambda_l}{\lambda_j} \right).
\end{align}
Introduce a $3\times3$ diagonal function
$$Q(z;x,t)=y(x,t)\Lambda(z)+tA(z),$$
where
\begin{align}
&Q_x=q^2\Lambda, \quad Q_t=-\left(u^2+2u \right) q^2\Lambda-A, \nonumber\\
	&y(x,t)=x-\int_{x}^{\infty}\left( q^2(s,t)-1\right) ds\label{y}, \\
	&c_+(x,t)=\int_{x}^{\infty}\left( q^2(s,t)-1\right) ds.\nonumber
\end{align}
 Making a transformation
\begin{align}
	\Phi(z;x,t)=\mu(z;x,t)  e ^{Q},\label{transmu}
\end{align}
then  Lax pair (\ref{laxphi}) is changed to
\begin{align}
	\mu_x-[Q_x,\mu]=U\mu,\quad \mu_t-[Q_t,\mu]=V\mu,\label{laxmu}
\end{align}
whose  solutions  satisfy   the  Fredholm integral equations
\begin{equation}
	\mu^{\pm}(z;x,t)=I-\int_{x}^{\pm \infty}e^{-\hat{\Lambda}(z)\int_{x}^{s}q^2(v,t)dv}[U\mu_\pm(s,t;z)]ds\label{intmu}.
\end{equation}
We  define six  rays at $z=0$
$$\Sigma =\cup_{n=1}^6 L_n, \ \  \ L_n=e^{\frac{\pi(n-1)i}{3}}\mathbb{R}^+, \ n=1,\cdots,6, $$
which  divide the complex plane $ \mathbb{C}$   into  six open cones
$$S_n=\{z\in\mathbb{C};\arg z\in(   {(n-1)\pi}/{3}, { n \pi }/{3})\}, \   n=1, \cdots, 6,$$
see Figure \ref{figC}.   Denote the matrix
$$\mu^{\pm}=\left(  \mu^{\pm}_1, \mu^{\pm}_2, \mu^{\pm}_3 \right), $$
where  $ \mu^{\pm}_1$,  $ \mu^{\pm}_3$ and $ \mu^{\pm}_3$ are
the first, second and third column of $ \mu^{\pm}(z)$  respectively.
Then  from  (\ref{intmu}),   we can show that  $\left(\mu^{+}_1, \mu^{+}_2, \mu^{+}_3 \right)$ is  analytical  in   the domains
$$\left(\bar{S}_1\cup \bar{S}_2,\ \bar{S}_5\cup \bar{S}_6,\ \bar{S}_3\cup \bar{S}_4 \right).$$
  And   $\left(\mu^{-}_1, \mu^{-}_2, \mu^{-}_3 \right)$ is  analytical  in the domains
$$\left(\bar{S}_4\cup \bar{S}_5,\ \bar{S}_3\cup \bar{S}_2,\ \bar{S}_1\cup \bar{S}_6 \right).$$
 Here, $\bar{S}_j$ denotes the  closure of $S_j$ for $j=1,...,6$, respectively.
The initial points of integration $\infty_{j,l}$ are specified as follows for each matrix entry $(j,l)$ for $j,l=1,2,3$:
\begin{align}
	\infty_{j,l}=\left\{ \begin{array}{ll}
		+\infty,   &\text{if Re}\lambda_j\geq \text{Re}\lambda_l,\\[12pt]
		-\infty  , &\text{if Re}\lambda_j< \text{Re}\lambda_l.\\
	\end{array}\right.
\end{align}
Then (\ref{intmu}) can be rewritten as a system of   Fredholm integral equations
\begin{equation}
	\mu_{jl}(z;x,t)=I_{jl}-\int_{x}^{ \infty_{jl}}e^{-\int_{x}^{s}q^2(v,t)dv (\lambda_j(z)-\lambda_l(z))}[U\mu(s,t;z)]_{jl}ds\label{intmujl}.
\end{equation}

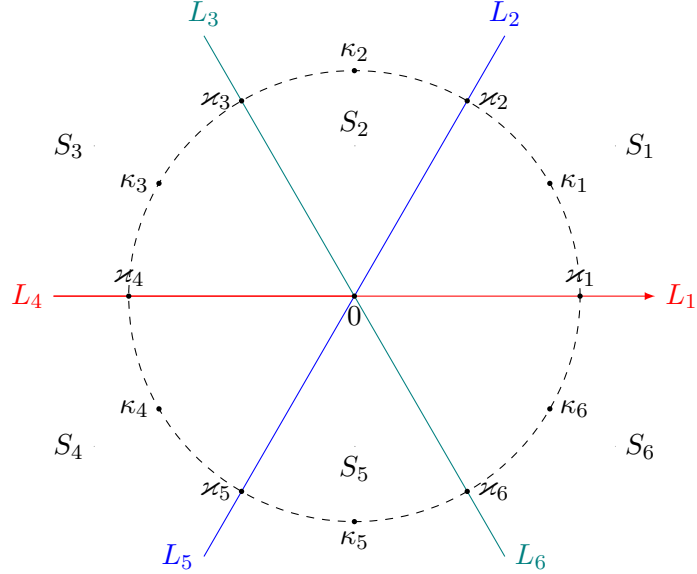
\begin{figure}[t]
	\centering
	\begin{tikzpicture}[node distance=2cm]
		\draw[red,-latex](-4,0)--(4,0)node[right]{$L_1$};
		\draw[blue](0,0)--(2,3.464)node[above]{$L_2$};
		\draw[dashed] (3,0) arc (0:360:3);
		\draw[teal](0,0)--(-2,3.464)node[above]{$L_3$};
		\draw[red ](0,0)--(-4,0)node[left]{$L_4$};
		\draw[teal ](0,0)--(2,-3.464)node[right]{$L_6$};
		\draw[blue ](0,0)--(-2,-3.464)node[left]{$L_5$};
		\coordinate (A) at (1.5,2.598);
		\coordinate (B) at (1.5,-2.598);
		\coordinate (C) at (-1.5,2.598);
		\coordinate (D) at (-1.5,-2.598);
		\coordinate (E) at (3,0);
		\coordinate (F) at (-3,0);
		\coordinate (G) at (-2.598,1.5);
		\coordinate (H) at (-2.598,-1.5);
		\coordinate (J) at (2.598,1.5);
		\coordinate (K) at (2.598,-1.5);
		\coordinate (L) at (0,3);
		\coordinate (M) at (0,-3);
		\coordinate (v) at (0,0);
		\fill (v) circle (1pt) node[below] {$0$};
		\fill (A) circle (1pt) node[right] {$\varkappa_2$};
		\fill (B) circle (1pt) node[right] {$\varkappa_6$};
		\fill (C) circle (1pt) node[left] {$\varkappa_3$};
		\fill (D) circle (1pt) node[left] {$\varkappa_5$};
		\fill (E) circle (1pt) node[above] {$\varkappa_1$};
		\fill (F) circle (1pt) node[above] {$\varkappa_4$};
		\fill (G) circle (1pt) node[left] {$\kappa_3$};
		\fill (H) circle (1pt) node[left] {$\kappa_4$};
		\fill (J) circle (1pt) node[right] {$\kappa_1$};
		\fill (K) circle (1pt) node[right] {$\kappa_6$};
		\fill (L) circle (1pt) node[above] {$\kappa_2$};
		\fill (M) circle (1pt) node[below] {$\kappa_5$};
		\coordinate (a) at (-3.464,2);
		\coordinate (s) at (-3.464,-2);
		\coordinate (d) at (3.464,2);
		\coordinate (f) at (3.464,-2);
		\coordinate (g) at (0,2);
		\coordinate (h) at (0,-2);
		\fill (a) circle (0pt) node[left] {$S_3$};
		\fill (s) circle (0pt) node[left] {$S_4$};
		\fill (d) circle (0pt) node[right] {$S_1$};
		\fill (f) circle (0pt) node[right] {$S_6$};
		\fill (g) circle (0pt) node[above] {$S_2$};
		\fill (h) circle (0pt) node[below] {$S_5$};
	\end{tikzpicture}
	\caption{\footnotesize  The $z$-plane is divided into six analytical  domains $S_n, n=1,...,6$ by
 six   rays  $L_n$.   $\varkappa_n$ and $\kappa_n$    are spectral   singularities.}
	\label{figC}
\end{figure}
 It was shown  that the eigenfunction  $\mu(z)$ defined by (\ref{transmu}) has the following properties  \cite{RHP}.

\begin{Proposition}
	 The	equations (\ref{intmu}) uniquely define a $3\times3$-matrix valued solutions $\mu (z) \triangleq \mu(z;x,t)$ of (\ref{laxmu})
	with the following properties:
\begin{itemize}
	\item [(1)] det $\mu(z)$=1;

\item [(2)]  $\mu(z)$ is piecewise meromorphic with respect to $\Sigma $, as function of the spectral parameter $z$;

\item [(3)]  $\mu(z)$ obeys the  symmetries $\mu(z)=\Gamma_1\overline{\mu(\bar{z})}\Gamma_1=\Gamma_2\overline{\mu(\omega^2\bar{z})}\Gamma_2=\Gamma_3\overline{\mu(\omega\bar{z})}\Gamma_3$ and $\mu(z)=\overline{\mu(\bar{z}^{-1})}$, where
	\begin{align}
		\Gamma_1=\left(\begin{array}{ccc}
			0&	1 & 0 \\
			1&0 & 0\\
			0&	0 & 1
		\end{array}\right),\quad \Gamma_2=\left(\begin{array}{ccc}
		0&	0 & 1 \\
		0&1 & 0\\
		1&	0 & 0
	\end{array}\right),\quad 					\Gamma_3=\left(\begin{array}{ccc}
	1&	0 & 0 \\
	0&0 & 1\\
	0&	1 & 0
\end{array}\right);
	\end{align}
\item [(4)] $\mu(z)$ has pole singularities at $\varkappa_n=e^{\frac{n\pi i}{3}}$ with $\mu=\mathcal{O}(\frac{1}{z-\varkappa_n})$ as $z\to\varkappa_n$;

\item [(5)] $\mu(z) \to I$ as $z\to\infty$, and for $z\in\mathbb{C}\setminus\Sigma $, $\mu(z)$ is bounded as $x\to-\infty$ and $\mu (z) \to I$ as $x\to+\infty$;
\end{itemize}
\end{Proposition}
Denote  $\mu_{\pm}(z;x,t)$ as   the limiting values of $\mu (z';x,t)$ as $z'\to z $  from the
positive or negative side of $L_n$,    then they   are related as follows
\begin{align}
	\mu_{+}(z)=\mu_{-}(z)e^Q v_ne^{-Q},\quad z\in L_n,\ n=1,\cdots,6,
\end{align}
where $v_n$ only depends on $z$ and  is completely determined by $u(x, 0)$, i.e., by the initial data for the Cauchy problem (\ref{Novikov}). Take  $L_1=\mathbb{R}^+$ and $L_4=\mathbb{R}^-$ as an example,  $v_n$ for $n=1,4$ has a   special matrix structure
\begin{align}
	v_n(z)= \left(\begin{array}{ccc}
		1&	-r_\pm(z) & 0 \\
		\bar{r}_\pm(z)& 1-|r_\pm(z)|^2 & 0\\
		0&	0 & 1
	\end{array}\right),\quad z\in\mathbb{R}^\pm.	
\end{align}
$r_\pm(z)$ are  single scalar
functions,  with $r_\pm(z)\in L^{\infty}(\mathbb{R}^\pm)$, and $r_\pm(z)=\mathcal{O}(z^{-1})$ as $z\to\pm\infty$. The symmetry of $\mu (z) $ gives that $r_\pm(z)=\overline{r_\pm(z^{-1})}$,
therefore it also has $r_\pm(z)\in L^{2}(\mathbb{R}^\pm)$ and $\lim_{z\to 0^\pm}r(z)_\pm=0$.
Moreover, the singularities at $\pm1$ give that $r_\pm(\pm1)=0$. So we may define  \textit{reflection coefficient }
\begin{align}
	r(z)=\left\{ \begin{array}{ll}
		r_\pm(z),   & z\in \mathbb{R}^\pm,\\[12pt]
		0  , & z=0.\\
	\end{array}\right.
\end{align}
Then $r\in L^\infty(\mathbb{R})\cap L^2(\mathbb{R})$ and $r(z)=\mathcal{O}(z^{-1})$ as $z\to\infty$.  In reference \cite{RHP, BC1984},  it was shown  that there exist at most a finite
number of simple poles $z_n$ of $\mu(z)$ lying in $S_1\cap\{z\in\mathbb{C};\ |z|>1\}$ and $w_m$  lying in $S_1\cap\{z\in\mathbb{C};\ |z|=1\}$.
And there are no poles except $\pm1$, $\pm\omega$ and $\pm\omega^2$ on the contour $\Sigma $. Note that, unlike the case of $2\times2$ matrix function,  the residue conditions have two special matrix forms with only one nonzero entry.

 To differentiate this two types of poles,  we denote them as $z_n$, $z_n^A$ and $w_m$, $w_m^A$, respectively.
Denote $N_1$, $N_1^A$, $N_2$ and $N_2^A$ as the number of $z_n$, $z_n^A$, $w_m$,  and $w_m^A$, respectively. The  symmetries  of $\mu$  imply  $\bar{z}_n^{-1}$ and $\frac{1}{\bar{z}_n^A}$ are also the poles of $\mu$ in $S_1$.  It is convenient to define  $\zeta_n=z_n$, and $\zeta_{n+N_1}=\bar{z}_n^{-1}$ for $n=1,\cdot\cdot\cdot,N_1$; $\zeta_{m+2N_1}=w_m$  for $m=1,\cdot\cdot\cdot,N_2$; $\zeta_{j+2N_1+N_2}=z_j^A$, and $\zeta_{j+2N_1+N_2+N_1^A}=\frac{1}{\bar{z}_n^A}$ for $j=1,\cdot\cdot\cdot,N_1^A$; $\zeta_{m+2N_1+2N_1^A+N_2}=w_l^A$  for $l=1,\cdot\cdot\cdot,N_2^A$. For the sake of brevity, let
\begin{align*}
	N_0=2N_1+2N_1^A+N_2+N_2^A.
\end{align*}
Moreover,  $\omega \zeta_n$, $\omega^2 \zeta_n$, $\bar{\zeta}_n$, $\omega\bar{\zeta}_n$,   $\omega^2\bar{\zeta}_n$  are also poles of $\mu(z)$ in $S_j$, $j=2,\cdots,6$. So for convenience, let $\zeta_{n+N_0}=\omega\bar{\zeta}_n$, $\zeta_{n+2N_0}=\omega\zeta_n,$ $\zeta_{n+3N_0}=\omega^2\bar{\zeta}_n$, $\zeta_{n+4N_0}=\omega^2\zeta_n$ and  $\zeta_{n+5N_0}=\bar{\zeta}_n$ for $n=1,...,N_0$.
 Therefore, the discrete spectrum is
\begin{equation}
	\mathcal{Z}=\left\{ \zeta_n\right\}_{n=1}^{6N_0}, \label{spectrals}
\end{equation}
with $\zeta_n\in S_1$ and $\bar{\zeta}_n\in S_6$. And the distribution  of $	\mathcal{Z}$ on the $z$-plane   is shown  in Figure \ref{fig:figure1}.
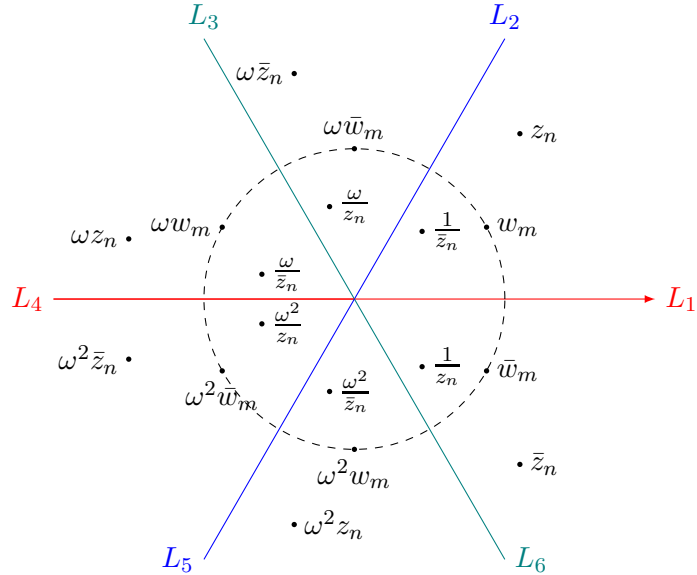
\begin{figure}[H]
	\centering
	\begin{tikzpicture}[node distance=2cm]
		\draw[red,-latex](-4,0)--(4,0)node[right]{$L_1$};
		\draw[blue](0,0)--(2,3.464)node[above]{$L_2$};
		\draw[teal](0,0)--(-2,3.464)node[above]{$L_3$};
		\draw[red](0,0)--(-4,0)node[left]{$L_4$};
		\draw[teal](0,0)--(2,-3.464)node[right]{$L_6$};
		\draw[blue](0,0)--(-2,-3.464)node[left]{$L_5$};
		\draw[dashed] (2,0) arc (0:360:2);
		\coordinate (A) at (2.2,2.2);
		\coordinate (B) at (2.2,-2.2);
		\coordinate (C) at (-0.8,3);
		\coordinate (D) at (-0.8,-3);
		\coordinate (E) at (0.9,0.9);
		\coordinate (F) at (0.9,-0.9);
		\coordinate (G) at (-3,0.8);
		\coordinate (H) at (-3,-0.8);
		\coordinate (J) at (1.7570508075688774,0.956);
		\coordinate (K) at (1.7570508075688774,-0.956);
		\coordinate (L) at (-1.7570508075688774,0.956);
		\coordinate (M) at (-1.7570508075688774,-0.956);
		\coordinate (a) at (0,2);
		\fill (a) circle (1pt) node[above] {$\omega \bar{w}_m$};
		\coordinate (s) at (0,-2);
		\fill (s) circle (1pt) node[below] {$\omega ^2w_m$};
		\coordinate (d) at (-0.33,1.23);
		\fill (d) circle (1pt) node[right] {$\frac{\omega}{z_n} $};
		\coordinate (f) at (-0.33,-1.23);
		\fill (f) circle (1pt) node[right] {$\frac{\omega^2}{\bar{z}_n}$};
		\coordinate (g) at (-1.23,0.33);
		\fill (g) circle (1pt) node[right] {$\frac{\omega}{\bar{z}_n} $};
		\coordinate (h) at (-1.23,-0.33);
		\fill (h) circle (1pt) node[right] {$\frac{\omega^2}{z_n}$};
		\fill (A) circle (1pt) node[right] {$z_n$};
		\fill (B) circle (1pt) node[right] {$\bar{z}_n$};
		\fill (C) circle (1pt) node[left] {$\omega \bar{z}_n$};
		\fill (D) circle (1pt) node[right] {$\omega^2 z_n$};
		\fill (E) circle (1pt) node[right] {$\frac{1}{\bar{z}_n}$};
		\fill (F) circle (1pt) node[right] {$\frac{1}{z_n}$};
		\fill (G) circle (1pt) node[left] {$\omega z_n$};
		\fill (H) circle (1pt) node[left] {$\omega^2\bar{z}_n$};
		\fill (J) circle (1pt) node[right] {$w_m$};
		\fill (K) circle (1pt) node[right] {$\bar{w}_m$};
		\fill (L) circle (1pt) node[left] {$\omega w_m$};
		\fill (M) circle (1pt) node[below] {$\omega^2\bar{w}_m$};
	\end{tikzpicture}
	\caption{\footnotesize Distribution of the discrete spectrum $\mathcal{Z}$ in the $z$-plane. }
	\label{fig:figure1}
\end{figure}
 Besides, the poles of $\mu(z)$ on $L_n$, $n=1,...,6$  may occur and they correspond to spectral singularities. Therefore, it is reasonable to give the following assumption:
As shown in  \cite{RHP},   denote   \textit{norming constant}    $c_n$.    Therefore, there are  residue conditions as
\begin{align}
	&\res_{z=\zeta_n}\mu(z) =\lim_{z\to \zeta_n}\mu(z) e^Q\left(\begin{array}{ccc}
		0&	-c_n & 0 \\
		0& 0 & 0\\
		0&	0 & 0
	\end{array}\right)e^{-Q},
\end{align}
for $n=1,...,2N_1+
N_2$ and
\begin{align}
	&\res_{z=\zeta_j}\mu(z)=\lim_{z\to \zeta_j}\mu(z) e^Q\left(\begin{array}{ccc}
		0&	0 & 0 \\
		0& 0 & -c_{j+2N_1+N_2}\\
		0&	0 & 0
	\end{array}\right)e^{-Q},\label{resrelation1}
\end{align}
for $j=1+2N_1+N_2,...,N_0$.
In addition,  the collection
$\sigma_d= \left\lbrace \zeta_n,C_n\right\rbrace^{6N_0}_{n=1}$
is called the \emph{scattering data} with $C_n=c_n$, $C_{n+N_0}=\omega\bar{c}_n$, $C_{n+2N_0}=\omega c_n,$ $C_{n+3N_0}=\omega^2\bar{c}_n$, $C_{n+4N_0}=\omega^2c_n$ and  $C_{n+5N_0}=\bar{c}_n$ for $n=1,...,N_0$.

 To deal with our following work,
we assume our initial data satisfies that $u_0\in\mathcal{S}(\mathbb{R})$ to generate  generic scattering data such  that $\mu(z)$  has no the poles
 on $L_n\setminus\{\varkappa_n\}$, $n=1,...,6$  and  at point  $z= e^{\frac{\pi i}{6}}$. And also  $r(z)$ belongs to $\mathcal{S}(\mathbb{R})$.
This statement will be given is section \ref{secr}.

\subsection{Set up of    RH problem}

\quad
We replace
the pair of variables  $(x, t)$ by $(y, t)$ with $y$ defined in (\ref{y}),
The price to pay for this is that the solution of the initial problem can be given only implicitly,
or parametrically. It will be given in terms of functions in the new scale, whereas the original scale will also be given in terms of functions in the new scale.
By the definition of the new scale $y(x, t)$, we denote the  phase function
\begin{equation}
	\theta_{jl}(z)=-i\left[ \xi (\lambda_j(z)-\lambda_l(z))+\left(\frac{1}{\lambda_j(z)}-\frac{1}{\lambda_l(z)} \right) \right] ,\label{theta}\quad  \xi=\frac{y}{t}.
\end{equation}
Especially,
\begin{equation}
	\theta_{12}(z)=\sqrt{3}\left(z-\frac{1}{z} \right) \left[ \xi -\frac{1}{z^2-1+z^{-2}} \right], \quad
\end{equation}
with $\theta_{23}(z)=\theta_{12}(\omega z),\text{ and }\theta_{31}(z)=\theta_{12}(\omega^2 z)$. Furthermore, denote $\theta_{12}(\zeta_n)\triangleq [\theta_{12}]_n$.
Let
\begin{equation}
	M(z;y,t)\triangleq \mu(z;x(y,t),t).
\end{equation}
which solves the following RH problem.
\begin{RHP}\label{RHP1}
	 Find a matrix-valued function $	M(z)\triangleq M(z;y,t)$ which satisfies:
	
	$\blacktriangleright$ Analyticity: $M(z)$ is meromorphic in $\mathbb{C}\setminus \Sigma $ and has single poles;
	
	$\blacktriangleright$ Symmetry: $M(z)=\Gamma_1\overline{M(\bar{z})}\Gamma_1=\Gamma_2\overline{M(\omega^2\bar{z})}\Gamma_2=\Gamma_3\overline{M(\omega\bar{z})}\Gamma_3$

\quad and $M(z)=\overline{M(\bar{z}^{-1})}$;
	
	$\blacktriangleright$ Jump condition: $M$ has continuous boundary values $M_\pm(z)$ on $L_n$
	\begin{equation}
		M_+(z)=M_-(z)V_n(z),\hspace{0.3cm} z \in L_n, \label{jumpv}
	\end{equation}
	where
	\begin{align}
		&V_1(z)=\left(\begin{array}{ccc}
			1&	-r(z)e^{it\theta_{12}} & 0 \\
			\bar{r}(z)e^{-it\theta_{12}} & 1-|r(z)|^2 & 0\\
			0&	0 & 1
		\end{array}\right),\\	
	&V_2(z)=\left(\begin{array}{ccc}
		1&	0 & 0 \\
		0 & 1 & -r(\omega z)e^{it\theta_{23}}\\
		0&	\bar{r}(\omega z)e^{-it\theta_{23}} & 1-|r(\omega z)|^2
	\end{array}\right),\\
	&V_3(z)=\left(\begin{array}{ccc}
	1-|r(\omega^2 z)|^2&	0 & \bar{r}(\omega^2 z)e^{it\theta_{13}} \\
	0 & 1 & 0\\
	-r(\omega^2 z)e^{-it\theta_{13}}	&	0 & 1
	\end{array}\right),\\
	&V_4(z)=\left(\begin{array}{ccc}
		1&	-r(z)e^{it\theta_{12}} & 0 \\
		\bar{r}(z)e^{-it\theta_{12}} & 1-|r(z)|^2 & 0\\
		0&	0 & 1
	\end{array}\right),\\	
&V_5(z)=\left(\begin{array}{ccc}
	1&	0 & 0 \\
	0 & 1 & -r(\omega z)e^{it\theta_{23}}\\
	0&	\bar{r}(\omega z)e^{-it\theta_{23}} & 1-|r(\omega z)|^2
\end{array}\right),\\
&V_6(z)=\left(\begin{array}{ccc}
	1-|r(\omega^2 z)|^2&	0 & \bar{r}(\omega^2 z)e^{it\theta_{13}} \\
	0 & 1 & 0\\
	-r(\omega^2 z)e^{-it\theta_{13}}	&	0 & 1
\end{array}\right)
	\end{align}
	
	$\blacktriangleright$ Asymptotic behaviors:
	\begin{align}
		&M(z) = I+\mathcal{O}(z^{-1}),\hspace{0.5cm}z \rightarrow \infty;
	\end{align}
	
	$\blacktriangleright$ Singularities: As   $z\to\varkappa_n =e^{\frac{i\pi(n-1)}{3}}$, $n = 1,...,6$,
 the  limit  of $M(z)$   has pole   singularities
	\begin{align}
		&M(z)=\frac{1}{z\mp1}\left(\begin{array}{ccc}
			\alpha_\pm &	\alpha_\pm & \beta_\pm \\
			-\alpha_\pm & -\alpha_\pm & -\beta_\pm\\
			0	&	0 & 0
		\end{array}\right)+\mathcal{O}(1),\ z\to\pm 1,\label{asyM1}\\
	&M(z)=\frac{1}{z\mp\omega^2}\left(\begin{array}{ccc}
		0 &	0 &  0\\
	\beta_\pm	 & \alpha_\pm &\alpha_\pm \\
		-\beta_\pm	&	-\alpha_\pm & -\alpha_\pm
	\end{array}\right)+\mathcal{O}(1),\ z\to\pm \omega^2,\\
	&M(z)=\frac{1}{z\mp\omega}\left(\begin{array}{ccc}
	-\alpha_\pm &	-\beta_\pm & -\alpha_\pm\\
	0	 & 0 &0 \\
	\alpha_\pm &	\beta_\pm & \alpha_\pm
	\end{array}\right)+\mathcal{O}(1),\ z\to\pm \omega\label{asymo},
	\end{align}
	with $\alpha_\pm=\alpha_\pm(y,t)=-\bar{\alpha}_\pm$, $\beta_\pm=\beta_\pm(y,t)=-\bar{\beta}_\pm$ and $M^{-1}$ has same specific
	matrix structure  with $\alpha_\pm$, $\beta_\pm$ replaced by $\tilde{\alpha}_\pm$, $\tilde{\beta}_\pm$. Moreover, $\left( \alpha_\pm,\ \beta_\pm\right) \neq0$ iff $\left( \tilde{\alpha}_\pm,\ \tilde{\beta}_\pm\right) \neq0$;
	
	$\blacktriangleright$ Residue conditions: $M(z)$ has simple poles at each point in $ \mathcal{Z}\cup \bar{\mathcal{Z}}$ with:\\
	for $n=1,...,2N_1+N_2$
	\begin{align}
		&\res_{z=\zeta_n}M=\lim_{z\to \zeta_n}M \left(\begin{array}{ccc}
			0&	-c_ne^{it[\theta_{12}]_n} & 0 \\
			0& 0 & 0\\
			0&	0 & 0
		\end{array}\right)\triangleq\lim_{z\to \zeta_n}M B_n,\label{RES1}\\
		&\res_{z=\omega\bar{\zeta}_n}M=\lim_{z\to \omega\bar{\zeta}_n}M \left(\begin{array}{ccc}
			0&	0 & -\omega\bar{c}_ne^{it\theta_{13}(\omega\bar{\zeta}_n)} \\
			0& 0 & 0\\
			0&	0 & 0
		\end{array}\right)\triangleq\lim_{z\to \zeta_n}M B_{n+N_0},\\
		&\res_{z=\omega\zeta_n}M=\lim_{z\to \omega\zeta_n}M \left(\begin{array}{ccc}
		0&	0 & 0 \\
		0& 0 & 0\\
		-\omega c_ne^{-it\theta_{13}(\omega\zeta_n)}&	0 & 0
		\end{array}\right)\triangleq\lim_{z\to \zeta_n}M B_{n+2N_0},\\
	&\res_{z=\omega^2\bar{\zeta}_n}M=\lim_{z\to \omega^2\bar{\zeta}_n}M \left(\begin{array}{ccc}
		0&	0 & 0 \\
		0& 0 & 0\\
		0& -\omega^2\bar{c}_ne^{-it\theta_{23}(\omega^2\bar{\zeta}_n)} & 0
	\end{array}\right)\triangleq\lim_{z\to \zeta_n}M B_{n+3N_0},\\
	&\res_{z=\omega^2\zeta_n}M=\lim_{z\to \omega^2\zeta_n}M \left(\begin{array}{ccc}
		0&	0 & 0 \\
		0& 0 & -\omega^2 c_ne^{it\theta_{23}(\omega^2\zeta_n)}\\
		0&	0 & 0
	\end{array}\right)\triangleq\lim_{z\to \zeta_n}MB_{n+4N_0},\\
	&\res_{z=\bar{\zeta}_n}M=\lim_{z\to \bar{\zeta}_n}M \left(\begin{array}{ccc}
	0&	0 & 0 \\
	-\bar{c}_ne^{-it[\theta_{12}]_n}& 0 & 0\\
	0&	0 & 0
\end{array}\right)\triangleq\lim_{z\to \zeta_n}M B_{n+5N_0},
\end{align}
and for $j=1+2N_1+N_2,...,2N_1^A+N_2^A+2N_1+N_2$,
\begin{align}
	&\res_{z=\omega\zeta_{j}}M=\lim_{z\to \omega\zeta_{j}}M \left(\begin{array}{ccc}
		0&	-\omega c_{j}e^{it\theta_{12}(\omega\zeta_{j})} & 0 \\
		0& 0 & 0\\
		0&	0 & 0
	\end{array}\right)\triangleq\lim_{z\to \zeta_n}M B_{n+2N_0},\label{RES2}\\
	&\res_{z=\bar{\zeta}_{j}}M=\lim_{z\to \bar{\zeta}_{j}}M \left(\begin{array}{ccc}
		0&	0 & -\bar{c}_{j}e^{it\theta_{13}(\bar{\zeta}_{j})} \\
		0& 0 & 0\\
		0&	0 & 0
	\end{array}\right)\triangleq\lim_{z\to \zeta_n}M B_{n+5N_0},\\
	&\res_{z=\omega^2\zeta_{j}}M=\lim_{z\to \omega^2\zeta_{j}}M \left(\begin{array}{ccc}
		0&	0 & 0 \\
		0& 0 & 0\\
		- \omega^2c_je^{-it\theta_{13}(\omega^2\zeta_{j})}&	0 & 0
	\end{array}\right)\triangleq\lim_{z\to \zeta_n}M B_{n+4N_0} ,\\
	&\res_{z=\omega\bar{\zeta}_j}M=\lim_{z\to \omega\bar{\zeta}_j}M \left(\begin{array}{ccc}
		0&	0 & 0 \\
		0& 0 & 0\\
		0& -\omega\bar{c}_je^{-it\theta_{23}(\omega\bar{\zeta}_j)} & 0
	\end{array}\right)\triangleq\lim_{z\to \zeta_n}M B_{n+N_0},\\
	&\res_{z=\zeta_j}M=\lim_{z\to \zeta_j}M \left(\begin{array}{ccc}
		0&	0 &  0\\
		0& 0 & - c_je^{it\theta_{23}(\zeta_j)}\\
		0&	0 & 0
	\end{array}\right)\triangleq\lim_{z\to \zeta_n}M B_{n},\\
	&\res_{z=\omega^2\bar{\zeta}_j}M=\lim_{z\to \omega^2\bar{\zeta}_j}M \left(\begin{array}{ccc}
		0&	0 & 0 \\
		-\omega^2\bar{c}_je^{-it\theta_{12}(\omega^2\bar{\zeta}_j)}& 0 & 0\\
		0&	0 & 0
	\end{array}\right)\triangleq\lim_{z\to \zeta_n}M B_{n+3N_0}\label{RES3}.
\end{align}
\end{RHP}
Denote $M(z;y,t)= (M_{jl}(z;y,t))_{jl=1}^3$,
then   solution of  Novikov equation (\ref{Novikov}) can be  obtained by the following reconstruction formula
\begin{align}
	u(x,t)=u(y(x,t),t)=&\frac{1}{2}\tilde{m}_1(y,t)\left(\frac{M_{33}(e^{\frac{i\pi}{6}};y,t)}{M_{11}(e^{\frac{i\pi}{6}};y,t)} \right)^{1/2}\nonumber\\
	&+ \frac{1}{2}\tilde{m}_3(y,t)\left(\frac{M_{33}(e^{\frac{i\pi}{6}};y,t)}{M_{11}(e^{\frac{i\pi}{6}};y,t)} \right)^{-1/2}-1 ,\label{recons u}
\end{align}
where
\begin{equation}
	x(y,t)=y+c_+(x,t)=y+\frac{1}{2} \ln\frac{M_{33}(e^{\frac{i\pi}{6}};y,t)}{M_{11}(e^{\frac{i\pi}{6}};y,t)} ,\label{recons x}
\end{equation}
and
\begin{align}
	\tilde{m}_l\triangleq\sum_{j=1}^3M_{jl}(e^{\frac{i\pi}{6}};y,t),\ l=1,2,3. \nonumber
\end{align}

\subsection{Scattering maps}\label{secr}
\quad In this section, our purpose is to demonstrate the following  proposition about reflection coefficient $r(z)$.
\begin{Proposition}\label{pror}
	If the initial data $u_0 \in  \mathcal{S}(\mathbb{R})$, then $r(z)$ belongs to $\mathcal{S}(\mathbb{R})$. In addition, there exist a fixed positive constant $C_{1,r}$ and a fixed  constant $C_{2,r}$ satisfying that if $u_0- u_{0,xx}> C_{2,r}>-1$, $C_{1,r}>\parallel u_0- u_{0,xx}\parallel_{L^1}$ and $C_{1,r}>\parallel u_0\parallel_{W^{j,3}}$, $j=1,\infty$, then $|r(z)|<1$ for $z\in\mathbb{R}$.
\end{Proposition}
The fact that $r(z)$ belongs to $\mathcal{S}(\mathbb{R})$ is obviously. In this Section, we will give that $|r(\pm1)|=0$ and $r(z)$ is continuous  near $z=\pm 1$. So, it is reasonable to assume that there a fixed small sufficiently positive constant $\epsilon$ satisfying:\\
1.  $(z-1)\mu_\pm(z)\neq 0 $ for $z$ in $(1-\epsilon,1+\epsilon)$;\\
2. for $z\in(1-\epsilon,1+\epsilon)$, $|r(z)|<\frac{1}{2}$.

Since   $\mu_\pm$ are two fundamental matrix solutions of the  Lax  pair (\ref{laxmu}),  there exists a linear  relation between $\mu_+$ and $\mu_-$, namely
\begin{equation}
	\mu_+(z;x,t)=\mu_-(z;x,t)S(z),\label{scattering}
\end{equation}
where $S(z)$ is called scattering matrix  and   only depended on $z$.
Take $x\to-\infty$, then we obtain the integral representation  of $S(z)$ as:
\begin{equation}
	S(z)=I-\int_{-\infty}^{+ \infty}e^{-\hat{\Lambda}(z)\int_{-\infty}^{x}(q^2(v,t)-1)dv}[U\mu_+(x,t;z)]ds\label{sz}.
\end{equation}
The entries of $S(z)$ are defined and continuous for $z$ in
\begin{align*}
	\left(\begin{array}{ccc}
		S_1\cup S_2&	\mathbb{R}^+ & \omega\mathbb{R}^+ \\
		\mathbb{R}^+& S_3\cup S_4 & \omega^2\mathbb{R}^+\\
		\omega\mathbb{R}^+&	\omega^2\mathbb{R}^+ & S_5\cup S_6
	\end{array}\right)\setminus\{\pm 1,\pm\omega,\pm\omega^2\}.
\end{align*}
Further  by  symmetries, $S(z)$ can be written as
\begin{align}
	S(z)=\left(\begin{array}{ccc}
		s_{11}(z) &	s_{12}(z) & \overline{s_{12}(\omega\bar{z})} \\
		\overline{s_{12}(\bar{z})} & \overline{s_{11}(\bar{z})} & s_{12}(\omega z)\\
		s_{12}(\omega^2 z)&	\overline{s_{12}(\omega^2\bar{z})} & \overline{s_{11}(\omega^2\bar{z})}
	\end{array}\right),\ s_{11}(z)=\overline{s_{11}(\omega\bar{z})}.
\end{align}
Rewrite (\ref{scattering}) as
$$S(z)=\mu_-(z;x,t)^{-1}\mu_+(z;x,t).$$
 Note that $\mu_\pm(z;x,t)$ has same asymptotic behavior  (\ref{asyM1})-(\ref{asymo}), we see that
  $\pm1$, $\pm\omega$, $\pm\omega^2$ are not the singularity of $S(z)$ any more.  Then when the initial data $u_0 \in  \mathcal{S}(\mathbb{R})$,  $r(z)$ belongs to $\mathcal{S}(\mathbb{R})$ too. In addition,  $S(\varkappa_j)\sim I$, for $j=1,...,6$. It gives that $r(\pm1)=0$.
Furthermore, $$\det S(z)=1,\ \ \ S(z)\sim I\  as\ z\to\infty.$$
And the reflection coefficient can be represented as $r(z)=\frac{s_{11}(z)}{s_{12}(z)}$ for $z\in\mathbb{R}^+$.

$S^A(z)=[s_{jl}^A(z)]_{j,l=1}^3=(S(z)^{-1})^T$ is the cofactor matrix of $S(z)$, with integral representation
\begin{equation}
	S^A(z)=I+\int_{-\infty}^{+ \infty}e^{\int_{-\hat{\Lambda}(z)\infty}^{x}q^2(v,t)-1dv}[U^T\mu_+^A(x,t;z)]ds\label{sAz}.
\end{equation}
The entries of it are defined and continuous for $z$ in
\begin{align*}
	\left(\begin{array}{ccc}
		S_4\cup S_5&	\mathbb{R}^- & \omega\mathbb{R}^- \\
		\mathbb{R}^-& S_1\cup S_6 & \omega^2\mathbb{R}^-\\
		\omega\mathbb{R}^-&	\omega^2\mathbb{R}^- & S_2\cup S_3
	\end{array}\right)\setminus\{\pm 1,\pm\omega,\pm\omega^2\}.
\end{align*}
It has same symmetry as $S(z)$ and $s_{11}^A(z)=\overline{s_{11}^A(\omega\bar{z})}$. And  it also does not have singularity at $\pm 1,\ \pm\omega,\ \pm\omega^2$.
Similarly, for $z\in\mathbb{R} $, the reflection coefficient can be represented as $r(z)=\frac{s_{11}^A(z)}{s_{12}^A(z)}$.

Now we begin to demonstrate Proposition \ref{pror}.  We only consider the  $x$-part of Lax pair to give the proof of proposition \ref{pror} in this section. In fact,  taking account of $t$-part of Lax pair and through the standard direct scattering transform, it can be deduced that  $r(z)$ have linear time evolution:  $r(z,t)=e^{\frac{i}{4\lambda^2(z)}(z-\frac{1}{z})}r(z,0)$. Recall that $r(z)$ has different representation on $\mathbb{R}^+$ and $\mathbb{R}^-$. In this section, we only consider the $L^2$-integrability and the boundedness of $r(z)$ for $z\in\mathbb{R}^+$. The case of $z\in\mathbb{R}^-$ can be obtain through same analysis.
Because of  the singularity on $1$ of the spectral problem (\ref{laxmu}), we divide our approach into three cases: \\
$$(0,1-\epsilon),\ (1-\epsilon,1+\epsilon),\ (1+\epsilon,+\infty).$$

Remark: If $(s_0,s_1)=0$, the analysis will be more simple. The property near $z=1$ do not need to be considered separately.

For $z>0$, $z$ away from $1$, recall $r(z)=\frac{s_{12}(z)}{s_{11}(z)}$. So,
the fact that the $L^{\infty}$-norm of $ r(z)$ for $z$ in $(1+\epsilon,+\infty)\cup(0,1-\epsilon)$ is controlled by $u_0$  is equivalent to the following two proposition:
the  maximal value, the  minimal value  of $s_{11}(z)$  and the  maximal value $s_{12}(z)$  for $z$ in $(1+\epsilon,+\infty)\cup(0,1-\epsilon)$ is controlled by $u_0$
 From  (\ref{sz}), we give the  integral equations of $s_{12}(z)$ and $s_{11}(z)$:
\begin{align}
	&s_{11}(z)=1-\int_{\mathbb{R}}\sum_{ k=1 }^3U_{1k}(y,z)[\mu_+]_{k1}(y,z)dy,\label{s11e}\\
	&s_{12}(z)=-\int_{\mathbb{R}}e^{i(z-1/z)p(y)}\sum_{ k=1 }^3U_{1k}(y,z)[\mu_+]_{k2}(y,z)dy\label{s12e}.
\end{align}
Here, $U_{ij}(y,z)$ is the entire of (\ref{laxU}), $p(y)=y-\int_{y}^{+\infty}\left( q^2(s,0)-1\right) ds$. It has special structure as $U_{ij}(y,z)=h_{ij}(z)c_2(y)+g_{ij}(z)c_1(y)$, with $c_1(y)=\frac{q_x}{q}(y,0)$, $c_2(y)=q^{-2}(y,0)-q^2(y,0)$ and
\begin{align}
	h_{ij}(z)=\frac{\lambda_j(z)}{3\lambda_i(z)^2-1},\ g_{ij}(z)=\frac{\lambda_j(z)\lambda_i(z)-\lambda_j(z)^2}{3\lambda_i(z)^2-1}.
\end{align}
\begin{lemma}\label{lemmac12}
	For $u_0\in\mathcal{S}(\mathbb{R})$, we have that
	\begin{align}
		\parallel c_1(y)\parallel_1\leq&(1+C_{r,2})^{-1}\parallel u_0\parallel_{W^{1,3}},\\
		\parallel c_2(y)\parallel_1\leq&(1+C_{r,2})^{-2/3}\left( 1+(1+2\parallel u_0\parallel_{W^{\infty,3}})^{1/3}\right) \nonumber\\
		&\left( 1+(1+2\parallel u_0\parallel_{W^{\infty,3}})^{2/3}\right)\parallel u_0-[u_0]_{xx}\parallel_1 .
	\end{align}

\end{lemma}
\begin{proof}
	Denote $q(y,0)\triangleq q_0(y)$. When $u_0\in\mathcal{S}(\mathbb{R})$ with $u_0-[u_0]_{xx}>C_{r,2}$,
	\begin{align*}
		&(1+C_{r,2})^{1/3}\leq|q_0(y)|\leq1+2\parallel u_0\parallel_{W^{\infty,3}},\
	\parallel	q_0'(y)\parallel_1\leq (1+C_{r,2})^{-2/3}\parallel u_0\parallel_{W^{1,3}},\\
	&\parallel1-q_0^4\parallel_1\leq \left( 1+(1+2\parallel u_0\parallel_{W^{\infty,3}})^{1/3}\right) \left( 1+(1+2\parallel u_0\parallel_{W^{\infty,3}})^{2/3}\right)\parallel u_0-[u_0]_{xx}\parallel_1 .
	\end{align*}
Here we use the fact:
\begin{align*}
	|1-q_0|=\frac{|u_0-[u_0]_{xx}|}{1+q_0+q_0^2}\leq|u_0-[u_0]_{xx}|.
\end{align*}
Thus it is easily to obtain the result.
\end{proof}
The following property comes soon by simple calculation:
\begin{lemma}\label{lemmahg}
	For $i,j=1,2,3$, $h_{ij}(z)$ and $g_{ij}(z)$ belong in $L^p(\mathbb{R}\setminus\{z;\ |z-1|<\epsilon\})$, $\infty\geq p>1$. $h_{ij}(z),\ g_{ij}(z)\sim \frac{1}{z-1}$ when $i\neq3$ as $z\to 1$. $h_{ij}(z)\sim \frac{1}{z}$, $g_{ij}(z)\sim 1$  as $z\to \infty$,  and   $h_{ij}(z)\sim z$, $g_{ij}(z)\sim 1$  as $z\to 0$.
\end{lemma}
We denote the positive constant
$C(h,g,\epsilon)$ relying on $h_{ij}(z),\ g_{ij}(z)$ and $\epsilon$ with
$$C(h,g,\epsilon)=\underset{i,j=1,2,3}{\max}\left\lbrace \parallel h_{ij}\parallel_{L^\infty\left( (1+\epsilon,+\infty)\cup(0,1-\epsilon)\right) },\  \parallel g_{ij}\parallel_{L^\infty\left( (1+\epsilon,+\infty)\cup(0,1-\epsilon)\right) } \right\rbrace .$$
Then by (\ref{s11e})-(\ref{s12e}), it requires  us to research the property of $\mu_+$ in following   subsections.

However  unlike the focusing  NLS  equation   \cite{fNLS},   the   Novikov equation admits  a  $3\times3$ matrix  Lax pair,
 $1-|r(z)|^2\geq 0$ does not always  hold.
  In fact, from the proof in following Subsection, we can   ensure $|r(z)| < 1$ for  $z$ not near $1$, because its $\infty$-norm can be controlled by initial value.
And as we assuming previous, $s_{11}^A(\omega^2z)\neq0$ on $l_1=\mathbb{R}^+\setminus\{1\}$, which implies $s_{11}^A(\omega^2z)=|s_{11}(z)|^2(1-|r(z)|^2)\neq0$ on $\mathbb{R}^+\setminus\{1\}$. So  $1-|r(z)|^2\neq0$  on $\mathbb{R}^+$ except $z=1$, and it does note change sign by its continuity. Thus,  $|r(z)| <1$  on $\mathbb{R}^+$. Similarly, $|r(z)| <1$  on $\mathbb{R}^-$.

First, our goal is to prove that $L^{\infty}$-norm of $ \mu^+(z)$ for $z$ in $(1+\epsilon,+\infty)\cup(0,1-\epsilon)$ is controlled by $u_0$. Here, We give the details of  $z\in(1+\epsilon,+\infty)$. The case $z\in(0,1-\epsilon)$ can be demonstrated following the same way.  Recall the integral equation   (\ref{intmu}) and set $t=0$,
 we only consider its first column $  \mu^{+}_1(z) $, the other columns can be obtained analogously.

Introduce the integral operator $T_+$:
\begin{align}\label{T1}
	T_+(f)(x,z)=\int_{x}^{+\infty}K_+(x,y,z)f(y,z)dy,  
\end{align}
where integral kernel $K_+(x,y,z)=\left( [K_+(x,y,z)]_{ij}\right)_{3\times3} $ is a $3\times3$ matrix with entire
\begin{align}\label{K}
	 [K_+(x,y,z)]_{ij}=e^{(\lambda_i(z)-\lambda_1(z))(p(x)-p(y))}U_{ij}(y,z).
\end{align}
Here, $U_{ij}(y,z)$ is the entire of (\ref{laxU}) with $t=0$. And we denote  $e_1=(1,0,0)^T$,
then (\ref{laxmu}) trans to
\begin{align}\label{eqn}
	&[\mu_+]_1=e_1+T_+([\mu_+]_1).	
\end{align}

\begin{lemma}\label{lemma4}
	$T_+$ is a integral operator defined above, then  it is a  bounded operator on $L^{\infty}\left( \mathbb{R}\times (1+\epsilon,+\infty)\cup(0,1-\epsilon)\right) $. And its norm is under control by $u_0$.
\end{lemma}
\begin{proof}
	For any function $f(x,z)=(f_1,f_2,f_3 )^T$ in $L^{\infty}\left( \mathbb{R}\times (1+\epsilon,+\infty)\cup(0,1-\epsilon)\right) $,
	by the definition (\ref{T1}), we have 
	\begin{align}
		T_+(f)(x,z)=&\int_{x}^{+\infty}
		\left(\begin{array}{cc}
			\sum_{j=1}^3 U_{1j}f_j\\
			\left( \sum_{j=1}^3 U_{2j}f_j\right)e^{(\lambda_2(z)-\lambda_1(z))(p(x)-p(y))} \\
			\left( \sum_{j=1}^3 U_{3j}f_j\right)e^{(\lambda_3(z)-\lambda_1(z))(p(x)-p(y))}
		\end{array}\right)
		dy.\label{Te1}
	\end{align}
It  immediately derive that  
\begin{align}
	|T_+(f)(x,z)|\lesssim& C(h,g,\epsilon)\left( \int_{x}^{+\infty}|c_2(y)|dy+\int_{x}^{+\infty}|c_1(y)|dy\right) \parallel f\parallel_\infty.
\end{align}
So $T_+\in\mathcal{B}\left( L^{\infty}\left( \mathbb{R}\times (1+\epsilon,+\infty)\cup(0,1-\epsilon)\right) \right) ,$ with $$\parallel T_+ \parallel\leq C(h,g,\epsilon)\left( \int_{x}^{+\infty}|c_2(y)|dy+\int_{x}^{+\infty}|c_1(y)|dy\right).$$
\end{proof}

Consider the Volterra operator $T_+^n$ with
\begin{align}
	T_+^n(f)(x,z)=T_+\left( T_+^{n-1}(f)\right) (x,z)
\end{align}
Denote $K^n_+$ is the integral kernel of   $[T_+]^n$ as
\begin{align}
	K^n_+(x,y,z)=\int_{x}^{y}\int_{x}^{y_{n-1}}...\int_{x}^{y_2}K_+(x,y_1,z)K_+(y_1,y_2,z)...K_+(y_{n-1},y,z)dy_1...dy_{n-1},\nonumber
\end{align}
with
\begin{align}
	|K^n_+(x,y,z)|\leq& \frac{C(h,g,\epsilon)^n}{(n-1)!}\left( |c_1(y)|+|c_2(y)|\right)\nonumber\\
	&\left[ \left( \int_{x}^{+\infty}|c_2(y)|dy\right)^{n-1} +\left( \int_{x}^{+\infty}|c_1(y)|dy\right)^{n-1}\right]  .
\end{align}
Then  the standard Volterra theory gives the of the operator $(I-T_+)^{-1}=\sum_{n=1}^\infty T^n$ with following operator norm:
\begin{align}
	&\parallel (I-T_+)^{-1} \parallel\leq e^{C(h,g,\epsilon)\left(\parallel c_1(y) \parallel_{L^1}+\parallel c_2(y) \parallel_{L^1}\right) }.
\end{align}

The equations (\ref{eqn})  are  solvable with:
\begin{align}
	& \mu^+_1(x,z)=(I-T_+)^{-1}(T_+(e_1))(x,z).
\end{align}
It  leads to the  boundedness of $[\mu_+]_1(x,z)$ on $\mathbb{R}\times\big(  (1+\epsilon,+\infty)\cup(0,1-\epsilon)\big) $ with
\begin{align}
	\parallel \mu^+_1(x,z)\parallel_{L^\infty}\leq& e^{C(h,g,\epsilon)\left(\parallel c_1(y) \parallel_{L^1}+\parallel c_2(y) \parallel_{L^1}\right) }C(h,g,\epsilon)\left( \parallel c_1(y) \parallel_{L^1}+\parallel c_2(y) \parallel_{L^1}\right). \label{mu}
\end{align}
Then we begin to estimate $s_{12}(z)$ and $s_{11}(z)$. (\ref{s11e}) and (\ref{s12e}) give that for $z\in(1+\epsilon,+\infty)\cup(0,1-\epsilon)$,
\begin{align}
	&|s_{11}(z)-1|\leq \sum_{ k=1 }^3\parallel U_{1k}\parallel_1 \parallel[\mu_+]_{k1}\parallel_\infty,\\
	&|s_{12}(z)|\leq \sum_{ k=1 }^3\parallel U_{1k}\parallel_1 \parallel[\mu_+]_{k2}\parallel_\infty.
\end{align}
So the $\infty$-norm of $|s_{11}(z)-1|$ and $|s_{12}(z)|$ are controlled by the initial value. And when $\sum_{ k=1 }^3\parallel U_{1k}\parallel_1 \parallel[\mu_+]_{k1}\parallel_\infty$ small enough,
we deduce that $2>s_{11}(z)>C(u_0)>0$ with a positive constant $C(u_0)$ depending on $u_0$. Thus,
\begin{align}
	|r(z)|=&|\frac{b(z)}{a(z)}|<\frac{\sum_{ k=1 }^3\parallel U_{1k}\parallel_1 \parallel[\mu_+]_{k2}\parallel_\infty}{C(u_0)}.
\end{align}
Together with (\ref{mu}), Lemma \ref{lemmac12} and Lemma \ref{lemmahg}, $|r(z)|$ is  proportional to $\parallel c_1 \parallel_1$ and $\parallel c_2 \parallel_1$, namely $\parallel u_0-[u_0]_{xx}\parallel_1$, $\parallel u_0\parallel_{W^{\infty,3}}$ and $\parallel u_0\parallel_{W^{1,3}}$. Hence, there exists a positive constant $C_{r,1}$ to control above three norm of $u_0$ and admit that $|r(z)|<1$ for $z\in\mathbb{R}^+\setminus(1-\epsilon,1+\epsilon).$

Thus when initial datum $u_0\in\mathcal{S}(\mathbb{R})$, $s_{11}(z)$ and $s_{12}(z)$ are  continuous on $\mathbb{R}^+\setminus(1-\epsilon,1+\epsilon)$. In fact, take $\epsilon\to 0$, then we obtain that  $s_{11}(z)$ and $s_{12}(z)$ are  continuous on $\mathbb{R}^+\setminus\{1\}$. Moreover, as we assuming  preceding, $s_{11}(z)\neq0$ on $\mathbb{R}^+$. So $r(z)$ is continuous on $\mathbb{R}^+\setminus\{1\}$. Then the asymptotic property: $r(z)\to0$ as $z\to1$ in $\mathbb{R}^+$ gives that $z=1$ is a removable singular point of $r(z)$.

\section{Normalization of the RH problem}\label{sec3}

\subsection{Stationary phase points and decay domains  }
The long-time asymptotic  of the  RHP \ref{RHP1}  is affected by the growth and decay of the exponential function $e^{\pm2it\theta_{jl}}$  with
\begin{align}
	&\theta_{12}(z)=\sqrt{3}\left(z-\frac{1}{z} \right) \left[ \xi -\frac{1}{z^2-1+z^{-2}} \right],\\
	&\theta_{13}(z)=-\theta_{12}(\omega^2z),\ \theta_{23}(z)=\theta_{12}(\omega z).\label{42}
\end{align}
which is appearing in both the jump relation and the residue conditions. So we need control the real part of the phase functions  $\pm2it\theta_{jl}$.
 In fact, from (\ref{42}), we only need consider  the phase points and signature tables of    $ \theta_{12}$. Therefore, in this section, we introduce  a new transform  $M(z)\to M^{(1)}(z)$,  which  make that the  $M^{(1)}(z)$ is well behaved as $t\to \infty$ along any characteristic line.
  For the oscillatory term $e^{2it\theta_{12}}$ as $t\to \infty$, we consider the real stationary points of $\theta_{12}(z)$ on $ \mathbb{R}=L_1\cup L_4$. Let
  $$\xi=\frac{y}{t}, \ \ \breve{k}=z-\frac{1}{z},$$
   then
\begin{align}
	\frac{\partial\theta_{12}}{\partial z}= \frac{\partial\theta_{12}}{\partial \breve{k}}\frac{\partial\breve{k}}{\partial z}=\sqrt{3}\left(\xi-\frac{1-\breve{k}^2}{(\breve{k}^2+1)^2} \right)\left(1+\frac{1}{z^2} \right).
\end{align}
Therefore the  stationary points satisfy the equation
\begin{align}
\xi(\breve{k}^2)^2+(2\xi+1)\breve{k}^2+\xi-1=0,
\end{align}
from which,  we found that
\begin{itemize}
\item[Case I:]  \  For the domains  $\xi<- {1}/{8}$ and $\xi>1$, there is  no  stationary point on $\mathbb{R}$;

\item[Case II:]\ For the domain   $- {1}/{8}<\xi<0$,   on $\mathbb{R}$ there are 8 stationary points,  which reorder  as
$$\xi_1>\xi_2>\xi_3>\xi_4>\xi_5>\xi_6>\xi_7>\xi_8,$$
 with $\xi_1=1/\xi_4=-1/\xi_5=-\xi_8$ and $\xi_2=1/\xi_3=-1/\xi_6=-\xi_7$;

\item[Case III:]\  For the domain  $0\leq\xi<1$,  on $\mathbb{R}$  there are 4 stationary points,  which reorder  as
$$\xi_1>\xi_2>\xi_3>\xi_4,$$
 with $\xi_1=1/\xi_2=-1/\xi_3=-\xi_4$;
\end{itemize}

Next we consider  the signature table  of $\text{Im}\theta_{12}$, which is determined by
the real part of $it\theta_{12}$
\begin{align}
&\text{Re}(it\theta_{12})=-\sqrt{3}t\text{Im}\theta_{12}=-t\text{Im}z\left(1+|z|^{-2}  \right)\xi+\nonumber\\
& \dfrac{\sqrt{3}t\text{Im}z\left(1+|z|^{-2}\right) \left( -|z|^6-|z|^4+4\text{Re}^2z|z|^2-|z|^2\right) }{|z|^8+1+2[(\text{Re}^2z-\text{Im}^2z)^2-4\text{Re}^2z\text{Im}^2z]-2(1+|z|^4)(\text{Re}^2z-\text{Im}^2z)+|z|^4 }  .\label{Reitheta}
\end{align}
For the above case I, the signature table  of $\text{Im}\theta_{12}(z)$ are shown by (a) and (d)  in Figure \ref{figtheta};
For  the above  Case II and Case III,  the signature tables  of $\text{Im}\theta_{12}(z)$
 are shown by (b) and (c) in Figure \ref{figtheta}, respectively.

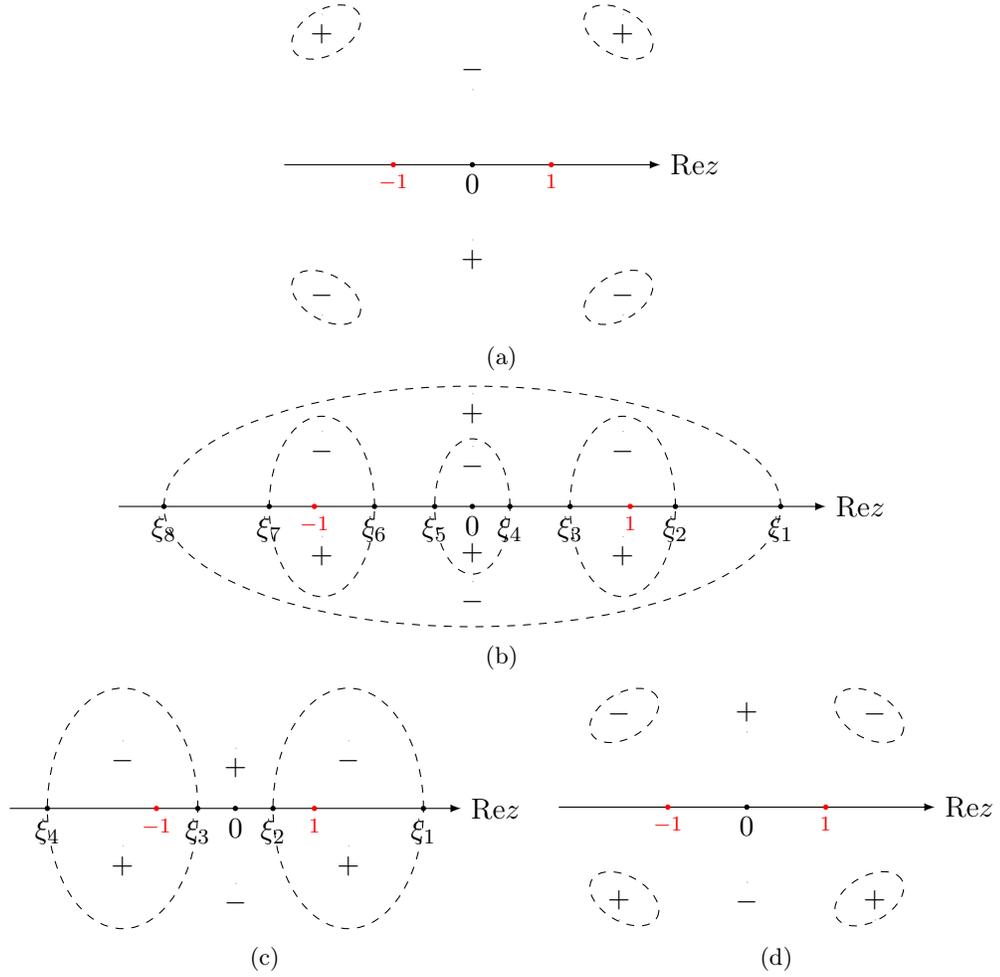
\begin{figure}[htp]
	\centering
		\subfigure[]{
			\begin{tikzpicture}
				\draw[-latex](-2.5,0)--(2.5,0)node[right]{ Re$z$};
				\draw [rotate=-30,dashed](0.8,2.5) ellipse (0.5 and 0.3);
				\draw [rotate=30,dashed](0.8,-2.5) ellipse (0.5 and 0.3);
				\draw [rotate=30,dashed](-0.8,2.5) ellipse (0.5 and 0.3);
				\draw [rotate=-30,dashed](-0.8,-2.5) ellipse (0.5 and 0.3);
				\coordinate (I) at (0,0);
				\fill (I) circle (1pt) node[below] {$0$};
				\coordinate (a) at (0,-1);
				\fill (a) circle (0pt) node[below] {$+$};
				\coordinate (aa) at (2,2);
				\fill (aa) circle (0pt) node[below] {$+$};
				\coordinate (as) at (-2,2);
				\fill (as) circle (0pt) node[below] {$+$};
				\coordinate (f) at (2,-2);
				\fill (f) circle (0pt) node[above] {$-$};
				\coordinate (ff) at (-2,-2);
				\fill (ff) circle (0pt) node[above] {$-$};
				\coordinate (s) at (0,1);
				\fill (s) circle (0pt) node[above] {$-$};
				\coordinate (c) at (-1.05,0);
				\fill[red] (c) circle (1pt) node[below] {\scriptsize$-1$};
				\coordinate (d) at (1.05,0);
				\fill[red] (d) circle (1pt) node[below] {\scriptsize$1$};
			\end{tikzpicture}
		}
	\subfigure[]{
		\begin{tikzpicture}
			\draw[-latex](-4.7,0)--(4.7,0)node[right]{ Re$z$};
			\draw [dashed](0,0) ellipse (0.5 and 0.9);
			\draw [dashed](2,0) ellipse (0.7 and 1.2);
			\draw [dashed](-2,0) ellipse (0.7 and 1.2);
			\draw [dashed](0,0) ellipse (4.1 and 1.6);
			\coordinate (a) at (0,-0.35);
			\fill (a) circle (0pt) node[below] {$+$};
			\coordinate (aa) at (0,1.5);
			\fill (aa) circle (0pt) node[below] {$+$};
			\coordinate (s) at (0,0.8);
			\fill (s) circle (0pt) node[below] {$-$};
			\coordinate (sa) at (0,-1);
			\fill (sa) circle (0pt) node[below] {$-$};
			\coordinate (f) at (2,-0.4);
			\fill (f) circle (0pt) node[below] {$+$};
			\coordinate (d) at (2,1);
			\fill (d) circle (0pt) node[below] {$-$};
			\coordinate (ff) at (-2,-0.4);
			\fill (ff) circle (0pt) node[below] {$+$};
			\coordinate (dd) at (-2,1);
			\fill (dd) circle (0pt) node[below] {$-$};
			\coordinate (I) at (0,0);
			\fill (I) circle (1pt) node[below] {$0$};
			\coordinate (c) at (-2.1,0);
			\fill[red] (c) circle (1pt) node[below] {\scriptsize$-1$};
			\coordinate (D) at (2.1,0);
			\fill[red] (D) circle (1pt) node[below] {\scriptsize$1$};
			\coordinate (A) at (-4.1,0);
			\fill (A) circle (1pt) node[below] {$\xi_8$};
			\coordinate (b) at (-2.7,0);
			\fill (b) circle (1pt) node[below] {$\xi_7$};
			\coordinate (C) at (-0.5,0);
			\fill (C) circle (1pt) node[below] {$\xi_5$};
			\coordinate (d) at (-1.3,0);
			\fill (d) circle (1pt) node[below] {$\xi_6$};
			\coordinate (E) at (4.1,0);
			\fill (E) circle (1pt) node[below] {$\xi_1$};
			\coordinate (R) at (2.7,0);
			\fill (R) circle (1pt) node[below] {$\xi_2$};
			\coordinate (T) at (0.5,0);
			\fill (T) circle (1pt) node[below] {$\xi_4$};
			\coordinate (Y) at (1.3,0);
			\fill (Y) circle (1pt) node[below] {$\xi_3$};
		\end{tikzpicture}
	}
	\subfigure[]{
	\begin{tikzpicture}
		\draw[-latex](-3,0)--(3,0)node[right]{ Re$z$};
		\draw [dashed](1.5,0) ellipse (1 and 1.6);
		\draw [dashed](-1.5,0) ellipse (1 and 1.6);
		\coordinate (I) at (0,0);
		\fill (I) circle (1pt) node[below] {$0$};
		\coordinate (A) at (-2.5,0);
		\fill (A) circle (1pt) node[below] {$\xi_4$};
		\coordinate (b) at (-0.5,0);
		\fill (b) circle (1pt) node[below] {$\xi_3$};
		\coordinate (e) at (2.5,0);
		\fill (e) circle (1pt) node[below] {$\xi_1$};
		\coordinate (f) at (0.5,0);
		\fill (f) circle (1pt) node[below] {$\xi_2$};
		\coordinate (c) at (-1.05,0);
		\fill[red] (c) circle (1pt) node[below] {\scriptsize$-1$};
		\coordinate (d) at (1.05,0);
		\fill[red] (d) circle (1pt) node[below] {\scriptsize$1$};
		\coordinate (s) at (0,0.8);
		\coordinate (sa) at (0,-1);
		\fill (sa) circle (0pt) node[below] {$-$};
		\fill (s) circle (0pt) node[below] {$+$};
		\coordinate (k) at (1.5,0.9);
		\coordinate (l) at (1.5,-0.5);
		\fill (k) circle (0pt) node[below] {$-$};
		\fill (l) circle (0pt) node[below] {$+$};
		\coordinate (g) at (-1.5,0.9);
		\coordinate (h) at (-1.5,-0.5);
		\fill (g) circle (0pt) node[below] {$-$};
		\fill (h) circle (0pt) node[below] {$+$};
	\end{tikzpicture}
}
	\subfigure[]{
	\begin{tikzpicture}
		\draw[-latex](-2.5,0)--(2.5,0)node[right]{ Re$z$};
		\draw [rotate=-30,dashed](0.8,1.87) ellipse (0.5 and 0.3);
		\draw [rotate=30,dashed](0.8,-1.87) ellipse (0.5 and 0.3);
		\draw [rotate=30,dashed](-0.8,1.87) ellipse (0.5 and 0.3);
		\draw [rotate=-30,dashed](-0.8,-1.87) ellipse (0.5 and 0.3);
		\coordinate (I) at (0,0);
		\fill (I) circle (1pt) node[below] {$0$};
		\coordinate (a) at (0,-1);
		\fill (a) circle (0pt) node[below] {$-$};
		\coordinate (s) at (0,1);
		\fill (s) circle (0pt) node[above] {$+$};
		\coordinate (c) at (-1.05,0);
		\fill[red] (c) circle (1pt) node[below] {\scriptsize$-1$};
		\coordinate (d) at (1.05,0);
		\fill[red] (d) circle (1pt) node[below] {\scriptsize$1$};
		\coordinate (aa) at (1.7,1.5);
		\fill (aa) circle (0pt) node[below] {$-$};
		\coordinate (as) at (-1.7,1.5);
		\fill (as) circle (0pt) node[below] {$-$};
		\coordinate (f) at (1.7,-1.5);
		\fill (f) circle (0pt) node[above] {$+$};
		\coordinate (ff) at (-1.7,-1.5);
		\fill (ff) circle (0pt) node[above] {$+$};
	\end{tikzpicture}
}
	\caption{\footnotesize The signature table  of $\text{Im}\theta_{12}$: The figure (a), (b), (c),(d) are corresponding to the cases $\xi<-\frac{1}{8}$,  $-\frac{1}{8}<\xi<0$, $0\leq\xi<1$,  $ \xi>1$, respectively. Notice that the line of dashes  only  represent the critical  lines of the sign of $\text{Im}\theta_{12}$ approximatively.  In the  region with sign "$+$", $\text{Im}\theta_{12}>0$. It implies that $|e^{it\theta_{12}}|\to 0$ as $t\to\infty$. And $\text{Im}\theta_{12}<0$ in the  region with sign "$-$", which implies  $|e^{-it\theta_{12}}|\to 0$ as $t\to\infty$. }
	\label{figtheta}
\end{figure}

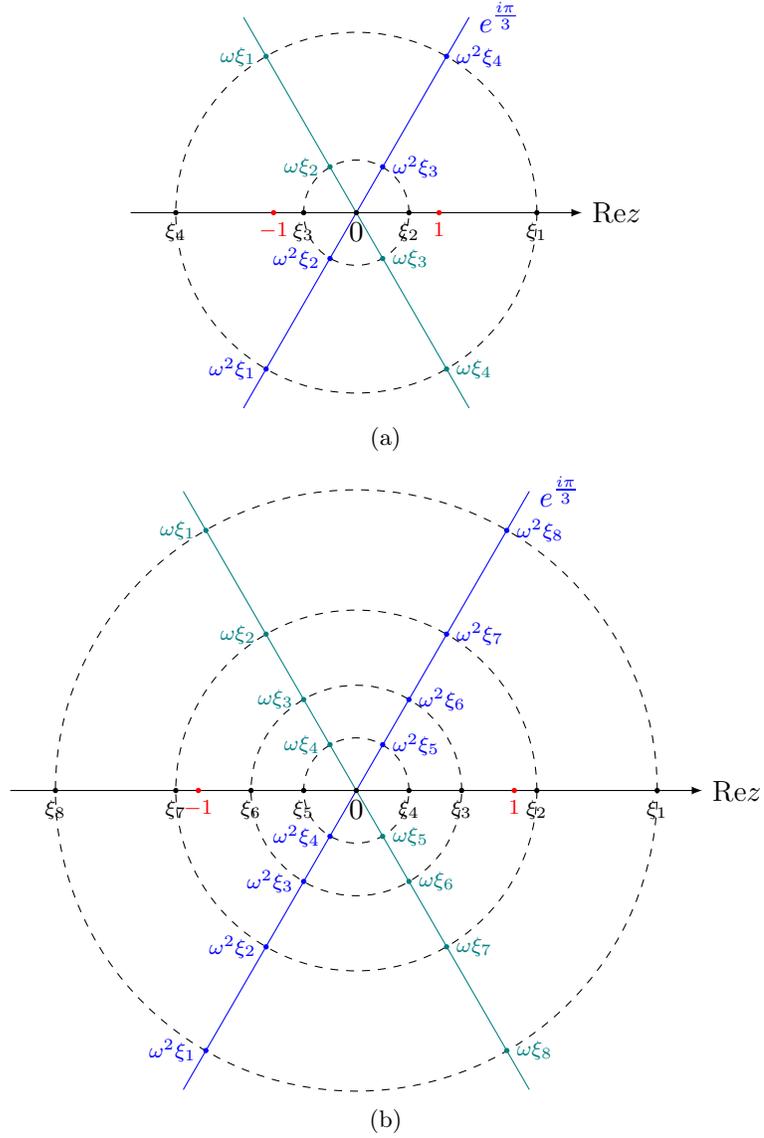
\begin{figure}[htp]
	\centering
	\subfigure[]{
		\begin{tikzpicture}
			\draw[-latex](0,0)--(3,0)node[right]{ Re$z$};
			\draw[blue ](0,0)--(1.5,2.6)node[right]{ $e^{\frac{i\pi}{3}}$};
			\draw[teal](0,0)--(-1.5,2.6);
			\draw[teal](0,0)--(1.5,-2.6);
			\draw[blue](0,0)--(-1.5,-2.6);
			\draw[ ](0,0)--(-3,0);
			\draw[dashed] (0.7,0) arc (0:360:0.7);
			\draw[dashed] (2.4,0) arc (0:360:2.4);
			\coordinate (I) at (0,0);
			\fill (I) circle (1pt) node[below] {$0$};
			\coordinate (A) at (-2.4,0);
			\fill (A) circle (1pt) node[below] {\scriptsize$\xi_4$};
			\coordinate (b) at (-0.7,0);
			\fill (b) circle (1pt) node[below] {\scriptsize$\xi_3$};
			\coordinate (e) at (2.4,0);
			\fill (e) circle (1pt) node[below] {\scriptsize$\xi_1$};
			\coordinate (f) at (0.7,0);
			\fill (f) circle (1pt) node[below] {\scriptsize$\xi_2$};
			\coordinate (c) at (-1.1,0);
			\fill[red] (c) circle (1pt) node[below] {\scriptsize$-1$};
			\coordinate (d) at (1.1,0);
			\fill[red] (d) circle (1pt) node[below] {\scriptsize$1$};
			\coordinate (A1) at (0.35,0.61);
			\fill[blue] (A1) circle (1pt) node[right] {\scriptsize$\omega^2\xi_3$};
			\coordinate (A2) at (0.35,-0.61);
			\fill[teal] (A2) circle (1pt) node[right] {\scriptsize$\omega\xi_3$};
			\coordinate (A3) at (-0.35,0.61);
			\fill[teal] (A3) circle (1pt) node[left] {\scriptsize$\omega\xi_2$};
			\coordinate (A4) at (-0.35,-0.61);
			\fill[blue] (A4) circle (1pt) node[left] {\scriptsize$\omega^2\xi_2$};
			\coordinate (n1) at (1.2,2.08);
			\fill[blue] (n1) circle (1pt) node[right] {\scriptsize$\omega^2\xi_4$};
			\coordinate (m1) at (1.2,-2.08);
			\fill[teal] (m1) circle (1pt) node[right] {\scriptsize$\omega\xi_4$};
			\coordinate (j1) at (-1.2,2.08);
			\fill[teal] (j1) circle (1pt) node[left] {\scriptsize$\omega\xi_1$};
			\coordinate[blue] (k1) at (-1.2,-2.08);
			\fill[blue] (k1) circle (1pt) node[left] {\scriptsize$\omega^2\xi_1$};
		\end{tikzpicture}
	}
	\subfigure[]{
		\begin{tikzpicture}
			\draw[-latex](0,0)--(4.6,0)node[right]{ Re$z$};
			\draw[ ](0,0)--(-4.6,0);
			\draw[blue ](0,0)--(2.3,3.98)node[right]{ $e^{\frac{i\pi}{3}}$};	
			\draw[teal ](0,0)--(2.3,-3.98);
			\draw[teal ](0,0)--(-2.3,3.98);
			\draw[blue ](0,0)--(-2.3,-3.98);
			\draw[dashed] (0.7,0) arc (0:360:0.7);
			\draw[dashed] (2.4,0) arc (0:360:2.4);
			\draw[dashed] (1.4,0) arc (0:360:1.4);
			\draw[dashed] (4,0) arc (0:360:4);
			\coordinate (I) at (0,0);
			\fill (I) circle (1pt) node[below] {$0$};
			\coordinate (c) at (-2.1,0);
			\fill[red] (c) circle (1pt) node[below] {\scriptsize$-1$};
			\coordinate (D) at (2.1,0);
			\fill[red] (D) circle (1pt) node[below] {\scriptsize$1$};
			\coordinate (A) at (-4,0);
			\fill (A) circle (1pt) node[below] {\scriptsize$\xi_8$};
			\coordinate (b) at (-2.4,0);
			\fill (b) circle (1pt) node[below] {\scriptsize$\xi_7$};
			\coordinate (C) at (-0.7,0);
			\fill (C) circle (1pt) node[below] {\scriptsize$\xi_5$};
			\coordinate (d) at (-1.4,0);
			\fill (d) circle (1pt) node[below] {\scriptsize$\xi_6$};
			\coordinate (E) at (4,0);
			\fill (E) circle (1pt) node[below] {\scriptsize$\xi_1$};
			\coordinate (R) at (2.4,0);
			\fill (R) circle (1pt) node[below] {\scriptsize$\xi_2$};
			\coordinate (T) at (0.7,0);
			\fill (T) circle (1pt) node[below] {\scriptsize$\xi_4$};
			\coordinate (Y) at (1.4,0);
			\fill (Y) circle (1pt) node[below] {\scriptsize$\xi_3$};
			\coordinate (A1) at (0.35,0.61);
			\fill[blue] (A1) circle (1pt) node[right] {\scriptsize$\omega^2\xi_5$};
			\coordinate (A2) at (0.35,-0.61);
			\fill[teal] (A2) circle (1pt) node[right] {\scriptsize$\omega\xi_5$};
			\coordinate (A3) at (-0.35,0.61);
			\fill[teal] (A3) circle (1pt) node[left] {\scriptsize$\omega\xi_4$};
			\coordinate (A4) at (-0.35,-0.61);
			\fill[blue] (A4) circle (1pt) node[left] {\scriptsize$\omega^2\xi_4$};
			\coordinate (n1) at (1.2,2.08);
			\fill[blue] (n1) circle (1pt) node[right] {\scriptsize$\omega^2\xi_7$};
			\coordinate (m1) at (1.2,-2.08);
			\fill[teal] (m1) circle (1pt) node[right] {\scriptsize$\omega\xi_7$};
			\coordinate (j1) at (-1.2,2.08);
			\fill[teal] (j1) circle (1pt) node[left] {\scriptsize$\omega\xi_2$};
			\coordinate[blue] (k1) at (-1.2,-2.08);
			\fill[blue] (k1) circle (1pt) node[left] {\scriptsize$\omega^2\xi_2$};
			\coordinate (n2) at (0.7,1.21);
			\fill[blue] (n2) circle (1pt) node[right] {\scriptsize$\omega^2\xi_6$};
			\coordinate (m2) at (0.7,-1.21);
			\fill[teal] (m2) circle (1pt) node[right] {\scriptsize$\omega\xi_6$};
			\coordinate (j2) at (-0.7,1.21);
			\fill[teal] (j2) circle (1pt) node[left] {\scriptsize$\omega\xi_3$};
			\coordinate[blue] (k2) at (-0.7,-1.21);
			\fill[blue] (k2) circle (1pt) node[left] {\scriptsize$\omega^2\xi_3$};
			\coordinate (n23) at (2,3.46);
			\fill[blue] (n23) circle (1pt) node[right] {\scriptsize$\omega^2\xi_8$};
			\coordinate (m23) at (2,-3.46);
			\fill[teal] (m23) circle (1pt) node[right] {\scriptsize$\omega\xi_8$};
			\coordinate (j23) at (-2,3.46);
			\fill[teal] (j23) circle (1pt) node[left] {\scriptsize$\omega\xi_1$};
			\coordinate[blue] (k3) at (-2,-3.46);
			\fill[blue] (k3) circle (1pt) node[left] {\scriptsize$\omega^2\xi_1$};
		\end{tikzpicture}
	}
	\caption{\footnotesize The distribution of stationary phase points  on the contour $\Sigma$:
  The figures (a) and (b) are corresponding to  the  cases   $- {1}/{8}<\xi<0$ and $0\leq\xi<1$, respectively.   }
	\label{figpoint}
\end{figure}

To  uniformly deal with above  the  Case I,  Case II and  Case III,  we introduce some necessary  notations.
We denote the number of stationary phase points on the $ \mathbb{R}$  as
\begin{align}
	p(\xi)=\left\{ \begin{array}{ll}
		0,   &\text{as } \xi>1 \text{ and } \xi<-\frac{1}{8},\\[10pt]
		4 , &\text{as } 0\leq\xi<1,\\[10pt]
		8,   &\text{as } -\frac{1}{8}<\xi<0.
	\end{array}\right. \label{r36}
\end{align}
From the relation  (\ref{42}),   if $\xi_j, j=1,...,p(\xi)$  is   phase points  on the real axis $ \mathbb{R}$,
then  $\omega \xi_j$  and $\omega^2 \xi_j$ for $j=1,...,p(\xi)$ are also the stationary phase points on $\omega\mathbb{R}$ and $\omega^2\mathbb{R}$ respectively.
We use  $\xi_{n,j} :=\omega^n\xi_j$   denoting    the  $j th$  phase point on the $\omega^n\mathbb{R}$ with  $ n=0,1,2;\ j=1,\cdots, p(\xi)$.
Altogether, there are $3p(\xi)$  phase points  corresponding to  12 phase points for the case $0\leq\xi<1$ and 24 phase points
  $-\frac{1}{8}<\xi<0$, respectively.   See  Figure \ref{figpoint}.

On the  $ \mathbb{R} $, denote $\xi_0=-\infty$, $\xi_{p(\xi)+1}=+\infty$, and introduce some  intervals  when $j=1,...,p(\xi)$, for $0\leq\xi<1$
\begin{align}
I_{j1}=I_{j2}=\left\{ \begin{array}{ll}
\left( \frac{\xi_j+\xi_{j+1}}{2},\xi_j\right) ,\    & j\text{ is  odd number} ,\\[10pt]
\left(\xi_j ,\frac{\xi_j+\xi_{j-1}}{2}\right),   &j\text{ is  even number},
\end{array}\right.\label{In1}\\
I_{j3}=I_{j4}=\left\{ \begin{array}{ll}
\left(\xi_j ,\frac{\xi_j+\xi_{j-1}}{2}\right),\    & j\text{ is  odd number} ,\\[10pt]
\left( \frac{\xi_j+\xi_{j+1}}{2},\xi_j\right) ,   &j\text{ is  even number},
\end{array}\right.
\end{align}
and for $- {1}/{8}<\xi<0$,
\begin{align}
I_{j1}=I_{j2}=\left\{ \begin{array}{ll}
\left(\xi_j ,\frac{\xi_j+\xi_{j-1}}{2}\right),\    & j\text{ is  odd number} ,\\[10pt]
\left( \frac{\xi_j+\xi_{j+1}}{2},\xi_j\right) ,   &j\text{ is  even number},
\end{array}\right.\\
I_{j3}=I_{j4}=\left\{ \begin{array}{ll}
\left( \frac{\xi_j+\xi_{j+1}}{2},\xi_j\right) ,\    & j\text{ is  odd number} ,\\[10pt]
\left(\xi_j ,\frac{\xi_j+\xi_{j-1}}{2}\right),   &j\text{ is  even number},
\end{array}\right..\label{In2}
\end{align}
 On the $\omega^n\mathbb{R}, n=0,1,2$,  the  interval division   can be got from those above
$$I_{jk}^\omega=\{\omega z:   z \in I_{jk}\}, \ \ I_{jk}^{\omega^2} =\{   \omega^2z:  z\in I_{jk}\}.$$
As illustrative example, the  interval   division  on the $\mathbb{R}$  are shown in Figure \ref{phase}.
\begin{figure}[h]
	\subfigure[]{
		\begin{tikzpicture}
			\draw[->](-5,0)--(5,0)node[right]{ Re$z$};
			\draw(2.5,0)--(2.5,0.1)node[above]{\scriptsize$\frac{\xi_1+\xi_2}{2}$};
			\draw(2.5,0)--(2.5,-0.1);
			\draw(-2.5,0)--(-2.5,0.1)node[above]{\scriptsize$\frac{\xi_3+\xi_4}{2}$};
			\draw(-2.5,0)--(-2.5,-0.1);
			\draw(0,0)--(0,0.1)node[above]{\scriptsize$\frac{\xi_3+\xi_2}{2}$};
			\draw(0,0)--(0,-0.1);
			\coordinate (I) at (0,0);
			\fill (I) circle (1pt) node[below] {$0$};
			\coordinate (A) at (-4,0);
			\fill (A) circle (1pt) node[below] {$\xi_4$};
			\coordinate (b) at (-1,0);
			\fill (b) circle (1pt) node[below] {$\xi_3$};
			\coordinate (e) at (4,0);
			\fill (e) circle (1pt) node[below] {$\xi_1$};
			\coordinate (f) at (1,0);
			\fill (f) circle (1pt) node[below] {$\xi_2$};
			\coordinate (ke) at (4.7,0.1);
			\fill (ke) circle (0pt) node[below] {\footnotesize$I_{13}$};
			\coordinate (k1e) at (4.7,-0.1);
			\fill (k1e) circle (0pt) node[above] {\footnotesize$I_{14}$};
			\coordinate (le) at (3.3,0.1);
			\fill (le) circle (0pt) node[below] {\footnotesize$I_{12}$};
			\coordinate (l1e) at (3.3,-0.1);
			\fill (l1e) circle (0pt) node[above] {\footnotesize$I_{11}$};
			\coordinate (n2) at (0.57,0.1);
			\fill (n2) circle (0pt) node[below] {\footnotesize$I_{23}$};
			\coordinate (n12) at (0.57,-0.1);
			\fill (n12) circle (0pt) node[above] {\footnotesize$I_{24}$};
			\coordinate (m2) at (1.8,0.1);
			\fill (m2) circle (0pt) node[below] {\footnotesize$I_{22}$};
			\coordinate (m12) at (1.8,-0.1);
			\fill (m12) circle (0pt) node[above] {\footnotesize$I_{21}$};
			\coordinate (k) at (-4.7,0.1);
			\fill (k) circle (0pt) node[below] {\footnotesize$I_{43}$};
			\coordinate (k1) at (-4.7,-0.1);
			\fill (k1) circle (0pt) node[above] {\footnotesize$I_{44}$};
			\coordinate (l) at (-3.2,0.1);
			\fill (l) circle (0pt) node[below] {\footnotesize$I_{42}$};
			\coordinate (l1) at (-3.2,-0.1);
			\fill (l1) circle (0pt) node[above] {\footnotesize$I_{41}$};
			\coordinate (n) at (-0.5,0.1);
			\fill (n) circle (0pt) node[below] {\footnotesize$I_{33}$};
			\coordinate (n1) at (-0.5,-0.1);
			\fill (n1) circle (0pt) node[above] {\footnotesize$I_{34}$};
			\coordinate (m) at (-1.7,0.1);
			\fill (m) circle (0pt) node[below] {\footnotesize$I_{32}$};
			\coordinate (m1) at (-1.7,-0.1);
			\fill (m1) circle (0pt) node[above] {\footnotesize$I_{31}$};
			\coordinate (c) at (-2.05,0);
			\fill[red] (c) circle (1pt) node[below] {\scriptsize$-1$};
			\coordinate (d) at (2.05,0);
			\fill[red] (d) circle (1pt) node[below] {\scriptsize$1$};
			\end{tikzpicture}
			\label{pcase1}}
		\subfigure[]{
		\begin{tikzpicture}
		\draw[->](-6.5,0)--(6.5,0)node[right]{ Re$z$};
		\coordinate (I) at (0,0);
		\fill (I) circle (1pt) node[below] {$0$};
		\coordinate (c) at (-3,0);
		\fill[red] (c) circle (1pt) node[below] {\scriptsize$-1$};
		\coordinate (D) at (3,0);
		\fill[red] (D) circle (1pt) node[below] {\scriptsize$1$};
		\draw(-1.5,0)--(-1.5,0.1)node[above]{\scriptsize$\frac{\xi_5+\xi_6}{2}$};
		\draw(-1.5,0)--(-1.5,-0.1);
		\draw(-4.7,0)--(-4.7,0.1)node[above]{\scriptsize$\frac{\xi_7+\xi_8}{2}$};
		\draw(-4.7,0)--(-4.7,-0.1);
		\draw(-3.1,0)--(-3.1,0.1)node[above]{\scriptsize$\frac{\xi_7+\xi_6}{2}$};
		\draw(-3.1,0)--(-3.1,-0.1);
		\draw(1.5,0)--(1.5,0.1)node[above]{\scriptsize$\frac{\xi_3+\xi_4}{2}$};
		\draw(1.5,0)--(1.5,-0.1);
		\draw(4.7,0)--(4.7,0.1)node[above]{\scriptsize$\frac{\xi_1+\xi_2}{2}$};
		\draw(4.7,0)--(4.7,-0.1);
		\draw(3.1,0)--(3.1,0.1)node[above]{\scriptsize$\frac{\xi_3+\xi_2}{2}$};
		\draw(3.1,0)--(3.1,-0.1);
		\draw(0,0)--(0,0.1)node[above]{\scriptsize$\frac{\xi_4+\xi_5}{2}$};
		\draw(0,0)--(0,-0.1);
		\coordinate (A) at (-5.4,0);
		\fill (A) circle (1pt) node[below] {$\xi_8$};
		\coordinate (b) at (-4,0);
		\fill (b) circle (1pt) node[below] {$\xi_7$};
		\coordinate (C) at (-0.8,0);
		\fill (C) circle (1pt) node[below] {$\xi_5$};
		\coordinate (d) at (-2.2,0);
		\fill (d) circle (1pt) node[below] {$\xi_6$};
		\coordinate (E) at (5.4,0);
		\fill (E) circle (1pt) node[below] {$\xi_1$};
		\coordinate (R) at (4,0);
		\fill (R) circle (1pt) node[below] {$\xi_2$};
		\coordinate (T) at (0.8,0);
		\fill (T) circle (1pt) node[below] {$\xi_4$};
		\coordinate (Y) at (2.2,0);
		\fill (Y) circle (1pt) node[below] {$\xi_3$};
		\coordinate (q) at (6.2,-0.1);
		\fill (q) circle (0pt) node[above] {\tiny$I_{11}$};
		\coordinate (q1) at (6.2,0.05);
		\fill (q1) circle (0pt) node[below] {\tiny$I_{12}$};
		\coordinate (w) at (4.95,-0.1);
		\fill (w) circle (0pt) node[above] {\tiny$I_{14}$};
		\coordinate (w1) at (4.95,0.1);
		\fill (w1) circle (0pt) node[below] {\tiny$I_{13}$};
		\coordinate (e) at (4.47,-0.1);
		\fill (e) circle (0pt) node[above] {\tiny$I_{24}$};
		\coordinate (e1) at (4.47,0.1);
		\fill (e1) circle (0pt) node[below] {\tiny$I_{23}$};
		\coordinate (r) at (3.4,-0.1);
		\fill (r) circle (0pt) node[above] {\tiny$I_{21}$};
		\coordinate (r1) at (3.4,0.05);
		\fill (r1) circle (0pt) node[below] {\tiny$I_{22}$};
		\coordinate (t) at (2.7,-0.1);
		\fill (t) circle (0pt) node[above] {\tiny$I_{31}$};
		\coordinate (t1) at (2.7,0.05);
		\fill (t1) circle (0pt) node[below] {\tiny$I_{32}$};
		\coordinate (y) at (1.75,-0.1);
		\fill (y) circle (0pt) node[above] {\tiny$I_{34}$};
		\coordinate (y1) at (1.75,0.1);
		\fill (y1) circle (0pt) node[below] {\tiny$I_{33}$};
		\coordinate (l) at (1.26,-0.1);
		\fill (l) circle (0pt) node[above] {\tiny$I_{44}$};
		\coordinate (l1) at (1.26,0.1);
		\fill (l1) circle (0pt) node[below] {\tiny$I_{43}$};
		\coordinate (k) at (0.3,-0.1);
		\fill (k) circle (0pt) node[above] {\tiny$I_{41}$};
		\coordinate (k1) at (0.3,0.1);
		\fill (k1) circle (0pt) node[below] {\tiny$I_{42}$};
		\coordinate (q8) at (-6.2,-0.1);
		\fill (q8) circle (0pt) node[above] {\tiny$I_{81}$};
		\coordinate (q18) at (-6.2,0.05);
		\fill (q18) circle (0pt) node[below] {\tiny$I_{82}$};
		\coordinate (w8) at (-4.95,-0.1);
		\fill (w8) circle (0pt) node[above] {\tiny$I_{84}$};
		\coordinate (w18) at (-4.95,0.1);
		\fill (w18) circle (0pt) node[below] {\tiny$I_{83}$};
		\coordinate (e7) at (-4.47,-0.1);
		\fill (e7) circle (0pt) node[above] {\tiny$I_{74}$};
		\coordinate (e17) at (-4.47,0.1);
		\fill (e17) circle (0pt) node[below] {\tiny$I_{73}$};
		\coordinate (7r) at (-3.4,-0.1);
		\fill (7r) circle (0pt) node[above] {\tiny$I_{71}$};
		\coordinate (r17) at (-3.4,0.05);
		\fill (r17) circle (0pt) node[below] {\tiny$I_{72}$};
		\coordinate (t6) at (-2.7,-0.1);
		\fill (t6) circle (0pt) node[above] {\tiny$I_{61}$};
		\coordinate (t16) at (-2.7,0.05);
		\fill (t16) circle (0pt) node[below] {\tiny$I_{62}$};
		\coordinate (y6) at (-1.75,-0.1);
		\fill (y6) circle (0pt) node[above] {\tiny$I_{64}$};
		\coordinate (y16) at (-1.75,0.1);
		\fill (y16) circle (0pt) node[below] {\tiny$I_{63}$};
		\coordinate (l5) at (-1.26,-0.1);
		\fill (l5) circle (0pt) node[above] {\tiny$I_{54}$};
		\coordinate (l15) at (-1.26,0.1);
		\fill (l15) circle (0pt) node[below] {\tiny$I_{53}$};
		\coordinate (k5) at (-0.3,-0.1);
		\fill (k5) circle (0pt) node[above] {\tiny$I_{51}$};
		\coordinate (k15) at (-0.3,0.1);
		\fill (k15) circle (0pt) node[below] {\tiny$I_{52}$};
		\end{tikzpicture}
		\label{phase2}}
	\caption{\footnotesize Figure (a) corresponds to  the case $0<\xi<1$  with  four stationary phase points $\xi_1,...\xi_4$.
     Figure  (b)   corresponds  to  the case  $- {1}/{8}<\xi<0$  with eight stationary phase points $\xi_1,...\xi_8$.}
	\label{phase}
\end{figure}
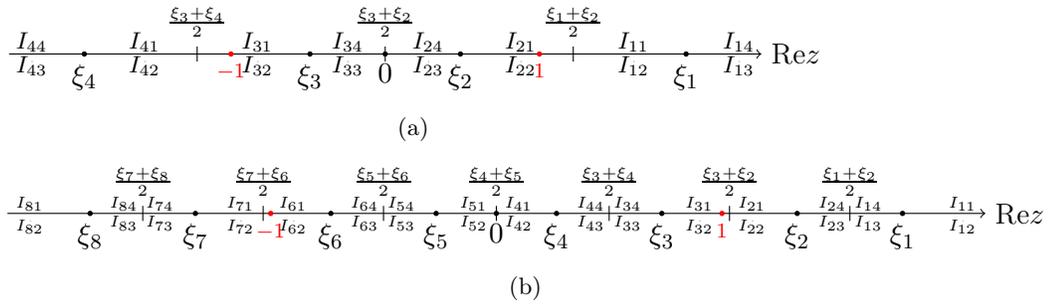

\subsection{Conjugation}

For convenience, we denote
$$\mathcal{N}\triangleq\left\lbrace 1,...,N_0\right\rbrace, \   \tilde{\mathcal{N}}\triangleq\left\lbrace 1,...,2N_1+N_2\right\rbrace, \ \tilde{\mathcal{N}}^A\triangleq\left\lbrace 1+2N_1+N_2,...,N_0\right\rbrace.$$
Further  to distinguish different  type of zeros, we introduce a small positive constant $\delta_0$ to give the  partitions $\Delta,\nabla$ and $\lozenge$  of $\mathcal{N}$   as follow

\begin{align}
&\nabla_1=\left\lbrace j \in \tilde{\mathcal{N}}  ;\text{Im}\theta_{12}(\zeta_j)> \delta_0\right\rbrace,
\Delta_1=\left\lbrace j \in \tilde{\mathcal{N}} ;\text{Im}\theta_{12}(\zeta_j)< 0\right\rbrace,\nonumber\\
&\nabla_2=\left\lbrace i \tilde{\mathcal{N}}^A  ;\text{Im}\theta_{23}(\zeta_i)> \delta_0\right\rbrace,
\Delta_2=\left\lbrace i \in \tilde{\mathcal{N}}^A ;\text{Im}\theta_{23}(\zeta_i)< 0\right\rbrace, \\
&\lozenge_1=\left\lbrace j_0 \in \tilde{\mathcal{N}} ;0\leq\text{Im}\theta_{12}(\zeta_{j_0})\leq\delta_0\right\rbrace,\nonumber\\	
&\lozenge_2=\left\lbrace i_0 \in \tilde{\mathcal{N}}^A  ;0\leq\text{Im}\theta_{23}(\zeta_{i_0})\leq\delta_0\right\rbrace, \\	
&\nabla=\nabla_1\cup\nabla_2,
\Delta=\Delta_1\cup\Delta_2,\nonumber\\
&\lozenge=\{n\in\mathcal{N};\exists k\in\{1,...,5\},n-kN_0\in\lozenge_1\cup\lozenge_2\}.\label{devide}
\end{align}
For $\zeta_n$ with $n\in\Delta$, the residue of $M(z)$ at $\zeta_n$ in (\ref{RES1}) and (\ref{RES2})  grows without bound as $t\to\infty$. Similarly, for $\zeta_n$ with $n\in\nabla$, the residue are  approaching to  $0$. Denote two constants $\mathcal{N}(\lozenge)=|\lozenge|$ and
\begin{equation}
	\rho_0=\min_{n\in\mathcal{N}\setminus \lozenge}\left\lbrace |\text{Im}\theta_n|\right\rbrace >0.\label{rho0}
\end{equation}

For the poles $\zeta_n$ with $ \mathcal{N}\setminus \lozenge$, we want to trap them for jumps along small closed loops enclosing themselves,
respectively. And the jump matrix $V_j(z)$ (\ref{jumpv})  also needs to  be restricted for $j=1,...,6$. Recall the   factorizations of  the jump matrix:\\
For, j=1,4,
\begin{align}
	&V_j(z)=\left(\begin{array}{ccc}
	1 & 0 &0\\
	\bar{r} e^{-it\theta_{12}} & 1 &0\\
	0 & 0 &1
\end{array}\right)\left(\begin{array}{ccc}
1 & -r e^{it\theta_{12}} &0\\
0 & 1&0\\
0 & 0 &1
\end{array}\right)\\
&=\left(\begin{array}{ccc}
	1 & \frac{-r e^{it\theta_{12}}}{1-|r|^2}&0\\
	0 & 1&0\\
	0 & 0 &1
\end{array}\right)\left(\begin{array}{ccc}
\frac{1 }{1-|r|^2} & 0 & 0\\
0 & 1-|r|^2 & 0\\
0 & 0 &1
\end{array}\right)\left(\begin{array}{ccc}
1 & 0 & 0\\
\frac{\bar{r} e^{-it\theta_{12}}}{1-|r|^2} & 1 & 0\\
0 & 0 &1
\end{array}\right);\nonumber
\end{align}
For, j=2,5,
\begin{align}
	&V_j(z)=\left(\begin{array}{ccc}
		1 & 0 &0\\
		0 & 1 &0\\
		0 & \bar{r}(\omega z) e^{-it\theta_{23}} &1
	\end{array}\right)\left(\begin{array}{ccc}
		1 & 0 &0\\
		0 & 1&-r(\omega z) e^{it\theta_{23}}\\
		0 & 0 &1
	\end{array}\right)\\
	&=\left(\begin{array}{ccc}
		1 & 0&0\\
		0 & 1&\frac{-r(\omega z) e^{it\theta_{23}}}{1-|r(\omega z)|^2}\\
		0 & 0 &1
	\end{array}\right)\left(\begin{array}{ccc}
		1 & 0 & 0\\
		0 & \frac{1 }{1-|r(\omega z)|^2} & 0\\
		0 & 0 &1-|r(\omega z)|^2
	\end{array}\right)\left(\begin{array}{ccc}
		1 & 0 & 0\\
	0 & 1 & 0\\
		0 & 	\frac{\bar{r}(\omega z) e^{-it\theta_{23}}}{1-|r(\omega z)|^2} &1
	\end{array}\right);\nonumber
\end{align}
For, j=3,6,
\begin{align}
	&V_j(z)=\left(\begin{array}{ccc}
		1 & 0 &\bar{r}(\omega^2 z) e^{it\theta_{13}}\\
		0 & 1 &0\\
		0 & 0 &1
	\end{array}\right)\left(\begin{array}{ccc}
		1 & 0 &0\\
		0 & 1&0\\
		-r(\omega^2 z) e^{-it\theta_{13}} & 0 &1
	\end{array}\right)\\
	&=\left(\begin{array}{ccc}
		1 & 0&0\\
		0 & 1&0\\
		\frac{-r(\omega^2 z) e^{-it\theta_{13}}}{1-|r(\omega^2 z)|^2} & 0 &1
	\end{array}\right)\left(\begin{array}{ccc}
		1-|r(\omega^2 z)|^2 & 0 & 0\\
		0 & 1 & 0\\
		0 & 0 &\frac{1}{1-|r(\omega^2 z)|^2}
	\end{array}\right)\left(\begin{array}{ccc}
		1 & 0 & 	\frac{\bar{r}(\omega^2 z) e^{it\theta_{13}}}{1-|r(\omega^2 z)|^2}\\
		0 & 1 & 0\\
		0 & 0 &1
	\end{array}\right).\nonumber
\end{align}
We will utilize  these factorizations to deform the jump contours so that the oscillating factor $e^{\pm it\theta_{12}}$ are decaying in corresponding region, respectively.

 Introduce the notations
\begin{align}
	I(\xi)=\left\{ \begin{array}{ll}
		\emptyset,   &\text{as } \xi>1,\\[4pt]
		(-\infty,\xi_8)\cup_{j=1}^3(\xi_{2j+1},\xi_{2j})\cup(\xi_1,+\infty),   &\text{as } -\frac{1}{8}<\xi<0,\\[4pt]
		(\xi_4,\xi_3)\cup(\xi_2,\xi_1),   &\text{as } 0\leq\xi<1,\\[4pt]
		\mathbb{R},   &\text{as } \xi<-\frac{1}{8},\\
	\end{array}\right.
\end{align}
and $I^\omega(\xi)=\{  \omega z:  z\in   I(\xi)\}$, $I^{\omega^2}(\xi)=\{ \omega^2 z: z\in I(\xi)\}$. Here  we denote
the integration on $\emptyset$ is zero.

Consider the following scalar RH problem \\
$\blacktriangleright$ $\delta_1(z)$	is analytic in $\mathbb{C}\setminus\mathbb{R}$;\\
$\blacktriangleright$ $\delta_1(z)$ has jump relation:
\begin{align}
& \delta_{1,+} (z)=\delta_{1,-}(z)(1-|r(z)|^2), \ z\in I(\xi),\nonumber\\
&\delta_{1,+}(z)=\delta_{1,-}(z), \ z\in \mathbb{R}\setminus I(\xi).\nonumber
\end{align}
$\blacktriangleright$ $\delta_1(z)\to 1$ as $z\to\infty$.

 By Plemelj  formula, this RH problem admits solution
\begin{align}
	\delta_1 (z)&=\delta_1 (z,\xi)=\exp\left(-i\int _{I(\xi)}\dfrac{\nu(s) ds}{s-z}\right),
\end{align}
where $\nu(z) =-\frac{1}{2\pi }\log (1-|r(z)|^2).$

We  define
\begin{align}
\Pi(z)&=\prod_{n\in \Delta_1}\dfrac{z-\zeta_n}{z-\bar{\zeta}_n},\
\Pi^A(z)=\prod_{m\in \Delta_2}\dfrac{z-\omega\zeta_m}{z-\omega^2\bar{\zeta}_m};\label{Pi}\\
H(z)&=H(z,\xi)=\Pi^A(z)\Pi(z)\delta_1 (z,\xi)^{-1};\\
T_1(z)&=T_1(z,\xi)= \frac{H(\omega^2z)}{H(z)};\label{T}\\ T_2(z)&=T_2(z,\xi)=T_1(\omega z)= \frac{H(z)}{H(\omega z)};\\ T_3(z)&=T_3(z,\xi)=T_1(\omega^2 z)= \frac{H(\omega z)}{H(\omega^2 z)};\\
T_{ij}(z)&=T_{ij}(z,\xi)=\frac{T_i(z)}{T_j(z)},\ i,j=1,2,3\label{Tij}.
\end{align}
In  the above formulas, we choose the principal branch of power and logarithm functions. Additionally, Introduce   a positive constant $\tilde{\varrho}=\frac{1}{6}\underset{j\neq i\in \mathcal{N}}{\min} |\zeta_i-\zeta_j|$ and a set of  characteristic functions $\mathcal{X}(z;\xi,j)$ on  the interval $\eta(\xi,j)\xi_j- \tilde{\varrho}<\eta(\xi,j)z<\eta(\xi,j)\xi_j$ for $j=1,...,p(\xi)$, respectively. And $\eta(\xi,j)$ is a constant depend on $\xi$ and $j$:
\begin{align}
\eta(\xi,j)=\left\{ \begin{array}{ll}
		(-1)^j,   &\text{as } -\frac{1}{8}<\xi<0;\\
		(-1)^{j+1},   &\text{as } 0\leq\xi<1.
	\end{array}\right.\label{eta}
\end{align}

\begin{Proposition}\label{proT}
	The function defined by (\ref{T}) and (\ref{Tij}) has following properties
\begin{itemize}
\item[{\rm (a)}]  $T_1$ is meromorphic in $\mathbb{C}\setminus \mathbb{R}$, and for each $n\in\Delta_1$, $T_1(z)$ has a simple pole at $\zeta_n$ and a simple zero
 at $\bar{\zeta}_n$, for each $m\in\Delta_2$, $T_1(z)$ has a simple pole at $\omega\zeta_m$ and a simple zero at $\omega\bar{\zeta}_m$;

\item[{\rm (b)}]  $\overline{T_1(\bar{z})}=T_1(\omega z)=T_1(z^{-1})$;

\item[{\rm (c)}] For $z\in I(\xi)$, as z approaching the real axis from above and below, $\delta_1(z)$ has boundary values $ \delta_{1,\pm}$, which satisfies:
	\begin{equation}
	 \delta_{1,-}(z)=(1-|r(z)|^2) \delta_{1,+}(z),\hspace{0.5cm}z\in I(\xi),
	\end{equation}
	which gives that
	\begin{align}
	& T_{1,-}(z)=(1-|r(z)|^2) T_{1,+ }(z),\hspace{0.5cm}z\in I(\xi),\\
&T_{1,+ }(z)=(1-|r(\omega^2z)|^2)T_{1,- }(z),\hspace{0.5cm}z\in I^{\omega}(\xi);
	\end{align}

	\item[{\rm (d)}]  $\lim_{z\to \infty}T_1(z)\triangleq T_1(\infty)$ with $T_1(\infty)=1$;

\item[{\rm (e)}]  $T_1(e^{\frac{i\pi}{6}})$ exists as a constants;

\item[{\rm (f)}]  $T_1(z)$ is continuous at $z=0$, and
	\begin{equation}
	T_1(0)=T_1(\infty)=1 \label{T0};
	\end{equation}

\item[{\rm (g)}]  As $z\to \xi_j$ along any ray $\xi_j+e^{i\phi}\mathbb{R}^+$ with $|\phi|<\pi$ ,
	\begin{align}
	|T_{12}(z,\xi)-T_{12}^{(j)}(\xi)\left[\eta(\xi,j) (z-\xi_j)\right] ^{i\eta(\xi,j)\nu(\xi_j)}|\lesssim \parallel r\parallel_{H^{1,1}(\mathbb{R})}|z-\xi_j|^{1/2},\label{T-TJ}
	\end{align}
	where $T_{12}^{(j)}(\xi) $ is the complex unit
	\begin{align}
	T_{12}^{(j)}(\xi)&= \prod_{n\in \Delta_1}\left(\dfrac{\xi_j-\zeta_n}{\xi_j-\bar{\zeta}_n}\right) ^{-2}\dfrac{\xi_j-\omega\zeta_n}{\xi_j-\omega\bar{\zeta}_n}\dfrac{\xi_j-\omega^2\zeta_n}{\xi_j-\omega^2\bar{\zeta}_n}\nonumber\\
	&\prod_{m\in \Delta_2}\left( \dfrac{\xi_j-\omega\zeta_m}{\xi_j-\omega^2\bar{\zeta}_m}\right) ^{-2}\dfrac{\xi_j-\omega^2\zeta_m}{\xi_j-\bar{\zeta}_m}\dfrac{\xi_j-\zeta_m}{\xi_j-\omega\bar{\zeta}_m}e^{2i\beta(\xi_j,\xi)}, \ \ j=1,...,p(\xi).\nonumber
	\end{align}
  In above function,
	\begin{align}
	\beta_j(z,\xi)
	&=\int_{I(\xi)}\frac{\nu(s)-\mathcal{X}(z;\xi,j)\nu(\xi_j)}{s-z}ds-\eta(\xi,j)\nu(\xi_j)\log\left(\eta(\xi,j) (z-\xi_j+\tilde{\varrho})\right).\nonumber
	\end{align}
\end{itemize}
\end{Proposition}
\begin{proof}
	 Properties (a), (b), (d) and (f)  can be obtain by simple calculation from the definition of $T_1(z)$ and $T_{ij}$ in (\ref{T}) and (\ref{Tij}).  (c)  follows from the Plemelj formula. And for (g), analogously to  \cite{fNLS}, rewrite
	 \begin{align}
	 \delta_1 (z,\xi)=\exp\left(i\beta_j(z,\xi)+\nu(\xi_j)\eta(\xi,j)i\log\left( \eta(\xi,j)(z-\xi_j)\right)  \right) ,
	 \end{align}
	 and note the fact that
	 \begin{align}
	 |(z-\xi_j)^{\eta(\xi,j)i\nu(\xi_j)}|\leq e^{-\pi\nu(\xi_j)=\sqrt{1+|r(\xi_j)|^2}}\label{key},
	 \end{align}
	 and
	 $$|\beta_j(z,\xi)-\beta_j(\xi_j,\xi)|\lesssim \parallel r\parallel_{H^{1,0}(\mathbb{R})}|z-\xi_j|^{1/2}.$$
	 Then, the result  follows promptly.
	 For brevity, we  omit computation.
\end{proof}

The (g) in  Proposition \ref{proT} and the definition in (\ref{T})-(\ref{Tij}) also give that  as $z\to \omega\xi_j$ along any ray $\omega\xi_j+e^{i\phi}\mathbb{R}^+$ with $|\phi|<\pi$ ,
\begin{align}
	|T_{31}(z,\xi)-T_{31}^{(j)}(\xi)\left( \eta(\xi,j)( z-\omega\xi_j)\right) ^{\eta(\xi,j)i\nu(\xi_j)}|\lesssim \parallel r\parallel_{H^{1,1}(\mathbb{R})}|z-\omega\xi_j|^{1/2},\label{T-T31}
\end{align}
where $T_{31}^{(j)}(\xi)= $ is the complex unit
\begin{align}
	T_{31}^{(j)}(\xi)&= \prod_{n\in \Delta_1}\dfrac{\omega\xi_j-\zeta_n}{\omega\xi_j-\bar{\zeta}_n}\left(\dfrac{\omega\xi_j-\omega\zeta_n}{\omega\xi_j-\omega\bar{\zeta}_n}\right) ^{-2}\dfrac{\omega\xi_j-\omega^2\zeta_n}{\omega\xi_j-\omega^2\bar{\zeta}_n}\nonumber\\
	&\prod_{m\in \Delta_2} \dfrac{\omega\xi_j-\omega\zeta_m}{\omega\xi_j-\omega^2\bar{\zeta}_m}\left(\dfrac{\omega\xi_j-\omega^2\zeta_m}{\omega\xi_j-\bar{\zeta}_m}\right) ^{-2}\dfrac{\omega\xi_j-\zeta_m}{\omega\xi_j-\omega\bar{\zeta}_m}e^{2i\beta(\omega\xi_j,\xi)},
\end{align}
for  $j=1,...,p(\xi)$.
 And as $z\to \omega^2\xi_j$ along any ray $\omega^2\xi_j+e^{i\phi}\mathbb{R}^+$ with $|\phi|<\pi$ ,
 \begin{align}
 	|T_{23}(z,\xi)-T_{23}^{(j)}(\xi)\left( \eta(\xi,j)( z-\omega^2\xi_j)\right) ^{\eta(\xi,j)i\nu(\xi_j)}|\lesssim \parallel r\parallel_{H^{1,1}(\mathbb{R})}|z-\omega\xi_j|^{1/2},\label{T-T23}
 \end{align}
 where $T_{23}^{(j)}(\xi)= $ is the complex unit
 \begin{align}
 	T_{23}^{(j)}(\xi)&= \prod_{n\in \Delta_1}\dfrac{\omega^2\xi_j-\zeta_n}{\omega^2\xi_j-\bar{\zeta}_n}\dfrac{\omega^2\xi_j-\omega\zeta_n}{\omega^2\xi_j-\omega\bar{\zeta}_n}\left(\dfrac{\omega^2\xi_j-\omega^2\zeta_n}{\omega^2\xi_j-\omega^2\bar{\zeta}_n}\right) ^{-2}\nonumber\\
 	&\prod_{m\in \Delta_2} \dfrac{\omega^2\xi_j-\omega\zeta_m}{\omega^2\xi_j-\omega^2\bar{\zeta}_m}\dfrac{\omega^2\xi_j-\omega^2\zeta_m}{\omega^2\xi_j-\bar{\zeta}_m}\left(\dfrac{\omega^2\xi_j-\zeta_m}{\omega^2\xi_j-\omega\bar{\zeta}_m}\right) ^{-2}e^{2i\beta(\omega^2\xi_j,\xi)},
 \end{align}
 for  $j=1,...,p(\xi)$.

Additionally, introduce   a positive constant $\varrho$:
\begin{align}
	\varrho=\frac{1}{4}\min&\left\lbrace \min_{j\in \mathcal{N}} |\text{Im}\zeta_j| ,\min_{j\in \mathcal{N},\arg(z)=\frac{\pi i}{3}}|\zeta_j-z|,\min_{j\in \mathcal{N}\setminus\lozenge,\text{Im}\theta_{ik}(z)=0}|\zeta_j-z|,\right.\nonumber\\
	&\left. \min_{  j\in \mathcal{N}}|\zeta_j-e^{\frac{i\pi}{6}}|,\min_{  j\neq k\in \mathcal{N}}|\zeta_j-\zeta_k|\right\rbrace .
\end{align}
By above definition and the symmetry of poles and $\theta$, for every $n$ satisfies that exist $k\in\{0,...,5\}$ make $n-kN_0\in\mathcal{N}\setminus\lozenge$,
 The small disks $\mathbb{D}_n\triangleq\mathbb{D}(\zeta_n,\varrho)$  are pairwise disjoint,  also  disjoint  with  critical lines
$\left\lbrace z\in \mathbb{C};\text{Im} \theta(z)=0 \right\rbrace $, as well as  the contours  $\mathbb{R}$, $\omega\mathbb{R}$ and $\omega^2\mathbb{R}$. Besides, $e^{\frac{i\pi}{6}}\notin \mathbb{D}_n$.  Denote
 a piecewise matrix function
\begin{equation}
	G(z)=\left\{ \begin{array}{ll}
		I-\frac{B_n}{z-\zeta_n},   & z\in\mathbb{D}_n,n-kN_0\in\nabla,k\in\{0,...,5\};\\[12pt]
		\left(\begin{array}{ccc}
			1& 0& 0 \\
			-\frac{z-\zeta_n}{C_{n}e^{it\theta_{12}(\zeta_{n})}}& 1 & 0\\
			0&	0 & 1
		\end{array}\right),   & z\in\mathbb{D}_n,n\in\Delta_1 \text{ or }n-2N_0\in\Delta_2;\\[12pt]
	\left(\begin{array}{ccc}
		1&	0 & 0 \\
		0& 1 & 0\\
		0& -\frac{z-\zeta_n}{C_{n}e^{it\theta_{23}(\zeta_{n})}} & 1
	\end{array}\right),   & z\in\mathbb{D}_n,n-N_0\in\Delta_1 \text{ or }n-5N_0\in\Delta_2;\\[12pt]
	\left(\begin{array}{ccc}
		1&	0 & -\frac{z-\zeta_n}{C_{n}e^{-it\theta_{13}(\zeta_{n})}} \\
		0& 1 & 0\\
		0&	0 & 1
	\end{array}\right),   & z\in\mathbb{D}_n,n-2N_0\in\Delta_1 \text{ or }n-4N_0\in\Delta_2;\\[12pt]
	\left(\begin{array}{ccc}
		1&	0 &  0\\
		0& 1 &  -\frac{z-\zeta_n}{C_{n}e^{-it\theta_{23}(\zeta_{n})}}\\
		0&	0 & 1
	\end{array}\right),   & z\in\mathbb{D}_n,n-3N_0\in\Delta_1 \text{ or }n-N_0\in\Delta_2;\\[12pt]
	\left(\begin{array}{ccc}
		1&	0 & 0 \\
		0& 1 & 0\\
		0 &-\frac{z-\zeta_n}{C_{n}e^{it\theta_{23}(\zeta_{n})}}	 & 1
	\end{array}\right),   & z\in\mathbb{D}_n,n-4N_0\in\Delta_1 \text{ or }n\in\Delta_2;\\[12pt]
	\left(\begin{array}{ccc}
	1&	-\frac{z-\zeta_n}{C_{n}e^{-it\theta_{12}(\zeta_{n})}} & 0 \\
	0& 1 & 0\\
	0&	0 & 1
	\end{array}\right),   & z\in\mathbb{D}_n,n-5N_0\in\Delta_1\text{ or } n-3N_0\in\Delta_2;\\[12pt]
	I & 	z \text{ in elsewhere};
	\end{array}\right.,\label{funcG}
\end{equation}
where $B_n$ are nilpotent matrixs defined in (\ref{RES1})-(\ref{RES3}).
Then by using $G(z)$ and $$T(z)=\text{diag}\{T_1(z),T_2(z),T_3(z)\},$$the new  matrix-valued   function $M^{(1)}(z)$  is defined as
\begin{equation}
M^{(1)}(z;y,t)\triangleq M^{(1)}(z)=M(z)G(z)T(z),\label{transm1}
\end{equation}
which then satisfies the following RH problem.

\begin{figure}[htp]
	\centering	
	\begin{tikzpicture}[node distance=2cm]
		\draw[red,-latex](-4,0)--(4,0)node[right]{$L_1$};
		\draw[blue ](0,0)--(2,3.464)node[above]{$L_2$};
		\draw[teal](0,0)--(-2,3.464)node[above]{$L_3$};
		\draw[red](0,0)--(-4,0)node[left]{$L_4$};
		\draw[teal](0,0)--(2,-3.464)node[right]{$L_6$};
		\draw[blue](0,0)--(-2,-3.464)node[left]{$L_5$};
		\draw[dashed] (2,0) arc (0:360:2);
		\coordinate (A) at (2.2,2.2);
		\draw (2.3,2.2) arc (0:360:0.1);
		\coordinate (B) at (2.2,-2.2);
		\draw (2.3,-2.2) arc (0:360:0.1);
		\coordinate (C) at (-0.8,3);
		\draw (-0.7,3) arc (0:360:0.1);
		\coordinate (D) at (-0.8,-3);
		\draw (-0.7,-3) arc (0:360:0.1);
		\coordinate (E) at (0.9,0.9);
		\draw (1,0.9) arc (0:360:0.1);
		\coordinate (F) at (0.9,-0.9);
		\draw (1,-0.9) arc (0:360:0.1);
		\coordinate (G) at (-3,0.8);
		\draw (-2.9,0.8) arc (0:360:0.1);
		\coordinate (H) at (-3,-0.8);
		\draw (-2.9,-0.8) arc (0:360:0.1);
		\coordinate (J) at (1.7570508075688774,0.956);
		\coordinate (K) at (1.7570508075688774,-0.956);
		\coordinate (L) at (-1.7570508075688774,0.956);
		\coordinate (M) at (-1.7570508075688774,-0.956);
		\coordinate (a) at (0,2);
		\fill (a) circle (1pt) node[above] {$\omega \bar{w}_m$};
		\coordinate (s) at (0,-2);
		\fill (s) circle (1pt) node[below] {$\omega ^2w_m$};
		\coordinate (d) at (-0.33,1.23);
		\fill (d) circle (1pt) node[right] {$\frac{\omega}{z_n} $};
		\draw (-0.23,1.23) arc (0:360:0.1);
		\coordinate (f) at (-0.33,-1.23);
		\fill (f) circle (1pt) node[right] {$\frac{\omega^2}{\bar{z}_n}$};
		\draw (-0.23,-1.23) arc (0:360:0.1);
		\coordinate (g) at (-1.23,0.33);
		\fill (g) circle (1pt) node[right] {$\frac{\omega}{\bar{z}_n} $};
		\draw (-1.13,0.33) arc (0:360:0.1);
		\coordinate (h) at (-1.23,-0.33);
		\fill (h) circle (1pt) node[right] {$\frac{\omega^2}{z_n}$};
		\draw (-1.13,-0.33) arc (0:360:0.1);
		\fill (A) circle (1pt) node[right] {$z_n$};
		\fill (B) circle (1pt) node[right] {$\bar{z}_n$};
		\fill (C) circle (1pt) node[left] {$\omega \bar{z}_n$};
		\fill (D) circle (1pt) node[right] {$\omega^2 z_n$};
		\fill (E) circle (1pt) node[right] {$\frac{1}{\bar{z}_n}$};
		\fill (F) circle (1pt) node[right] {$\frac{1}{z_n}$};
		\fill (G) circle (1pt) node[left] {$\omega z_n$};
		\fill (H) circle (1pt) node[left] {$\omega^2\bar{z}_n$};
		\fill (J) circle (1pt) node[right] {$w_m$};
		\fill (K) circle (1pt) node[right] {$\bar{w}_m$};
		\fill (L) circle (1pt) node[left] {$\omega w_m$};
		\fill (M) circle (1pt) node[below] {$\omega^2\bar{w}_m$};
	\end{tikzpicture}
	\caption{\footnotesize  $\mathbb{R}$, $\omega\mathbb{R}$, $\omega^2\mathbb{R}$ and  the small circles  constitute $\Sigma^{(1)}$. the curve $\{z;\text{Im}\theta_{kj}(z)=0\}$ for $k,j=1,2,3$ absences because we can't obtain its accurate   graph. In this figure, $0\leq$Im$\theta_{12}(w_m)\leq\delta_0 $, it remains the pole of $M^{(1)}$. Then  $\omega w_m$, $\omega^2 w_m$, $\bar{w}_m$, $\omega\bar{w}_m$ and   $\omega^2\bar{w}_m$ also remain. Simultaneously, the pole $z_n$ and its symmetry point are changed to jump on  small circles.   }	\label{fig:zero}
\end{figure}
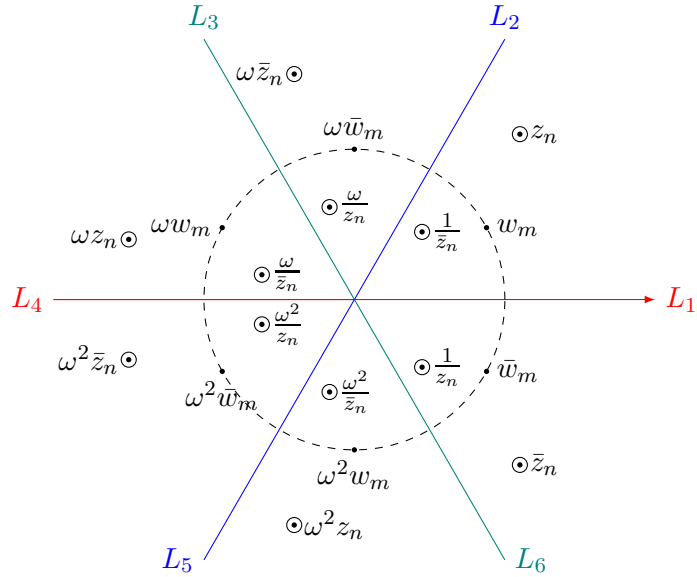
\begin{RHP}\label{RHP3}
	Find a matrix-valued function  $  M^{(1)}(z )$ which satisfies:
	
	$\blacktriangleright$ Analyticity: $M^{(1)}(z)$ is meromorphic in $\mathbb{C}\setminus \Sigma^{(1)}$, where
	\begin{equation}
		\Sigma^{(1)}=\Sigma \cup\left( \underset{n\in\mathcal{N}\setminus\lozenge,k=0,..,5}{\cup}\partial\mathbb{D}_{n+kN_0} \right)  ,
	\end{equation}
	is shown in Figure \ref{fig:zero};
	
	$\blacktriangleright$ Symmetry: $M^{(1)}(z)=\Gamma_1\overline{M^{(1)}(\bar{z})}\Gamma_1=\Gamma_2\overline{M^{(1)}(\omega^2\bar{z})}\Gamma_2=\Gamma_3\overline{M^{(1)}(\omega\bar{z})}\Gamma_3$

\quad and $M^{(1)}(z)=\overline{M^{(1)}(\bar{z}^{-1})}$;
	
	$\blacktriangleright$ Jump condition: $M^{(1)}$ has continuous boundary values $M^{(1)}_\pm$ on $\Sigma^{(1)}$ and
	\begin{equation}
		M^{(1)}_+(z)=M^{(1)}_-(z)V^{(1)}(z),\hspace{0.5cm}z \in \Sigma^{(1)},
	\end{equation}
	where
	\footnotesize{\begin{equation}
		V^{(1)}(z)=\left\{\begin{array}{ll}	\left(\begin{array}{ccc}
				1 & 0 &0\\
				\bar{r}T_{12} e^{-it\theta_{12}} & 1 &0\\
				0 & 0 &1
			\end{array}\right)\left(\begin{array}{ccc}
				1 & -r T_{21}e^{it\theta_{12}} &0\\
				0 & 1&0\\
				0 & 0 &1
			\end{array}\right),   & z\in 	\mathbb{R}\setminus I(\xi);\\[12pt]
			\left(\begin{array}{ccc}
				1 & \frac{rT_{21} e^{it\theta_{12}}}{1-|r|^2}&0\\
				0 & 1&0\\
				0 & 0 &1
			\end{array}\right)\left(\begin{array}{ccc}
				1 & 0 & 0\\
				\frac{\bar{r}T_{12} e^{-it\theta_{12}}}{1-|r|^2} & 1 & 0\\
				0 & 0 &1
			\end{array}\right),   & z\in I(\xi);\\[12pt]
		\left(\begin{array}{ccc}
			1 & 0 &0\\
			0 & 1 &0\\
			0 & \bar{r}(\omega z) e^{-it\theta_{23}}T_{23} &1
		\end{array}\right)\left(\begin{array}{ccc}
			1 & 0 &0\\
			0 & 1&-r(\omega z)T_{32} e^{it\theta_{23}}\\
			0 & 0 &1
		\end{array}\right),   & z\in 	\omega^2\mathbb{R}\setminus I^\omega(\xi);\\[12pt]
		\left(\begin{array}{ccc}
			1 & 0&0\\
			0 & 1&\frac{-r(\omega z)T_{32} e^{it\theta_{23}}}{1-|r(\omega z)|^2}\\
			0 & 0 &1
		\end{array}\right)\left(\begin{array}{ccc}
			1 & 0 & 0\\
			0 & 1 & 0\\
			0 & 	\frac{\bar{r}(\omega z) e^{-it\theta_{23}}T_{23}}{1-|r(\omega z)|^2} &1
		\end{array}\right),   & z\in I^\omega(\xi);\\[12pt]
	\left(\begin{array}{ccc}
		1 & 0 &\bar{r}(\omega^2 z)T_{31} e^{it\theta_{13}}\\
		0 & 1 &0\\
		0 & 0 &1
	\end{array}\right)\left(\begin{array}{ccc}
		1 & 0 &0\\
		0 & 1&0\\
		-r(\omega^2 z)T_{13} e^{-it\theta_{13}} & 0 &1
	\end{array}\right),   & z\in 	\omega\mathbb{R}\setminus I^{\omega^2}(\xi);\\[12pt]
	\left(\begin{array}{ccc}
		1 & 0&0\\
		0 & 1&0\\
		\frac{-r(\omega^2 z)T_{13} e^{-it\theta_{13}}}{1-|r(\omega^2 z)|^2} & 0 &1
	\end{array}\right)\left(\begin{array}{ccc}
		1 & 0 & 	\frac{\bar{r}(\omega^2 z)T_{31} e^{it\theta_{13}}}{1-|r(\omega^2 z)|^2}\\
		0 & 1 & 0\\
		0 & 0 &1
	\end{array}\right),   & z\in I^{\omega^2}(\xi);\\[12pt]
	T^{-1}(z)G(z)T(z),   & 	z\in\partial\mathbb{D}_n\cap (\underset{k=1}{\cup^3}S_{2k});\\[12pt]
		T^{-1}(z)G^{-1}(z)T(z),    & z\in\partial\mathbb{D}_n\cap (\underset{k=1}{\cup^3}S_{2k-1});\\
		\end{array}\right. \label{jumpv1}
	\end{equation}}

	\normalsize$\blacktriangleright$ Asymptotic behaviors:
	\begin{align}
		M^{(1)}(z;y,t) =& I+\mathcal{O}(z^{-1}),\hspace{0.5cm}z \rightarrow \infty;
\end{align}

		$\blacktriangleright$ Singularities:  As $z\to \varkappa_l=e^{\frac{i\pi(l-1)}{3}}$, $l = 1,...,6$,   the limit of $M^{(1)}(z)$  has  pole singularities
	\begin{align}
		&M^{(1)}(z)=\frac{1}{z\mp1}\left(\begin{array}{ccc}
			\alpha_\pm &	\alpha_\pm & \beta_\pm \\
			-\alpha_\pm & -\alpha_\pm & -\beta_\pm\\
			0	&	0 & 0
		\end{array}\right)T(\pm1)+\mathcal{O}(1),\ z\to\pm 1,\\
		&M^{(1)}(z)=\frac{1}{z\mp\omega^2}\left(\begin{array}{ccc}
			0 &	0 &  0\\
			\beta_\pm	 & \alpha_\pm &\alpha_\pm \\
			-\beta_\pm	&	-\alpha_\pm & -\alpha_\pm
		\end{array}\right)T(\pm\omega^2)+\mathcal{O}(1),\ z\to\pm \omega^2,\\
		&M^{(1)}(z)=\frac{1}{z\mp\omega}\left(\begin{array}{ccc}
			-\alpha_\pm &	-\beta_\pm & -\alpha_\pm\\
			0	 & 0 &0 \\
			\alpha_\pm &	\beta_\pm & \alpha_\pm
		\end{array}\right)T(\pm\omega)+\mathcal{O}(1),\ z\to\pm \omega,
	\end{align}
	with $\alpha_\pm=\alpha_\pm(y,t)=-\bar{\alpha}_\pm$, $\beta_\pm=\beta_\pm(y,t)=-\bar{\beta}_\pm$ and $M^{(1)}(z)^{-1}$ has same specific
	matrix structure with $\alpha_\pm$, $\beta_\pm$ replaced by $\tilde{\alpha}_\pm$, $\tilde{\beta}_\pm$;
	
	$\blacktriangleright$  Residue conditions: $M^{(1)}$ has simple poles at each point $\zeta_n$ and $\bar{\zeta}_n$ for $n\in\lozenge$ with
	\begin{align}
		&\res_{z=\zeta_n}M^{(1)}(z)=\lim_{z\to \zeta_n}M^{(1)}(z)\left[ T^{-1}(z)B_nT(z)\right] .
	\end{align}
\end{RHP}

\begin{proof}
	The triangular factors  (\ref{funcG})  trades 	poles $\zeta_n$  and $\bar{\zeta}_n$  to jumps on the disk boundaries $\partial \mathbb{D}_n$ and $\partial \overline{\mathbb{D}}_n$, respectively  for $n\in\mathcal{N}\setminus\lozenge$. Then by simple calculation we can obtain the residues condition and jump condition from (\ref{RES1}), (\ref{RES2}) (\ref{jumpv}), (\ref{funcG}) and (\ref{transm1}). The   analyticity and symmetry of $M^{(1)}(z)$ are directly from its definition, the Proposition \ref{proT}, (\ref{funcG}) and the identities of $M$. As for asymptotic behaviors, from $\lim_{z\to i}G(z)=\lim_{z\to \infty}G(z)=I$ and (d) of  Proposition \ref{proT}, we  obtain the asymptotic behaviors of $M^{(1)}(z)$.
\end{proof}

\section{Hybrid $\bar{\partial}$-RH problem }\label{sec4}

\quad  In this section, by defining  a  new  transformation for
$M^{(1)}(z)$,   we make continuous extension to the jump matrix $V^{(1)}(z)$  to   remove the jump contours  $ \mathbb{R}\cup\omega\mathbb{R}\cup\omega^2\mathbb{R}$,
so that we can  take  advantage of the decay/growth of  $e^{2it\theta_{jk}(z)}$  on new jump contours.

\subsection{Opening $\bar\partial$-lenses in space-time   regions  $\xi<- {1}/{8}$ and $\xi>1$}

  From (\ref{r36}),  there is no phase point in the space-time   regions  $\xi<- {1}/{8}$ and $\xi>1$,
 for which  we can  open the contours at $z=0$. For this purpose, we  introduce  some notations.

Fixed the angle  $ 0<\varphi< \frac{\pi}{6} $   sufficiently small,  near the
jump contours  $ \mathbb{R}\cup\omega\mathbb{R}\cup\omega^2\mathbb{R}$,
we define domains
\begin{align}
	&\Omega_{ 2n+1}^k =\left\lbrace z\in\mathbb{C}: \  n\pi \leq \arg (\omega^k z)  \leq   n\pi+\varphi \right\rbrace,  \label{dom1}\\
	&\Omega_{2n+2}^k =\left\lbrace z\in\mathbb{C}:  \  (n+1)\pi -\varphi \leq \arg (\omega^k z)  \leq  (n+1)\pi   \right\rbrace. \label{dom6}
\end{align}
where $n=0,1; \ k=0, 1, 2 $. The   boundaries  of these domains are  the following rays
\begin{align}
&\Sigma_n^k=e^{(n-1)\pi i/2+i\varphi}R_+,\hspace{0.5cm} n=1,3,\\
&\Sigma_n^k=e^{n  \pi i/2-i\varphi}R_+,\hspace{0.5cm} n=2,4.
\end{align}
where $ k=0, 1, 2 $.   In addition, for these cases, let
\begin{align}
&\Omega(\xi)=\underset{n=1,...,4}{\cup}\left( \Omega_n^0\cup\Omega_n^1\cup\Omega_n^{2}\right) ,\\
&\Sigma^{(2)}(\xi)=\underset{n\in\mathcal{N}\setminus\lozenge,k=0,..,5}{\cup}\partial\mathbb{D}_{n+kN_0}  ,\\
& {\Sigma}^k(\xi)=\Sigma_1^k\cup\Sigma_2^k\cup\Sigma_{3}^k\cup\Sigma_{4}^k, \ \ k=0, 1, 2,\\
& {\Sigma}(\xi)= {\Sigma}^0(\xi) \cup  {\Sigma}^1 (\xi)\cup  {\Sigma}^2 (\xi),
\end{align}
 which are shown in Figure \ref{figR2}.

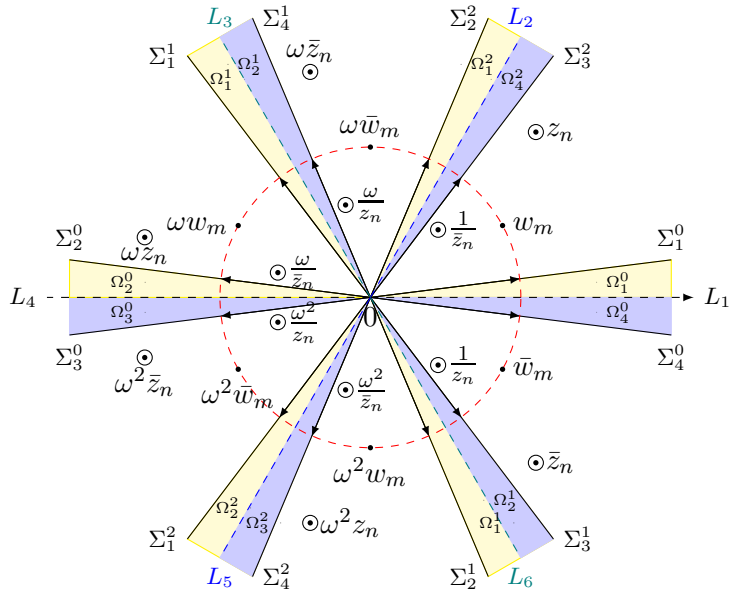
\begin{figure}[htp]
	\centering
	\begin{tikzpicture}[node distance=2cm]
		%\draw[yellow, fill=yellow!20] (0,0)--(4,-0.5)--(4,0.5)--(0,0)--(-4,-0.5)--(-4,0.5)--(0,0);
        \draw[yellow, fill=yellow!20] (0,0)--(4,0)--(4,0.5)--(0,0);
        \draw[blue!20, fill=blue!20] (0,0)--(4,0)--(4,-0.5)--(0,0);
         \draw[blue!20, fill=blue!20] (0,0)--(-4,0)--(-4,-0.5)--(0,0);
         \draw[yellow, fill=yellow!20] (0,0)--(-4,0)--(-4, 0.5)--(0,0);
         	\draw[yellow, fill=yellow!20] (0,0)--(2.433,3.214)--(1.567,3.714)--(0,0)--(-2.433,-3.214)--(-1.567,-3.714)--(0,0);
		\draw[blue!20, fill=blue!20] (0,0)--(2.433,3.214)--(2,3.464)--(0,0);
	  \draw[blue!20, fill=blue!20] (0,0)--(-1.567,-3.714)--(-2,-3.464)--(0,0);
		\draw[yellow, fill=yellow!20] (0,0)--(-2.433,3.214)--(-1.567,3.714)--(0,0)--(2.433,-3.214)--(1.567,-3.714)--(0,0);
	\draw[blue!20, fill=blue!20] (0,0)--(-1.567,3.714)--(-2,3.464)--(0,0);
	\draw[blue!20, fill=blue!20] (0,0)--(2.433,-3.214)--(2,-3.464)--(0,0);
		\draw(0,0)--(-1.567,3.714)node[right]{\footnotesize $\Sigma_4^1$};
	    \draw[-latex](0,0)--(-1.567/2,3.714/2);
		\draw(0,0)--(-2.433,3.214)node[left]{\footnotesize $\Sigma_1^1$};
         \draw[-latex](0,0)--(-2.433/2, 3.214/2);
		\draw(0,0)--(2.433,-3.214)node[right]{\footnotesize $\Sigma_3^1$};
   \draw[-latex](0,0)--(2.433/2, -3.214/2);
		\draw(0,0)--(1.567,-3.714)node[left]{\footnotesize $\Sigma_2^1$};
   \draw[-latex](0,0)--(1.567/2,-3.714/2);
		\draw(0,0)--(1.567,3.714)node[left]{\footnotesize $\Sigma_2^2$};
  \draw[-latex](0,0)--(1.567/2,3.714/2);
		\draw(0,0)--(-2.433,-3.214)node[left]{\footnotesize $\Sigma_1^2$};
  \draw[-latex](0,0)--(-2.433/2,-3.214/2);
		\draw(0,0)--(2.433,3.214)node[right]{\footnotesize $\Sigma_3^2$};
 \draw[-latex](0,0)--( 2.433/2, 3.214/2);
		\draw(0,0)--(-1.567,-3.714)node[right]{\footnotesize $\Sigma_4^2$};
 \draw[-latex](0,0)--(-1.567/2,-3.714/2);
		\draw(0,0)--(4,0.5)node[above]{\footnotesize $\Sigma_1^0$};
		\draw(0,0)--(-4,0.5)node[above]{\footnotesize $\Sigma_2^0$};
		\draw(0,0)--(-4,-0.5)node[below]{\footnotesize $\Sigma_3^0$};
		\draw(0,0)--(4,-0.5)node[below]{\footnotesize $\Sigma_4^0$};
		\draw[-latex](0,0)--(-2,-0.25);
		\draw[-latex](0,0)--(-2,0.25);
		\draw[-latex](0,0)--(2,0.25);
		\draw[-latex](0,0)--(2,-0.25);
		\draw[ dashed,-latex](0,0)--(4.3,0)node[right]{\footnotesize $L_1$};
		\draw[blue,dashed ](0,0)--(2,3.464)node[above]{\footnotesize $L_2$};
		\draw[teal,dashed ](0,0)--(-2,3.464)node[above]{\footnotesize $L_3$};
		\draw[ dashed ](0,0)--(-4.3,0)node[left]{\footnotesize $L_4$};
		\draw[teal,dashed ](0,0)--(2,-3.464)node[below]{\footnotesize $L_6$};
		\draw[blue,dashed ](0,0)--(-2,-3.464)node[below]{\footnotesize $L_5$};
		\draw[red,dashed] (2,0) arc (0:360:2);
		\coordinate (A) at (2.2,2.2);
		\draw (2.3,2.2) arc (0:360:0.1);
		\coordinate (B) at (2.2,-2.2);
		\draw (2.3,-2.2) arc (0:360:0.1);
		\coordinate (C) at (-0.8,3);
		\draw (-0.7,3) arc (0:360:0.1);
		\coordinate (D) at (-0.8,-3);
		\draw (-0.7,-3) arc (0:360:0.1);
		\coordinate (E) at (0.9,0.9);
		\draw (1,0.9) arc (0:360:0.1);
		\coordinate (F) at (0.9,-0.9);
		\draw (1,-0.9) arc (0:360:0.1);
		\coordinate (G) at (-3,0.8);
		\draw (-2.9,0.8) arc (0:360:0.1);
		\coordinate (H) at (-3,-0.8);
		\draw (-2.9,-0.8) arc (0:360:0.1);
		\coordinate (J) at (1.7570508075688774,0.956);
		\coordinate (K) at (1.7570508075688774,-0.956);
		\coordinate (L) at (-1.7570508075688774,0.956);
		\coordinate (M) at (-1.7570508075688774,-0.956);
		\coordinate (a) at (0,2);
		\fill (a) circle (1pt) node[above] {$\omega \bar{w}_m$};
		\coordinate (s) at (0,-2);
		\fill (s) circle (1pt) node[below] {$\omega ^2w_m$};
		\coordinate (d) at (-0.33,1.23);
		\fill (d) circle (1pt) node[right] {$\frac{\omega}{z_n} $};
		\draw (-0.23,1.23) arc (0:360:0.1);
		\coordinate (f) at (-0.33,-1.23);
		\fill (f) circle (1pt) node[right] {$\frac{\omega^2}{\bar{z}_n}$};
		\draw (-0.23,-1.23) arc (0:360:0.1);
		\coordinate (g) at (-1.23,0.33);
		\fill (g) circle (1pt) node[right] {$\frac{\omega}{\bar{z}_n} $};
		\draw (-1.13,0.33) arc (0:360:0.1);
		\coordinate (h) at (-1.23,-0.33);
		\fill (h) circle (1pt) node[right] {$\frac{\omega^2}{z_n}$};
		\draw (-1.13,-0.33) arc (0:360:0.1);
		\fill (A) circle (1pt) node[right] {$z_n$};
		\fill (B) circle (1pt) node[right] {$\bar{z}_n$};
		\fill (C) circle (1pt) node[above] {$\omega \bar{z}_n$};
		\fill (D) circle (1pt) node[right] {$\omega^2 z_n$};
		\fill (E) circle (1pt) node[right] {$\frac{1}{\bar{z}_n}$};
		\fill (F) circle (1pt) node[right] {$\frac{1}{z_n}$};
		\fill (G) circle (1pt) node[below] {$\omega z_n$};
		\fill (H) circle (1pt) node[below] {$\omega^2\bar{z}_n$};
		\fill (J) circle (1pt) node[right] {$w_m$};
		\fill (K) circle (1pt) node[right] {$\bar{w}_m$};
		\fill (L) circle (1pt) node[left] {$\omega w_m$};
		\fill (M) circle (1pt) node[below] {$\omega^2\bar{w}_m$};
		\coordinate (C) at (-0.2,2.2);
		\coordinate (D) at (3,0.2);
		\fill (D) circle (0pt) node[right] {\tiny $\Omega_1^0$};
		\coordinate (J) at (-3,-0.2);
		\fill (J) circle (0pt) node[left] {\tiny $\Omega_3^0$};
		\coordinate (k) at (-3,0.2);
		\fill (k) circle (0pt) node[left] {\tiny $\Omega_2^0$};
		\coordinate (k0) at (3,-0.2);
		\fill (k0) circle (0pt) node[right] {\tiny $\Omega_4^0$};
		\coordinate (D1) at (1.6,2.9);
		\fill (D1) circle (0pt) node[right] {\tiny $\Omega_4^2$};
		\coordinate (J1) at (-1.6,-2.8);
		\fill (J1) circle (0pt) node[left] {\tiny $\Omega_2^2$};
		\coordinate (l1) at (-1.2,-3);
		\fill (l1) circle (0pt) node[left] {\tiny $\Omega_3^2$};
		\coordinate (k1) at (1.2,3.1 );
		\fill (k1) circle (0pt) node[right] {\tiny $\Omega_1^2$};
		\coordinate (D2) at (-1.9,3.1);
		\fill (D2) circle (0pt) node[right] {\tiny $\Omega_2^1$};
		\coordinate (J2) at (-1.68,2.9);
		\fill (J2) circle (0pt) node[left] {\tiny  $\Omega_1^1$};
		\coordinate (l2) at (1.9,-3);
		\fill (l2) circle (0pt) node[left] {\tiny $\Omega_1^1$};
		\coordinate (k2) at (1.5,-2.7);
		\fill (k2) circle (0pt) node[right] {\tiny $\Omega_2^1$};
		\coordinate (I) at (0,0);
		\fill (I) circle (0pt) node[below] {$0$};
	\end{tikzpicture}
	\caption{\footnotesize   For the case   $\xi>1$,   the signature table $ \text{Im }\theta_{kj}(z)< 0$  in  yellow domains
 and  $ \text{Im }\theta_{kj}(z)> 0$ in blue domains;  For the case   $\xi<-1/8$,   the signature table $ \text{Im }\theta_{kj}(z)>0$  in  yellow domains
 and  $ \text{Im }\theta_{kj}(z)< 0$ in blue domains.    The  circles around poles   off critical lines  $ \text{Im }\theta_{kj}(z)=0  $ for $j,k=1,2,3$  (here take $z_n$ as an example). }
	\label{figR2}
\end{figure}

Next we show  that when    angle $\varphi$ is sufficiently small,
  then  $\text{Im }\theta_{jk}(z)>0$  or   $\text{Im }\theta_{jk}(z)< 0$  in  opening domains  (\ref{dom1})-(\ref{dom6}) to
  guarantee  decay/growth of  $e^{2it\theta_{jk}(z)}$  on new jump contours.

\begin{lemma}  Let $f(x)=x+ {1}/{x}$ is a  real-valued function for $x\in\mathbb{R}$,
and each $\Omega_i$ doesn't intersect critical lines  $\left\lbrace z\in\mathbb{C};\text{Im }\theta_{12}(z)=0\right\rbrace$   and    disks   $\mathbb{D}_n$ or $\overline{\mathbb{D}}_n$.

{\rm (i)}\   For $\xi>1$,  fixed the angle $\varphi$ sufficiently small and   satisfy  $ 1+ {1}/{\xi}< 2\cos2\varphi<2$, then
  there exists a positive constant $c(\xi)$ such  that
	\begin{align}
		& {\rm Im }\theta_{12}(z) \leq -c(\xi)|\sin \varphi|f(|z|),\hspace{0.5cm} \text{as }z\in\Omega_1, \Omega_2,\\
		& {\rm Im }\theta_{12}(z)\geq c(\xi) |\sin \varphi|f(|z|),\hspace{0.5cm} \text{as }z\in\Omega_3, \Omega_4.
	\end{align}

{\rm (ii)}\   For $\xi<- {1}/{8}$,  let  $\mathcal{C}_0(\xi)$ be  the solution of equation $2\sqrt{x^2-5x+4}+x-1=1/|\xi|$,
  let  $\varphi$  satisfy  $\mathcal{C}_0(\xi)< 2\cos2\varphi+3<5$. Then  there exists a positive constant $c (\xi)$  such that
	\begin{align}
		&{\rm Im }\theta_{12}(z)\geq |\sin \varphi|f(|z|)c (\xi) ,\hspace{0.5cm} \text{as }z\in\Omega_1, \Omega_2,\\
		&{\rm Im }\theta_{12}(z)\leq -|\sin \varphi|f(|z|)c (\xi) ,\hspace{0.5cm} \text{as }z\in\Omega_3, \Omega_4.
	\end{align}
\end{lemma}
\begin{proof}  We prove the case (i) and take $z\in\Omega_1$ as an example, and the other regions are similarly.
	From (\ref{Reitheta}),  for $z=le^{i\phi}$, rewrite $\text{Im }\theta_{12}(z)$ as
	\begin{align}
	{\rm Im }\theta_{12}(z)=\sqrt{3}f(l)\sin\phi\left(\xi+\frac{f(l)^3-(2\cos2\phi+3)}{f(l)^4-(2\cos2\phi+4)f(l)^2+ (2\cos2\phi+1)^2} \right).\nonumber
	\end{align}
	Denote
	\begin{align}
	h(x;a)=\frac{x-a}{x^2-x(a+1)+(a-2)^2},
	\end{align}
	with $x\geq 4$ and $1<a\leq 5$. Then $$\frac{f(l)^3-(2\cos2\phi+3)}{f(l)^4-(2\cos2\phi+4)f(l)^2+ (2\cos2\phi+1)^2}=h(f(l)^2;2\cos2\phi+3).$$  By simple calculation,
	we  have that $$h(x;a)\in\left(\frac{1}{4-a},\frac{1}{2\sqrt{a^2-5a+4}+a-1} \right). $$
By using $ 1+ {1}/{\xi}< 2\cos2\varphi<2$,  we then have
$$\xi+h(f(l)^2;2\cos2\phi+3)\geq0.$$
 Similarly from assumption above, the estimate  in the case   $\xi <- {1}/{8} $ also holds.	
	Thus the result is obtained.
\end{proof}
\begin{corollary}\label{Imtheta}
		Let   $z=le^{i\phi}=u+vi$, then
there exists a constant $c(\xi)>0$ such that

For $\xi \in(-\infty,-\frac{1}{8})$, we have
\begin{align}
	&{\rm Im }\theta(z)\geq c(\xi)v ,\hspace{0.5cm} \text{as }z\in\Omega_3, \Omega_4,\\
	&{\rm Im }\theta(z)\leq  -c(\xi)v ,\hspace{0.5cm} \text{as }z\in\Omega_1, \Omega_2.
 \end{align}

For $\xi \in(1,+\infty)$, it holds that
	 \begin{align}
	 &{\rm Im }\theta(z)\geq c(\xi)v ,\hspace{0.5cm} \text{as }z\in\Omega_1, \Omega_2 \label{2},\\
	 &{\rm Im }\theta(z)\leq  -c(\xi)v ,\hspace{0.5cm} \text{as }z\in\Omega_3, \Omega_4.
	 \end{align}
\end{corollary}

The next key  step is  to make continuous extensions to the  jump matrix  $V^{(1)}(z)$  off the
jump contours  $\mathbb{R}\cup\omega\mathbb{R}\cup\omega^2\mathbb{R}$
onto new contours along which the jump matrices are decaying. For this purpose, we
 define the   new matrix unknown functions  $R^{(2)}(z,\xi)$ as

For  $  \xi >1 $,
\begin{equation}
R^{(2)}(z,\xi)=\left\{\begin{array}{lll}
\left(\begin{array}{ccc}
1 & R_j(z,\xi)e^{it\theta_{12}} & 0\\
0 & 1 & 0\\
0 & 0 & 1
\end{array}\right), & z\in \Omega_j^0,j=1,2;\\
\\
\left(\begin{array}{ccc}
1 & 0 & 0\\
R_j(z,\xi)e^{-it\theta_{12}} & 1 & 0\\
0 & 0 & 1
\end{array}\right),  &z\in \Omega_j^0,j=3,4;\\
\\
\left(\begin{array}{ccc}
	1 & 0 &0\\
	0 & 1 & 0\\
	0 & R_j(\omega z,\xi)e^{-it\theta_{23}} & 1
\end{array}\right), & z\in \Omega^1_j,j=1,2;\\
\\
\left(\begin{array}{ccc}
	1 & 0 & 0\\
	0 & 1 & R_j(\omega z,\xi)e^{it\theta_{23}}\\
	0 & 0 & 1
\end{array}\right),  &z\in \Omega^1_j,j=3,4;\\
\\
\left(\begin{array}{ccc}
	1 & 0 & R_j(\omega^2z,\xi)e^{it\theta_{13}}\\
	0 & 1 & 0\\
	0 & 0 & 1
\end{array}\right), & z\in \Omega^2_j,j=1,2;\\
\\
\left(\begin{array}{ccc}
	1 & 0 & 0\\
	0 & 1 & 0\\
	R_j(\omega^2z,\xi)e^{-it\theta_{13}} & 0 & 1
\end{array}\right),  &z\in \Omega^2_j,j=3,4;\\
\\
I,  &elsewhere;\\
\end{array}\right.\label{R(2)-}
\end{equation}

for  $\xi <-1/8 $,
\begin{equation}
	R^{(2)}(z,\xi)=\left\{\begin{array}{lll}
		\left(\begin{array}{ccc}
			1 & R_j(z,\xi)e^{it\theta_{12}} & 0\\
			0 & 1 & 0\\
			0 & 0 & 1
		\end{array}\right), & z\in \Omega_j^0,j=3,4;\\
		\\
		\left(\begin{array}{ccc}
			1 & 0 & 0\\
			R_j(z,\xi)e^{-it\theta_{12}} & 1 & 0\\
			0 & 0 & 1
		\end{array}\right),  &z\in \Omega_j^0,j=1,2;\\
		\\
	\left(\begin{array}{ccc}
			1 & 0 &0\\
			0 & 1 & 0\\
			0 & R_j(\omega z,\xi)e^{-it\theta_{23}} & 1
		\end{array}\right), & z\in \Omega^1_j,j=3,4;\\
		\\
		\left(\begin{array}{ccc}
			1 & 0 & 0\\
			0 & 1 & R_j(\omega z,\xi)e^{it\theta_{23}}\\
			0 & 0 & 1
		\end{array}\right),  &z\in \Omega^1_j,j=1,2;\\
		\\
		\left(\begin{array}{ccc}
			1 & 0 & R_j(\omega^2z,\xi)e^{it\theta_{13}}\\
			0 & 1 & 0\\
			0 & 0 & 1
		\end{array}\right), & z\in \Omega^2_j,j=3,4;\\
		\\
		\left(\begin{array}{ccc}
			1 & 0 & 0\\
			0 & 1 & 0\\
			R_j(\omega^2z,\xi)e^{-it\theta_{13}} & 0 & 1
		\end{array}\right),  &z\in \Omega^2_j,j=1,2;\\
		\\
		I,  &elsewhere;\\
	\end{array}\right.\label{R(2)+}
\end{equation}

where  the functions $R_j$, $j=1,2,3,4$, is defined in following Proposition.
\begin{Proposition}\label{proR}
	 $R_j$: $\bar{\Omega}_j^k\to C$, $j=1,2,3,4; k=0, 1, 2$ have boundary values as follows:
(i)\ For $\xi>1$,
	\begin{align}
&R_1(z,\xi)=\Bigg\{\begin{array}{ll}
	p_1(z,\xi)T_{21}(z) & z\in \mathbb{R}^+,\\
	0  &z\in \Sigma_1,\\
\end{array} ,\hspace{0.2cm}
R_2(z,\xi)=\Bigg\{\begin{array}{ll}
	0  &z\in \Sigma_2,\\
	p_2(z,\xi)T_{21}(z) &z\in  \mathbb{R}^-,\\
\end{array}\nonumber \\
&R_3(z,\xi)=\Bigg\{\begin{array}{ll}
	p_3(z,\xi)T_{12}(z) &z\in \mathbb{R}^-, \\
	0 &z\in \Sigma_3,\\
\end{array} ,\hspace{0.2cm}
R_4(z,\xi)=\Bigg\{\begin{array}{ll}
	0  &z\in \Sigma_4,\\
	p_4(z,\xi)T_{12}(z) &z\in  \mathbb{R}^+,\\
\end{array}. \nonumber
	\end{align}	
The functions $R_j, j=1,2,3,4 $  have following property
	\begin{align}
	&|\bar{\partial}R_j(z)|\lesssim|p_j'(sign(\text{Re}z)|z|)|+|z|^{-1/2}+|\mathcal{X}'(\text{Re}z)|, \text{ for all $z\in \Omega_j^k$,}\label{dbarRj}\\
		&R_j(z)=\bar{\partial}R_j(z)=0, \text{ for } z\in \Omega_j \text{ with } |  |z|-1|\leq\varepsilon,\label{R1}\\	
	&\bar{\partial}R_j(z)=0,\hspace{0.5cm}\text{for } z\in elsewhere. \nonumber
	\end{align}
(ii)\ For $\xi  < -1/8 $,
	\begin{align}
		&R_1(z,\xi)=\Bigg\{\begin{array}{ll}
			p_1(z,\xi)[T_{12}]_+(z) & z\in \mathbb{R}^+,\\
			0  &z\in \Sigma_1,\\
		\end{array} ,\hspace{0.2cm}
		R_2(z,\xi)=\Bigg\{\begin{array}{ll}
			0  &z\in \Sigma_2,\nonumber\\
			p_2(z,\xi)[T_{12}]_+(z) &z\in  \mathbb{R}^-,\\
		\end{array} \\
		&R_3(z,\xi)=\Bigg\{\begin{array}{ll}
			p_3(z,\xi)[T_{21}]_-(z) &z\in \mathbb{R}^-, \\
			0 &z\in \Sigma_3,\\
		\end{array} ,\hspace{0.2cm}
		R_4(z,\xi)=\Bigg\{\begin{array}{ll}
			0  &z\in \Sigma_4,\\
			p_4(z,\xi)[T_{21}(z)]_- &z\in  \mathbb{R}^+,\\
		\end{array}.\nonumber
	\end{align}	
The functions $R_j, j=1,2,3,4 $  have following property
	\begin{align}
	 	&|\bar{\partial}R_j(z)|\lesssim|p_j'(sign(\text{Re}z)|z|)|+|z|^{-1}+|\mathcal{X}'(\text{Re}z)|, \text{ for all $z\in \Omega_j^k$.}\label{dbarRj2}\\
		&R_j(z)=\bar{\partial}R_j(z)=0, \text{ for } z\in \Omega_j \text{ with } |  |z|-1|\leq\varepsilon,\label{R12}\\	
	&\bar{\partial}R_j(z)=0,\hspace{0.5cm}\text{for } z\in elsewhere.\nonumber
	\end{align}
\end{Proposition}

\begin{proof}
	For brief, we
	 take $R_1(z)$ as an example. The extension of $R_1(z)$ can be constructed by:
	 \begin{align}
	 	R_{1}(z)=(1-\mathcal{X}(\text{Re}z))R_{11}(z),
	 \end{align}
 	where
	\begin{align}
		R_{11}(z)=&p_1(|z|)T_{12}(z)\cos(k_0 \arg z),\hspace{0.5cm}k_0=\frac{ \pi}{2\varphi}.
	\end{align}
	Obviously, this construction make $R_1$ admits (\ref{R1}).
	The other cases of $R_j$ are easily inferred.  Denote $z=le^{i\phi}$, then we have $\bar{\partial}=\frac{e^{i\phi}}{2}\left(\partial_r+\frac{i}{r} \partial_\phi\right) $. therefore,
	\begin{align}
	\bar{\partial}\left( (1-\mathcal{X}(l\cos\phi))R_{11}(z)\right) =&\frac{e^{i\phi}}{2}T_{12}(z)p_1'(l)(1-\mathcal{X}(l\cos\phi))\cos(k_0\phi)\nonumber\\
	&-\frac{e^{i\phi}ik_0}{2l}T_{12}(z)p_1(l)\sin(k_0\phi)\nonumber\\
	&+\frac{e^{i\phi}}{2}T_{12}(z)p_1(l)\mathcal{X}'(l\cos\phi)\left(\sin\phi-\cos(k_0\phi) \right)  .
	\end{align}
	There are two ways to bound last term. First we use Cauchy-Schwarz inequality and obtain
	\begin{equation}
	|p_1(l)|=  |p_1(l)-p_1(0)|=|\int_{0}^lp_1'(s)ds|\leq \parallel p_1'(s)\parallel_{L^2} l^{1/2}\lesssim l^{1/2}.
	\end{equation}
	
	And note that $T(z)$ is a bounded function in $\bar{\Omega}_1$. Then the boundedness of (\ref{dbarRj})  follows immediately. On the other side, $p_1(l)\in L^\infty$ ,  which implies (\ref{dbarRj2}).
\end{proof}

\subsection{Opening $\bar\partial$-lenses  in space-time   region     $- {1}/{8}<\xi<1$ }

  From (\ref{r36}),  there are four and eight  phase points
 in the space-time   regions $0< \xi <1$ and  $   - {1}/{8} <\xi<0 $ respectively,
 for which  we should   open  three   jump contours    $  \omega^n \mathbb{R}, n=0,1,2$  at phase points $ \xi_j \in \mathbb{R},  j=2,\cdots, p(\xi)-1$.
In the follow, to simplify notations, we set
$A^{  \omega^0} = A, \ \  A^{  \omega^1} = A^{  \omega },  $ for a  symbol $A$.

\noindent $\blacktriangleright$ At phase points  $ \xi_j \in  \omega^n  \mathbb{R},  j=2,\cdots, p(\xi)-1$, we define rays

For  $k=1,3$,
 $\rho\in(0,\frac{|\xi_{j+(-1)^{j+(k+1)/2}}-\xi_j|}{4\cos\varphi})$,
\begin{align}
&\Sigma_{jk}^{\omega^n} (\xi)=\left\{\begin{array}{lll}
 \xi_j+e^{i[(k/2+1/2+j)\pi+(-1)^{j+1}\varphi]}\rho,&\ 0>\xi>-1/8\\
	 \xi_j+e^{i[(k/2-1/2+j)\pi+(-1)^{j}\varphi]}\rho,&\ 0\leq\xi< 1.\end{array}\right.
\end{align}

For  $k=2,4,$  $\rho\in(0,\frac{|\xi_{j+(-1)^{j+(k+1)/2}}-\xi_j|}{4\cos\varphi})$,
\begin{align}
&\Sigma_{jk}^{\omega^n} (\xi)=\left\{\begin{array}{lll}
 \xi_j+e^{i[(k/2-1+j)\pi+(-1)^{j+1}\varphi]}\rho,&\ 0>\xi>-1/8,\\
 \xi_j+e^{i[(k/2+j)\pi+(-1)^j\varphi]}\rho,&\ 0\leq\xi< 1, \end{array}\right.
\end{align}

\noindent $\blacktriangleright$  At the phase  points $ \xi_j \in \mathbb{R},  j=1,  p(\xi) $, we define rays
\begin{align}
&\Sigma_{j1}^{\omega^n}(\xi)=\left\{\begin{array}{lll}
\omega^n\xi_j+e^{(1+j)\pi i+(-1)^{j+1}i\varphi}\mathbb{R}^+,&\ 0>\xi>-1/8\\
\omega^n\xi_j+e^{j\pi i+(-1)^ji\varphi}\rho,&\ 0\leq\xi< 1\end{array}\right.,\\
&\Sigma_{j2}^{\omega^n}(\xi)=\left\{\begin{array}{lll}
 \xi_j+e^{(1+j)\pi i+(-1)^ji\varphi}\mathbb{R}^+,&\ 0>\xi>-1/8\\
 \xi_j+e^{j\pi i+(-1)^{j+1}i\varphi }\rho,&\ 0\leq\xi< 1\end{array}\right.,\\
&\Sigma_{j3}^{\omega^n}(\xi)=\left\{\begin{array}{lll}
 \xi_j+e^{j\pi i+(-1)^{j+1}i\varphi }\rho,&\ 0>\xi>-1/8\\
 \xi_j+e^{(1+j)\pi i+(-1)^ji\varphi}\mathbb{R}^+,&\ 0\leq\xi< 1\end{array}\right.,\\
&\Sigma_{j4}^{\omega^n}(\xi)=\left\{\begin{array}{lll}
 \xi_j+e^{j\pi i+(-1)^ji\varphi}\rho,&\ 0>\xi>-1/8\\
 \xi_j+e^{(1+j)\pi i+(-1)^{j+1}i\varphi}\mathbb{R}^+,&\ 0\leq\xi< 1.\end{array}\right.
\end{align}

\noindent $\blacktriangleright$ Denote $ \xi_{(j+1)/2}=(   \xi_j+ \xi_{j+1} )/2$, we define rays

For k=1,3,   $\rho \in(0,\frac{|\xi_{j+(-1)^{j+(k+1)/2}}-\xi_j|}{4\cos\varphi})$,
\begin{align}
&\Sigma_{jk}^{'\omega^n}(\xi)=\left\{\begin{array}{lll}
 \xi_{(j+1)/2}+e^{i[(k/2-1/2+j)\pi+(-1)^{j}\varphi]}\rho,&\ 0>\xi>-1/8,\\
 \xi_{(j+1)/2}+e^{i[(k/2+1/2+j)\pi+(-1)^{j+1}\varphi]}\rho,&\ 0\leq\xi< 1,\end{array}\right.
\end{align}

For $k=2,4,$ $\rho\in(0,\frac{|\xi_{j+(-1)^{j+(k+1)/2}}-\xi_j|}{4\cos\varphi})$.
\begin{align}
	&\Sigma_{jk}^{'\omega^n}(\xi)=\left\{\begin{array}{lll}
	 \xi_{(j+1)/2}+e^{i[(k/2-1+j)\pi+(-1)^{j+1}\varphi]}\rho,&\ 0>\xi>-1/8,\\
	 \xi_{(j+1)/2}+e^{i[(k/2+j)\pi+(-1)^{j}\varphi]}\rho,&\ 0\leq\xi< 1.\end{array}\right. \hspace{0.2cm}
\end{align}
These  contours separate complex plane $\mathbb{C}$ into  sectors  denoted by $\Omega_{jk }^{\omega^n}, j= 1,\cdots, p(\xi);$ $ k=1,\cdots, 4$, which
   are shown in Figure \ref{FigOmig}.
We fix  a sufficiently small angle    $ 0< \varphi < \frac{\pi}{6} $ such that
each $\Omega_{jk }^{\omega^n}$  doesn't   intersect with critical lines  $ {\rm Im }\theta_{jk}(z)=0 $ as well as
disks  $\mathbb{D}_n$ or $\overline{\mathbb{D}}_n$.

By using  these  rays defined above, we then define  new contours   obtained when  opening   jump contours
  $\omega^n \mathbb{R}, n=0,1,2$:
\begin{align}
	&\tilde{\Sigma}^{\omega^n}(\xi)=( \underset{j=1,..,p(\xi)}{\underset{k=1,...,4,}{\cup}}\Sigma_{jk}^{\omega^n}) \cup(\underset{j=1,..,p(\xi)}{\underset{k=1,...,4,}{\cup}}\Sigma_{jk}^{'\omega^n} ),\ \ n=0,1,2\hspace{0.5cm}\\	
	&\ddot{\Sigma}(\xi)=\tilde{\Sigma}(\xi)\cup\tilde{\Sigma}^{\omega}(\xi)\cup\tilde{\Sigma}^{\omega^2}(\xi),\\
	&\Sigma^{(2)}(\xi)=\ddot{\Sigma}(\xi)\underset{n\in\mathcal{N}\setminus\lozenge}{\cup}\left( \partial\overline{\mathbb{D}}_n\cup\partial\mathbb{D}_n\right).
\end{align}
Further we also  define opened domains along  jump contours  $\omega^n \mathbb{R}, n=0,1,2$:
\begin{align}
	&\Omega(\xi)=\underset{j=1,..,p(\xi)}{\underset{k=1,...,4,}{\cup}}\left( \Omega_{jk}\cup\Omega_{jk}^\omega\cup\Omega_{jk}^{\omega^2}\right),\nonumber\\
	& \Omega^{(m)}(\xi)=S_m\setminus\Omega,\ m=1,...,6.\nonumber
\end{align}

\begin{figure}[htp]
\centering
	\subfigure[]{
		\begin{tikzpicture}
		\draw[yellow!20, fill=yellow!20] (-5,0.45)--(-5,-0.45)--(-4,0)--(-2.5,-0.6)--(-1,0)--(0,-0.5)--(1,0)--(2.5,-0.6)--(4,0)--(5,-0.5)
--(5,0.5)--(4,0)--(2.5,0.6)--(1,0)--(0,0.5)--(-1,0)--(-2.5,0.6)--(-4,0)--(-5,0.5);
\draw[blue!20, fill=blue!20] (-1,0)--(0,-0.5)--( 1,0)--(0, 0.5)--(-1,0);
%\draw[blue!20, fill=blue!20] (-4,0)--(-2.5,-0.6)--(-1,0)--(-2.5, 0.6)--(-4,0);
\draw[blue!20, fill=blue!20] ( 4,0)--( 2.5,-0.6)--( 1,0)--( 2.5, 0.6)--( 4,0);
		\draw(-4,0)--(-5,0.45)node[above]{\tiny$\Sigma_{34'}$};
		\draw[-latex](-4,0)--(-4.5,0.225);
		\draw(-4,0)--(-2.5,0.6);
		\draw[-latex ] (-2.5,-0.6)--(-3.25,-0.3) node[below]{\tiny$\Sigma_{23}$};
		\draw(-4,0)--(-5,-0.45)node[below]{\tiny$\Sigma_{33}'$};
		\draw[-latex ] (-2.5, 0.6)--(-3.25, 0.3)node[above]{\tiny$\Sigma_{24}$};
		\draw(-4,0)--(-2.5,-0.6);
		\draw[-latex](-4,0)--(-4.5,-0.225);
		\draw(-1,0)--(0,0.5);
		\draw[-latex] (0,0.5)--(-0.5,0.25) node[above]{\tiny$\Sigma_{21}$};
		\draw(-1,0)--(-2.5,0.6);
		\draw[-latex](-1,0)--(-1.75,-0.3)node[below]{\tiny$\Sigma_{23}$};
		\draw(-1,0)--(0,-0.5);
		\draw[-latex](-1,0)--(-1.75,0.3)node[above]{\tiny$\Sigma_{24}$};
		\draw(-1,0)--(-2.5,-0.6);
		\draw[-latex](0,-0.5)--(-0.5,-0.25)node[below]{\tiny$\Sigma_{22}$};
		\draw[-latex, dashed](-5.5,0)--(5.5,0)node[right]{\footnotesize Re$z$};
		\draw(1,0)--(0,0.5);
		\draw[-latex](1,0)--(0.5,0.25)node[above]{\tiny$\Sigma_{21}'$};
		\draw(1,0)--(2.5,0.6);
		\draw[-latex](2.5,-0.6)--(1.75,-0.3)node[below]{\tiny$\Sigma_{12}'$};
		\draw(1,0)--(0,-0.5);
		\draw[-latex](2.5,0.6)--(1.75,0.3)node[above]{\tiny$\Sigma_{11}'$};
		\draw(1,0)--(2.5,-0.6);
		\draw[-latex](1,0)--(0.5,-0.25)node[below]{\tiny$\Sigma_{22}'$};
		\draw(4,0)--(5,0.5)node[above]{\tiny$\Sigma_{14}$};
		\draw[-latex](5,0.5)--(4.5,0.25);
		\draw(4,0)--(2.5,0.6);
		\draw[-latex](4,0)--(3.25,-0.3)node[below]{\tiny$\Sigma_{12}$};
		\draw(4,0)--(5,-0.5)node[below]{\tiny$\Sigma_{13}$};
		\draw[-latex](4,0)--(3.25,0.3)node[above]{\tiny$\Sigma_{11}$};
		\draw(4,0)--(2.5,-0.6);
		\draw[-latex](5,-0.5)--(4.5,-0.25);
		\coordinate (A) at (-4,0);
		\fill (A) circle (1pt) node[below] {\footnotesize $0$};
		\coordinate (b) at (-1,0);
		\fill (b) circle (1pt) node[below] {\footnotesize $\xi_2$};
		\coordinate (e) at (4,0);
		\fill (e) circle (1pt) node[below] {\footnotesize $\xi_1$};
		\coordinate (f) at (1,0);
		\fill (f) circle (1pt) node[below] {\footnotesize $\xi_{\frac{3}{2}}$};
		\coordinate (ke) at (4.7,0.1);
		\fill (ke) circle (0pt) node[below] {\tiny$\Omega_{13}$};
		\coordinate (k1e) at (4.7,-0.1);
		\fill (k1e) circle (0pt) node[above] {\tiny$\Omega_{14}$};
		\coordinate (le) at (2.5,0.1);
		\fill (le) circle (0pt) node[below] {\tiny$\Omega_{12}$};
		\coordinate (l1e) at (2.5,-0.1);
		\fill (l1e) circle (0pt) node[above] {\tiny$\Omega_{11}$};
		\coordinate (k) at (-4.7,0.1);
		\fill (k) circle (0pt) node[below] {\tiny$\Omega_{33}$};
		\coordinate (k1) at (-4.7,-0.1);
		\fill (k1) circle (0pt) node[above] {\tiny$\Omega_{34}$};
			\coordinate (n) at (0,0.1);
		\fill (n) circle (0pt) node[below] {\tiny$\Omega_{22}$};
		\coordinate (n1) at (0,-0.1);
		\fill (n1) circle (0pt) node[above] {\tiny$\Omega_{21}$};
		\coordinate (m) at (-2.5,0.1);
		\fill (m) circle (0pt) node[below] {\tiny$\Omega_{23}$};
		\coordinate (m1) at (-2.5,-0.1);
		\fill (m1) circle (0pt) node[above] {\tiny$\Omega_{24}$};
		\end{tikzpicture}
		\label{case1}}
	\subfigure[]{
		\begin{tikzpicture}
		\draw[yellow, fill=yellow!20](-6.5,0.9)--(-6.5,-0.9)--(-5.4,0)--(-4.7,-0.6)--(-4,0)--(-3.1,-0.7)--(-2.2,0)--(-1.5,-0.6)--
(-0.8,0)--(-0,-0.7)--(0.8,0)--(1.5,-0.6)--(2.2,0)--(3.1,-0.7)--(4,0)--(4.7,-0.6)--(5.4,0)--(6.5,-0.9)--(6.5,0.9)--(5.4,0)--(4.7,0.6)--
(4,0)--(3.1,0.7)--(2.2,0)--(1.5,0.6)--(0.8,0)--(-0,0.7)--(-0.8,0)--(-1.5,0.6)--(-2.2,0)--(-3.1,0.7)--(-4,0)--(-4.7,0.6)--(-5.4,0)--(-6.5,0.9);
\draw[blue!20, fill=blue!20](-6.5,0.9)--(-6.5,-0.9)--(-5.4,0)--(-6.5,0.9);
 \draw[blue!20, fill=blue!20] (-5.4,0)--(-4.7,-0.6)--(-4,0)--(-4.7, 0.6)--(-5.4,0);
\draw[blue!20, fill=blue!20](-0.8,0)--(-0,-0.7)--(0.8,0)--(-0, 0.7)--(-0.8,0);
\draw[blue!20, fill=blue!20]( 6.5,0.9)--( 6.5,-0.9)--( 5.4,0)--( 6.5,0.9);
 \draw[blue!20, fill=blue!20](0.8,0)--(1.5,-0.6)--(2.2,0)--(1.5, 0.6)--(0.8,0);
		\draw[-latex, dashed](-7,0)--(7,0)node[right]{\footnotesize Re$z$};
		\draw(-0.8,0)--(-0,0.7);
		\draw[-latex](-0,0.7)--(-0.4,0.35)node[above]{\tiny$\Sigma_{31}$};
		\draw(-0.8,0)--(-1.5,0.6);
		\draw[-latex](-0.8,0)--(-1.15,-0.3)node[below]{\tiny$\Sigma_{33}$};
		\draw(-0.8,0)--(-0,-0.7);
		\draw[-latex](-0.8,0)--(-1.15,0.3)node[above]{\tiny$\Sigma_{34}$};
		\draw(-0.8,0)--(-1.5,-0.6);
		\draw[-latex](0,-0.7)--(-0.4,-0.35)node[below]{\tiny$\Sigma_{32}$};
		\draw(-2.2,0)--(-1.5,0.6);
		\draw[-latex](-2.2,0)--(-2.65,0.35)node[above]{\tiny$\Sigma_{44}'$};
		\draw(-2.2,0)--(-1.5,-0.6);
		\draw[-latex](-1.5,-0.6)--(-1.85,-0.3)node[below]{\tiny$\Sigma_{33}'$};
		\draw(-2.2,0)--(-3.1,0.7);
		\draw[-latex](-1.5, 0.6)--(-1.85, 0.3)node[above]{\tiny$\Sigma_{34}'$};
		\draw(-2.2,0)--(-3.1,-0.7);
		\draw[-latex](-2.2,0)--(-2.65,-0.35)node[below]{\tiny$\Sigma_{43}'$};
		\draw(-5.4,0)--(-6.5,0.9)node[above]{\tiny$\Sigma_{51}'$};
		\draw[-latex](-5.4,0)--(-5.95,0.45);
		\draw(-5.4,0)--(-4.7,0.6);
		\draw[-latex](-4.7,-0.6)--(-5.05,-0.3)node[below]{\tiny$\Sigma_{42}'$};
		\draw(-5.4,0)--(-6.5,-0.9)node[below]{\tiny$\Sigma_{52}'$};
		\draw[-latex](-4.7, 0.6)--(-5.05,0.3)node[above]{\tiny$\Sigma_{41}'$};
		\draw(-5.4,0)--(-4.7,-0.6);
		\draw[-latex](-5.4,0)--(-5.95,-0.45);
		\draw(-4,0)--(-3.1,0.7);
		\draw[-latex ](-3.1, 0.7)--(-3.55,0.35)node[above]{\tiny$\Sigma_{44}$};
		\draw(-4,0)--(-4.7,0.6);
		\draw[-latex](-4,0)--(-4.35,-0.3)node[below]{\tiny$\Sigma_{42}$};
		\draw(-4,0)--(-3.1,-0.7);
		\draw[-latex](-4,0)--(-4.35,0.3)node[above]{\tiny$\Sigma_{41}$};
		\draw(-4,0)--(-4.7,-0.6);
		\draw[-latex ](-3.1, -0.7)--(-3.55,-0.35)node[below]{\tiny$\Sigma_{43}$};
		\draw(0.8,0)--(0,0.7);
		\draw[-latex](0.8,0)--(0.4,0.35)node[above]{\tiny$\Sigma_{31}'$};
		\draw(0.8,0)--(1.5,0.6);
		\draw[-latex](1.5,-0.6)--(1.15,-0.3)node[below]{\tiny$\Sigma_{22}'$};
		\draw(0.8,0)--(0,-0.7);
		\draw[-latex](1.5, 0.6)--(1.15,0.3)node[above]{\tiny$\Sigma_{21}'$};
		\draw(0.8,0)--(1.5,-0.6);
		\draw[-latex](0.8,0)--(0.4,-0.35)node[below]{\tiny$\Sigma_{32}'$};
		\draw(2.2,0)--(1.5,0.6);
		\draw[-latex](3.1, 0.7)--(2.65,0.35)node[above]{\tiny$\Sigma_{24}$};
		\draw(2.2,0)--(1.5,-0.6);
		\draw[-latex](2.2,0)--(1.85,-0.3)node[below]{\tiny$\Sigma_{22}$};
		\draw(2.2,0)--(3.1,0.7);
		\draw[-latex](2.2,0)--(1.85,0.3)node[above]{\tiny$\Sigma_{21}$};
		\draw(2.2,0)--(3.1,-0.7);
		\draw[-latex](3.1, -0.7)--(2.65,-0.35)node[below]{\tiny$\Sigma_{23}$};
		\draw(5.4,0)--(6.5,0.9)node[above]{\tiny$\Sigma_{11}$};
		\draw[-latex](6.5, 0.9)--(5.95,0.45);
		\draw(5.4,0)--(4.7,0.6);
		\draw[-latex](5.4,0)--(5.05,-0.3)node[below]{\tiny$\Sigma_{13}$};
		\draw(5.4,0)--(6.5,-0.9)node[below]{\tiny$\Sigma_{12}$};
		\draw[-latex](5.4,0)--(5.05,0.3)node[above]{\tiny$\Sigma_{14}$};
		\draw(5.4,0)--(4.7,-0.6);
		\draw[-latex](6.5, -0.9)--(5.95,-0.45);
		\draw(4,0)--(3.1,0.7);
		\draw[-latex](4,0)--(3.55,0.35)node[above]{\tiny$\Sigma_{24}'$};
		\draw(4,0)--(4.7,0.6);
		\draw[-latex](4.7,-0.6)--(4.35,-0.3)node[below]{\tiny$\Sigma_{13}'$};
		\draw(4,0)--(3.1,-0.7);
		\draw[-latex](4.7, 0.6)--(4.35,0.3)node[above]{\tiny$\Sigma_{14}'$};
		\draw(4,0)--(4.7,-0.6);
		\draw[-latex](4,0)--(3.55,-0.35)node[below]{\tiny$\Sigma_{23}'$};
		\coordinate (A) at (-5.4,0);
		\fill (A) circle (1pt) node[below] {\footnotesize $0$};
		\coordinate (b) at (-4,0);
		\fill (b) circle (1pt) node[below] {\footnotesize $\xi_4$};
		\coordinate (C) at (-0.8,0);
		\fill (C) circle (1pt) node[below] {\footnotesize $\xi_3$};
		\coordinate (d) at (-2.2,0);
		\fill (d) circle (1pt) node[below] {\footnotesize $\xi_{\frac{7}{2}}$};
		\coordinate (E) at (5.4,0);
		\fill (E) circle (1pt) node[below] {\footnotesize $\xi_1$};
		\coordinate (R) at (4,0);
		\fill (R) circle (1pt) node[below] {\footnotesize $\xi_{\frac{3}{2}}$};
		\coordinate (T) at (0.8,0);
		\fill (T) circle (1pt) node[below] {\footnotesize $\xi_{\frac{5}{2}}$};
		\coordinate (Y) at (2.2,0);
		\fill (Y) circle (1pt) node[below] {\footnotesize $\xi_2$};
		\coordinate (q) at (6.2,-0.1);
		\fill (q) circle (0pt) node[above] {\tiny$\Omega_{11}$};
		\coordinate (q1) at (6.2,0.05);
		\fill (q1) circle (0pt) node[below] {\tiny$\Omega_{12}$};
		\coordinate (w) at (4.75,-0.1);
		\fill (w) circle (0pt) node[above] {\tiny$\Omega_{14}$};
		\coordinate (w1) at (4.75,0.1);
		\fill (w1) circle (0pt) node[below] {\tiny$\Omega_{13}$};
		\coordinate (t) at (3,-0.1);
		\fill (t) circle (0pt) node[above] {\tiny$\Omega_{24}$};
		\coordinate (t1) at (3,0.05);
		\fill (t1) circle (0pt) node[below] {\tiny$\Omega_{23}$};
		\coordinate (y) at (1.55,-0.1);
		\fill (y) circle (0pt) node[above] {\tiny$\Omega_{21}$};
		\coordinate (y1) at (1.55,0.1);
		\fill (y1) circle (0pt) node[below] {\tiny$\Omega_{22}$};
			\coordinate (q8) at (-6.2,-0.1);
		\fill (q8) circle (0pt) node[above] {\tiny$\Omega_{51}$};
		\coordinate (q18) at (-6.2,0.05);
		\fill (q18) circle (0pt) node[below] {\tiny$\Omega_{52}$};
		\coordinate (e7) at (-4.65,-0.1);
		\fill (e7) circle (0pt) node[above] {\tiny$\Omega_{41}$};
		\coordinate (e17) at (-4.65,0.1);
		\fill (e17) circle (0pt) node[below] {\tiny$\Omega_{42}$};
		\coordinate (7r) at (-3.14,-0.1);
		\fill (7r) circle (0pt) node[above] {\tiny$\Omega_{44}$};
		\coordinate (r17) at (-3.14,0.05);
		\fill (r17) circle (0pt) node[below] {\tiny$\Omega_{43}$};
		\coordinate (l5) at (-1.43,-0.1);
		\fill (l5) circle (0pt) node[above] {\tiny$\Omega_{34}$};
		\coordinate (l15) at (-1.43,0.1);
		\fill (l15) circle (0pt) node[below] {\tiny$\Omega_{33}$};
		\coordinate (k5) at (-0.02,-0.1);
		\fill (k5) circle (0pt) node[above] {\tiny$\Omega_{31}$};
		\coordinate (k15) at (-0.02,0.1);
		\fill (k15) circle (0pt) node[below] {\tiny$\Omega_{32}$};
		\end{tikzpicture}
		\label{case2}}
	\caption{\footnotesize Figures (a) and (b) are corresponding to the  $0\leq\xi<1$ and  $- {1}/{8}<\xi<0$, respectively.
As illustrative example, here  we only give the graph as Re$z>0$. $\Sigma_{ij}$ separate
         complex plane $\mathbb{C}$ into some  sectors denoted by $\Omega_{ij}$. The yellow domains are produced  by the  first  factorizations of jump matrices;
  The yellow domains are produced  by the  second  factorizations of jump matrices}
	\label{FigOmig}
\end{figure}

\begin{lemma}\label{theta2}
 There exists a constant $c(\xi)>0$ relied on $\xi= \in(-1/8,1)$ and a positive real-value function $g(x;\xi)$ with $\lim_{x\to \infty}g(x;\xi)=0$ and $\lim_{x\to \infty}(\text{Re}^2z-\xi_i^2)g(x;\xi)=g_0\in(0,1)$. Then the imaginary part of phase function (\ref{Reitheta}) ${\rm Im }\theta_{12}(z)$ have following estimation for $i=1,...,p(\xi)$:
	\begin{align}
	&{\rm Im }\theta_{12}(z)\geq c(\xi){\rm Im }z({\rm  Re }^2z-\xi_i^2)g({\rm  Re }z;\xi),\hspace{0.5cm} \text{as }z\in\Omega_{i2}, \Omega_{i4};\\
	&{\rm Im }\theta_{12}(z)\leq -c(\xi){\rm Im }z({\rm  Re }^2z-\xi_i^2)g({\rm  Re }z;\xi),\hspace{0.5cm} \text{as }z\in\Omega_{i1}, \Omega_{i3}.
	\end{align}	
\end{lemma}
\begin{proof}
	We only give the detail of Case III  and take $z\in\Omega_{14}$ as an example, and the other regions are similarly.
	Denote $z=x+yi$ with $x,y\in \mathbb{R}$ and
	\begin{align}
	k=-\frac{1}{4}\left( z-\frac{1}{z}\right) ,\hspace{0.5cm}	k_1=-\frac{1}{4}\left( \xi_1-\frac{1}{\xi_1}\right).
	\end{align}
	Take notice  that $\xi_1>1$, so $0<k_1<1$.
	 And denote
	\begin{align}
	&u\triangleq\text{Re}(k)=x\left( 1-\frac{1}{|z|^2}\right) ,\\
	&v\triangleq\text{Im}(k)=y\left( 1+\frac{1}{|z|^2}\right),
	\end{align}
	with $0<v\leq(x-\xi_1)\tan\varphi\left(1+\frac{1}{x^2+\tan^2\varphi(x-\xi_1)^2} \right)  $ and $u>1$.
	Then  the imaginary part of phase function (\ref{Reitheta}) ${\rm Im }\theta(z)$ can be rewritten as
	\begin{align}
	{\rm Im }\theta_{12}(z)=\sqrt{3}v\left[ \xi+\frac{u^2+v^2-1}{u^4+v^4+2u^2+1-2v^2+2u^2v^2}\right] .
	\end{align}
	Denote
	\begin{align}
		f_0(u,v)\triangleq\frac{u^2+v^2-1}{u^4+v^4+2u^2+1-2v^2+2u^2v^2}>0.
	\end{align}
	Our first goal is to show ${\rm Im }\theta_{12}(z)\geq c(\xi+f(u,0))$, with a positive constant $c$. Note that $u^2+v^2>1$,
	\begin{align}
		f_0(u,v)=\frac{1}{u^2+v^2-1+\frac{4u^2}{u^2+v^2-1}}.
	\end{align}
	therefore, $f_0(u,v)$ has  bounded nonzero maximum and  minimum value, then there exists a positive constant $c'$ making that $f_0(u,v)\geq c' f_0(u,0)$. Next, we rewrite
	\begin{align*}
		\xi=\frac{1-k_1^2}{(k_1^2+1)^2},
	\end{align*}
	 then together with $v=0\Leftrightarrow y=0$,
	 \begin{align*}
	 	\xi+f(u,0)&=\frac{u^2-1}{(u^2+1)^2}-\frac{k_1^2-1}{(k_1^2+1)^2}\\
	 	&=\dfrac{u^2-k_1^2}{(k_1^2+1)^2(u^2+1)^2}\left( 3+k_1^2+(1-k_1^2)u^2\right) \\
	 	&=(x^2-\xi_1^2)(1-\dfrac{1}{x^2\xi_1^2})\dfrac{1+3k_1^2+(1-k_1^2)(x^2+x^{-2})}{(k_1^2+1)^2(x^2+x^{-2}-1)^2}.
	 \end{align*}
	 Let $g(x;\xi)=(1-\dfrac{1}{x^2\xi_1^2})\dfrac{1+3k_1^2+(1-k_1^2)(x^2+x^{-2})}{(k_1^2+1)^2(x^2+x^{-2}-1)^2}>0$. Then the proof is completed.
\end{proof}
\begin{corollary}\label{theta2c}
	There exist three positive constants $c_1(\xi)$, $c_2(\xi)$ and a    large  $R(\xi)>>\xi_1$  relied on $\xi\in(-1/8,1)$, such that the imaginary part of phase function (\ref{Reitheta}) ${\rm Im }\theta_{12}(z)$ have following estimation for $i=1,...,p(\xi)$:
	\begin{align}
		&{\rm Im }\theta_{12}(z)\geq c_1(\xi){\rm Im }z({\rm Re } z-\xi_i) ,\hspace{0.5cm}  z\in\Omega_{i2}, \Omega_{i4},\ |{\rm Re } z|<R(\xi);\\
		&{\rm Im }\theta_{12}(z)\leq -c_1(\xi){\rm Im }z({\rm Re } z-\xi_i)  ,\hspace{0.5cm}  z\in\Omega_{i1}, \Omega_{i3},\ |{\rm Re } z|<R(\xi);
	\end{align}	
and
\begin{align}
	&{\rm Im }\theta_{12}(z)\geq c_2(\xi){\rm Im }z ,\hspace{0.5cm} \text{as }z\in\Omega_{i2}, \Omega_{i4},\ |{\rm Re }z|>R(\xi);\\
	&{\rm Im }\theta_{12}(z)\leq -c_2(\xi){\rm Im }z  ,\hspace{0.5cm} \text{as }z\in\Omega_{i1}, \Omega_{i3},\ |{\rm Re }z|>R(\xi).
\end{align}	
\end{corollary}

Introduce following functions for brief:\\
for $\xi>1$
\begin{align}
	&p_1(z,\xi)=p_2(z,\xi)=r(z),\ p_3(z,\xi)=p_4(z,\xi)=\bar{r}(z).
\end{align}
for $\xi<-1/8$
\begin{align}
	&p_1(z,\xi)=p_2(z,\xi)=\dfrac{-\bar{r}(z) }{1-|r(z)|^2},\ p_3(z,\xi)=p_4(z,\xi)=\dfrac{-r(z) }{1-|r(z)|^2}.
\end{align}
As in Case II and Case III, for $j=1,...,p(\xi)$,
\begin{align}
&p_{j1}(z,\xi)=\dfrac{-\bar{r}(z) }{1-|r(z)|^2},\hspace{0.5cm}p_{j3}(z,\xi)=\bar{r}(z),\\
&p_{j2}(z,\xi)=\dfrac{-r(z) }{1-|r(z)|^2},\hspace{0.6cm}p_{j4}(z,\xi)=r(z).
\end{align}
Besides, from $r\in \mathbb{R}(\mathbb{R})$, it also has  $p_1'(z)$ and $p_2'(z)$ exist and are in $L^2(\mathbb{R})\cup L^\infty(\mathbb{R})$. And their $L^2$-norm and $L^\infty$-norm  can be controlled by the norm of $r(z)$.
By symmetry, denote
\begin{align}
	&p_1^\omega(z,\xi)=p_2^\omega(z,\xi)=r(\omega^2z),\ p_3^\omega(z,\xi)=p_4^\omega(z,\xi)=\bar{r}(\omega^2z),\\
	&p_1^{\omega^2}(z,\xi)=p_2^{\omega^2}(z,\xi)=r(\omega z),\ p_3^{\omega^2}(z,\xi)=p_4^{\omega^2}(z,\xi)=\bar{r}(\omega z),
\end{align}
for $\xi>1$. And
\begin{align}
	&p_{j1}^\omega(z,\xi)=\dfrac{-\bar{r}(\omega^2z) }{1-|r(\omega^2z)|^2},\hspace{0.5cm}p_{j3}^\omega(z,\xi)=\bar{r}(\omega^2z),\\
	&p_{j2}^\omega(z,\xi)=\dfrac{-r(\omega^2z) }{1-|r(\omega^2z)|^2},\hspace{0.6cm}p_{j4}^\omega(z,\xi)=r(\omega^2z),\\
	&p_{j1}^{\omega^2}(z,\xi)=\dfrac{-\bar{r}(\omega z) }{1-|r(\omega z)|^2},\hspace{0.5cm}p_{j3}^{\omega^2}(z,\xi)=\bar{r}(\omega z),\\
	&p_{j2}^{\omega^2}(z,\xi)=\dfrac{-r(\omega z) }{1-|r(\omega z)|^2},\hspace{0.6cm}p_{j4}^{\omega^2}(z,\xi)=r(\omega z),
\end{align}
for $1>\xi>-\frac{1}{8}$ and $j=1,...,p(\xi)$.

The next step is to construct a matrix function $R^{(2)}$.
Usually, we need to remove jump on $\mathbb{R}$,  $\omega\mathbb{R}$ and $\omega^2\mathbb{R}$, and  have some mild control on $\bar{\partial}R^{(2)}$ sufficiently to ensure that the $\bar{\partial}$-contribution to the long-time asymptotics of $q(x, t)$ is negligible like in \cite{fNLS}. But we have extra singularity on the  boundary.
Hence, to deal with the singularity at $\varkappa_k$, $k=1,...,6$, we introduce a fixed cutoff function $\mathcal{X}(z)$   in $C^\infty_0(\mathbb{R},[0,1])$ with  support near $1$ with
\begin{align}
	\mathcal{X}(z)=\left\{\begin{array}{llll}
		0, & |z-1|>2\varepsilon,\\[4pt]
		1,  & |z-1|<\varepsilon,
	\end{array}\right.
\end{align}
where  $\varepsilon $ is a small enough positive constant satisfying
 the   support of $\mathcal{X}(z)$ doesn't contain any of phase points with
$$\varepsilon <\frac{1}{8}\underset{1\leq j\neq k\leq p(\xi)}{\min}|\xi_k-\xi_j|,$$
and the  support of $\mathcal{X}(\text{Im}z+1)\mathcal{X}(\text{Re}z)$  is disjoint with $\ddot{\Sigma}(\xi)$.
Such constant $\varepsilon$  indeed  exists.

In fact, at the case $\xi\in (-\infty,-1/8)\cup(1,+\infty)$ without no phase point, so these  requirements are easy to achieve. At the case $\xi\in(-1/8,0)$, there are two
 phase points $\xi_2$, $\xi_3=1/\xi_2$ near $1$. So $1\neq\frac{\xi_2+\xi_3}{2}$. Take $\varepsilon $ small enough, then it achieve above requirements. And the case $\xi\in[0,1) $ is same as $\xi\in(-1/8,0)$.

We now define continuous extension functions for the  Case II  and Case III,
\begin{equation}
R^{(2)}(z,\xi)=\left\{\begin{array}{lll}
\left(\begin{array}{ccc}
	1 & R_{kj}(z,\xi)e^{it\theta_{12}} & 0\\
	0 & 1 & 0\\
	0 & 0 & 1
\end{array}\right), & z\in \Omega_{kj},j=2,4,\ k=1,...,p(\xi);\\
\\
\left(\begin{array}{ccc}
	1 & 0 & 0\\
	R_{kj}(z,\xi)e^{-it\theta_{12}} & 1 & 0\\
	0 & 0 & 1
\end{array}\right),  &z\in \Omega_{kj},j=1,3\ k=1,...,p(\xi);\\
\\
\left(\begin{array}{ccc}
	1 & 0 & R_{kj}(\omega^2z,\xi)e^{it\theta_{13}}\\
	0 & 1 & 0\\
	0 & 0 & 1
\end{array}\right), & z\in \Omega^\omega_j,j=2,4\ k=1,...,p(\xi);\\
\\
\left(\begin{array}{ccc}
	1 & 0 & 0\\
	0 & 1 & 0\\
	R_{kj}(\omega^2z,\xi)e^{-it\theta_{13}} & 0 & 1
\end{array}\right),  &z\in \Omega^\omega_j,j=1,3\ k=1,...,p(\xi);\\
\\
\left(\begin{array}{ccc}
	1 & 0 &0\\
	0 & 1 & 0\\
	0 & R_{kj}(\omega z,\xi)e^{-it\theta_{23}} & 1
\end{array}\right), & z\in \Omega^{\omega^2}_j,j=2,4\ k=1,...,p(\xi);\\
\\
\left(\begin{array}{ccc}
	1 & 0 & 0\\
	0 & 1 & R_{kj}(\omega z,\xi)e^{it\theta_{23}}\\
	0 & 0 & 1
\end{array}\right),  &z\in \Omega^{\omega^2}_j,j=1,3\ k=1,...,p(\xi);\\
\\
I,  &elsewhere;\\
\end{array}\right.\label{R(2)1}
\end{equation}
where  the functions $R_{kj}$, $k=1,...,p(\xi)$, $j=1,2,3,4$  are given by  following Proposition.
\begin{Proposition}\label{proR1}
	As in Case II and Case III,  the functions $R_{kj}$: $\bar{\Omega}_{kj}\to \mathbb{C}$,$k=1,...,p(\xi)$, $j=1,2,3,4$  have boundary values as follows:
	\begin{align}
	&R_{k1}(z,\xi)=\Bigg\{\begin{array}{ll}
		p_{k1}(\xi_k,\xi)T_{12}^{(k)}(\xi)\left( \eta(z-\xi_k)\right) ^{2i\eta\nu(\xi_k)}  &z\in \Sigma_{k1},\\
		p_{k1}(z,\xi)[T_{12}]_+(z) &z\in I_{k1},\\
	\end{array} \\
	&R_{k2}(z,\xi)=\Bigg\{\begin{array}{ll}
	p_{k2}(z,\xi)[T_{12}]_-(z)^{-1} &z\in I_{k2}, \\
	p_{k2}(\xi_k,\xi)T_{12}^{(k)}(\xi)^{-1}\left( \eta(z-\xi_k)\right)^{-2i\eta\nu(\xi_k)} &z\in \Sigma_{k2},\\
	\end{array} \\
&R_{k3}(z,\xi)=\Bigg\{\begin{array}{ll}
p_{k3}(\xi_k,\xi)T_{12}^{(k)}(\xi)\left( \eta(z-\xi_k)\right)^{2i\eta\nu(\xi_k)}  &z\in \Sigma_{k3},\\
p_{k3}(z,\xi)T_{12}(z) &z\in  I_{k3},\\
\end{array} \\
	&R_{k4}(z,\xi)=\Bigg\{\begin{array}{ll}
	p_{k4}(z,\xi)T_{12}(z)^{-1} & z\in I_{k4},\\
	p_{k4}(\xi_k,\xi)T_{12}^{(k)}(\xi)^{-1}\left( \eta(z-\xi_k)\right)^{-2i\eta\nu(\xi_k)}  &z\in \Sigma_{k4},
\end{array}
	\end{align}	
	where $I_{kj}$ is specified in (\ref{In1})-(\ref{In2}),  $\eta=\eta(\xi,k)$ defined in (\ref{eta}). And $R_{kj}$  has following property:
	\begin{align}
	&|R_{kj}(z,\xi)|\lesssim \sin^2(k_0\arg(z-\xi_k))+ \left(1+ \text{Re}(z)^2\right) ^{-1/2}, \text{for all $z\in \Omega_{kj}$},\label{R}\\
	&|\bar{\partial}R_{kj}(z,\xi)|\lesssim|p_{kj}'(\text{Re}z)|+|\mathcal{X}'(\text{Re}z)|+|z-\xi_k|^{-1/2}, \text{for all $z\in \Omega_{kj}$.}\label{dbarRj3}
	\\
	&|\bar{\partial}R_{kj}(z,\xi)|\lesssim|p_{kj}'(\text{Re}z)|+|\mathcal{X}'(\text{Re}z)|+|z-\xi_k|^{-1}, \text{for all $z\in \Omega_{kj}$.}\label{dbarRj4}
	\end{align}
	And
	\begin{align}
	&R_{kj}(z)=\bar{\partial}R_{kj}(z)=0, \text{ for all }z\in\Omega_{kj} \text{ with } |\text{Re}z\pm1|<\varepsilon,\label{R11}\\	
	&\bar{\partial}R_{kj}(z,\xi)=0,\hspace{0.5cm}\text{if } z\in elsewhere.
	\end{align}
\end{Proposition}
\begin{proof}
	For the region $\omega_{kj}$ containing $\pm1$, we construct the function $R_{kj}$ by $R_{kj}(z)=(1-\mathcal{X}(\text{Re}z))\tilde{R}_{kj}(z)$ as same as the case in Proposition \ref{proR}. Through multiply by the cutoff function $(1-\mathcal{X}(\text{Re}z))$, we achieve (\ref{R11}). And for the regions away from 1,
	we give the  details for $R_{14}$ when $-\frac{1}{8}<\xi<0$ only. The other cases are easily inferred. Note that $1\notin\Omega_{14}$, so $1-\mathcal{X}(\text{Re}z)\equiv1$ in $\Omega_{14}$. Using the constants $T_{12}^{(k)}(\xi)$ defined in proposition \ref{proT}, we give the extension of $R_{14}(z,\xi)$ on $\Omega_{14}$:
	\begin{align}
	R_{14}(z,\xi)=&p_{14}(\xi_1,\xi)T_{12}^{(1)}(\xi)^{-2}(z-\xi_1)^{-2i\nu(\xi_1)}\left[1-\cos\big(k_0\arg(z-\xi_1) \big)\right] \nonumber\\
	&+\cos\big(k_0\arg(z-\xi_1)\big)p_{14}(\text{Re}z,\xi)T(z)^{-2}.
\end{align}
	Let $z-\xi_1=le^{i\psi}=u+vi$, $l,\psi,u,v\in\mathbb{R}$. And from $r\in H^{2,1}(\mathbb{R})$, which means $p_{14}\in H^{2,1}(R)$ we have $|p_{14}(u)|\lesssim (1+u^2)^{-1/2}$. Together with (\ref{key}) we have (\ref{R}). Since
	\begin{equation*}
	\bar{\partial}=\frac{1}{2}\left( \partial_u+i\partial_v\right) =\frac{e^{i\psi}}{2}\left( \partial_l+il^{-1}\partial_\psi\right),
	\end{equation*}
	we have
	\begin{align}
	\bar{\partial}R_{14}^{(1)}&=\left(p_{14}(u,\xi)T(z)^{-2}-p_{14}(\xi_1,\xi)T_1(\xi)^{-2}(z-\xi_1)^{-2i\nu(\xi_1)} \right)\bar{\partial}\cos (k_0\psi)\nonumber\\
	& +\frac{1}{2}T(z)^{-2}p_{14}'(u,\xi)\cos (k_0\psi).
	\end{align}
	Substitute (\ref{T-TJ}) into above equation, and using
	\begin{align}
	|p_{14}(u,\xi)-p_{14}(\xi_1,\xi)|=|\int_{\xi_1}^{u}	p_{14}'(s,\xi)ds|\leq\parallel p_{14}'\parallel_2|u-\xi|^{1/2},
	\end{align}
(\ref{dbarRj3}) comes immediately. If we simply use the boundedness of $p_{14}$, we obtain (\ref{dbarRj4}).
\end{proof}
In addition, $R^{(2)}$ achieves the symmetry:
\begin{equation}
	R^{(2)}(z)=\Gamma_1\overline{R^{(2)}(\bar{z})}\Gamma_1=\Gamma_2\overline{R^{(2)}(\omega^2\bar{z})}\Gamma_2=\Gamma_3\overline{R^{(2)}(\omega\bar{z})}\Gamma_3=\overline{R^{(2)}(\bar{z}^{-1})}.
\end{equation}

\subsection{A hybrid $\bar{\partial}$-RH problem and its decomposition  }\label{sec5}

We now  use $R^{(2)}$ to define the new transformation \begin{equation}
	M^{(2)}(z;y,t)\triangleq M^{(2)}(z)=M^{(1)}(z)R^{(2)}(z)\label{transm2},
\end{equation}
which satisfies the following mixed $\bar{\partial}$-RH problem.

\begin{RHP}\label{RHP4}
Find a matrix valued function  $ M^{(2)}(z)$ with following properties:

$\blacktriangleright$ Analyticity:  $M^{(2)}(z)$ is continuous and has  sectionally continuous first partial derivatives in
$\mathbb{C}\setminus \left( \Sigma^{(2)}(\xi)\cup \left\lbrace\zeta_n,\bar{\zeta}_n \right\rbrace_{n\in\lozenge} \right) $,  and is meromorphic out $\bar{\Omega}$;

$\blacktriangleright$ Symmetry: $M^{(2)}(z)=\Gamma_1\overline{M^{(2)}(\bar{z})}\Gamma_1=\Gamma_2\overline{M^{(2)}(\omega^2\bar{z})}\Gamma_2=\Gamma_3\overline{M^{(2)}(\omega\bar{z})}\Gamma_3$ and $M^{(2)}(z)=\overline{M^{(2)}(\bar{z}^{-1})}$;

$\blacktriangleright$ Jump condition: $M^{(2)}$ has continuous boundary values $M^{(2)}_\pm$ on $\Sigma^{(2)}(\xi)$ and
\begin{equation}
	M^{(2)}_+(z)=M^{(2)}_-(z)V^{(2)}(z),\hspace{0.5cm}z \in \Sigma^{(2)}(\xi),
\end{equation}
where for  $\xi \in(1,+\infty)$
\footnotesize\begin{equation}
	V^{(2)}(z)=\left\{ \begin{array}{ll}
		T^{-1}(z)G(z)T(z),   & 	z\in\partial\mathbb{D}_n\cap (\underset{k=1}{\cup^3}S_{2k});\\[12pt]
		T^{-1}(z)G^{-1}(z)T(z),    & z\in\partial\mathbb{D}_n\cap (\underset{k=1}{\cup^3}S_{2k-1});\\[12pt]
			\left(\begin{array}{ccc}
			1 & \frac{rT_{21}\mathcal{X}(|z|) e^{it\theta_{12}}}{1-|r|^2}&0\\
			0 & 1&0\\
			0 & 0 &1
		\end{array}\right)\left(\begin{array}{ccc}
			1 & 0 & 0\\
			\frac{\bar{r}T_{12}\mathcal{X}(|z|) e^{-it\theta_{12}}}{1-|r|^2} & 1 & 0\\
			0 & 0 &1
		\end{array}\right),   & z\in \mathbb{R};\\[12pt]
		\left(\begin{array}{ccc}
			1 & 0&0\\
			0 & 1&\frac{-r(\omega z)T_{32} e^{it\theta_{23}}\mathcal{X}(|z|)}{1-|r(\omega z)|^2}\\
			0 & 0 &1
		\end{array}\right)\left(\begin{array}{ccc}
			1 & 0 & 0\\
			0 & 1 & 0\\
			0 & 	\frac{\bar{r}(\omega z) e^{-it\theta_{23}}T_{23}\mathcal{X}(|z|)}{1-|r(\omega z)|^2} &1
		\end{array}\right),   & z\in \omega^2\mathbb{R};\\[12pt]
		\left(\begin{array}{ccc}
			1 & 0&0\\
			0 & 1&0\\
			\frac{-r(\omega^2 z)T_{13}\mathcal{X}(|z|) e^{-it\theta_{13}}}{1-|r(\omega^2 z)|^2} & 0 &1
		\end{array}\right)\left(\begin{array}{ccc}
			1 & 0 & 	\frac{\bar{r}(\omega^2 z)T_{31}\mathcal{X}(|z|) e^{it\theta_{13}}}{1-|r(\omega^2 z)|^2}\\
			0 & 1 & 0\\
			0 & 0 &1
		\end{array}\right),   & z\in \omega\mathbb{R};\\
	\end{array}\right.,\label{jumpv2}
\end{equation}	\normalsize
and for $\xi \in(-1/8,1)$
\footnotesize\begin{equation}
V^{(2)}(z)=\left\{
\begin{array}{ll}
\underset{z'\in\Omega\to z\in\ddot{\Sigma}(\xi)}{\lim}R^{(2)}(z')&	z\in\ddot{\Sigma}(\xi)\cap (\underset{m=1}{\cup^3}S_{2m-1});\\[12pt]
\underset{z'\in\Omega\to z\in\ddot{\Sigma}(\xi)}{\lim}R^{(2)}(z')^{-1}& 	z\in\ddot{\Sigma}(\xi)\cap (\underset{m=1}{\cup^3}S_{2m});\\[12pt]
T^{-1}(z)G(z)T(z),   & 	z\in\partial\mathbb{D}_n\cap (\underset{m=1}{\cup^3}S_{2m});\\[12pt]
T^{-1}(z)G^{-1}(z)T(z),    & z\in\partial\mathbb{D}_n\cap (\underset{m=1}{\cup^3}S_{2m-1});\\[12pt]
\left(\begin{array}{ccc}
	1 & \frac{rT_{21}\mathcal{X}(|z|) e^{it\theta_{12}}}{1-|r|^2}&0\\
	0 & 1&0\\
	0 & 0 &1
\end{array}\right)\left(\begin{array}{ccc}
	1 & 0 & 0\\
	\frac{\bar{r}T_{12}\mathcal{X}(|z|) e^{-it\theta_{12}}}{1-|r|^2} & 1 & 0\\
	0 & 0 &1
\end{array}\right),   & z\in \mathbb{R};\\[12pt]
\left(\begin{array}{ccc}
	1 & 0&0\\
	0 & 1&\frac{-r(\omega z)T_{32} e^{it\theta_{23}}\mathcal{X}(|z|)}{1-|r(\omega z)|^2}\\
	0 & 0 &1
\end{array}\right)\left(\begin{array}{ccc}
	1 & 0 & 0\\
	0 & 1 & 0\\
	0 & 	\frac{\bar{r}(\omega z) e^{-it\theta_{23}}T_{23}\mathcal{X}(|z|)}{1-|r(\omega z)|^2} &1
\end{array}\right),   & z\in \omega^2\mathbb{R};\\[12pt]
\left(\begin{array}{ccc}
	1 & 0&0\\
	0 & 1&0\\
	\frac{-r(\omega^2 z)T_{13}\mathcal{X}(|z|) e^{-it\theta_{13}}}{1-|r(\omega^2 z)|^2} & 0 &1
\end{array}\right)\left(\begin{array}{ccc}
	1 & 0 & 	\frac{\bar{r}(\omega^2 z)T_{31}\mathcal{X}(|z|) e^{it\theta_{13}}}{1-|r(\omega^2 z)|^2}\\
	0 & 1 & 0\\
	0 & 0 &1
\end{array}\right),   & z\in \omega\mathbb{R};\\
\end{array}\right.;\label{jumpv21}
\end{equation}	\normalsize

$\blacktriangleright$ Asymptotic behaviors:
	\begin{align}
	M^{(2)}(z) =& I+\mathcal{O}(z^{-1}),\hspace{0.5cm}z \rightarrow \infty;\label{asyM2}
\end{align}

$\blacktriangleright$ $\bar{\partial}$-Derivative: For $z\in\mathbb{C}$
we have
\begin{align}
	\bar{\partial}M^{(2)}=M^{(2)}\bar{\partial}R^{(2)};
\end{align}

		$\blacktriangleright$ Singularities: The limiting values of $M^{(2)}(z)$ as $z$ approaches one	of the points $\varkappa_l=e^{\frac{i\pi(l-1)}{3}}$, $l = 1,...,6$ have pole singularities with leading terms of a specific
matrix structure
\begin{align}
	&M^{(2)}(z)=\frac{1}{z\mp1}\left(\begin{array}{ccc}
		\alpha^{(2)}_\pm &	\alpha^{(2)}_\pm & \beta^{(2)}_\pm \\
		-\alpha^{(2)}_\pm & -\alpha^{(2)}_\pm & -\beta^{(2)}_\pm\\
		0	&	0 & 0
	\end{array}\right)+\mathcal{O}(1),\ z\to\pm 1,\\
	&M^{(2)}(z)=\frac{1}{z\mp\omega^2}\left(\begin{array}{ccc}
		0 &	0 &  0\\
		\beta^{(2)}_\pm	 & \alpha^{(2)}_\pm &\alpha^{(2)}_\pm \\
		-\beta^{(2)}_\pm	&	-\alpha^{(2)}_\pm & -\alpha^{(2)}_\pm
	\end{array}\right)+\mathcal{O}(1),\ z\to\pm \omega^2,\\
	&M^{(2)}(z)=\frac{1}{z\mp\omega}\left(\begin{array}{ccc}
		-\alpha^{(2)}_\pm &	-\beta^{(2)}_\pm & -\alpha^{(2)}_\pm\\
		0	 & 0 &0 \\
		\alpha^{(2)}_\pm &	\beta^{(2)}_\pm & \alpha^{(2)}_\pm
	\end{array}\right)+\mathcal{O}(1),\ z\to\pm \omega,
\end{align}
with $\alpha^{(2)}_\pm=\alpha^{(2)}_\pm(y,t)=-\bar{\alpha}^{(2)}_\pm$, $\beta^{(2)}_\pm=\beta^{(2)}_\pm(y,t)=-\bar{\beta}^{(2)}_\pm$ and $M^{(2)}(z)^{-1}$ has same specific
matrix structure with $\alpha^{(2)}_\pm$, $\beta^{(2)}_\pm$ replaced by $\tilde{\alpha}^{(2)}_\pm$, $\tilde{\beta}^{(2)}_\pm$;

$\blacktriangleright$ Residue conditions: $M^{(2)}$ has simple poles at each point $\zeta_n$ and $\bar{\zeta}_n$ for $n\in\lozenge$ with:
\begin{align}
	&\res_{z=\zeta_n}M^{(2)}(z)=\lim_{z\to \zeta_n}M^{(2)}(z)\left[ T^{-1}(z)B_nT(z)\right] .
\end{align}
	
\end{RHP}

Unlike the classical case in \cite{fNLS}, our construction of $R^{(2)}$ will not completely  remove the jump of $M^{(1)}$. As shown above,  $M^{(2)}$ still has jump near the singularity $\varkappa_k$ on $l_k$, $k=1,...,6$. As reward for it, near  singularity $\varkappa_k$, $R^{(2)}\equiv I$ with $\bar{\partial}R^{(2)}\equiv0$.

\quad To solve RHP \ref{RHP4},  we decompose it into a model   RH  problem  for $M^{R}(z;y,t)\triangleq M^{R}(z)$  with $\bar\partial R^{(2)}\equiv0$   and a pure $\bar{\partial}$-Problem with nonzero $\bar{\partial}$-derivatives.
First  we establish  a   RH problem  for the  $M^{R}(z)$   as follows.

\begin{RHP}\label{RHP5}
Find a matrix-valued function  $  M^{R}(z)$ with following properties:

$\blacktriangleright$ Analyticity: $M^{R}(z)$ is  meromorphic  in $\mathbb{C}\setminus \Sigma^{(2)}$;

$\blacktriangleright$ Jump condition: $M^{R}$ has continuous boundary values $M^{R}_\pm$ on $\Sigma^{(2)}$ and
\begin{equation}
	M^{R}_+(z)=M^{R}_-(z)V^{(2)}(z),\hspace{0.5cm}z \in \Sigma^{(2)};\label{jump5}
\end{equation}

$\blacktriangleright$ Symmetry: $M^{R}(z)=\Gamma_1\overline{M^{R}(\bar{z})}\Gamma_1=\Gamma_2\overline{M^{R}(\omega^2\bar{z})}\Gamma_2=\Gamma_3\overline{M^{R}(\omega\bar{z})}\Gamma_3$

\quad and $M^{R}(z)=\overline{M^{R}(\bar{z}^{-1})}$;

$\blacktriangleright$ $\bar{\partial}$-Derivative:  $\bar{\partial}R^{(2)}=0$, for $ z\in \mathbb{C}$;

$\blacktriangleright$ Asymptotic behaviors:
	\begin{align}
	M^{R}(z) =& I+\mathcal{O}(z^{-1}),\hspace{0.5cm}z \rightarrow \infty;\label{asyMr}
\end{align}

$\blacktriangleright$ Singularities:  As $z\rightarrow \varkappa_l=e^{\frac{i\pi(l-1)}{3}}, l = 1,...,6$,   the  limit  of $M^{R}(z)$ have

\quad   pole singularities
\begin{align}
	&M^{R}(z)=\frac{1}{z\mp1}\left(\begin{array}{ccc}
		\alpha^{(2)}_\pm &	\alpha^{(2)}_\pm & \beta^{(2)}_\pm \\
		-\alpha^{(2)}_\pm & -\alpha^{(2)}_\pm & -\beta^{(2)}_\pm\\
		0	&	0 & 0
	\end{array}\right)+\mathcal{O}(1),\ z\to\pm 1,\label{MR1}\\
	&M^{R}(z)=\frac{1}{z\mp\omega^2}\left(\begin{array}{ccc}
		0 &	0 &  0\\
		\beta^{(2)}_\pm	 & \alpha^{(2)}_\pm &\alpha^{(2)}_\pm \\
		-\beta^{(2)}_\pm	&	-\alpha^{(2)}_\pm & -\alpha^{(2)}_\pm
	\end{array}\right)+\mathcal{O}(1),\ z\to\pm \omega^2,\\
	&M^{R}(z)=\frac{1}{z\mp\omega}\left(\begin{array}{ccc}
		-\alpha^{(2)}_\pm &	-\beta^{(2)}_\pm & -\alpha^{(2)}_\pm\\
		0	 & 0 &0 \\
		\alpha^{(2)}_\pm &	\beta^{(2)}_\pm & \alpha^{(2)}_\pm
	\end{array}\right)+\mathcal{O}(1),\ z\to\pm \omega,\label{MR3}
\end{align}
with $\alpha^{(2)}_\pm=\alpha^{(2)}_\pm(y,t)=-\bar{\alpha}^{(2)}_\pm$, $\beta^{(2)}_\pm=\beta^{(2)}_\pm(y,t)=-\bar{\beta}^{(2)}_\pm$ and $M^{(2)}(z)^{-1}$ has same specific
matrix structure with $\alpha^{(2)}_\pm$, $\beta^{(2)}_\pm$ replaced by $\tilde{\alpha}^{(2)}_\pm$, $\tilde{\beta}^{(2)}_\pm$;

$\blacktriangleright$ Residue conditions: $M^{R}$ has simple poles at each point $\zeta_n$ and $\bar{\zeta}_n$ for $n\in\lozenge$ with:
\begin{align}
	&\res_{z=\zeta_n}M^{R}(z)=\lim_{z\to \zeta_n}M^{R}(z)\left[ T^{-1}(z)B_nT(z)\right] \label{resMr}.
\end{align}

\end{RHP}

Compared with the case $\xi \in(1,+\infty)$,     the  jump matrix $V^{(2)}$  in   the case $\xi \in(-1/8,1)$
  has additional portion on $\Sigma_{jk}$ and $\Sigma_{jk}'$. So this case is  more difficult  to deal with. Denote $\mathbb{B}_j$ as the neighborhood of $\varkappa_j$, $j=1,...,6$ with
 \begin{align}
 	\mathbb{B}_j=\{ z \in \mathbb{C}\ ; \text{ |Re}\left( {z}/{\varkappa_j}\right) -1|<2\varepsilon,\ |\text{Im}\left(  {z}/{\varkappa_j}\right) |<2\varepsilon\}.
 \end{align}
Denote $ U(\xi)$ as the union set of neighborhood of $\xi_j$, $\omega\xi_j$ and $\omega^2\xi_j$,
\begin{align}
&U(\xi)=\underset{j=1,...,p(\xi)}{\cup}\left( U_{\xi_j}\cup U_{\omega\xi_j}\cup U_{\omega^2\xi_j}\right) ,\ U_{\xi_j}= \left\lbrace z;|z-\xi_j|\leq \varrho^{0} \right\rbrace,\nonumber\\
&U_{\omega\xi_j}= \left\lbrace z;|z-\omega\xi_j|\leq \varrho^{0} \right\rbrace,\ U_{\omega^2\xi_j}= \left\lbrace z;|z-\omega^2\xi_j|\leq  \varrho^{0} \right\rbrace, \  j=1,...,p(\xi). \nonumber
\end{align}
Here, $\varrho^{0}$ is a small enough positive constant make that $U(\xi)$ and $\mathbb{B}_j$ are disjoint, and $\varrho^{0}\leq\min\left\lbrace \varrho,\  \frac{1}{8}\underset{j\neq i\in \mathcal{N}}{\min}|\xi_i-\xi_j|,\  \frac{1}{8}\underset{j\in \mathcal{N}}{\min}|\xi_j\pm1| \right\rbrace$. Then this additional part of  jump matrix $V^{(2)}$ has following  estimation.
\begin{Proposition}\label{prov2}
	For $1\leq q\leq+\infty$, there exists a positive constant $K_q$   such  that the jump matrix $V^{(2)}$ defined in (\ref{jumpv21})
has  the  estimate
	\begin{align}
	\parallel V^{(2)}-I\parallel_{L^q(\Sigma_{kj}\setminus U(\xi_k))}= \mathcal{O}( e^{-K_qt}), \ t\to\infty,
	\end{align}
where  $k=1,...,p(\xi)$ and $j=1,...,4$.
	Also for  $1\leq q<+\infty$, there also exists a positive constant $K_q'$   such  that the jump matrix $V^{(2)}$  has the estimate
	\begin{align}
	\parallel V^{(2)}-I\parallel_{L^q(\Sigma_{kj}')}= \mathcal{O}( e^{-K_q't}),
	\end{align}
where $k=1,...,p(\xi)$, $j=1,...,4$.
\end{Proposition}

\begin{proof}
	We demonstrate the case $\xi \in[0,1)$, and the another case can be proved in similar way.
 For $z\in\Sigma_{11}\setminus U_{\xi_1}$, when $1\leq q<+\infty$, by using  definition of $V^{(2)}$ and (\ref{R}),  we have
	\begin{align}
	\parallel V^{(2)}-I\parallel_{L^q(\Sigma_{11}\setminus U_{\xi_1})}&=\parallel p_{11}(\xi_1,\xi)T^{(i)}_{12}(\xi)(z-\xi_1)^{-2i\nu(\xi_1)} e^{2it\theta}\parallel_{L^q(\Sigma_{11}\setminus U_{\xi_1})}\nonumber\\
	&\lesssim \parallel e^{2it\theta}\parallel_{L^q(\Sigma_{11}\setminus U_{\xi_1})}.\nonumber
	\end{align}
	For $z\in\Sigma_{11}\setminus U_{\xi_1}$, denote $z=\xi_1+s e^{i\varphi}$, $s\in(\varrho,+\infty)$, by lemma \ref{theta2}, we get
	\begin{align}
	\parallel V^{(2)}-I\parallel_{L^q(\Sigma_{11}\setminus U_{\xi_1})}^q &\lesssim \int_{\Sigma_{11}\setminus U_{\xi_1}}
\exp\left( -q c(\xi)t\text{Im}z(\text{Re}^2z-\xi_i^2)g(\text{Re}z;\xi)\right)dz \nonumber\\
	&\lesssim \int_{\varrho}^{+\infty}\exp\left( -q c'(\xi)t s\right)ds\lesssim t^{-1}\exp\left( -q c'(\xi)t\varrho\right). \nonumber
	\end{align}
	The second step is from  $g(\text{Re}z;\xi)(\text{Re}^2z-\xi_i^2)$ has nonzero  boundary on $\Sigma_{11}\setminus U_{\xi_1}$. And when $q =+\infty$ is obviously.
For $z\in\Sigma_{kj}'$, we only give the details of $\Sigma_{14}'$.
	\begin{align}
	\parallel V^{(2)}-I\parallel_{L^q(\Sigma_{14}')}&=\parallel R_{14}(z)e^{-2it\theta}\parallel_{L^q(\Sigma_{14}')}
\lesssim\parallel e^{-2it\theta}\parallel_{L^q(\Sigma_{1+}')}\nonumber\\
	&\lesssim t^{-1/q}\exp\left( -c''(\xi)t\right).\nonumber
	\end{align}
\end{proof}
This proposition means that the jump matrix $V^{(2)}(z)$    uniformly goes to  $I$  on     $\ddot{\Sigma}\setminus U(\xi)$.
therefore, outside of  $U(\xi)$,  there is only exponentially small error (in $t$) by completely ignoring the jump condition of  $M^{R}(z)$.
And this proposition enlightens  us to construct the solution $M^{R}(z)$ as follows:
\begin{equation}
M^{R}(z)=\left\{\begin{array}{ll}
E(z,\xi)M^{(r)}(z), & z\notin U(\xi)\cup \mathbb{B}_j,\\
E(z,\xi)M^{(r)}(z)M^{lo}(z),  &z\in U(\xi),\\
E(z,\xi)M^{(r)}(z)M^{B}_j(z),  &z\in \mathbb{B}_j.\\
\end{array}\right. \label{transm4}
\end{equation}
Note that, when  $\xi \in(1,+\infty)$, $M^{(r)}(z)$ has no jump except the circle around poles not in $\lozenge$, and it has  no  phase point. So $U(\xi)=\emptyset$ in these case, which means $M^{R}(z)=M^{(r)}(z)$. And it is more easy.  On the other hand, for the case $\xi\in(-1/8,1)$,  the definition above implies  that $M^{lo}$ is pole free and has no singularity. This construction  decomposes $M^{R}$ to two parts: $M^{(r)}$ solves the pure RHP obtained by ignoring the jump conditions of RHP \ref{RHP5}, which is shown in Section \ref{sec6}; $M^{lo}$ uses parabolic cylinder functions to build a matrix to match  jumps  of $M^{(2)}$ in a neighborhood of each critical point  which is shown in Section \ref{secpc}. And
$M^{B}_j(z)$ is a solution of a RHP which only has jump near $\varkappa_j$ in Section \ref{secB}. Finally, as the error function,
$E(z,\xi)$  will be different in different case of $\xi$ and is a solution of a small-norm Riemann-Hilbert problem shown in Section \ref{sec7}.

We now use $M^{R}(z)$ to construct  a new matrix function
\begin{equation}
M^{(3)}(z;y,t)\triangleq M^{(3)}(z)=M^{(2)}(z)M^{R}(z)^{-1}.\label{transm3}
\end{equation}
which   removes   analytical component  $M^{R}$    to get  a  pure $\bar{\partial}$-problem.

\noindent\textbf{$\bar{\partial}$-problem}. Find a matrix-valued function  $ M^{(3)}(z;y,t)\triangleq M^{(3)}(z)$ with following identities:

$\blacktriangleright$ Analyticity: $M^{(3)}(z)$ is continuous   and has sectionally continuous first partial derivatives in $\mathbb{C}$.

$\blacktriangleright$ Asymptotic behavior:
\begin{align}
&M^{(3)}(z) \sim I+\mathcal{O}(z^{-1}),\hspace{0.5cm}z \rightarrow \infty;\label{asymbehv7}
\end{align}

$\blacktriangleright$ $\bar{\partial}$-Derivative: We have
$$\bar{\partial}M^{(3)}=M^{(3)}W^{(3)},\ \ z\in \mathbb{C},$$
where
\begin{equation}
W^{(3)}=M^{R}(z)\bar{\partial}R^{(2)}(z)M^{R}(z)^{-1}.
\end{equation}

\begin{proof}
	By using  properties  of  the   solutions   $M^{(2)}$ and $M^{R}$  for  RHP \ref{RHP5}  and $\bar{\partial}$-problem,
 the analyticity is obtained   immediately.
Since $M^{(2)}$ and $M^{R}$ achieve same jump matrix, we have
	\begin{align*}
	M_-^{(3)}(z)^{-1}M_+^{(3)}(z)&=M_-^{(2)}(z)^{-1}M_-^{R}(z)M_+^{R}(z)^{-1}M_+^{(2)}(z)\\
	&=M_-^{(2)}(z)^{-1}V^{(2)}(z)^{-1}M_+^{(2)}(z)=I,
	\end{align*}
	which means $ M^{(3)}$ has no jumps and is continuously everywhere.  We can also show  that $ M^{(3)}$ has no pole.
	 For
 $\zeta_n$ with $ n\in\lozenge $,  let $\mathcal{W}_n=T^{-1}(\zeta_n)B_nT(\zeta_n)$  which appears in the left side of the
corresponding residue condition of RHP \ref{RHP4}  and  RHP \ref{RHP5}. We take $\zeta_n\in S_1$ as an example. Then
 we have the Laurent expansions in $z-\zeta_n$
 \begin{align}
 	M^{(2)}(z)=\frac{\res_{z=\zeta_n}M^{(2)}(z)}{z-\zeta_n}+a(\zeta_n)+\mathcal{O}(z-\zeta_n),
 \end{align}
with  $a(\lambda)$ is a constant  matrix
\begin{align}
	\res_{z=\zeta_n}M^{(2)}(z)&=\left(0,\res_{z=\zeta_n}\frac{-C_nT_{21}(\zeta_n)e^{it[\theta_{12}]_n}[M^{(2)}(\zeta_n)]_1}{z-\zeta_n},0 \right)=a(\zeta_n)\mathcal{W}_n,\nonumber
\end{align}
then
	\begin{align}
&M^{(2)}(z)=a(\zeta_n) \left[ \dfrac{\mathcal{W}_n}{z-\zeta_n}+I\right] +\mathcal{O}(z-\zeta_n).
\end{align}
Similarly,   for $M^R(z)^{-1}$,
	\begin{align}
	M^R(z)^{-1}&= \left[ \dfrac{\mathcal{W}_n}{z-\zeta_n}+I\right]^{-1}a(\zeta_n)^{-1} +\mathcal{O}(z-\zeta_n)\nonumber\\
	&= \left[ \dfrac{-\mathcal{W}_n}{z-\zeta_n}+I\right]a(\zeta_n)^{-1} +\mathcal{O}(z-\zeta_n).
\end{align}
Then
	\begin{align}
	M^{(3)}(z)&=\left\lbrace a(\zeta_n) \left[ \dfrac{\mathcal{W}}{z-\lambda}+I\right]\right\rbrace \left\lbrace\left[ \dfrac{-\mathcal{W}}{z-\lambda}+I\right]a(\zeta_n)^{-1}\right\rbrace + \mathcal{O}(z-\lambda)\nonumber\\
	&=\mathcal{O}(1),\nonumber
	\end{align}
	which  implies that  $M^{(3)}(z)$ has removable singularities at $\zeta_n$. Similarly, $M^{(3)}(z)$ has no singularities on $\varkappa_l=e^{\frac{i\pi(l-1)}{3}}$, $l = 1,...,6$.
 And the $\bar{\partial}$-derivative of  $ M^{(3)}(z)$ comes  from    $ M^{(3)}(z)$  due to   analyticity of $M^{R}(z)$.
\end{proof}
The unique existence  and asymptotic  of  $M^{(3)}(z)$  will be shown in   section \ref{sec7}.

\section{Contribution from  discrete spectrum   } \label{sec6}

\quad In this section, we build a reflectionless case of  RHP \ref{RHP1} to  show that  its solution can be approximated  with  $M^{(r)}(z)$. As $V^{(2)}(z)\equiv0$ except on $\partial\mathbb{D}_n$, RHP \ref{RHP5} reduces to a the sectionally meromorphic function $M^{(r)}(z)$ with jump discontinuities on the union of circles. Then, by relating $M^{(r)}(z)$ with original RHP \ref{RHP1}, we  show the existence and uniqueness of solution of the above RHP \ref{RHP5}.

\subsection{$M^{(r)}(z)$-  and  $ M^{(r)}_\lozenge(z)$-solitons   }

\begin{Proposition} \label{rer}
	If $M^{(r)}(z)$ is the solution of the RH problem \ref{RHP5} with scattering data $\mathcal{D}=\left\lbrace  r(z),\left\lbrace \zeta_n,C_n\right\rbrace_{n\in\{1,...,6N_0\}}\right\rbrace$,
  then   $M^{(r)}(z)$ exists and is  unique.
\end{Proposition}
\begin{proof}
	To transform $M^{(r)}(z)$ to the soliton-solution  of RHP \ref{RHP1}, the jumps and poles need to be restored. We reverse the triangularity effected in (\ref{transm1}) and (\ref{transm2}):
	\begin{equation}
		N(z;\tilde{\mathcal{D}})= M^{(r)}(z)T^{-1}(z)G^{-1}(z)T(z)\tilde{\Pi}(z),\label{N}
	\end{equation}
with $G(z)$ defined in (\ref{funcG}), $\tilde{\mathcal{D}}=\left\lbrace  r(z),\left\lbrace \zeta_n,C_n\delta_{\zeta_n}\right\rbrace_{n\in\{1,...,6N_0\}}\right\rbrace$ and
\begin{align}
	\tilde{\Pi}(z)=\text{diag}\left(\frac{\Pi(z)\Pi^A(z)}{\Pi(\omega^2z)\Pi^A(\omega^2z)},\frac{\Pi(\omega z)\Pi^A(\omega z)}{\Pi(z)\Pi^A(z)}, \frac{\Pi(\omega^2z)\Pi^A(\omega^2z)}{\Pi(\omega z)\Pi^A(\omega z)} \right) .
\end{align}
Here, for $n\in\mathcal{N}=\{1,...,2N_1+N_2\}$, $$\delta_{\zeta_n}=\frac{\delta_3(\zeta_n)\delta_5(\zeta_n)}{\delta_1^{2}(\zeta_n)}=\frac{\delta_1(\omega^2\zeta_n)\delta_1(\omega\zeta_n)}{\delta_1^{2}(\zeta_n)}.$$
 Similarly, for $n\in\mathcal{N}=\{2N_1+N_2,...,N_0\}$, $$\delta_{\zeta_n}=\frac{\delta_3(\zeta_n)\delta_1(\zeta_n)}{\delta_5^{2}(\zeta_n)}=\frac{\delta_1(\omega^2\zeta_n)\delta_1(\zeta_n)}{\delta_1^{2}(\omega\zeta_n)}.$$
 Then from the symmetry  $\delta_1(z)^{-1}=\overline{\delta_1(\bar{z})}$, $$\delta_{\zeta_{n+kN_0}}=\frac{\delta_1(\omega^2\zeta_{n+kN_0})\delta_1(\omega\zeta_{n+kN_0})}{\delta_1^{2}(\zeta_{n+kN_0})},$$  for $n\in\mathcal{N}=\{1,...,2N_1+N_2\}$, $k=1,...,5$. And
\begin{align*}
	\delta_{\zeta_{n+kN_0}}=\frac{\delta_1(\omega^2\zeta_{n+kN_0})\delta_1(\zeta_{n+kN_0})}{\delta_1^{2}(\omega\zeta_{n+kN_0})},
\end{align*}
for $n\in\mathcal{N}=\{2N_1+N_2,...,N_0\}$, $k=1,...,5$.
 First we verify $N(z;\tilde{\mathcal{D}})$ satisfying RHP \ref{RHP1}. This transformation to $N(z;\tilde{\mathcal{D}})$ preserves the normalization conditions at the origin and infinity obviously. And comparing with (\ref{transm1}), this transformation  restores the jump on   $\overline{\mathbb{D}}_n$ and $\mathbb{D}_n$ to residue for $n\notin\lozenge$.  As for $n\in\lozenge$, take $\zeta_n\in S_1$ and $n\in\lozenge_1$ as an example. Substitute (\ref{resMr}) into the transformation:
\begin{align}
	\res_{z=\zeta_n}N(z;\tilde{\mathcal{D}})=&\res_{z=\zeta_n}M^{(r)}(z)T^{-1}(z)G^{-1}(z)T(z)\tilde{\Pi}(z)\nonumber\\
		=&\lim_{z\to \zeta_n}N(z;\tilde{\mathcal{D}})\left(\begin{array}{cc}
			0 & 0\\
			C_n\delta_{\zeta_n}e^{-2it\theta_n} & 0
		\end{array}\right).
\end{align}
Its analyticity and symmetry follow from the Proposition of $M^{(r)}(z)$, $T(z)$ and $G(z)$ immediately.  Then $N(z;\tilde{\mathcal{D}})$ is solution of RHP \ref{RHP1} with absence of reflection, whose  exact solution  exists and can be obtained as described similarly in Appendix A \cite{SandRNLS}. Its uniqueness
  can be affirmed  with  the Liouville's theorem. Then the uniqueness and existences of $M^{(r)}(z)$  come from (\ref{N}).
\end{proof}

The contribution to  the  $M^{(r)}(z)$  comes  from  all  discrete spectrum,  some of which
  are  restored    from  the jump on   $\partial\overline{\mathbb{D}}_n$ and $\partial\mathbb{D}_n$.  We will show
that these poles make   very small   contribution to the    $M^{(r)}(z)$  as $t\to\infty$.

We put all  circles together and define
$$ \partial \mathbb{D}  =   \cup_{n\in   \mathcal{Z} \setminus \lozenge}  ( \partial\mathbb{D}_n \cup \partial\overline{\mathbb{D}}_n),$$
then  we  have the following  estimate.

\begin{lemma}\label{lemmav2} For $1\leq q\leq +\infty$,
	the jump matrix $ V^{(2)}(z)\big|_{\partial\mathbb{D} }$   given in    (\ref{jumpv2})-(\ref{jumpv21}) admits the estimate
	\begin{align}
	&\parallel V^{(2)}(z)-I\parallel_{L^q(\partial\mathbb{D} )}=\mathcal{\mathcal{O}}(e^{- \min\{\rho_0,\delta_0\}t} ), \ \ n\in  \mathcal{Z} \setminus \lozenge.   \label{7.1}
\end{align}
where   $\rho_0$ is given by  (\ref{rho0})
\end{lemma}
\begin{proof}  Since $\partial\mathbb{D}_n$ is bounded, we just need to prove that the estimate (\ref{7.1}) for $ L^\infty(\partial\mathbb{D}_n)$ is true.
	Here take $z\in\partial\mathbb{D}_n$, $n\in\nabla_1$ as an illustrative  example.
	\begin{align}
		&\parallel V^{(2)}(z)-I\parallel_{L^\infty(\partial\mathbb{D}_n)}=|C_n(z-\zeta_n)^{-1}T_{21}(z)e^{it[\theta_{12}]_n}|\nonumber\\
	&\lesssim \varrho^{-1}e^{\text{Re}(it[\theta_{12}]_n)}\lesssim e^{-t\text{Im}([\theta_{12}]_n)} \leq e^{-\delta_0t}.\nonumber
	\end{align}
\end{proof}
This Lemma  inspires us to   completely ignore the   jump conditions on $\partial\mathbb{D} $, because they    exponentially decay to $I$   as  $t \to \infty$.
The main  contribution to $M^{(r)} (z)$  comes from  discrete spectrum   $\left\lbrace\zeta_n \right\rbrace_{n\in\lozenge}$.
Let  $M^{(r)}_\lozenge(z) = M^{(r)} (z)\big|_{V^{(2)}(z)\equiv0}$, then  the  RHP \ref{RHP5}
   reduces to the following RH problem.

\begin{RHP}\label{RHP6}
Find a matrix-valued function  $ M^{(r)}_\lozenge(z)$ with following properties:

$\blacktriangleright$ Analyticity: $M^{(r)}_\lozenge(z)$ is analytical  in $\mathbb{C}\setminus \left\lbrace\zeta_n \right\rbrace_{n\in\lozenge} $;

$\blacktriangleright$ Symmetry: $M^{(r)}_\lozenge(z)=\Gamma_1\overline{M^{(r)}_\lozenge(\bar{z})}\Gamma_1
=\Gamma_2\overline{M^{(r)}_\lozenge(\omega^2\bar{z})}\Gamma_2=\Gamma_3\overline{M^{(r)}_\lozenge(\omega\bar{z})}\Gamma_3$

\quad and $M^{(r)}_\lozenge(z)=\overline{M^{(r)}_\lozenge(\bar{z}^{-1})}$;

$\blacktriangleright$ Asymptotic behaviors:
	\begin{align}
	M^{(r)}_\lozenge(z) =& I+\mathcal{O}(z^{-1}),\hspace{0.5cm}z \rightarrow \infty;\label{asyMrL}
\end{align}

$\blacktriangleright$ Singularities: As $z\rightarrow \varkappa_l=e^{\frac{i\pi(l-1)}{3}}, l = 1,...,6$,   the  limit  of $M^{(r)}_\lozenge(z)$ have

\quad   pole singularities
\begin{align}
	&M^{(r)}_\lozenge(z)=\frac{1}{z\mp1}\left(\begin{array}{ccc}
		\alpha^{(2)}_\pm &	\alpha^{(2)}_\pm & \beta^{(2)}_\pm \\
		-\alpha^{(2)}_\pm & -\alpha^{(2)}_\pm & -\beta^{(2)}_\pm\\
		0	&	0 & 0
	\end{array}\right)+\mathcal{O}(1),\ z\to\pm 1,\\
	&M^{(r)}_\lozenge(z)=\frac{1}{z\mp\omega^2}\left(\begin{array}{ccc}
		0 &	0 &  0\\
		\beta^{(2)}_\pm	 & \alpha^{(2)}_\pm &\alpha^{(2)}_\pm \\
		-\beta^{(2)}_\pm	&	-\alpha^{(2)}_\pm & -\alpha^{(2)}_\pm
	\end{array}\right)+\mathcal{O}(1),\ z\to\pm \omega^2,\\
	&M^{(r)}_\lozenge(z)=\frac{1}{z\mp\omega}\left(\begin{array}{ccc}
		-\alpha^{(2)}_\pm &	-\beta^{(2)}_\pm & -\alpha^{(2)}_\pm\\
		0	 & 0 &0 \\
		\alpha^{(2)}_\pm &	\beta^{(2)}_\pm & \alpha^{(2)}_\pm
	\end{array}\right)+\mathcal{O}(1),\ z\to\pm \omega,
\end{align}
with $\alpha^{(2)}_\pm=\alpha^{(2)}_\pm(y,t)=-\bar{\alpha}^{(2)}_\pm$, $\beta^{(2)}_\pm=\beta^{(2)}_\pm(y,t)=-\bar{\beta}^{(2)}_\pm$ and $M^{(r)}_\lozenge(z)^{-1}$ has same specific
matrix structure with $\alpha^{(2)}_\pm$, $\beta^{(2)}_\pm$ replaced by $\tilde{\alpha}^{(2)}_\pm$, $\tilde{\beta}^{(2)}_\pm$;

$\blacktriangleright$ Residue conditions: $M^{(r)}_\lozenge$ has simple poles at each point $\zeta_n$ and $\bar{\zeta}_n$ for $n\in\lozenge$ with:
\begin{align}
	&\res_{z=\zeta_n}M^{(r)}_\lozenge(z)=\lim_{z\to \zeta_n}M^{(r)}_\lozenge(z)\left[ T^{-1}(z)B_nT(z)\right] \label{resMrsol}.
\end{align}
\end{RHP}

By using the Proposition  \ref{rer}, we then have
\begin{Proposition}\label{unim}	 The RHP \ref{RHP6}  exists an  unique solution. The  $M^{(r)}_\lozenge(z)$ is equivalent  to a  solution of the original reflectionless RHP \ref{RHP1} with modified scattering data $\tilde{\mathcal{D}}_\lozenge=\left\lbrace  0,\left\lbrace \zeta_n,C_n^\lozenge\right\rbrace_{n\in\lozenge}\right\rbrace$ with
$C_n^\lozenge = C_nT^2(\zeta_n)$.
\end{Proposition}

When $r(s)\equiv0$, denote $u^r(x,t;\tilde{\mathcal{D}})$ as the $\mathcal{N}(\lozenge)$-soliton with   scattering data $\tilde{\mathcal{D}}_\lozenge$.
By the reconstruction formula (\ref{recons u}) and (\ref{recons x}), we then have
\begin{corollary}\label{sol}
The soliton solution   of the Novikov equation  (\ref{Novikov}) is given by
\begin{align}
	u^r(x,t;\tilde{\mathcal{D}}_\lozenge)
	=&\frac{1}{2}[\tilde{m}^{(r)}_\lozenge]_1(y,t)\left(\frac{[M^{(r)}_\lozenge]_{33}(e^{\frac{i\pi}{6}};y,t)}{[M^{(r)}_\lozenge]_{11}(e^{\frac{i\pi}{6}};y,t)} \right)^{1/2}\nonumber\\
	&+ \frac{1}{2}[\tilde{m}^{(r)}_\lozenge]_3(y,t)\left(\frac{[M^{(r)}_\lozenge]_{33}(e^{\frac{i\pi}{6}};y,t)}{[M^{(r)}_\lozenge]_{11}(e^{\frac{i\pi}{6}};y,t)} \right)^{-1/2}-1,\label{recons ur}
\end{align}
in which
\begin{align}
&x(y,t;\tilde{\mathcal{D}}_\lozenge)=y+c^r_+(x,t;\tilde{\mathcal{D}}_\lozenge)=y+\frac{1}{2} \ln\frac{M_{33}(e^{\frac{i\pi}{6}};y,t)}{M_{11}(e^{\frac{i\pi}{6}};y,t)},\label{recons xr}\\
&[\tilde{m}^{(r)}_\lozenge]_j\triangleq\sum_{n=1}^3[M^{(r)}_\lozenge]_{nj}(e^{\frac{i\pi}{6}};y,t),\ j=1,2,3.\nonumber
\end{align}
\end{corollary}

\subsection{ Residual error    between $M^{(r)}$- and $M^{(r)}_\lozenge$-solitons}\label{sec61}

To obtain   the residual error  between $M^{(r)}$   and $M^{(r)}_\lozenge$, we make the  factorization
\begin{equation}
	M^{(r)}(z)=\tilde{E}(z)M^{(r)}_\lozenge(z),\label{transMr}
\end{equation}
where $\tilde{E}(z)$ is a error function, which  can be obtained by solving
the following  small-norm RH problem

\begin{RHP}\label{RHP7}
	Find a matrix-valued function $\tilde{E}(z)$  with following  properties
	
	$\blacktriangleright$ Analyticity: $\tilde{E}(z)$ is analytical  in $\mathbb{C}\setminus  \partial\mathbb{D}  $;

	$\blacktriangleright$ Asymptotic behaviors:
	\begin{align}
	&\tilde{E}(z) \sim I+\mathcal{O}(z^{-1}),\hspace{0.5cm}|z| \rightarrow \infty;
	\end{align}

	$\blacktriangleright$ Jump condition: $\tilde{E}(z)$ has continuous boundary values $\tilde{E}_\pm(z)$ on $\partial\mathbb{D} $ satisfying
	$$\tilde{E}_+(z)=\tilde{E}_-(z)V^{\tilde{E}}(z),$$
	where the jump matrix $V^{\tilde{E}}(z)$ is given by
	\begin{equation}
	V^{\tilde{E}}(z)=M^{(r)}_\lozenge(z)V^{(2)}(z)M^{(r)}_\lozenge(z)^{-1}. \label{tVE}
	\end{equation}
\end{RHP}

{Proposition \ref{unim}} shows  that $M^{(r)}_\lozenge(z)$ is bound on $\Sigma^{(2)}$. By using Lemma \ref{lemmav2}, we have the following uniformly decaying
evaluate
\begin{equation}
\parallel V^{\tilde{E}}(z)-I \parallel_{L^q(\partial\mathbb{D} )}\lesssim \parallel V^{(2)}-I \parallel_{L^q(\partial\mathbb{D} )}=\mathcal{O}(e^{- \min\{\rho_0,\delta_0\}t} ).  \label{tVE-I}
\end{equation}
Therefore,
the   existence and uniqueness  of  the RHP \ref{RHP7}  is  shown  by using  a  small-norm RH problem \cite{RN9,RN10}, moreover its  solution can be given by
\begin{equation}
\tilde{E}(z)=I+ \frac{1}{2\pi i}\int_{\partial\mathbb{D} }\dfrac{ \eta(s)  (V^{\tilde{E}}-I)}{s-z}ds,\label{tEz}
\end{equation}
where  $\eta\in L^2(\partial\mathbb{D})$ is a unique solution of Fredholm  equation
\begin{equation}
(1-C_{\tilde{E}})\eta=I.
\end{equation}
The integral operator  $C_{\tilde{E}}$: $L^2 \to L^2 $ is  given by
\begin{equation}
C_{\tilde{E}}(\eta) =\mathcal{P}^-\left( \eta (V^{\tilde{E}}-I)\right),\nonumber
\end{equation}
where $\mathcal{P}^-$ is a Plemelj projection operator
\begin{equation}
\mathcal{P}^-( \eta (V^{\tilde{E}}-I) )= \lim_{\varepsilon\to 0}\frac{1}{2\pi i}\int_{\partial\mathbb{D} }\dfrac{\eta (V^{\tilde{E}}-I) }{s-(z-i\varepsilon )}ds.
\end{equation}
Then by (\ref{tVE}),  we have
\begin{equation}
\parallel C_{\tilde{E}}\parallel_{L^2 } \leq \parallel \mathcal{P}^- \parallel_{L^2 } \parallel V^{\tilde{E}}-I\parallel_{L^\infty} \lesssim \mathcal{O}(e^{- \min\{\rho_0,\delta_0\}t} ),
\end{equation}
which means $\parallel C_{\tilde{E}}\parallel_{L^2 }<1$ for sufficiently large $t$, so   $\eta$ uniquely  exists and
\begin{equation}
\parallel \eta\parallel_{L^2 } \lesssim\mathcal{O}(e^{- \min\{\rho_0,\delta_0\}t} ).\label{normeta}
\end{equation}
In order to reconstruct the solution $u(x,t)$ of the Novikov equation (\ref{Novikov}),  we need the asymptotic
   of $\tilde{E}(z)$ as $z\to \infty$ and the long time asymptotic   of $\tilde{E}(i)$.
\begin{Proposition}\label{tasyE}   The residual error   $\tilde{E}(z)$  defined  by  (\ref{tEz}) admits estimate
	\begin{equation}
	|\tilde{E}(z)-I|\lesssim\mathcal{O}(e^{- \min\{\rho_0,\delta_0\}t}), \ t \to \infty. \nonumber
	\end{equation}
	Moreover $\tilde{E}(z)$ has expansion  at $z=e^{-\frac{i\pi}{6}}$,
	\begin{align}
	\tilde{E}(z)=\tilde{E}(e^{-\frac{i\pi}{6}})+\tilde{E}_1(z-e^{-\frac{i\pi}{6}})+\mathcal{O}( (z-e^{-\frac{i\pi}{6}})^2),\label{texpE}
	\end{align}
where
	\begin{align}
	&\tilde{E}(e^{-\frac{i\pi}{6}})=I+\frac{1}{2\pi i}\int_{\Sigma^{(2)}}\dfrac{ \eta(s)  (V^{\tilde{E}}-I)}{s-e^{-\frac{i\pi}{6}}}ds,\label{tEi}\\	
    &\tilde{E}_1=\frac{1}{2\pi i}\int_{\Sigma^{(2)}}\frac{ \eta(s)  (V^{\tilde{E}}-I)}{(s-e^{-\frac{i\pi}{6}})^2}ds, \nonumber
	\end{align}
which  satisfy the  following  estimates
	\begin{equation}
	|\tilde{E}(e^{-\frac{i\pi}{6}})-I|\lesssim\mathcal{O}(e^{- \min\{\rho_0,\delta_0\}t}),\hspace{0.5cm}\tilde{E}_1\lesssim\mathcal{O}(e^{- \min\{\rho_0,\delta_0\}t}).\label{tE1t}
	\end{equation}
\end{Proposition}
\begin{proof}
	By combining (\ref{normeta}) and (\ref{tVE-I}), we obtain
	\begin{equation}
	|\tilde{E}(z)-I|\leq|(1-C_{\tilde{E}})(\eta)|+|C_{\tilde{E}}(\eta)|\lesssim\mathcal{O}(e^{- \min\{\rho_0,\delta_0\}t}).
	\end{equation}
The estimate  $\tilde{E}(e^{-\frac{i\pi}{6}})$ in (\ref{tE1t}) is obtained by taking $z=e^{-\frac{i\pi}{6}}$ in above estimate.
As $z\to e^{-\frac{i\pi}{6}}$, geometrically expanding $(s-z)^{-1}$
	for $z$ large in (\ref{tEz}) leads to (\ref{texpE}). Finally for $\tilde{E}_1$, noting that $|s-e^{-\frac{i\pi}{6}}|^{-2}$ is bounded on $\Sigma^{(2)}$, then
	\begin{align}
	|\tilde{E}_1|\lesssim  \parallel \eta \parallel_2\parallel V^{\tilde{E}}-I \parallel_2\lesssim\mathcal{O}(e^{- \min\{\rho_0,\delta_0\}t}). \nonumber
	\end{align}
\end{proof}

\begin{corollary}  Assume that $M^{(r)}(z)(z)$ and  $ M^{(r)}_\lozenge(z)$ are two solutions of  the  RHP \ref{RHP6} and
  RHP \ref{RHP5} respectively, then they have the relation
\begin{equation}\label{asymsol}
M^{(r)}(z)(z)=M^{(r)}_\lozenge(z)\left[I+\mathcal{O}(e^{- \min\{\rho_0,\delta_0\}t})\right], \ t \to \infty.
\end{equation}
 The  corresponding  soliton  solutions  have the relation
\begin{equation}
u^r(x,t;\tilde{\mathcal{D}})=u^r(x,t;\tilde{\mathcal{D}}_\lozenge)+\mathcal{O}(e^{- \min\{\rho_0,\delta_0\}t}).
\end{equation}
\end{corollary}

\section{Contribution  from  jump contours    }\label{secpc}

\subsection{Local solvable  RH problem  near phase points }

\quad When $\xi\in(-1/8,1)$, proposition  \ref{prov2} gives that out of $U(\xi)$, the jumps are exponentially close to the identity.
Hence we need to continue our investigation near  the stationary phase points in this section.

We define a  new local contour
$$\Sigma^{lo}= {\Sigma}^{lo}_{\omega^0} \cup  {\Sigma}^{lo}_{\omega^1}\cup {\Sigma}^{lo}_{\omega^2}, $$
where $  {\Sigma}^{lo}_{\omega^n}, \ n=0,1,2$ are the local contours on  jump  contours $\omega^n \mathbb{R}$, respectively
\begin{align}
	&  {\Sigma}^{lo}_{\omega^n}= (\underset{k=1,..,p(\xi)}{\underset{j=1,...,4,}{\cup}}\Sigma_{kj}^{\omega^n} )\cap U(\xi),\  n=0, 1,2. \nonumber
\end{align}
 Since there are  $3p(\xi)$ phase points $\xi_{n,k}=\omega^n\xi_k, n=0,1,2; k=1,\cdots,p(\xi)$,  so
the   local jump contour $\Sigma^{lo}$  is consist    of  $3p(\xi)$   crosses  with $p(\xi)=4$ for the case $0\leq\xi<1$
and $p(\xi)=8$  for the case $- {1}/{8}<\xi<0$.  See Figure \ref{sigma0}.
Further  denote  the  local  jump   cross  for each  phase point $\xi_{n,k}$
\begin{align}
	& \Sigma^{lo,n}_k= (\underset{j=1,...,4}{\cup} \Sigma_{kj}^{\omega^n})\cap U(\xi), \ n=0, 1,2; \ k=1, \cdots, p(\xi), \nonumber\\
&\Sigma^{lo,n}_k \cap \Sigma^{lo,n}_j=\varnothing, k\not= j.\nonumber
\end{align}
 We consider the following  local  RH problem
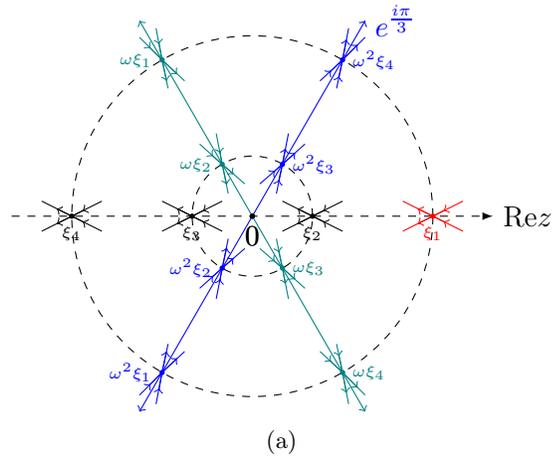
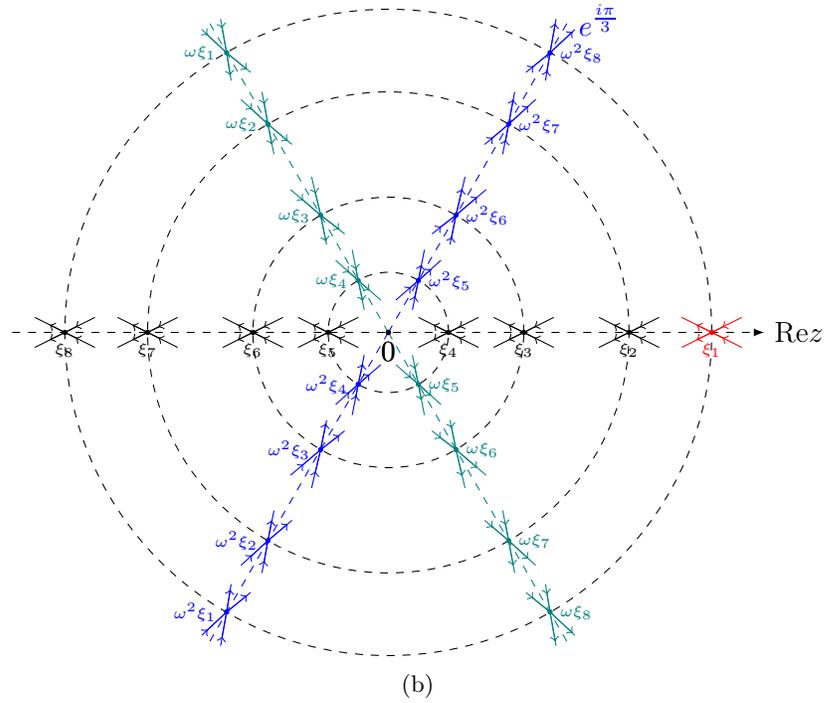
\begin{figure}[htp]
	\centering
	\subfigure[]{
		\begin{tikzpicture}
		\draw(-2.4,0)--(-2.8,0.2);
		\draw[->](-2.4,0)--(-2.6,0.1);
		\draw(-2.4,0)--(-2,0.2);
		\draw[-<](-2.4,0)--(-2.2,-0.1);
		\draw(-2.4,0)--(-2.8,-0.2);
		\draw[->](-2.4,0)--(-2.6,-0.1);
		\draw(-2.4,0)--(-2,-0.2);
		\draw[-<](-2.4,0)--(-2.2,0.1);
		\draw(-0.8,0)--(-0.4,0.2);
		\draw[-<](-0.8,0)--(-0.6,0.1);
		\draw(-0.8,0)--(-1.2,0.2);
		\draw[->](-0.8,0)--(-1,-0.1);
		\draw(-0.8,0)--(-0.4,-0.2);
		\draw[->](-0.8,0)--(-1,0.1);
		\draw(-0.8,0)--(-1.2,-0.2);
		\draw[-<](-0.8,0)--(-0.6,-0.1);
		\draw[dashed,-latex](-3.2,0)--(3.2,0)node[right]{ Re$z$};
		\draw(0.8,0)--(0.4,0.2);
		\draw[->](0.8,0)--(0.6,0.1);
		\draw(0.8,0)--(0.4,-0.2);
		\draw[-<](0.8,0)--(1,-0.1);
		\draw(0.8,0)--(1.2,0.2);
		\draw[-<](0.8,0)--(1,0.1);
		\draw(0.8,0)--(1.2,-0.2);
		\draw[->](0.8,0)--(0.6,-0.1);
		\draw[red] (2.4,0)--(2.8,0.2);
		\draw[red,-<](2.4,0)--(2.6,0.1);
		\draw[red] (2.4,0)--(2,0.2);
		\draw[red,->](2.4,0)--(2.2,-0.1);
		\draw[red] (2.4,0)--(2.8,-0.2);
		\draw[red,-<](2.4,0)--(2.6,-0.1);
		\draw[red] (2.4,0)--(2,-0.2);
		\draw[red,->](2.4,0)--(2.2,0.1);
		\coordinate (I) at (0,0);
		\fill (I) circle (1pt) node[below] {$0$};
		\draw[blue,->](0,0)--(1.5,2.6)node[right]{ $e^{\frac{i\pi}{3}}$};
		\draw[teal,->](0,0)--(-1.5,2.6);
		\draw[teal,->](0,0)--(1.5,-2.6);
		\draw[blue,->](0,0)--(-1.5,-2.6);
		\draw[dashed] (0.8,0) arc (0:360:0.8);
		\draw[dashed] (2.4,0) arc (0:360:2.4);
		\coordinate (I) at (0,0);
		\fill (I) circle (1pt) node[below] {$0$};
		\coordinate (A) at (-2.4,0);
		\fill (A) circle (1pt) node[below] {\tiny$\xi_4$};
		\coordinate (b) at (-0.8,0);
		\fill (b) circle (1pt) node[below] {\tiny$\xi_3$};
		\coordinate (e) at (2.4,0);
		\fill[red] (e) circle (1pt) node[below] {\tiny$\xi_1$};
		\coordinate (f) at (0.8,0);
		\fill (f) circle (1pt) node[below] {\tiny$\xi_2$};
		\coordinate (A1) at (0.4,0.69);
		\fill[blue] (A1) circle (1pt) node[right] {\tiny$\omega^2\xi_3$};
		\draw[blue](0.4,0.69)--(0.32,0.3);
		\draw[blue,-<](0.4,0.69)--(0.34,0.42);
		\draw[blue](0.4,0.69)--(0.1,0.43);
		\draw[blue,-<](0.4,0.69)--(0.25,0.56);
		\draw[blue](0.4,0.69)--(0.48,1.1);
		\draw[blue,->](0.4,0.69)--(0.44,0.9);
		\draw[blue](0.4,0.69)--(0.7,0.96);
		\draw[blue,->](0.4,0.69)--(0.57,0.84);
		\coordinate (A2) at (0.4,-0.69);
		\fill[teal] (A2) circle (1pt) node[right] {\tiny$\omega\xi_3$};
		\draw[teal](0.4,-0.69)--(0.32,-0.3);
		\draw[teal,-<](0.4,-0.69)--(0.34,-0.42);
		\draw[teal](0.4,-0.69)--(0.1,-0.43);
		\draw[teal,-<](0.4,-0.69)--(0.25,-0.56);
		\draw[teal](0.4,-0.69)--(0.48,-1.1);
		\draw[teal,->](0.4,-0.69)--(0.44,-0.9);
		\draw[teal](0.4,-0.69)--(0.7,-0.96);
		\draw[teal,->](0.4,-0.69)--(0.57,-0.84);
		\coordinate (A3) at (-0.4,0.69);
		\fill[teal] (A3) circle (1pt) node[left] {\tiny$\omega\xi_2$};
		\draw[teal](-0.4,0.69)--(-0.32,0.3);
		\draw[teal,->](-0.4,0.69)--(-0.34,0.42);
		\draw[teal](-0.4,0.69)--(-0.1,0.43);
		\draw[teal,->](-0.4,0.69)--(-0.25,0.56);
		\draw[teal](-0.4,0.69)--(-0.48,1.1);
		\draw[teal,-<](-0.4,0.69)--(-0.44,0.9);
		\draw[teal](-0.4,0.69)--(-0.7,0.96);
		\draw[teal,-<](-0.4,0.69)--(-0.57,0.84);
		\coordinate (A4) at (-0.4,-0.69);
		\fill[blue] (A4) circle (1pt) node[left] {\tiny$\omega^2\xi_2$};
		\draw[blue](-0.4,-0.69)--(-0.32,-0.3);
		\draw[blue,->](-0.4,-0.69)--(-0.34,-0.42);
		\draw[blue](-0.4,-0.69)--(-0.1,-0.43);
		\draw[blue,->](-0.4,-0.69)--(-0.25,-0.56);
		\draw[blue](-0.4,-0.69)--(-0.48,-1.1);
		\draw[blue,-<](-0.4,-0.69)--(-0.44,-0.9);
		\draw[blue](-0.4,-0.69)--(-0.7,-0.96);
		\draw[blue,-<](-0.4,-0.69)--(-0.57,-0.84);
		\coordinate (n1) at (1.2,2.08);
		\fill[blue] (n1) circle (1pt) node[right] {\tiny$\omega^2\xi_4$};
		\draw[blue](1.2,2.08)--(1.12,1.7);
		\draw[blue,-<](1.2,2.08)--(1.14,1.8);
		\draw[blue](1.2,2.08)--(0.9,1.8);
		\draw[blue,-<](1.2,2.08)--(1.0,1.89);
		\draw[blue](1.2,2.08)--(1.28,2.5);
		\draw[blue,->](1.2,2.08)--(1.25,2.35);
		\draw[blue](1.2,2.08)--(1.5,2.38);
		\draw[blue,->](1.2,2.08)--(1.4,2.28);
		\coordinate (m1) at (1.2,-2.08);
		\fill[teal] (m1) circle (1pt) node[right] {\tiny$\omega\xi_4$};
		\draw[teal](1.2,-2.08)--(1.12,-1.7);
		\draw[teal,-<](1.2,-2.08)--(1.14,-1.8);
		\draw[teal](1.2,-2.08)--(0.9,-1.8);
		\draw[teal,-<](1.2,-2.08)--(1.0,-1.89);
		\draw[teal](1.2,-2.08)--(1.28,-2.5);
		\draw[teal,->](1.2,-2.08)--(1.25,-2.35);
		\draw[teal](1.2,-2.08)--(1.5,-2.38);
		\draw[teal,->](1.2,-2.08)--(1.4,-2.28);
		\coordinate (j1) at (-1.2,2.08);
		\fill[teal] (j1) circle (1pt) node[left] {\tiny$\omega\xi_1$};
		\draw[teal](-1.2,2.08)--(-1.12,1.7);
		\draw[teal,->](-1.2,2.08)--(-1.14,1.8);
		\draw[teal](-1.2,2.08)--(-0.9,1.8);
		\draw[teal,->](-1.2,2.08)--(-1.0,1.89);
		\draw[teal](-1.2,2.08)--(-1.28,2.5);
		\draw[teal,-<](-1.2,2.08)--(-1.25,2.35);
		\draw[teal](-1.2,2.08)--(-1.5,2.38);
		\draw[teal,-<](-1.2,2.08)--(-1.4,2.28);
		\coordinate[blue] (k1) at (-1.2,-2.08);
		\fill[blue] (k1) circle (1pt) node[left] {\tiny$\omega^2\xi_1$};
		\draw[blue](-1.2,-2.08)--(-1.12,-1.7);
		\draw[blue,->](-1.2,-2.08)--(-1.14,-1.8);
		\draw[blue](-1.2,-2.08)--(-0.9,-1.8);
		\draw[blue,->](-1.2,-2.08)--(-1.0,-1.89);
		\draw[blue](-1.2,-2.08)--(-1.28,-2.5);
		\draw[blue,-<](-1.2,-2.08)--(-1.25,-2.35);
		\draw[blue](-1.2,-2.08)--(-1.5,-2.38);
		\draw[blue,-<](-1.2,-2.08)--(-1.4,-2.28);
		\end{tikzpicture}
		\label{si1}}\\
	\subfigure[]{
		\begin{tikzpicture}
		\draw[dashed,-latex](-5,0)--(5,0)node[right]{ Re$z$};
		\coordinate (I) at (0,0);
		\fill (I) circle (1pt) node[below] {$0$};
		\coordinate (c) at (-3,0);
		\draw(-0.8,0)--(-0.4,0.2);
		\draw[-<](-0.8,0)--(-0.6,0.1);
		\draw(-0.8,0)--(-1.2,0.2);
		\draw[->](-0.8,0)--(-1,-0.1);
		\draw(-0.8,0)--(-0.4,-0.2);
		\draw[->](-0.8,0)--(-1,0.1);
		\draw(-0.8,0)--(-1.2,-0.2);
		\draw[-<](-0.8,0)--(-0.6,-0.1);
		\draw(-1.8,0)--(-1.4,0.2);
		\draw[->](-1.8,0)--(-2,0.1);
		\draw(-1.8,0)--(-2.2,-0.2);
		\draw[-<](-1.8,0)--(-1.6,-0.1);
		\draw(-1.8,0)--(-2.2,0.2);
		\draw[-<](-1.8,0)--(-1.6,0.1);
		\draw(-1.8,0)--(-1.4,-0.2);
		\draw[->](-1.8,0)--(-2,-0.1);
		\draw(-4.3,0)--(-4.7,0.2);
		\draw[-<](-4.3,0)--(-4.1,0.1);
		\draw(-4.3,0)--(-3.9,0.2);
		\draw[->](-4.3,0)--(-4.5,-0.1);
		\draw(-4.3,0)--(-4.7,-0.2);
		\draw[->](-4.3,0)--(-4.5,0.1);
		\draw(-4.3,0)--(-3.9,-0.2);
		\draw[-<](-4.3,0)--(-4.1,-0.1);
		\draw(-3.2,0)--(-2.8,0.2);
		\draw[->](-3.2,0)--(-3.4,0.1);
		\draw(-3.2,0)--(-3.6,0.2);
		\draw[-<](-3.2,0)--(-3,-0.1);
		\draw(-3.2,0)--(-3.6,-0.2);
		\draw[-<](-3.2,0)--(-3,0.1);
		\draw(-3.2,0)--(-2.8,-0.2);
		\draw[->](-3.2,0)--(-3.4,-0.1);
		\draw(0.8,0)--(0.4,0.2);
		\draw[->](0.8,0)--(0.6,0.1);
		\draw(0.8,0)--(0.4,-0.2);
		\draw[-<](0.8,0)--(1,-0.1);
		\draw(0.8,0)--(1.2,0.2);
		\draw[-<](0.8,0)--(1,0.1);
		\draw(0.8,0)--(1.2,-0.2);
		\draw[->](0.8,0)--(0.6,-0.1);
		\draw(1.8,0)--(1.4,0.2);
		\draw[-<](1.8,0)--(2,0.1);
		\draw(1.8,0)--(2.2,-0.2);
		\draw[->](1.8,0)--(1.6,-0.1);
		\draw(1.8,0)--(2.2,0.2);
		\draw[->](1.8,0)--(1.6,0.1);
		\draw(1.8,0)--(1.4,-0.2);
		\draw[-<](1.8,0)--(2,-0.1);
		\draw[red](4.3,0)--(4.7,0.2);
		\draw[red,->](4.3,0)--(4.1,0.1);
		\draw[red](4.3,0)--(3.9,0.2);
		\draw[red,-<](4.3,0)--(4.5,-0.1);
		\draw[red](4.3,0)--(4.7,-0.2);
		\draw[red,-<](4.3,0)--(4.5,0.1);
		\draw[red](4.3,0)--(3.9,-0.2);
		\draw[red,->](4.3,0)--(4.1,-0.1);
		\draw(3.2,0)--(2.8,0.2);
		\draw[-<](3.2,0)--(3.4,0.1);
		\draw(3.2,0)--(3.6,0.2);
		\draw[->](3.2,0)--(3,-0.1);
		\draw(3.2,0)--(3.6,-0.2);
		\draw[->](3.2,0)--(3,0.1);
		\draw(3.2,0)--(2.8,-0.2);
		\draw[-<](3.2,0)--(3.4,-0.1);
		\draw[blue,dashed](0,0)--(2.4,4.157)node[right]{ $e^{\frac{i\pi}{3}}$};	
		\draw[teal,dashed](0,0)--(2.4,-4.157);
		\draw[teal,dashed](0,0)--(-2.4,4.157);
		\draw[blue,dashed](0,0)--(-2.4,-4.157);
		\draw[dashed] (0.8,0) arc (0:360:0.8);
		\draw[dashed] (3.2,0) arc (0:360:3.2);
		\draw[dashed] (1.8,0) arc (0:360:1.8);
		\draw[dashed] (4.3,0) arc (0:360:4.3);
		\coordinate (I) at (0,0);
		\fill (I) circle (1pt) node[below] {$0$};
		\coordinate (A) at (-4.3,0);
		\fill (A) circle (1pt) node[below] {\tiny$\xi_8$};
		\coordinate (b) at (-3.2,0);
		\fill (b) circle (1pt) node[below] {\tiny$\xi_7$};
		\coordinate (C) at (-0.8,0);
		\fill (C) circle (1pt) node[below] {\tiny$\xi_5$};
		\coordinate (d) at (-1.8,0);
		\fill (d) circle (1pt) node[below] {\tiny$\xi_6$};
		\coordinate (E) at (4.3,0);
		\fill[red] (E) circle (1pt) node[below] {\tiny$\xi_1$};
		\coordinate (R) at (3.2,0);
		\fill (R) circle (1pt) node[below] {\tiny$\xi_2$};
		\coordinate (T) at (0.8,0);
		\fill (T) circle (1pt) node[below] {\tiny$\xi_4$};
		\coordinate (Y) at (1.8,0);
		\fill (Y) circle (1pt) node[below] {\tiny$\xi_3$};
		\coordinate (A1) at (0.4,0.69);
		\fill[blue] (A1) circle (1pt) node[right] {\tiny$\omega^2\xi_5$};
		\draw[blue](0.4,0.69)--(0.32,0.3);
		\draw[blue,-<](0.4,0.69)--(0.34,0.42);
		\draw[blue](0.4,0.69)--(0.1,0.43);
		\draw[blue,-<](0.4,0.69)--(0.25,0.56);
		\draw[blue](0.4,0.69)--(0.48,1.1);
		\draw[blue,->](0.4,0.69)--(0.44,0.9);
		\draw[blue](0.4,0.69)--(0.7,0.96);
		\draw[blue,->](0.4,0.69)--(0.57,0.84);
		\coordinate (A2) at (0.4,-0.69);
		\fill[teal] (A2) circle (1pt) node[right] {\tiny$\omega\xi_5$};
		\draw[teal](0.4,-0.69)--(0.32,-0.3);
		\draw[teal,-<](0.4,-0.69)--(0.34,-0.42);
		\draw[teal](0.4,-0.69)--(0.1,-0.43);
		\draw[teal,-<](0.4,-0.69)--(0.25,-0.56);
		\draw[teal](0.4,-0.69)--(0.48,-1.1);
		\draw[teal,->](0.4,-0.69)--(0.44,-0.9);
		\draw[teal](0.4,-0.69)--(0.7,-0.96);
		\draw[teal,->](0.4,-0.69)--(0.57,-0.84);
		\coordinate (A3) at (-0.4,0.69);
		\fill[teal] (A3) circle (1pt) node[left] {\tiny$\omega\xi_4$};
		\draw[teal](-0.4,0.69)--(-0.32,0.3);
		\draw[teal,->](-0.4,0.69)--(-0.34,0.42);
		\draw[teal](-0.4,0.69)--(-0.1,0.43);
		\draw[teal,->](-0.4,0.69)--(-0.25,0.56);
		\draw[teal](-0.4,0.69)--(-0.48,1.1);
		\draw[teal,-<](-0.4,0.69)--(-0.44,0.9);
		\draw[teal](-0.4,0.69)--(-0.7,0.96);
		\draw[teal,-<](-0.4,0.69)--(-0.57,0.84);
		\coordinate (A4) at (-0.4,-0.69);
		\fill[blue] (A4) circle (1pt) node[left] {\tiny$\omega^2\xi_4$};
		\draw[blue](-0.4,-0.69)--(-0.32,-0.3);
		\draw[blue,->](-0.4,-0.69)--(-0.34,-0.42);
		\draw[blue](-0.4,-0.69)--(-0.1,-0.43);
		\draw[blue,->](-0.4,-0.69)--(-0.25,-0.56);
		\draw[blue](-0.4,-0.69)--(-0.48,-1.1);
		\draw[blue,-<](-0.4,-0.69)--(-0.44,-0.9);
		\draw[blue](-0.4,-0.69)--(-0.7,-0.96);
		\draw[blue,-<](-0.4,-0.69)--(-0.57,-0.84);
		\coordinate (n1) at (1.6,2.77);
		\fill[blue] (n1) circle (1pt) node[right] {\tiny$\omega^2\xi_7$};
		\draw[blue](1.6,2.77)--(1.52,2.37);
		\draw[blue,-<](1.6,2.77)--(1.55,2.5);
		\draw[blue](1.6,2.77)--(1.28,2.5);
		\draw[blue,-<](1.6,2.77)--(1.34,2.55);
		\draw[blue](1.6,2.77)--(1.68,3.2);
		\draw[blue,->](1.6,2.77)--(1.66,3.1);
		\draw[blue](1.6,2.77)--(1.96,3.1);
		\draw[blue,->](1.6,2.77)--(1.86,3.01);
		\coordinate (m1) at (1.6,-2.77);
		\fill[teal] (m1) circle (1pt) node[right] {\tiny$\omega\xi_7$};
		\coordinate (j1) at (-1.6,2.77);
		\fill[teal] (j1) circle (1pt) node[left] {\tiny$\omega\xi_2$};
		\draw[teal](1.6,-2.77)--(1.52,-2.37);
		\draw[teal,-<](1.6,-2.77)--(1.55,-2.5);
		\draw[teal](1.6,-2.77)--(1.28,-2.5);
		\draw[teal,-<](1.6,-2.77)--(1.34,-2.55);
		\draw[teal](1.6,-2.77)--(1.68,-3.2);
		\draw[teal,->](1.6,-2.77)--(1.66,-3.1);
		\draw[teal](1.6,-2.77)--(1.96,-3.1);
		\draw[teal,->](1.6,-2.77)--(1.86,-3.01);
		\draw[teal](-1.6,2.77)--(-1.52,2.37);
		\draw[teal,->](-1.6,2.77)--(-1.55,2.5);
		\draw[teal](-1.6,2.77)--(-1.28,2.5);
		\draw[teal,->](-1.6,2.77)--(-1.34,2.55);
		\draw[teal](-1.6,2.77)--(-1.68,3.2);
		\draw[teal,-<](-1.6,2.77)--(-1.66,3.1);
		\draw[teal](-1.6,2.77)--(-1.96,3.1);
		\draw[teal,-<](-1.6,2.77)--(-1.86,3.01);
		\coordinate[blue] (k1) at (-1.6,-2.77);
		\fill[blue] (k1) circle (1pt) node[left] {\tiny$\omega^2\xi_2$};
		\draw[blue](-1.6,-2.77)--(-1.52,-2.37);
		\draw[blue,->](-1.6,-2.77)--(-1.55,-2.5);
		\draw[blue](-1.6,-2.77)--(-1.28,-2.5);
		\draw[blue,->](-1.6,-2.77)--(-1.34,-2.55);
		\draw[blue](-1.6,-2.77)--(-1.68,-3.2);
		\draw[blue,-<](-1.6,-2.77)--(-1.66,-3.1);
		\draw[blue](-1.6,-2.77)--(-1.96,-3.1);
		\draw[blue,-<](-1.6,-2.77)--(-1.86,-3.01);
		\coordinate (n2) at (0.9,1.56);
		\fill[blue] (n2) circle (1pt) node[right] {\tiny$\omega^2\xi_6$};
		\draw[blue](0.9,1.56)--(0.82,1.17);
		\draw[blue,-<](0.9,1.56)--(0.83,1.21);
		\draw[blue](0.9,1.56)--(0.58,1.31);
		\draw[blue,-<](0.9,1.56)--(0.65,1.37);
		\draw[blue](0.9,1.56)--(1.02,2.07);
		\draw[blue,->](0.9,1.56)--(0.98,1.91);
		\draw[blue](0.9,1.56)--(1.3,1.9);
		\draw[blue,->](0.9,1.56)--(1.18,1.8);
		\draw[blue](-0.9,-1.56)--(-0.82,-1.17);
		\draw[blue,->](-0.9,-1.56)--(-0.83,-1.21);
		\draw[blue](-0.9,-1.56)--(-0.58,-1.31);
		\draw[blue,->](-0.9,-1.56)--(-0.65,-1.37);
		\draw[blue](-0.9,-1.56)--(-1.02,-2.07);
		\draw[blue,-<](-0.9,-1.56)--(-0.98,-1.91);
		\draw[blue](-0.9,-1.56)--(-1.3,-1.9);
		\draw[blue,-<](-0.9,-1.56)--(-1.18,-1.8);
		\coordinate (m2) at (0.9,-1.56);
		\fill[teal] (m2) circle (1pt) node[right] {\tiny$\omega\xi_6$};
		\draw[teal](0.9,-1.56)--(0.82,-1.17);
		\draw[teal,-<](0.9,-1.56)--(0.83,-1.21);
		\draw[teal](0.9,-1.56)--(0.58,-1.31);
		\draw[teal,-<](0.9,-1.56)--(0.65,-1.37);
		\draw[teal](0.9,-1.56)--(1.02,-2.07);
		\draw[teal,->](0.9,-1.56)--(0.98,-1.91);
		\draw[teal](0.9,-1.56)--(1.3,-1.9);
		\draw[teal,->](0.9,-1.56)--(1.18,-1.8);
		\coordinate (j2) at (-0.9,1.56);
		\fill[teal] (j2) circle (1pt) node[left] {\tiny$\omega\xi_3$};
		\coordinate[blue] (k2) at (-0.9,-1.56);
		\fill[blue] (k2) circle (1pt) node[left] {\tiny$\omega^2\xi_3$};
		\draw[teal](-0.9,1.56)--(-0.82,1.17);
		\draw[teal,->](-0.9,1.56)--(-0.83,1.21);
		\draw[teal](-0.9,1.56)--(-0.58,1.31);
		\draw[teal,->](-0.9,1.56)--(-0.65,1.37);
		\draw[teal](-0.9,1.56)--(-1.02,2.07);
		\draw[teal,-<](-0.9,1.56)--(-0.98,1.91);
		\draw[teal](-0.9,1.56)--(-1.3,1.9);
		\draw[teal,-<](-0.9,1.56)--(-1.18,1.8);
		\coordinate (n23) at (2.15,3.72);
		\fill[blue] (n23) circle (1pt) node[right] {\tiny$\omega^2\xi_8$};
		\draw[blue](2.15,3.72)--(2.23,4.17);
		\draw[blue,->](2.15,3.72)--(2.22,4.11);
		\draw[blue](2.15,3.72)--(2.48,4.03);
		\draw[blue,->](2.15,3.72)--(2.45,4);
		\draw[blue](2.15,3.72)--(2.09,3.35);
		\draw[blue,-<](2.15,3.72)--(2.1,3.42);
		\draw[blue](2.15,3.72)--(1.83,3.46);
		\draw[blue,-<](2.15,3.72)--(1.88,3.5);
		\coordinate (m23) at (2.15,-3.72);
		\fill[teal] (m23) circle (1pt) node[right] {\tiny$\omega\xi_8$};
		\coordinate (j23) at (-2.15,3.72);
		\fill[teal] (j23) circle (1pt) node[left] {\tiny$\omega\xi_1$};
		\draw[teal](2.15,-3.72)--(2.23,-4.17);
		\draw[teal,->](2.15,-3.72)--(2.22,-4.11);
		\draw[teal](2.15,-3.72)--(2.48,-4.03);
		\draw[teal,->](2.15,-3.72)--(2.45,-4);
		\draw[teal](2.15,-3.72)--(2.09,-3.35);
		\draw[teal,-<](2.15,-3.72)--(2.1,-3.42);
		\draw[teal](2.15,-3.72)--(1.83,-3.46);
		\draw[teal,-<](2.15,-3.72)--(1.88,-3.5);
		\draw[teal](-2.15,3.72)--(-2.23,4.17);
		\draw[teal,-<](-2.15,3.72)--(-2.22,4.11);
		\draw[teal](-2.15,3.72)--(-2.48,4.03);
		\draw[teal,-<](-2.15,3.72)--(-2.45,4);
		\draw[teal](-2.15,3.72)--(-2.09,3.35);
		\draw[teal,->](-2.15,3.72)--(-2.1,3.42);
		\draw[teal](-2.15,3.72)--(-1.83,3.46);
		\draw[teal,->](-2.15,3.72)--(-1.88,3.5);
		\coordinate[blue] (k3) at (-2.15,-3.72);
		\fill[blue] (k3) circle (1pt) node[left] {\tiny$\omega^2\xi_1$};
		\draw[blue](-2.15,-3.72)--(-2.23,-4.17);
		\draw[blue,-<](-2.15,-3.72)--(-2.22,-4.11);
		\draw[blue](-2.15,-3.72)--(-2.48,-4.03);
		\draw[blue,-<](-2.15,-3.72)--(-2.45,-4);
		\draw[blue](-2.15,-3.72)--(-2.09,-3.35);
		\draw[blue,->](-2.15,-3.72)--(-2.1,-3.42);
		\draw[blue](-2.15,-3.72)--(-1.83,-3.46);
		\draw[blue,->](-2.15,-3.72)--(-1.88,-3.5);
		\end{tikzpicture}
		\label{si2}}
	\caption{\footnotesize Figure (a) is the local jump contour $\Sigma^{lo}$ consisting of   12   crosses  for the case $0\leq\xi<1$;
  Figure (b)  is the jump contour $\Sigma^{lo}$ consisting of  24  crosses for the case $- {1}/{8}<\xi<0$.}
	\label{sigma0}
\end{figure}

\begin{RHP} \label{rh7}
	Find a matrix-valued function  $ M^{lo}(z)$ with following properties:
	
	$\blacktriangleright$ Analyticity: $M^{lo}(z)$ is analytical  in $\mathbb{C}\setminus \Sigma^{ lo } $;
	
	$\blacktriangleright$ Symmetry: $M^{lo}(z)=\Gamma_1\overline{M^{lo}(\bar{z})}\Gamma_1=\Gamma_2\overline{M^{lo}(\omega^2\bar{z})}\Gamma_2=\Gamma_3\overline{M^{lo}(\omega\bar{z})}\Gamma_3$

\quad and $M^{lo}(z)=\overline{M^{lo}(\bar{z}^{-1})}$;
	
	$\blacktriangleright$ Jump condition: $M^{lo}$ has continuous boundary values $M^{lo}_\pm$ on $\Sigma $ and
	\begin{equation}
	M^{lo}_+(z)=M^{lo}_-(z)V^{lo}(z),\hspace{0.5cm}z \in \Sigma^{lo},\label{jump6}
	\end{equation}
	where $ V^{lo}(z)= V^{(2)}(z)\big|_{\Sigma^{lo}}$.

	$\blacktriangleright$ Asymptotic behaviors:
	\begin{align}
	M^{lo}(z) =& I+\mathcal{O}(z^{-1}),\hspace{0.5cm}z \rightarrow \infty.
	\end{align}
\end{RHP}	

This  local  RH problem, which consists of $3p(\xi)$ local models on $\Sigma^{lo,n}_k$ about
phase point $\xi_{n,k}$,   has the jump condition and no poles.
First, we show as $t \to \infty$, the interaction between $\Sigma^{lo,n}_k$ and $\Sigma^{lo,n}_j$  reduces to  0 to higher order and the contribution to
 the solution of $M^{ lo }(z)$ is simply the sum of the separate contributions from $\Sigma^{lo,n}_k$.

     We consider the trivial decomposition of the jump matrix
           \begin{align*}
           &V^{lo}(z)=b_-^{-1}(z) b_+(z), \ \  \; b_-(z)=I,  \ \ \; b_+(z)=V^{lo}(z),\\
          	&w_-(z)=0,\ \   w_+(z)= V^{lo}(z)-I,\ \  w(z) =V^{lo}(z)-I.
          \end{align*}
Define
\begin{align}
&w_k^{n}  (z) =w(z)\big|_{\Sigma^{lo,n}_k},\ \  w_{k,\pm}^{n}  (z) =w_\pm (z)\big|_{\Sigma^{lo,n}_k},\nonumber\\
&  V^{lo,n}_k (z)=V^{lo}(z)|_{\Sigma^{lo,n}_k}, \ \ w_k^{n}  (z) =0, \ z\in \Sigma^{lo,n}_j, j\not= k.\nonumber
\end{align}
which  determine a  local model with the  jump contour  $\Sigma^{lo,n}_k$  at
phase point $\xi_{n,k}$
\begin{RHP} \label{rh71}
	Find a matrix-valued function  $ M^{lo,n}_k (z)$ with following properties:
	
	$\blacktriangleright$ Analyticity: $M^{lo,n}_k (z)$ is analytical  in $\mathbb{C}\setminus \Sigma^{lo,n}_k $;

	$\blacktriangleright$ Jump condition:  $M^{lo,n}_k (z)$ has continuous boundary values $M^{lo,n}_{k,\pm} $  and
	\begin{equation}
	M^{lo,n}_{k,+} (z) =M^{lo,n}_{k,-}(z) V^{lo,n}_k (z),\hspace{0.5cm}z \in \Sigma^{lo,n}_k.\label{jump6}
	\end{equation}

	$\blacktriangleright$ Asymptotic behaviors:
	\begin{align}
	M^{lo,n}_k(z) =& I+\mathcal{O}(z^{-1}),\hspace{0.5cm}z \rightarrow \infty.
	\end{align}
\end{RHP}	

By using above factorization,  $w(z)$  can be written as
\begin{align}
&w=  \underset{n=0,1,2}{\underset{k=1,..,p(\xi)}{\Sigma}}  w_k^{n}  (z), \  \ \ k=1,...,p(\xi), n=0,1,2. \nonumber
\end{align}
 and   the   Cauchy  operators are given by
\begin{align}
C_w(f)= \mathcal{P}^-(fw_+),\hspace{0.5cm} C_{w_k^n}(f)= \mathcal{P}^- (fw_{k,+}^n), \ \ C_w=  \underset{n=0,1,2}{\underset{k=1,..,p(\xi)}{\Sigma}} C_{w_k^n},\label{eew}
\end{align}
where  $\mathcal{P}^-$ is the Cauchy projection operator defined by
\begin{equation}
\mathcal{P}^- (f)(z)=\lim_{\varepsilon\to 0}\frac{1}{2\pi i}\int_{\Sigma^{lo}/\Sigma^{lo,n}_k}\dfrac{f(s)}{s-(z-i\varepsilon) }ds,  \ \  z\in \Sigma^{lo}/\Sigma^{lo,n}_k,
\end{equation}

A simple calculation gives
\begin{lemma} The matrix functions  $w   (z)$ and  $w_k^{n}  (z)$ defined above admits following  estimate
	\begin{align}
	\parallel w (z) \parallel_{L^q(\Sigma^{lo})}, \ \ \parallel w_k^{n}(z) \parallel_{L^p(\Sigma^{lo,n}_k )}=\mathcal{O}(t^{-1/2}),\ 1\leq q<+\infty.
	\end{align}
\end{lemma}
 This lemma  implies that $I-C_w$ and $I-C_{w_k^n}$ are reversible,  so  both   RHP \ref{rh7} and   RHP \ref{rh71}  exists   a  unique solution and   can be written as
\begin{align}
M^{lo}=I+\frac{1}{2\pi i}\int_{\Sigma^{lo} }\frac{(I-C_w)^{-1}I\ w}{s-z}ds.
\end{align}
Then following the step of \cite{RN6},  we can directly derive that
\begin{Proposition}\label{dividepc}
	As $t\to+\infty$,  we have
	\begin{align}
	\parallel C_{w_k^n}C_{w_j^n}\parallel_{ L^2\to L^2 }\lesssim t^{-1},\hspace{0.5cm}\parallel C_{w_k^n}C_{w_j^n}\parallel_{L^\infty \to L^2 }\lesssim t^{-1}.
	\end{align}
	\begin{align}
	\int_{\Sigma^{lo} }\frac{(I-C_w)^{-1}I\ w}{s-z}ds
=\underset{n=0,1,2}{\underset{k=1,..,p(\xi)}{\Sigma}}\int_{\Sigma^{lo,n}_k}\frac{(I-C_{w_k^n})^{-1}I\ w_k^n}{s-z}ds+\mathcal{O}(t^{-3/2}).
	\end{align}
\end{Proposition}
This Proposition implies  that
 contributions of every cross  $\Sigma^{lo,n}_k$ can be separated out, and  the solution   $M^{lo}$ of the RHP \ref{rh7}
can be given by  the sum of  the  local  RHP \ref{rh71} of  all phase points $\xi_{n,k}$.
We further show that the solution of each local model can be given explicitly in terms of parabolic cylinder
functions on every contour $\Sigma^{lo,n}_k$  when  $ t \to \infty$.
As illustrative example,  we   only consider a local model  at phase point $\xi_{0,1} :=  \xi_1$ in the  RHP \ref{rh71}, its jump cross
denotes  (red cross in Figure \ref{sigma0})
$$\Sigma^{lo }_1= (\underset{j=1,...,4}{\cup} \Sigma_{1j} )\cap U(\xi)=\{z=\xi_1+le^{\pm\varphi i},\ l\in\mathbb{R}\},$$
which corresponds to the following   local RH problem
\begin{RHP}\label{RHPlo1}
	Find a matrix-valued function  $ M^{lo,1}(z)$ with following properties:
	
$\blacktriangleright$ Analyticity: $M^{lo,1}(z)$ is analytical  in $\mathbb{C}\setminus \Sigma^{lo }_1 $;
	
	$\blacktriangleright$ Jump condition: $M^{lo,1}$ has continuous boundary values $M^{lo,1}_\pm$ and
	\begin{equation}
	M^{lo,1}_+(z)=M^{lo,1}_-(z)V^{lo,1}(z),\hspace{0.5cm}z \in \Sigma^{lo }_1,
	\end{equation}
	where  when $|z-\xi_1|\leq \varrho^{(0)}$,  for the case  $\xi\in(-1/8,0)$,
	\begin{align}
	V^{lo,1}(z)=\left\{\begin{array}{ll}
	\left(\begin{array}{ccc}
	1 & r(\xi_1)T_{12}^{(1)}(\xi)^{-1}\left(\eta (z-\xi_1)\right)^{-2i\eta\nu(\xi_i)} e^{it\theta_{12}}& 0\\
	0 & 1 & 0\\
	0 & 0&1
	\end{array}\right),  & z\in \Sigma_{11}\cap U(\xi),\\[10pt]
	\left(\begin{array}{ccc}
	1 & 0 &0\\
	-\bar{r}(\xi_1)T_{12}^{(1)}(\xi)\left(\eta (z-\xi_1)\right)^{2i\eta\nu(\xi_1)}e^{-it\theta_{12}} & 1 & 0\\
	0 & 0&1
	\end{array}\right),   & z\in \Sigma_{12}\cap U(\xi),\\[10pt]
	\left(\begin{array}{ccc}
	1& \frac{r(\xi_1)}{1-|r(\xi_1)|^2} T_{12}^{(1)}(\xi)^{-1}\left(\eta (z-\xi_1)\right)^{-2i\eta\nu(\xi_i)} e^{it\theta_{12}} &0\\
	0&1 & 0\\
	0 & 0&1
	\end{array}\right),   & z\in \Sigma_{13}\cap U(\xi),\\[10pt]
	\left(\begin{array}{ccc}
	1 & 0&0\\
	\frac{-\bar{r}(\xi_1)}{1-|r(\xi_1)|^2}T_{12}^{(1)}(\xi)\left(\eta (z-\xi_1)\right) ^{2\eta i\nu(\xi_1)}e^{-it\theta_{12}} & 1 & 0\\
	0 & 0&1
	\end{array}\right),   & z\in \Sigma_{14}\cap U(\xi).
	\end{array}\right. \label{efe}
	\end{align}
	For the case  $\xi\in(0,1)$,  the corresponding jump matrix   $V^{lo,1}(z)$ from above   (\ref{efe})
by exchanging    $\Sigma_{11} \longleftrightarrow \Sigma_{14}, \ \ \Sigma_{12} \longleftrightarrow \Sigma_{13}$.

	$\blacktriangleright$ Asymptotic behaviors:
	\begin{align}
	M^{lo,1}(z) =& I+\mathcal{O}(z^{-1}),\hspace{0.5cm}z \rightarrow \infty.
	\end{align}
\end{RHP}

 In  near $\xi_1$,   we rewrite phase function as
\begin{align}
\theta_{12}(z)=\theta_{12}(\xi_1)+(z-\xi_1)^2\theta_{12}''(\xi_1)+\mathcal{O}((z-\xi_1)^3).
\end{align}
For  $\xi\in[0,1)$, $\theta_{12}''(\xi_1)>0$ and when $\xi\in(-1/8,0)$, $\theta_{12}''(\xi_1)<0$.
Recall the $\eta(\xi,j)$ defined in (\ref{eta}).

The RHP \ref{RHPlo1} is a local model  and does not possess the symmetry.
In order to motivate the model, let $\zeta = \zeta(z)$  denote the  rescaled   local variable
\begin{align}
\zeta(z)=t^{1/2}\sqrt{-4\eta(\xi,1)\theta''(\xi_1)}(z-\xi_1),
\end{align}
where  $\eta(\xi, 1)=1$  as $\xi\in(-1/8,0)$   and $\eta(\xi, 1)=-1$ as  $\xi\in[0,1)$ in (\ref{eta}).
This change of variable maps $U_{\xi_1}$ to an   neighborhood of $\zeta= 0$. Additionally, let
\begin{align}
r_{\xi_1}=r(\xi_1)T_1(\xi)^{2}e^{-2it\theta(\xi_1)}\eta^{-2i\eta\nu(\xi_1)}\exp\left\lbrace -i\eta\nu(\xi_1)\log \left( 4t\theta''(\xi_1)\tilde{\eta}(\xi_1)\right) \right\rbrace ,
\end{align}
with $|r_{\xi_1}|=|r(\xi_1)|$.
In the above expression, the complex powers are defined by choosing the branch of
the logarithm with  $-\pi< \arg \zeta < \pi$ in the cases $\xi\in(-1/8,0)$, and the branch of the logarithm with $0 < \arg \zeta < 2\pi$ in the case $\xi\in[0,1)$.

Through this change of variable,  the jump matrix  $V^{lo,1}(z)$ approximates to  the jump matrix of a parabolic cylinder model problem as follows.
 \begin{RHP}\label{RHPpc}
 	Find a matrix-valued function  $ M^{pc}(\zeta;\xi)$ with following properties:
 	
 	$\blacktriangleright$ Analyticity: $M^{pc}(\zeta;\xi)$ is analytical  in $\mathbb{C}\setminus \Sigma^{pc} $ with $\Sigma^{pc}=\left\lbrace\mathbb{R}e^{\varphi i} \right\rbrace \cup \left\lbrace\mathbb{R}e^{(\pi-\varphi) i} \right\rbrace$ shown in Figure \ref{sigpc};

 	$\blacktriangleright$ Jump condition: $M^{pc}(\zeta;\xi)$ has continuous boundary values $M^{pc}_\pm$ on $\Sigma^{pc}$ and
 	\begin{equation}
 	M^{pc}_+(\zeta;\xi)=M^{pc}_-(\zeta;\xi)V^{pc}(\zeta),\hspace{0.5cm}\zeta \in \Sigma^{\zeta},
 	\end{equation}
 	where   in the case $\xi\in(-1/8,0)$
 	\begin{align}
 		V^{pc}(\zeta;\xi)=\left\{\begin{array}{ll}
 	\left(\begin{array}{ccc}
 	1 & 	r_{\xi_1}\zeta^{-2i\nu(\xi_1)}e^{-\frac{i}{2}\zeta^2}& 0 \\
  0& 1 &0 \\
  0&  0  &1
 	\end{array}\right),  & \zeta\in\mathbb{R}^+e^{\varphi i},\\[10pt]
 	\left(\begin{array}{ccc}
 	1 &0 &0 \\
 	-\bar{r}_{\xi_1}\zeta^{2i\nu(\xi_1)}e^{\frac{i}{2}\zeta^2}&1 &0 \\
 	0&  0  &1
 	\end{array}\right),   & \zeta\in \mathbb{R}^+e^{-\varphi i},\\[10pt]
 	\left(\begin{array}{ccc}
 	1 & \frac{r_{\xi_1}}{1-|r_{\xi_1}|^2}\zeta^{-2i\nu(\xi_1)}e^{-\frac{i}{2}\zeta^2} &0\\
 	0 & 1 &0\\
 	0 & 0&1
 	\end{array}\right),   & \zeta\in \mathbb{R}^+e^{(-\pi+\varphi) i},\\[10pt]
  \left(\begin{array}{ccc}
 	1 & 0&0\\
 	\frac{-\bar{r}_{\xi_1}}{1-|r_{\xi_1}|^2}\zeta^{2i\nu(\xi_1)}e^{\frac{i}{2}\zeta^2} & 1 & 0\\
 	0 & 0&1
 \end{array}\right),   & \zeta\in \mathbb{R}^+e^{(\pi-\varphi) i}.
 	\end{array}\right.
 	\end{align}
 	and in the case $\xi\in[0,1)$
 		\begin{align}
	V^{pc}(\zeta;\xi)=\left\{\begin{array}{ll}
 \left(\begin{array}{ccc}
 	1 & 0&0\\
 	\frac{-\bar{r}_{\xi_1}}{1-|r_{\xi_1}|^2}\zeta^{-2i\nu(\xi_1)}e^{-\frac{i}{2}\zeta^2} & 1 & 0\\
 	0 & 0&1
 \end{array}\right),  & \zeta\in\mathbb{R}^+e^{\varphi i},\\[10pt]
 \left(\begin{array}{ccc}
 	1 & \frac{r_{\xi_1}}{1-|r_{\xi_1}|^2}\zeta^{2i\nu(\xi_1)}e^{\frac{i}{2}\zeta^2} & 0\\
 	0 & 1 & 0\\
 	0 & 0& 1
 \end{array}\right),   & \zeta\in \mathbb{R}^+e^{(2\pi-\varphi) i},\\[10pt]
 	\left(\begin{array}{ccc}
 	1 &0 &0 \\
 	-\bar{r}_{\xi_1}\zeta^{-2i\nu(\xi_1)}e^{\frac{i}{2}\zeta^2} &1 &0 \\
 	0&  0  &1
 \end{array}\right),   & \zeta\in \mathbb{R}^+e^{(\pi+\varphi) i},\\[10pt]
\left(\begin{array}{ccc}
	1 & 	r_{\xi_1}\zeta^{2i\nu(\xi_1)}e^{-\frac{i}{2}\zeta^2}& 0 \\
	0& 1 &0 \\
	0&  0  &1
\end{array}\right),   & \zeta\in \mathbb{R}^+e^{ (\pi-\varphi)i}.
 \end{array}\right.
 	\end{align}

 	$\blacktriangleright$ Asymptotic behaviors:
 	\begin{align}
 	M^{pc}(\zeta;\xi) =& I+M^{pc}_1\zeta^{-1}+\mathcal{O}(\zeta^{-2}),\hspace{0.5cm}\zeta \rightarrow \infty.\label{asyMpc1}
 	\end{align}
 \end{RHP}	
\begin{figure}[t]
	\centering
	\subfigure[]{
		\begin{tikzpicture}[node distance=2cm]
		\draw[](0,2.7)node[above]{$\xi\in(-1/8,0)$};
		\draw(0,0)--(2,1.2)node[above]{$\mathbb{R}^+e^{\varphi i}$};
		\draw(0,0)--(-2,1.2)node[above]{$\mathbb{R}^+e^{(\pi-\varphi)i}$};
		\draw(0,0)--(-2,-1.2)node[below]{$\mathbb{R}^+e^{(-\pi+\varphi)i}$};
		\draw(0,0)--(2,-1.2)node[below]{$\mathbb{R}^+e^{-\varphi i}$};
		\draw[-<](0,0)--(1,-0.6);
		\draw[-<](0,0)--(1,0.6);
		\draw[dashed](-2,0)--(2,0)node[right]{\scriptsize Re$z$};
		\draw[-<](-2,-1.2)--(-1,-0.6);
		\draw[-<](-2,1.2)--(-1,0.6);
		\coordinate (A) at (-1.2,0.5);
		\coordinate (B) at (-1.2,-0.5);
		\coordinate (G) at (1.4,0.5);
		\coordinate (H) at (1.4,-0.5);
		\coordinate (I) at (0,0);
		\fill (I) circle (1pt) node[below] {$0$};
		\end{tikzpicture}
	}
	\subfigure[]{
		\begin{tikzpicture}[node distance=2cm]
		\draw[](0,2.7) node[above]{$\xi\in[0,1)$};
		\draw(0,0)--(2,1.5)node[above]{$\mathbb{R}^+e^{\varphi i}$};
		\draw(0,0)--(-2,1.5)node[above]{$\mathbb{R}^+e^{(\pi-\varphi)i}$};
		\draw[-<](0,0)--(1,0.75);
		\draw(0,0)--(-2,-1.5)node[below]{$\mathbb{R}^+e^{(\pi+\varphi)i}$};
		\draw(0,0)--(2,-1.5)node[below]{$\mathbb{R}^+e^{(2\pi-\varphi)i}$};
		\draw[-<](0,0)--(1,-0.75);
		\draw[dashed](-2,0)--(2,0)node[right]{\scriptsize Re$z$};
		\draw[-<](-2,-1.5)--(-1,-0.75);
		\draw[-<](-2,1.5)--(-1,0.75);
		\coordinate (A) at (1,0.5);
		\coordinate (B) at (1,-0.5);
		\coordinate (G) at (-1,0.5);
		\coordinate (H) at (-1,-0.5);
		\coordinate (I) at (0,0);
		\fill (I) circle (1pt) node[below] {$0$};
		\end{tikzpicture}
	}
	\caption{\footnotesize The contour $\Sigma^{pc}$ in case $\xi\in[0,1)$ and $\xi\in(-1/8,0)$, respectively.}
	\label{sigpc}
\end{figure}
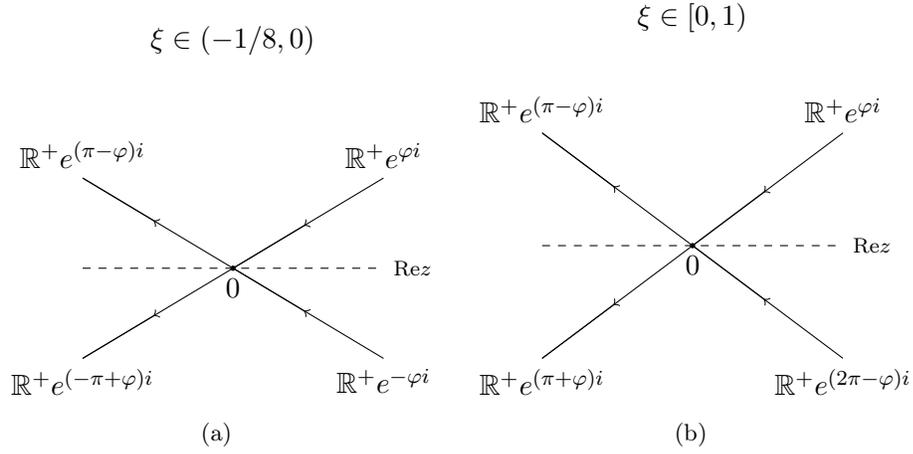
Moreover, change  variable of $M^{lo,1}(z)=M^{lo,1}(\zeta(z))$ and take $\zeta\to \infty$,
\begin{align}
	M^{lo,1}(\zeta)=I+\frac{M^{lo,1}_{(1)}}{\zeta}+\mathcal{O}(\zeta^{-2}),\label{asyMlo}
\end{align}
where $M^{lo,1}_{(1)}$ is unknown.

\begin{Proposition}
	As $t\to\infty$, the error between  $M^{lo,1}_{(1)}$ defined in (\ref{asyMlo}) and $M^{pc}_1$ defined in (\ref{asyMpc1}) is
	\begin{align}
		M^{lo,1}_{(1)}=M^{pc}_1+t^{-1/2}.
	\end{align}
\end{Proposition}
The proof of above proposition is similar as Theorem A.1. in \cite{HG2009}.

Then
\begin{align}
M^{lo,1}(z)=I+\frac{t^{-1/2}}{z-\xi_1}\frac{i\tilde{\eta}}{2} \left(\begin{array}{ccc}
0 & [M^{pc}_1]_{12} & 0\\
-[M^{pc}_1]_{21} & 0 & 0 \\
0&0&0
\end{array}\right)+\mathcal{O}(t^{-1}).\label{asyMlo1}
\end{align}
The RHP \ref{RHPpc} has an explicit solution $	M^{pc}(\zeta)$, which is expressed in terms of solutions of the parabolic cylinder equation $\left(\frac{\partial^2}{\partial z^2}+\left(\frac{1}{2}-\frac{z^2}{2}+a \right)  \right)D_a(z)=0 $. In fact, let
\begin{align}
M^{pc}(\zeta;\xi)=\Psi(\zeta;\xi)P(\xi)e^{-\frac{i}{4}\eta\zeta^2\sigma_3}\zeta^{-i\eta\nu(\xi_1)\sigma_3},
\end{align}
where  in the case $\xi\in[0,1)$
\begin{align}
P(\xi)=\left\{\begin{array}{ll}
\left(\begin{array}{ccc}
1 & 0 &0\\
-\frac{\bar{r}_{\xi_1}}{1-|r_{\xi_1}|^2} & 1&0\\
0&0&1
\end{array}\right),  & \arg\zeta\in(0,\varphi) ,\\[10pt]
\left(\begin{array}{ccc}
1&  -\frac{r_{\xi_1}}{1-|r_{\xi_1}|^2}&0\\
0 &1&0\\
0&0&1
\end{array}\right),   & \arg\zeta\in (2\pi-\varphi,2\pi) ,\\[10pt]
\left(\begin{array}{ccc}
1 &0&0\\
\bar{r}_{\xi_1} & 1&0\\
0&0&1
\end{array}\right),   & \arg\zeta\in (\pi,\pi+\varphi),\\[10pt]
\left(\begin{array}{ccc}
1 & r_{\xi_1}&0 \\
0 & 1 &0\\
0&0&1
\end{array}\right),   & \arg\zeta\in (\pi-\varphi,\pi),\\[10pt]
I,   & else.
\end{array}\right.
\end{align}
and in the case $\xi\in(-1/8,0)$
\begin{align}
P(\xi)=\left\{\begin{array}{ll}
\left(\begin{array}{ccc}
	1 & r_{\xi_1}&0 \\
	0 & 1 &0\\
	0&0&1
\end{array}\right),  & \arg\zeta\in(0,\varphi) ,\\[10pt]
\left(\begin{array}{ccc}
	1 &0&0\\
	\bar{r}_{\xi_1} & 1&0\\
	0&0&1
\end{array}\right),   & \arg\zeta\in(-\varphi,0) ,\\[10pt]
\left(\begin{array}{ccc}
	1&  -\frac{r_{\xi_1}}{1-|r_{\xi_1}|^2}&0\\
	0 &1&0\\
	0&0&1
\end{array}\right),   & \arg\zeta\in (-\pi+\varphi,-\pi) ,\\[10pt]
\left(\begin{array}{ccc}
	1 & 0 &0\\
	-\frac{\bar{r}_{\xi_1}}{1-|r_{\xi_1}|^2} & 1&0\\
	0&0&1
\end{array}\right),   & \arg\zeta\in(\pi-\varphi,\pi),\\[10pt]
I,   & else.
\end{array}\right.
\end{align}
And \begin{align}
	\sigma_3=\left(\begin{array}{ccc}
		1 & 0 &0\\
		0 & -1&0\\
		0&0&0
	\end{array}\right),\ with\ a^{\sigma_3}=\left(\begin{array}{ccc}
	a & 0 &0\\
	0 & a^{-1}&0\\
	0&0&1
\end{array}\right),\text{ for nozero constant }a.
\end{align}
By construction, the matrix $\Psi$ is continuous along the rays of $\Sigma^{pc}$. And Due to the branch cut of the logarithmic function along $\tilde{\eta}\mathbb{R}^+$, the matrix $\Psi$  has the same (constant) jump matrix along the negative and positive real axis. The function $\Psi(\zeta;\xi)$  satisfies the following model RHP.
 \begin{RHP}\label{RHPpsi}
	Find a matrix-valued function  $ \Psi(\zeta;\xi)$ with following properties:
	
	$\blacktriangleright$ Analyticity: $\Psi(\zeta;\xi)$ is analytical  in $\mathbb{C}\setminus \mathbb{R}$;

	$\blacktriangleright$ Jump condition: $\Psi(\zeta;\xi)$ has continuous boundary values $\Psi_\pm(\zeta;\xi)$ on $\mathbb{R}$ and
	\begin{equation}
	\Psi_+(\zeta;\xi)=\Psi_-(\zeta;\xi)V^{\Psi}(\zeta),\hspace{0.5cm}\zeta \in \mathbb{R},
	\end{equation}
	where
	\begin{align}
	V^{\Psi}(\xi)=	\left(\begin{array}{ccc}
	1 & -r_{\xi_1} &0\\
	\bar{r}_{\xi_1} & 1-|r_{\xi_1}|^2 &0\\
	0&0&1
	\end{array}\right).
	\end{align}

	$\blacktriangleright$ Asymptotic behaviors:
	\begin{align}
	\Psi(\zeta;\xi) \sim& \left( I+M^{pc}_1\zeta^{-1}\right) \zeta^{i\eta\nu(\xi_1)\sigma_3}e^{\frac{i}{4}\eta\zeta^2\sigma_3},\hspace{0.5cm}\zeta \rightarrow \infty.
	\end{align}
\end{RHP}	
For  brevity, denote $\tilde{\beta}^1_{12}=i\eta[M^{pc}_1]_{12}$, $\tilde{\beta}^1=-i\eta[M^{pc}_1]_{21}$ and $\nu_1=\nu(\xi_1)$. The unique solution to Problem \ref{RHPpsi} is:

1. $\xi\in[0,1)$,
when $\zeta\in\mathbb{C}^+$,
\begin{align}
\Psi(\zeta;\xi)=	\left(\begin{array}{ccc}
e^{\frac{3}{4}\pi\nu_1}D_{-i\nu_1}(e^{-\frac{3}{4}\pi i}\zeta) & \frac{i\nu_1}{\tilde{\beta}^1_{21}}e^{-\frac{\pi}{4}(\nu_1+i)}D_{i\nu_1-1}(e^{-\frac{\pi i}{4} }\zeta) &0\\
-\frac{i\nu_1}{\tilde{\beta}^1_{12}}e^{\frac{3\pi}{4}(\nu_1-i)}D_{-i\nu_1-1}(e^{-\frac{3\pi i}{4} }\zeta) & e^{-\frac{\pi}{4}\nu_1}D_{i\nu_1}(e^{-\frac{\pi}{4} i}\zeta)&0\\
0&0&1
\end{array}\right),\nonumber
\end{align}
when $\zeta\in\mathbb{C}^-$,
\begin{align}
\Psi(\zeta;\xi)=	\left(\begin{array}{ccc}
e^{\frac{7\pi}{4}\nu_1}D_{-i\nu_1}(e^{-\frac{7\pi}{4} i}\zeta) & \frac{i\nu_1}{\tilde{\beta}^1_{21}}e^{-\frac{5\pi}{4}(\nu_1+i)}D_{i\nu_1-1}(e^{-\frac{5\pi i}{4} }\zeta)&0\\
-\frac{i\nu_1}{\tilde{\beta}^1_{12}}e^{\frac{7\pi}{4}(\nu_1-i)}D_{-i\nu_1-1}(e^{-\frac{7\pi i}{4} }\zeta) & e^{-\frac{5}{4}\pi\nu_1}D_{i\nu_1}(e^{-\frac{5}{4}\pi i}\zeta)&0\\
0&0&1
\end{array}\right).\nonumber
\end{align}

2. $\xi\in(-1/8,0)$,
when $\zeta\in\mathbb{C}^+$,
\begin{align}
\Psi(\zeta;\xi)=	\left(\begin{array}{ccc}
e^{-\frac{\pi}{4}\nu_1}D_{i\nu_1}(e^{-\frac{\pi}{4} i}\zeta) & -\frac{i\nu_1}{\tilde{\beta}^1_{21}}e^{\frac{3\pi}{4}(\nu_1-i)}D_{-i\nu_1-1}(e^{-\frac{3\pi i}{4} }\zeta)&0\\
\frac{i\nu_1}{\tilde{\beta}^1_{12}}e^{-\frac{\pi}{4}(\nu_1+i)}D_{i\nu_1-1}(e^{-\frac{\pi i}{4} }\zeta) & e^{\frac{3}{4}\pi\nu_1}D_{-i\nu_1}(e^{-\frac{3}{4}\pi i}\zeta)&0\\
0&0&1
\end{array}\right),\nonumber
\end{align}
when $\zeta\in\mathbb{C}^-$,
\begin{align}
\Psi(\zeta;\xi)=	\left(\begin{array}{ccc}
 e^{\frac{3}{4}\pi\nu_1}D_{i\nu_1}(e^{\frac{3}{4}\pi i}\zeta) & -\frac{i\nu_1}{\tilde{\beta}^1_{21}}e^{-\frac{\pi}{4}(\nu_1-i)}D_{-i\nu_1-1}(e^{\frac{\pi i}{4} }\zeta)&0\\
\frac{i\nu_1}{\tilde{\beta}^1_{12}}e^{\frac{3\pi}{4}(\nu_1+i)}D_{i\nu_1-1}(e^{\frac{3\pi i}{4} }\zeta) &e^{-\frac{\pi}{4}\nu_1}D_{-i\nu_1}(e^{\frac{\pi}{4} i}\zeta)&0\\
0&0&1
\end{array}\right).\nonumber
\end{align}
And when $\xi\in[0,1)$,
\begin{align}
&\tilde{\beta}^1_{21}=\frac{\sqrt{2\pi}e^{\frac{5}{2}\pi\nu_1}e^{-\frac{7\pi}{4} i}}{r_{\xi_1}\Gamma(-i\nu_1)},\hspace{0.5cm}\tilde{\beta}^1_{21}\tilde{\beta}^1_{12}=-\nu_1,\\
&|\tilde{\beta}^1_{21}|=-\frac{\nu_1}{(1-|r(\xi_1)|^2)^3},\\
&\arg(\tilde{\beta}^1_{21})=\frac{5}{2}\pi\nu_1-\frac{7\pi}{4} i-\arg r_{\xi_1}-\arg \Gamma(-i\nu_1);
\end{align}
when $\xi\in(-1/8,0)$,
\begin{align}
&\tilde{\beta}^1_{21}=\frac{\sqrt{2\pi}e^{\frac{\pi}{2}\nu_1}e^{-\frac{\pi}{4} i}}{r_{\xi_1}\Gamma(i\nu_1)},\hspace{0.5cm}\tilde{\beta}^1_{21}\tilde{\beta}^1_{12}=-\nu_1,\\
&|\tilde{\beta}^1_{21}|=-\frac{\nu_1}{1-|r(\xi_1)|^2},\\
&\arg(\tilde{\beta}^1_{21})=\frac{\pi}{2}\nu_1-\frac{\pi}{4} i-\arg r_{\xi_1}-\arg \Gamma(i\nu_1).\nonumber
\end{align}

A derivation of this result is given in \cite{RN6}, and a direct verification of the solution is given in  \cite{Liu3}. Substitute above results into (\ref{asyMlo1}) and obtain:
\begin{align}
M^{lo,1}(z)=I+\frac{t^{-1/2}}{z-\xi_1} \left(\begin{array}{ccc}
0 & \tilde{\beta}^1_{12}&0\\
\tilde{\beta}^1_{21} & 0&0\\
0&0&0
\end{array}\right)+\mathcal{O}(t^{-1}).\label{asyMpc}
\end{align}
For the model around other stationary phase points, it  also admits
\begin{align}
M^{lo,k}(z)=I+\frac{t^{-1/2}}{z-\xi_k} \left(\begin{array}{ccc}
0 & \tilde{\beta}^k_{12}&0\\
\tilde{\beta}^k_{21} & 0&0\\
0&0&0
\end{array}\right)+\mathcal{O}(t^{-1}),\label{asyMpck}
\end{align}
for $k=2,...,p(\xi)$.
When $\xi\in(-1/8,0)$, $k$ is  even number or  $\xi\in[0,1)$, $k$ is odd  number,
\begin{align}
r_{\xi_k}=r(\xi_k)T_k(\xi)^{2}e^{-2it\theta(\xi_k)}\exp\left\lbrace -i\nu(\xi_k)\log \left( 4t\theta''(\xi_k)\right) \right\rbrace,
\end{align}
and
\begin{align}
&\tilde{\beta}^k_{21}=\frac{\sqrt{2\pi}e^{\frac{5}{2}\pi\nu(\xi_k)}e^{-\frac{7\pi}{4} i}}{r_{\xi_k}\Gamma(-i\nu(\xi_k))},\hspace{0.5cm}\tilde{\beta}^k_{21}\tilde{\beta}^k_{12}=-\nu(\xi_k),\\
&|\tilde{\beta}^k_{21}|=-\frac{\nu(\xi_k)}{(1-|r(\xi_k)|^2)^3},\\
&\arg(\tilde{\beta}^k_{21})=\frac{5}{2}\pi\nu(\xi_k)-\frac{7\pi}{4} i-\arg r_{\xi_k}-\arg \Gamma(-i\nu(\xi_k));
\end{align}
and when $\xi\in(-0.25,0)$, $k$ is  even number or  $\xi\in[0,2)$, $k$ is  odd number,
\begin{align}
r_{\xi_k}=r(\xi_k)T_k(\xi)^{2}e^{-2it\theta(\xi_k)}\exp\left\lbrace -i\nu(\xi_k)\log \left(- 4t\theta''(\xi_k)\right) \right\rbrace,
\end{align}
and
\begin{align}
&\tilde{\beta}^k_{21}=\frac{\sqrt{2\pi}e^{\frac{\pi}{2}\nu(\xi_k)}e^{-\frac{\pi}{4} i}}{r_{\xi_k}\Gamma(i\nu(\xi_k))},\hspace{0.5cm}\tilde{\beta}^k_{21}\tilde{\beta}^k_{12}=-\nu(\xi_k),\\
&|\tilde{\beta}^k_{21}|=-\frac{\nu(\xi_k)}{1-|r(\xi_k)|^2},\\
&\arg(\tilde{\beta}^k_{21})=\frac{\pi}{2}\nu(\xi_k)-\frac{\pi}{4} i-\arg r_{\xi_k}-\arg \Gamma(i\nu(\xi_k)).
\end{align}
Then together with proposition \ref{dividepc}, wo final obtain
\begin{Proposition}\label{asymlo}
	As $t\to+\infty$,
	\begin{align}
	M^{lo}(z)=I+t^{-1/2}\sum_{ k=1 }^{p(\xi)}\left( \frac{A_k(\xi)}{z-\xi_k}+\frac{\omega\Gamma_3\overline{A_k(\xi)}\Gamma_3}{z-\omega\xi_k}+\frac{\omega^2\Gamma_2\overline{A_k(\xi)}\Gamma_2}{z-\omega^2\xi_k}\right)  +\mathcal{O}(t^{-1}),
	\end{align}
	where
	\begin{align}
	A_k(\xi)=\left(\begin{array}{ccc}
	0 & \tilde{\beta}^k_{12} &0\\
	\tilde{\beta}^k_{21} & 0 &0\\
	0&0&0
	\end{array}\right).
	\end{align}
\end{Proposition}

\subsection{RH problem  near singularities  $\varkappa_j$ }\label{secB}
\quad The matrix function $M^B_j(z)$ is the solution of following RH problem
\begin{RHP}
	Find a matrix-valued function  $  M^B_j(z)$ with following properties:
	
	$\blacktriangleright$ Analyticity: $M^B_j(z)$ is  meromorphic  in $\mathbb{C}\setminus \left(\mathbb{B}_j\cap l_j \right) $;
	
	$\blacktriangleright$ Jump condition: $M^B_j$ has continuous boundary values $[M^B_j]_\pm$ on $\mathbb{B}_j\cap l_j$ and
	\begin{equation}
	[M^B_j]_+(z)=[M^B_j]_-(z)V^{(2)}(z),\hspace{0.5cm}z \in \mathbb{B}_j\cap l_j;\label{jumpB}
	\end{equation}

	$\blacktriangleright$ Asymptotic behaviors:
	\begin{align}
		M^{B}_j(z) =& I+\mathcal{O}(z^{-1}),\hspace{0.5cm}z \rightarrow \infty;\label{asyMb}
	\end{align}
\end{RHP}
By the symmetry of $M^R$, we obtain that $M^B_3(z)=\Gamma_3\overline{M^B_1(\omega \bar{z})}\Gamma_3$ for $z\in\mathbb{B}_3$; $M^B_5(z)=\Gamma_2\overline{M^B_1(\omega^2 \bar{z})}\Gamma_2$ for $z\in\mathbb{B}_5$ and $M^B_6(z)=\Gamma_3\overline{M^B_4(\omega \bar{z})}\Gamma_3$ for $z\in\mathbb{B}_6$; $M^B_2(z)=\Gamma_2\overline{M^B_4(\omega^2 \bar{z})}\Gamma_2$ for $z\in\mathbb{B}_2$. Thus, it is reasonably to only give the detail of $M^{B}_1(z)$. Then $M^{B}_4(z)$ can be obtained analogously and the others  can be obtained by symmetry. $M^{B}_1(z)$ only has jump $V^{b}$ on $(1-2\varepsilon,1+2\varepsilon)$:
\begin{align*}
	V^B(z)=\left(\begin{array}{ccc}
		1 & \frac{rT_{21}\mathcal{X} e^{it\theta_{12}}}{1-|r|^2}&0\\
		0 & 1&0\\
		0 & 0 &1
	\end{array}\right)\left(\begin{array}{ccc}
		1 & 0 & 0\\
		\frac{\bar{r}T_{12}\mathcal{X} e^{-it\theta_{12}}}{1-|r|^2} & 1 & 0\\
		0 & 0 &1
	\end{array}\right).
\end{align*}
To get the long-time behavior of $M^{B}_1(e^{\frac{\pi i}{6}})$, we trans it to a pure $\bar{\partial}$-problem by multiplying the function $R^{(B)}$ defined as follow:
\begin{align*}
	R^{(B)}(z)=\left\{\begin{array}{lll}
		\left(\begin{array}{ccc}
			1 & R^{(B)}_{+}e^{it\theta_{12}} & 0\\
			0 & 1 & 0\\
			0 & 0 & 1
		\end{array}\right), & z\in \mathbb{C}^-;\\[12pt]
	\left(\begin{array}{ccc}
	1 & 0 & 0\\
	 R^{(B)}_{-}e^{-it\theta_{12}} & 1 & 0\\
	0 & 0 & 1
	\end{array}\right), & z\in \mathbb{C}^+;
	\end{array}\right.
\end{align*}
Here,
\begin{align}
	&R^{(B)}_{+}(z)=\mathcal{X}(\text{Re}z)\mathcal{X}(\text{Im}z+1)f(\text{Re}z)g(z),\\
	&R^{(B)}_{-}(z)=\overline{R^{(B)}_{+}(\bar{z})},
\end{align}
with
\begin{align}
	f(z)=\frac{r(z)}{1-|r(z)|^2},\hspace{0.5cm} g(z)=T_{12}(z).
\end{align}
 Then by simply  calculation, we have
\begin{align}
	|\bar{\partial}R^{(B)}_{+}|\lesssim |\mathcal{X}'(\text{Re}z)\mathcal{X}(\text{Im}z+1)|+|\mathcal{X}(\text{Re}z)\mathcal{X}'(\text{Im}z+1)|.
\end{align}
Obviously, the support of  $R^{(B)}_{+}$ and  $\bar{\partial}R^{(B)}_{+}$ are contained in $\mathbb{B}_1$.
Denote
\begin{align}
	\tilde{M}^{B}_1=M^{B}_1R^{(B)},
\end{align}
with
\begin{align}
	&\bar{\partial}\tilde{M}^{B}_1=M^{B}_1\bar{\partial}R^{(B)},\ \ \  \tilde{M}^{B}_1(z)\sim I,\ z\to\infty.
\end{align}
Specially, $\tilde{M}^{B}_1$ has no jump. Therefore, it is  solution of the integral equation
\begin{equation}
\tilde{M}^{B}_1(z)=I+\frac{1}{\pi}\iint_\mathbb{C}\dfrac{\tilde{M}^{B}_1(s)\bar{\partial}R^{(B)} (s)}{s-z}dm(s).\label{tMB}
\end{equation}
Denote $C^B$ is the integral  operator: $L^\infty(\mathbb{C})\to L^\infty(\mathbb{C})$ as
\begin{equation}
	C^Bf(z)=\frac{1}{\pi}\iint_\mathbb{C}\dfrac{f(s)\bar{\partial}R^{(B)} (s)}{s-z}dm(s).\label{CB}
\end{equation}
\begin{Proposition}
	$C^B$ is a bounded integral  operator from $L^\infty(\mathbb{C})\to L^\infty(\mathbb{C})$ with:
	\begin{align}
		\parallel C^B \parallel\lesssim t^{-1/p}, \ p>1.
	\end{align}
It implies that when $t\to\infty$, $\left( 1-C^B\right)^{-1} $ exists.
\end{Proposition}
\begin{proof}
	Obviously,
	\begin{align}
		\parallel C^B \parallel\lesssim \iint_{\mathbb{C}^+}\dfrac{|\bar{\partial}R^{(B)} (s)|}{|s-z|}dm(s)+\iint_{\mathbb{C}^-}\dfrac{|\bar{\partial}R^{(B)} (s)|}{|s-z|}dm(s).
	\end{align}
Take the first term as an example. Let $s=u+vi=le^{i\vartheta}$, $z=x+yi$. In the following computation,  we will use the inequality
\begin{align}
	\parallel |s-z|^{-1}\parallel_{L^q(0,+\infty)}&=\left\lbrace \int_{0}^{+\infty}\left[  \left( \frac{u-x}{v-y}\right) ^2+1\right]^{-q/2}
	d\left( \frac{u-x}{|v-y|}\right)\right\rbrace ^{1/q}|v-y|^{1/q-1}\nonumber\\
	&\lesssim |v-y|^{1/q-1},\label{s-z}
\end{align}
with $1< q<+\infty$ and $\frac{1}{p}+\frac{1}{q}=1$. Thus, by either Lemma \ref{Imtheta} or Corollary \ref{theta2c}, we both have
\begin{align}
	\iint_{\mathbb{C}^+}\dfrac{|\bar{\partial}R^{(B)} (s)|}{|s-z|}dm(s)&=\int_{0}^{2\varepsilon}\int_{1-2\varepsilon}^{1+2\varepsilon}|s-z|^{-1}|\bar{\partial}R^{(B)} (s)|due^{-c(\xi)v}dv\nonumber\\
	&\lesssim\int_{0}^{2\varepsilon}v^{-1/q}e^{-c(\xi)v}dv\lesssim t^{-1/p}.
\end{align}
\end{proof}
Hence, from $\tilde{M}^{B}_1(z)=\left( 1-C^B\right)^{-1}\dot I$, we get the existence and uniqueness of $\tilde{M}^{B}_1(z)$. Take $z=e^{\frac{\pi i}{6}}$ in (\ref{tMB}), then
 \begin{equation}
 	\tilde{M}^{B}_1(z)-I=\frac{1}{\pi}\iint_\mathbb{C}\dfrac{\tilde{M}^{B}_1(s)\bar{\partial}R^{(B)} (s)}{s-e^{\frac{\pi i}{6}}}dm(s).
 \end{equation}
\begin{Proposition}
	There exists  a constant $T_1$, such that for all $t>T_1$,   $\tilde{M}^{B}(z)$   admits the following estimate
	\begin{align}
		\parallel \tilde{M}^{B}(e^{\frac{i\pi}{6}})-I\parallel\lesssim t^{-1}.\label{mBi}
	\end{align}
\end{Proposition}
\begin{proof}
	Because $e^{\frac{\pi i}{6}}\notin\mathbb{B}_1$, $|s-e^{\frac{\pi i}{6}}|\lesssim1$ in $\mathbb{B}_1$. So
	\begin{align}
		\frac{1}{\pi}\iint_\mathbb{C}\dfrac{\tilde{M}^{B}_1(s)\bar{\partial}R^{(B)} (s)}{s-e^{\frac{\pi i}{6}}}dm(s)\lesssim\iint_{\mathbb{C}^+}|\bar{\partial}R^{(B)} (s)|dm(s)+\iint_{\mathbb{C}^-}|\bar{\partial}R^{(B)} (s)|dm(s).\nonumber
	\end{align}
Similarly, we take the first term as an example
	\begin{align}
	\iint_{\mathbb{C}^+}|\bar{\partial}R^{(B)} (s)|dm(s)&\lesssim\int_{0}^{2\varepsilon}\int_{1-2\varepsilon}^{1+2\varepsilon}|\bar{\partial}R^{(B)} (s)|due^{-c(\xi)v}dv\nonumber\\
	&\lesssim\int_{0}^{2\varepsilon}e^{-c(\xi)v}dv\lesssim t^{-1}.
\end{align}
\end{proof}
Finally, we obtain
\begin{align}
	M^{B}_1(e^{\frac{i\pi}{6}})=
	\tilde{M}^{B}_1(e^{\frac{i\pi}{6}})R^{(B)}(e^{\frac{i\pi}{6}})=I+\mathcal{O}(t^{-1}).\label{asyMB}
\end{align}
Then for$j=1,...,6$,
\begin{align}
	M^{B}_j(e^{\frac{i\pi}{6}})=I+\mathcal{O}(t^{-1}).
\end{align}
In addition, on the  boundary of $\mathbb{B}_1$, $R^{(B)}(z)=I$, so $M^{B}_1(z)=\tilde{M}^{B}_1(z)$ on it.  In fact, for $z\in\partial\mathbb{B}$, when $s\in\partial\mathbb{B}$, we still have $\dfrac{|\bar{\partial}R^{(B)} (s)|}{|s-z|}=0$. It implies that $\dfrac{|\bar{\partial}R^{(B)} (s)|}{|s-z|}$ is bounded for $z\in\partial\mathbb{B}$,  $s\in\mathbb{B}$. Then similar as in above Proposition, we have for $z\in\partial\mathbb{B}$,
\begin{align}
	M^{B}_j(z)=I+\mathcal{O}(t^{-1}).
\end{align}
$r(\pm1)=0$, gives that $\bar{\partial}R^{(B)}(\pm1)=I$, so $M^{B}_1(1)=
\tilde{M}^{B}_1(1)=I+\mathcal{O}(t^{-1/p})$. Here, $p>1$ is a  arbitrary  constant. So we can denote $p$ as  $1/p=1-\rho$ with a small enough positive constant $\rho<\frac{1}{4}$. Then $M^{B}_1(1)=
\tilde{M}^{B}_1(1)=I+\mathcal{O}(t^{-1+\rho})$. Further, we have for $l=1,...,6$,
\begin{align}
	M^{B}_l(\kappa_l)=
	\tilde{M}^{B}_l(\kappa_l)=I+\mathcal{O}(t^{-1+\rho}).\label{MBkap}
\end{align}

\subsection{Small norm RH problem for residual error  }\label{sec7}

\quad  In this section,  we consider the error matrix-function $E(z;\xi)$.
From the  definition (\ref{transm4}), we can obtain a RH problem  for the   $E(z;\xi)$:

\begin{RHP}
	 Find a matrix-valued function $E(z;\xi)$  with following properties:

$\blacktriangleright$ Analyticity: $E(z;\xi)$ is analytical  in $\mathbb{C}\setminus  \Sigma^{E}(\xi) $, where
\begin{align}
&\Sigma^{E}(\xi)=  \cup_{j=1}^6\partial \mathbb{B}_j,\ \ \xi\in(-\infty,-1/8)\cup(1,+\infty),\nonumber\\
&\Sigma^{E}(\xi)=\left( \cup_{j=1}^6\partial \mathbb{B}_j\right) \cup \partial U(\xi)\cup
(\Sigma^{(2)}\setminus U(\xi)), \ \ \xi\in(-1/8,1).\nonumber
\end{align}

$\blacktriangleright$ Asymptotic behaviors:
\begin{align}
&E(z;\xi) \sim I+\mathcal{O}(z^{-1}),\hspace{0.5cm}|z| \rightarrow \infty;
\end{align}

$\blacktriangleright$ Jump condition: $E(z;\xi)$  has continuous boundary values $E_\pm(z;\xi)$  satisfying
$$E_+(z;\xi)=E_-(z;\xi)V^{E}(z), \ \ z\in \Sigma^{E}(\xi),$$
where the jump matrix $V^{E}(z)$ is given by
\begin{equation}
V^{E}(z)=\left\{\begin{array}{llll}
M^{(r)}(z)V^{(2)}(z)M^{(r)}(z)^{-1}, & z\in \Sigma^{(2)}\setminus\left(  U(\xi)\cup\mathbb{B}_j\right) ,\\[4pt]
M^{(r)}(z)M^{lo}(z)M^{(r)}(z)^{-1},  & z\in \partial U(\xi),\\[4pt]
M^{(r)}(z)M^{B}_j(z)M^{(r)}(z)^{-1},  & z\in \partial \mathbb{B}_j,
\end{array}\right. \label{deVE}
\end{equation}
see  Figure \ref{figE} and Figure \ref{figE1}. Especially  we define
$$\Sigma^{(2)}=U(\xi) :=\emptyset, \ \ \xi\in(-\infty,-1/8)\cup(1,+\infty).$$

$\blacktriangleright$ Singularities: As $z \rightarrow \varkappa_l=e^{\frac{i\pi(l-1)}{3}}$, $l = 1,...,6$, the limit  of $E(z;\xi)$
 has pole singularities with leading terms of a specific
matrix structure
\begin{align}
	&\lim_{z\to \varkappa_l}E(z;\xi)=\lim_{z\to \varkappa_l}M^R(z)M^B_l(\varkappa_l)M^{(r)}(z)^{-1}=\mathcal{O} (  (z-\varkappa_l)^{-2}).
\end{align}
\end{RHP}

Here  $E(z;\xi)$ has both   jump condition and singularities, its asymptotic behavior   can  be  solved through
 a  small  norm RH problem.  A natural idea is  to consider a matrix function
$E^{(2)}(z;\xi)$ which only has same jump with $E(z;\xi)$, but without pole singularities.

\begin{RHP}\label{E2}
	Find a matrix-valued function $E^{(2)}(z;\xi)$  with following properties

	$\blacktriangleright$ Analyticity: $E^{(2)}(z;\xi)$ is analytical  in $\mathbb{C}\setminus  \Sigma^{E}(\xi) $, where
\begin{align}
&\Sigma^{E}(\xi)=  \cup_{j=1}^6\partial \mathbb{B}_j,\ \ \xi\in(-\infty,-1/8)\cup(1,+\infty),\nonumber\\
&\Sigma^{E}(\xi)=\left( \cup_{j=1}^6\partial \mathbb{B}_j\right) \cup \partial U(\xi)\cup
(\Sigma^{(2)}\setminus U(\xi)), \ \ \xi\in(-1/8,1).\nonumber
\end{align}

	$\blacktriangleright$ Asymptotic behaviors:
	\begin{align}
		&E^{(2)}(z;\xi) \sim I+\mathcal{O}(z^{-1}),\hspace{0.5cm}|z| \rightarrow \infty;
	\end{align}

	$\blacktriangleright$ Jump condition: $E(z;\xi)$ has continuous boundary values $E^{(2)}_\pm(z;\xi)$   satisfying
	$$E^{(2)}_+(z;\xi)=E^{(2)}_-(z;\xi)V^{E}(z), \ \ z\in\Sigma^{E}(\xi).$$
\end{RHP}

\begin{figure}[htp]
	\centering
\subfigure[]{
	\begin{tikzpicture}
	\draw(-2.57,0.09)--(-2.8,0.2);
	\draw[->](-2.57,0.09)--(-2.7,0.15);
	\draw(-2.21,0.1)--(-2,0.2);
	\draw[-<](-2.21,-0.1)--(-2.1,-0.15);
	\draw(-2.59,-0.1)--(-2.8,-0.2);
	\draw[->](-2.59,-0.1)--(-2.7,-0.15);
	\draw(-2.21,-0.1)--(-2,-0.2);
	\draw[-<](-2.21,0.1)--(-2.1,0.15);
	\draw(-1.6,0)--(-2,0.2);
	\draw[->](-1.6,0)--(-1.9,0.15);
	\draw(-1.6,0)--(-2,-0.2);
	\draw[->](-1.6,0)--(-1.9,-0.15);
	\draw(-1.6,0)--(-1.2,0.2);
	\draw[-<](-1.6,0)--(-1.3,0.15);
	\draw(-1.6,0)--(-1.2,-0.2);
	\draw[-<](-1.6,0)--(-1.3,-0.15);
	\draw(-0.61,0.1)--(-0.4,0.2);
	\draw[-<](-0.61,0.1)--(-0.5,0.15);
	\draw(-0.99,0.1)--(-1.2,0.2);
	\draw[->](-0.99,-0.1)--(-1.1,-0.15);
	\draw(-0.61,-0.1)--(-0.4,-0.2);
	\draw[->](-0.99,0.1)--(-1.1,0.15);
	\draw(-0.99,-0.1)--(-1.2,-0.2);
	\draw[-<](-0.61,-0.1)--(-0.5,-0.15);
	\draw(-0,0)--(-0.4,-0.2);
	\draw[->](-0,0)--(-0.3,-0.15);
	\draw(-0,0)--(-0.4,0.2);
	\draw[->](-0,0)--(-0.3,0.15);
	\draw(-0,0)--(0.4,-0.2);
	\draw[-<](-0,0)--(0.3,-0.15);
	\draw(-0,0)--(0.4,0.2);
	\draw[-<](-0,0)--(0.3,0.15);
	\draw[dashed](-3.2,0)--(3.2,0)node[right]{ Re$z$};
\draw(2.57,0.09)--(2.8,0.2);
\draw[-<](2.57,0.09)--(2.7,0.15);
\draw(2.21,0.1)--(2,0.2);
\draw[->](2.21,-0.1)--(2.1,-0.15);
\draw(2.59,-0.1)--(2.8,-0.2);
\draw[-<](2.59,-0.1)--(2.7,-0.15);
\draw(2.21,-0.1)--(2,-0.2);
\draw[->](2.21,0.1)--(2.1,0.15);
\draw(1.6,0)--(2,0.2);
\draw[-<](1.6,0)--(1.9,0.15);
\draw(1.6,0)--(2,-0.2);
\draw[-<](1.6,0)--(1.9,-0.15);
\draw(1.6,0)--(1.2,0.2);
\draw[->](1.6,0)--(1.3,0.15);
\draw(1.6,0)--(1.2,-0.2);
\draw[->](1.6,0)--(1.3,-0.15);
\draw(0.61,0.1)--(0.4,0.2);
\draw[->](0.61,0.1)--(0.5,0.15);
\draw(0.99,0.1)--(1.2,0.2);
\draw[-<](0.99,-0.1)--(1.1,-0.15);
\draw(0.61,-0.1)--(0.4,-0.2);
\draw[-<](0.99,0.1)--(1.1,0.15);
\draw(0.99,-0.1)--(1.2,-0.2);
\draw[->](0.61,-0.1)--(0.5,-0.15);
	\draw[blue ](0,0)--(1.5,2.6)node[above]{ $e^{\frac{ \pi i}{3}}$};
	\draw[teal ](0,0)--(-1.5,2.6)node[above]{ $e^{\frac{2 \pi i}{3}}$};
	\draw[teal ](0,0)--(1.5,-2.6);
	\draw[blue ](0,0)--(-1.5,-2.6);
	\draw[dashed] (0.8,0) arc (0:360:0.8);
	\draw[dashed] (2.4,0) arc (0:360:2.4);	
	\draw[rotate=60](0.61,-0.1)--(0.4,-0.2);
	\draw[rotate=60,-<](0.99,0.1)--(1.1,0.15);
	\draw[rotate=60](0.99,-0.1)--(1.2,-0.2);
	\draw[rotate=60,->](0.61,-0.1)--(0.5,-0.15);
	\draw[rotate=60](0.61,0.1)--(0.4,0.2);
	\draw[rotate=60,->](0.61,0.1)--(0.5,0.15);
	\draw[rotate=60](0.99,0.1)--(1.2,0.2);
	\draw[rotate=60,-<](0.99,-0.1)--(1.1,-0.15);
	\draw[rotate=60](2.57,0.09)--(2.8,0.2);
	\draw[rotate=60,-<](2.57,0.09)--(2.7,0.15);
	\draw[rotate=60](2.21,0.1)--(2,0.2);
	\draw[rotate=60,->](2.21,-0.1)--(2.1,-0.15);
	\draw[rotate=60](2.59,-0.1)--(2.8,-0.2);
	\draw[rotate=60,-<](2.59,-0.1)--(2.7,-0.15);
	\draw[rotate=60](2.21,-0.1)--(2,-0.2);
	\draw[rotate=60,->](2.21,0.1)--(2.1,0.15);
	\draw[rotate=60](1.6,0)--(2,0.2);
	\draw[rotate=60,-<](1.6,0)--(1.9,0.15);
	\draw[rotate=60](1.6,0)--(2,-0.2);
	\draw[rotate=60,-<](1.6,0)--(1.9,-0.15);
	\draw[rotate=60](1.6,0)--(1.2,0.2);
	\draw[rotate=60,->](1.6,0)--(1.3,0.15);
	\draw[rotate=60](1.6,0)--(1.2,-0.2);
	\draw[rotate=60,->](1.6,0)--(1.3,-0.15);
	\draw[rotate=60](-2.57,0.09)--(-2.8,0.2);
	\draw[rotate=60,->](-2.57,0.09)--(-2.7,0.15);
	\draw[rotate=60](-2.21,0.1)--(-2,0.2);
	\draw[rotate=60,-<](-2.21,-0.1)--(-2.1,-0.15);
	\draw[rotate=60](-2.59,-0.1)--(-2.8,-0.2);
	\draw[rotate=60,->](-2.59,-0.1)--(-2.7,-0.15);
	\draw[rotate=60](-2.21,-0.1)--(-2,-0.2);
	\draw[rotate=60,-<](-2.21,0.1)--(-2.1,0.15);
	\draw[rotate=60](-1.6,0)--(-2,0.2);
	\draw[rotate=60,->](-1.6,0)--(-1.9,0.15);
	\draw[rotate=60](-1.6,0)--(-2,-0.2);
	\draw[rotate=60,->](-1.6,0)--(-1.9,-0.15);
	\draw[rotate=60](-1.6,0)--(-1.2,0.2);
	\draw[rotate=60,-<](-1.6,0)--(-1.3,0.15);
	\draw[rotate=60](-1.6,0)--(-1.2,-0.2);
	\draw[rotate=60,-<](-1.6,0)--(-1.3,-0.15);
	\draw[rotate=60](-0.61,0.1)--(-0.4,0.2);
	\draw[rotate=60,-<](-0.61,0.1)--(-0.5,0.15);
	\draw[rotate=60](-0.99,0.1)--(-1.2,0.2);
	\draw[rotate=60,->](-0.99,-0.1)--(-1.1,-0.15);
	\draw[rotate=60](-0.61,-0.1)--(-0.4,-0.2);
	\draw[rotate=60,->](-0.99,0.1)--(-1.1,0.15);
	\draw[rotate=60](-0.99,-0.1)--(-1.2,-0.2);
	\draw[rotate=60,-<](-0.61,-0.1)--(-0.5,-0.15);
	\draw[rotate=60](-0,0)--(-0.4,-0.2);
	\draw[rotate=60,->](-0,0)--(-0.3,-0.15);
	\draw[rotate=60](-0,0)--(-0.4,0.2);
	\draw[rotate=60,->](-0,0)--(-0.3,0.15);
	\draw[rotate=60](-0,0)--(0.4,-0.2);
	\draw[rotate=60,-<](-0,0)--(0.3,-0.15);
	\draw[rotate=60](-0,0)--(0.4,0.2);
	\draw[rotate=60,-<](-0,0)--(0.3,0.15);
		\draw[rotate=120](0.61,-0.1)--(0.4,-0.2);
	\draw[rotate=120,-<](0.99,0.1)--(1.1,0.15);
	\draw[rotate=120](0.99,-0.1)--(1.2,-0.2);
	\draw[rotate=120,->](0.61,-0.1)--(0.5,-0.15);
	\draw[rotate=120](0.61,0.1)--(0.4,0.2);
	\draw[rotate=120,->](0.61,0.1)--(0.5,0.15);
	\draw[rotate=120](0.99,0.1)--(1.2,0.2);
	\draw[rotate=120,-<](0.99,-0.1)--(1.1,-0.15);
	\draw[rotate=120](2.57,0.09)--(2.8,0.2);
	\draw[rotate=120,-<](2.57,0.09)--(2.7,0.15);
	\draw[rotate=120](2.21,0.1)--(2,0.2);
	\draw[rotate=120,->](2.21,-0.1)--(2.1,-0.15);
	\draw[rotate=120](2.59,-0.1)--(2.8,-0.2);
	\draw[rotate=120,-<](2.59,-0.1)--(2.7,-0.15);
	\draw[rotate=120](2.21,-0.1)--(2,-0.2);
	\draw[rotate=120,->](2.21,0.1)--(2.1,0.15);
	\draw[rotate=120](1.6,0)--(2,0.2);
	\draw[rotate=120,-<](1.6,0)--(1.9,0.15);
	\draw[rotate=120](1.6,0)--(2,-0.2);
	\draw[rotate=120,-<](1.6,0)--(1.9,-0.15);
	\draw[rotate=120](1.6,0)--(1.2,0.2);
	\draw[rotate=120,->](1.6,0)--(1.3,0.15);
	\draw[rotate=120](1.6,0)--(1.2,-0.2);
	\draw[rotate=120,->](1.6,0)--(1.3,-0.15);
	\draw[rotate=120](-2.57,0.09)--(-2.8,0.2);
	\draw[rotate=120,->](-2.57,0.09)--(-2.7,0.15);
	\draw[rotate=120](-2.21,0.1)--(-2,0.2);
	\draw[rotate=120,-<](-2.21,-0.1)--(-2.1,-0.15);
	\draw[rotate=120](-2.59,-0.1)--(-2.8,-0.2);
	\draw[rotate=120,->](-2.59,-0.1)--(-2.7,-0.15);
	\draw[rotate=120](-2.21,-0.1)--(-2,-0.2);
	\draw[rotate=120,-<](-2.21,0.1)--(-2.1,0.15);
	\draw[rotate=120](-1.6,0)--(-2,0.2);
	\draw[rotate=120,->](-1.6,0)--(-1.9,0.15);
	\draw[rotate=120](-1.6,0)--(-2,-0.2);
	\draw[rotate=120,->](-1.6,0)--(-1.9,-0.15);
	\draw[rotate=120](-1.6,0)--(-1.2,0.2);
	\draw[rotate=120,-<](-1.6,0)--(-1.3,0.15);
	\draw[rotate=120](-1.6,0)--(-1.2,-0.2);
	\draw[rotate=120,-<](-1.6,0)--(-1.3,-0.15);
	\draw[rotate=120](-0.61,0.1)--(-0.4,0.2);
	\draw[rotate=120,-<](-0.61,0.1)--(-0.5,0.15);
	\draw[rotate=120](-0.99,0.1)--(-1.2,0.2);
	\draw[rotate=120,->](-0.99,-0.1)--(-1.1,-0.15);
	\draw[rotate=120](-0.61,-0.1)--(-0.4,-0.2);
	\draw[rotate=120,->](-0.99,0.1)--(-1.1,0.15);
	\draw[rotate=120](-0.99,-0.1)--(-1.2,-0.2);
	\draw[rotate=120,-<](-0.61,-0.1)--(-0.5,-0.15);
	\draw[rotate=120](-0,0)--(-0.4,-0.2);
	\draw[rotate=120,->](-0,0)--(-0.3,-0.15);
	\draw[rotate=120](-0,0)--(-0.4,0.2);
	\draw[rotate=120,->](-0,0)--(-0.3,0.15);
	\draw[rotate=120](-0,0)--(0.4,-0.2);
	\draw[rotate=120,-<](-0,0)--(0.3,-0.15);
	\draw[rotate=120](-0,0)--(0.4,0.2);
	\draw[rotate=120,-<](-0,0)--(0.3,0.15);
	\draw[thick,red](0.8,0) circle (0.2);
	\draw[thick,red](2.4,0) circle (0.2);
	\draw[thick,red](-0.8,0) circle (0.2);
	\draw[thick,red](-2.4,0) circle (0.2);
	\draw[thick,red](0.4,0.69) circle (0.2);
	\draw[thick,red](0.4,-0.69) circle (0.2);
	\draw[thick,red](-0.4,0.69) circle (0.2);
	\draw[thick,red](-0.4,-0.69) circle (0.2);
	\draw[thick,red](1.2,2.08) circle (0.2);
	\draw[thick,red](1.2,-2.08) circle (0.2);
	\draw[thick,red](-1.2,2.08) circle (0.2);
	\draw[thick,red](-1.2,-2.08) circle (0.2);
	\coordinate (A) at (-2.4,0);
	\fill (A) circle (1pt) node[below] {\scriptsize$\xi_4$};
	\coordinate (b) at (-0.8,0);
	\fill (b) circle (1pt) node[below] {\scriptsize$\xi_3$};
	\coordinate (e) at (2.4,0);
	\fill (e) circle (1pt) node[below] {\scriptsize$\xi_1$};
	\coordinate (f) at (0.8,0);
	\fill (f) circle (1pt) node[below] {\scriptsize$\xi_2$};
	\coordinate (A1) at (0.4,0.69);
	\fill[blue] (A1) circle (1pt) node[right] {\scriptsize$\omega^2\xi_3$};
	\coordinate (A2) at (0.4,-0.69);
	\fill[teal] (A2) circle (1pt) node[right] {\scriptsize$\omega\xi_3$};
	\coordinate (A3) at (-0.4,0.69);
	\fill[teal] (A3) circle (1pt) node[left] {\scriptsize$\omega\xi_2$};
	\coordinate (A4) at (-0.4,-0.69);
	\fill[blue] (A4) circle (1pt) node[left] {\scriptsize$\omega^2\xi_2$};
	\coordinate (n1) at (1.2,2.08);
	\fill[blue] (n1) circle (1pt) node[right] {\scriptsize$\omega^2\xi_4$};
	\coordinate (m1) at (1.2,-2.08);
	\fill[teal] (m1) circle (1pt) node[right] {\scriptsize$\omega\xi_4$};
	\coordinate (j1) at (-1.2,2.08);
	\fill[teal] (j1) circle (1pt) node[left] {\scriptsize$\omega\xi_1$};
	\coordinate[blue] (k1) at (-1.2,-2.08);
	\fill[blue] (k1) circle (1pt) node[left] {\scriptsize$\omega^2\xi_1$};
	\coordinate (kap1) at (1.2,0);
	\fill[orange] (kap1) circle (1pt) node[below] {\scriptsize$1$};
	\draw[orange](1.1,-0.1)--(1.1,0.1)--(1.3,0.1)--(1.3,-0.1)--(1.1,-0.1);
	\coordinate (kap4) at (-1.2,0);
	\fill[orange] (kap4) circle (1pt) node[below] {\scriptsize$-1$};
	\draw[orange](-1.1,-0.1)--(-1.1,0.1)--(-1.3,0.1)--(-1.3,-0.1)--(-1.1,-0.1);
	\draw[rotate=60,orange](1.1,-0.1)--(1.1,0.1)--(1.3,0.1)--(1.3,-0.1)--(1.1,-0.1);		\draw[rotate=60,orange](-1.1,-0.1)--(-1.1,0.1)--(-1.3,0.1)--(-1.3,-0.1)--(-1.1,-0.1);
	\draw[rotate=120,orange](1.1,-0.1)--(1.1,0.1)--(1.3,0.1)--(1.3,-0.1)--(1.1,-0.1);		\draw[rotate=120,orange](-1.1,-0.1)--(-1.1,0.1)--(-1.3,0.1)--(-1.3,-0.1)--(-1.1,-0.1);
	\coordinate[orange] (kap2) at (0.6,1.039);
	\fill[orange] (kap2) circle (1pt) node[right] {\scriptsize$\varkappa_2$};
	\coordinate[orange] (kap3) at (-0.6,1.039);
	\fill[orange] (kap3) circle (1pt) node[right] {\scriptsize$\varkappa_3$};
	\coordinate[orange] (kap6) at (0.6,-1.039);
	\fill[orange] (kap6) circle (1pt) node[right] {\scriptsize$\varkappa_6$};
	\coordinate[orange] (kap5) at (-0.6,-1.039);
	\fill[orange] (kap5) circle (1pt) node[right] {\scriptsize$\varkappa_5$};
\end{tikzpicture}
}
\subfigure[]{
\begin{tikzpicture}
\draw[dashed](-6,0)--(6,0)node[right]{ Re$z$};
	\draw(-2.57,0.09)--(-2.8,0.2);
	\draw[->](-2.57,0.09)--(-2.7,0.15);
\draw(-2.21,0.1)--(-2,0.2);
\draw[-<](-2.21,-0.1)--(-2.1,-0.15);
\draw(-2.59,-0.1)--(-2.8,-0.2);
\draw[->](-2.59,-0.1)--(-2.7,-0.15);
\draw(-2.21,-0.1)--(-2,-0.2);
\draw[-<](-2.21,0.1)--(-2.1,0.15);
\draw(-1.6,0)--(-2,0.2);
\draw[->](-1.6,0)--(-1.9,0.15);
\draw(-1.6,0)--(-2,-0.2);
\draw[->](-1.6,0)--(-1.9,-0.15);
\draw(-1.6,0)--(-1.2,0.2);
\draw[-<](-1.6,0)--(-1.3,0.15);
\draw(-1.6,0)--(-1.2,-0.2);
\draw[-<](-1.6,0)--(-1.3,-0.15);
\draw(-0.61,0.1)--(-0.4,0.2);
\draw[-<](-0.61,0.1)--(-0.5,0.15);
\draw(-0.99,0.1)--(-1.2,0.2);
\draw[->](-0.99,-0.1)--(-1.1,-0.15);
\draw(-0.61,-0.1)--(-0.4,-0.2);
\draw[->](-0.99,0.1)--(-1.1,0.15);
\draw(-0.99,-0.1)--(-1.2,-0.2);
\draw[-<](-0.61,-0.1)--(-0.5,-0.15);
\draw(-0,0)--(-0.4,-0.2);
\draw[->](-0,0)--(-0.3,-0.15);
\draw(-0,0)--(-0.4,0.2);
\draw[->](-0,0)--(-0.3,0.15);
\draw(-0,0)--(0.4,-0.2);
\draw[-<](-0,0)--(0.3,-0.15);
\draw(-0,0)--(0.4,0.2);
\draw[-<](-0,0)--(0.3,0.15);
\draw(2.57,0.09)--(2.8,0.2);
\draw[-<](2.57,0.09)--(2.7,0.15);
\draw(2.21,0.1)--(2,0.2);
\draw[->](2.21,-0.1)--(2.1,-0.15);
\draw(2.59,-0.1)--(2.8,-0.2);
\draw[-<](2.59,-0.1)--(2.7,-0.15);
\draw(2.21,-0.1)--(2,-0.2);
\draw[->](2.21,0.1)--(2.1,0.15);
\draw(1.6,0)--(2,0.2);
\draw[-<](1.6,0)--(1.9,0.15);
\draw(1.6,0)--(2,-0.2);
\draw[-<](1.6,0)--(1.9,-0.15);
\draw(1.6,0)--(1.2,0.2);
\draw[->](1.6,0)--(1.3,0.15);
\draw(1.6,0)--(1.2,-0.2);
\draw[->](1.6,0)--(1.3,-0.15);
\draw(0.61,0.1)--(0.4,0.2);
\draw[->](0.61,0.1)--(0.5,0.15);
\draw(0.99,0.1)--(1.2,0.2);
\draw[-<](0.99,-0.1)--(1.1,-0.15);
\draw(0.61,-0.1)--(0.4,-0.2);
\draw[-<](0.99,0.1)--(1.1,0.15);
\draw(0.99,-0.1)--(1.2,-0.2);
\draw[->](0.61,-0.1)--(0.5,-0.15);
\draw[rotate=60](0.61,-0.1)--(0.4,-0.2);
\draw[rotate=60,-<](0.99,0.1)--(1.1,0.15);
\draw[rotate=60](0.99,-0.1)--(1.2,-0.2);
\draw[rotate=60,->](0.61,-0.1)--(0.5,-0.15);
\draw[rotate=60](0.61,0.1)--(0.4,0.2);
\draw[rotate=60,->](0.61,0.1)--(0.5,0.15);
\draw[rotate=60](0.99,0.1)--(1.2,0.2);
\draw[rotate=60,-<](0.99,-0.1)--(1.1,-0.15);
\draw[rotate=60](2.57,0.09)--(2.8,0.2);
\draw[rotate=60,-<](2.57,0.09)--(2.7,0.15);
\draw[rotate=60](2.21,0.1)--(2,0.2);
\draw[rotate=60,->](2.21,-0.1)--(2.1,-0.15);
\draw[rotate=60](2.59,-0.1)--(2.8,-0.2);
\draw[rotate=60,-<](2.59,-0.1)--(2.7,-0.15);
\draw[rotate=60](2.21,-0.1)--(2,-0.2);
\draw[rotate=60,->](2.21,0.1)--(2.1,0.15);
\draw[rotate=60](1.6,0)--(2,0.2);
\draw[rotate=60,-<](1.6,0)--(1.9,0.15);
\draw[rotate=60](1.6,0)--(2,-0.2);
\draw[rotate=60,-<](1.6,0)--(1.9,-0.15);
\draw[rotate=60](1.6,0)--(1.2,0.2);
\draw[rotate=60,->](1.6,0)--(1.3,0.15);
\draw[rotate=60](1.6,0)--(1.2,-0.2);
\draw[rotate=60,->](1.6,0)--(1.3,-0.15);
\draw[rotate=60](-2.57,0.09)--(-2.8,0.2);
\draw[rotate=60,->](-2.57,0.09)--(-2.7,0.15);
\draw[rotate=60](-2.21,0.1)--(-2,0.2);
\draw[rotate=60,-<](-2.21,-0.1)--(-2.1,-0.15);
\draw[rotate=60](-2.59,-0.1)--(-2.8,-0.2);
\draw[rotate=60,->](-2.59,-0.1)--(-2.7,-0.15);
\draw[rotate=60](-2.21,-0.1)--(-2,-0.2);
\draw[rotate=60,-<](-2.21,0.1)--(-2.1,0.15);
\draw[rotate=60](-1.6,0)--(-2,0.2);
\draw[rotate=60,->](-1.6,0)--(-1.9,0.15);
\draw[rotate=60](-1.6,0)--(-2,-0.2);
\draw[rotate=60,->](-1.6,0)--(-1.9,-0.15);
\draw[rotate=60](-1.6,0)--(-1.2,0.2);
\draw[rotate=60,-<](-1.6,0)--(-1.3,0.15);
\draw[rotate=60](-1.6,0)--(-1.2,-0.2);
\draw[rotate=60,-<](-1.6,0)--(-1.3,-0.15);
\draw[rotate=60](-0.61,0.1)--(-0.4,0.2);
\draw[rotate=60,-<](-0.61,0.1)--(-0.5,0.15);
\draw[rotate=60](-0.99,0.1)--(-1.2,0.2);
\draw[rotate=60,->](-0.99,-0.1)--(-1.1,-0.15);
\draw[rotate=60](-0.61,-0.1)--(-0.4,-0.2);
\draw[rotate=60,->](-0.99,0.1)--(-1.1,0.15);
\draw[rotate=60](-0.99,-0.1)--(-1.2,-0.2);
\draw[rotate=60,-<](-0.61,-0.1)--(-0.5,-0.15);
\draw[rotate=60](-0,0)--(-0.4,-0.2);
\draw[rotate=60,->](-0,0)--(-0.3,-0.15);
\draw[rotate=60](-0,0)--(-0.4,0.2);
\draw[rotate=60,->](-0,0)--(-0.3,0.15);
\draw[rotate=60](-0,0)--(0.4,-0.2);
\draw[rotate=60,-<](-0,0)--(0.3,-0.15);
\draw[rotate=60](-0,0)--(0.4,0.2);
\draw[rotate=60,-<](-0,0)--(0.3,0.15);
\draw[rotate=120](0.61,-0.1)--(0.4,-0.2);
\draw[rotate=120,-<](0.99,0.1)--(1.1,0.15);
\draw[rotate=120](0.99,-0.1)--(1.2,-0.2);
\draw[rotate=120,->](0.61,-0.1)--(0.5,-0.15);
\draw[rotate=120](0.61,0.1)--(0.4,0.2);
\draw[rotate=120,->](0.61,0.1)--(0.5,0.15);
\draw[rotate=120](0.99,0.1)--(1.2,0.2);
\draw[rotate=120,-<](0.99,-0.1)--(1.1,-0.15);
\draw[rotate=120](2.57,0.09)--(2.8,0.2);
\draw[rotate=120,-<](2.57,0.09)--(2.7,0.15);
\draw[rotate=120](2.21,0.1)--(2,0.2);
\draw[rotate=120,->](2.21,-0.1)--(2.1,-0.15);
\draw[rotate=120](2.59,-0.1)--(2.8,-0.2);
\draw[rotate=120,-<](2.59,-0.1)--(2.7,-0.15);
\draw[rotate=120](2.21,-0.1)--(2,-0.2);
\draw[rotate=120,->](2.21,0.1)--(2.1,0.15);
\draw[rotate=120](1.6,0)--(2,0.2);
\draw[rotate=120,-<](1.6,0)--(1.9,0.15);
\draw[rotate=120](1.6,0)--(2,-0.2);
\draw[rotate=120,-<](1.6,0)--(1.9,-0.15);
\draw[rotate=120](1.6,0)--(1.2,0.2);
\draw[rotate=120,->](1.6,0)--(1.3,0.15);
\draw[rotate=120](1.6,0)--(1.2,-0.2);
\draw[rotate=120,->](1.6,0)--(1.3,-0.15);
\draw[rotate=120](-2.57,0.09)--(-2.8,0.2);
\draw[rotate=120,->](-2.57,0.09)--(-2.7,0.15);
\draw[rotate=120](-2.21,0.1)--(-2,0.2);
\draw[rotate=120,-<](-2.21,-0.1)--(-2.1,-0.15);
\draw[rotate=120](-2.59,-0.1)--(-2.8,-0.2);
\draw[rotate=120,->](-2.59,-0.1)--(-2.7,-0.15);
\draw[rotate=120](-2.21,-0.1)--(-2,-0.2);
\draw[rotate=120,-<](-2.21,0.1)--(-2.1,0.15);
\draw[rotate=120](-1.6,0)--(-2,0.2);
\draw[rotate=120,->](-1.6,0)--(-1.9,0.15);
\draw[rotate=120](-1.6,0)--(-2,-0.2);
\draw[rotate=120,->](-1.6,0)--(-1.9,-0.15);
\draw[rotate=120](-1.6,0)--(-1.2,0.2);
\draw[rotate=120,-<](-1.6,0)--(-1.3,0.15);
\draw[rotate=120](-1.6,0)--(-1.2,-0.2);
\draw[rotate=120,-<](-1.6,0)--(-1.3,-0.15);
\draw[rotate=120](-0.61,0.1)--(-0.4,0.2);
\draw[rotate=120,-<](-0.61,0.1)--(-0.5,0.15);
\draw[rotate=120](-0.99,0.1)--(-1.2,0.2);
\draw[rotate=120,->](-0.99,-0.1)--(-1.1,-0.15);
\draw[rotate=120](-0.61,-0.1)--(-0.4,-0.2);
\draw[rotate=120,->](-0.99,0.1)--(-1.1,0.15);
\draw[rotate=120](-0.99,-0.1)--(-1.2,-0.2);
\draw[rotate=120,-<](-0.61,-0.1)--(-0.5,-0.15);
\draw[rotate=120](-0,0)--(-0.4,-0.2);
\draw[rotate=120,->](-0,0)--(-0.3,-0.15);
\draw[rotate=120](-0,0)--(-0.4,0.2);
\draw[rotate=120,->](-0,0)--(-0.3,0.15);
\draw[rotate=120](-0,0)--(0.4,-0.2);
\draw[rotate=120,-<](-0,0)--(0.3,-0.15);
\draw[rotate=120](-0,0)--(0.4,0.2);
\draw[rotate=120,-<](-0,0)--(0.3,0.15);
\draw[rotate=120](3.4,0.2)--(2.8,-0.2);
\draw[rotate=120](3.4,-0.2)--(2.8,0.2);
\draw[rotate=120,->](3.1,0)--(2.9,-0.14);
\draw[rotate=120,->](3.1,0)--(2.9,0.14);
\draw[rotate=120,-<](3.1,0)--(3.3,-0.14);
\draw[rotate=120,-<](3.1,0)--(3.3,0.14);
\draw[rotate=120](3.6,-0.1)--(3.4,-0.2);
\draw[rotate=120,->](3.6,-0.1)--(3.5,-0.14);
\draw[rotate=120](3.6,0.1)--(3.4,0.2);
\draw[rotate=120,->](3.6,0.1)--(3.5,0.14);
\draw[rotate=120](4,0.1)--(4.2,0.2);
\draw[rotate=120,-<](4,0.1)--(4.1,0.14);
\draw[rotate=120](4,-0.1)--(4.2,-0.2);
\draw[rotate=120,-<](4,-0.1)--(4.1,-0.14);
\draw[rotate=120](4.8,0.2)--(4.2,-0.2);
\draw[rotate=120](4.8,-0.2)--(4.2,0.2);
\draw[rotate=120,-<](4.5,-0)--(4.7,-0.14);
\draw[rotate=120,-<](4.5,-0)--(4.7,0.14);
\draw[rotate=120,->](4.5,-0)--(4.3,-0.14);
\draw[rotate=120,->](4.5,-0)--(4.3,0.14);
\draw[rotate=120](5,-0.1)--(4.8,-0.2);
\draw[rotate=120,->](5,-0.1)--(4.9,-0.15);
\draw[rotate=120](5,0.1)--(4.8,0.2);
\draw[rotate=120,->](5,0.1)--(4.9,0.15);
\draw[rotate=120](5.6,-0.2)--(5.4,-0.1);
\draw[rotate=120,-<](5.4,-0.1)--(5.5,-0.15);
\draw[rotate=120](5.6,0.2)--(5.4,0.1);
\draw[rotate=120,-<](5.4,0.1)--(5.5,0.15);
\draw[rotate=120](-3.4,0.2)--(-2.8,-0.2);
\draw[rotate=120](-3.4,-0.2)--(-2.8,0.2);
\draw[rotate=120,-<](-3.1,0)--(-2.9,-0.14);
\draw[rotate=120,-<](-3.1,0)--(-2.9,0.14);
\draw[rotate=120,->](-3.1,0)--(-3.3,-0.14);
\draw[rotate=120,->](-3.1,0)--(-3.3,0.14);
\draw[rotate=120](-3.6,-0.1)--(-3.4,-0.2);
\draw[rotate=120,-<](-3.6,-0.1)--(-3.5,-0.14);
\draw[rotate=120](-3.6,0.1)--(-3.4,0.2);
\draw[rotate=120,-<](-3.6,0.1)--(-3.5,0.14);
\draw[rotate=120](-4,0.1)--(-4.2,0.2);
\draw[rotate=120,->](-4,0.1)--(-4.1,0.14);
\draw[rotate=120](-4,-0.1)--(-4.2,-0.2);
\draw[rotate=120,->](-4,-0.1)--(-4.1,-0.14);
\draw[rotate=120](-4.8,0.2)--(-4.2,-0.2);
\draw[rotate=120](-4.8,-0.2)--(-4.2,0.2);
\draw[rotate=120,->](-4.5,-0)--(-4.7,-0.14);
\draw[rotate=120,->](-4.5,-0)--(-4.7,0.14);
\draw[rotate=120,-<](-4.5,-0)--(-4.3,-0.14);
\draw[rotate=120,-<](-4.5,-0)--(-4.3,0.14);
\draw[rotate=120](-5,-0.1)--(-4.8,-0.2);
\draw[rotate=120,-<](-5,-0.1)--(-4.9,-0.15);
\draw[rotate=120](-5,0.1)--(-4.8,0.2);
\draw[rotate=120,-<](-5,0.1)--(-4.9,0.15);
\draw[rotate=120](-5.6,-0.2)--(-5.4,-0.1);
\draw[rotate=120,->](-5.4,-0.1)--(-5.5,-0.15);
\draw[rotate=120](-5.6,0.2)--(-5.4,0.1);
\draw[rotate=120,->](-5.4,0.1)--(-5.5,0.15);
\draw(3.4,0.2)--(2.8,-0.2);
\draw(3.4,-0.2)--(2.8,0.2);
\draw[->](3.1,0)--(2.9,-0.14);
\draw[->](3.1,0)--(2.9,0.14);
\draw[-<](3.1,0)--(3.3,-0.14);
\draw[-<](3.1,0)--(3.3,0.14);
\draw(3.6,-0.1)--(3.4,-0.2);
\draw[->](3.6,-0.1)--(3.5,-0.14);
\draw(3.6,0.1)--(3.4,0.2);
\draw[->](3.6,0.1)--(3.5,0.14);
\draw(4,0.1)--(4.2,0.2);
\draw[-<](4,0.1)--(4.1,0.14);
\draw(4,-0.1)--(4.2,-0.2);
\draw[-<](4,-0.1)--(4.1,-0.14);
\draw(4.8,0.2)--(4.2,-0.2);
\draw(4.8,-0.2)--(4.2,0.2);
\draw[-<](4.5,-0)--(4.7,-0.14);
\draw[-<](4.5,-0)--(4.7,0.14);
\draw[->](4.5,-0)--(4.3,-0.14);
\draw[->](4.5,-0)--(4.3,0.14);
\draw(5,-0.1)--(4.8,-0.2);
\draw[->](5,-0.1)--(4.9,-0.15);
\draw(5,0.1)--(4.8,0.2);
\draw[->](5,0.1)--(4.9,0.15);
\draw(5.6,-0.2)--(5.4,-0.1);
\draw[-<](5.4,-0.1)--(5.5,-0.15);
\draw(5.6,0.2)--(5.4,0.1);
\draw[-<](5.4,0.1)--(5.5,0.15);
\draw(-3.4,0.2)--(-2.8,-0.2);
\draw(-3.4,-0.2)--(-2.8,0.2);
\draw[-<](-3.1,0)--(-2.9,-0.14);
\draw[-<](-3.1,0)--(-2.9,0.14);
\draw[->](-3.1,0)--(-3.3,-0.14);
\draw[->](-3.1,0)--(-3.3,0.14);
\draw(-3.6,-0.1)--(-3.4,-0.2);
\draw[-<](-3.6,-0.1)--(-3.5,-0.14);
\draw(-3.6,0.1)--(-3.4,0.2);
\draw[-<](-3.6,0.1)--(-3.5,0.14);
\draw(-4,0.1)--(-4.2,0.2);
\draw[->](-4,0.1)--(-4.1,0.14);
\draw(-4,-0.1)--(-4.2,-0.2);
\draw[->](-4,-0.1)--(-4.1,-0.14);
\draw(-4.8,0.2)--(-4.2,-0.2);
\draw(-4.8,-0.2)--(-4.2,0.2);
\draw[->](-4.5,-0)--(-4.7,-0.14);
\draw[->](-4.5,-0)--(-4.7,0.14);
\draw[-<](-4.5,-0)--(-4.3,-0.14);
\draw[-<](-4.5,-0)--(-4.3,0.14);
\draw(-5,-0.1)--(-4.8,-0.2);
\draw[-<](-5,-0.1)--(-4.9,-0.15);
\draw(-5,0.1)--(-4.8,0.2);
\draw[-<](-5,0.1)--(-4.9,0.15);
\draw(-5.6,-0.2)--(-5.4,-0.1);
\draw[->](-5.4,-0.1)--(-5.5,-0.15);
\draw(-5.6,0.2)--(-5.4,0.1);
\draw[->](-5.4,0.1)--(-5.5,0.15);
\draw[rotate=60](3.4,0.2)--(2.8,-0.2);
\draw[rotate=60](3.4,-0.2)--(2.8,0.2);
\draw[rotate=60,->](3.1,0)--(2.9,-0.14);
\draw[rotate=60,->](3.1,0)--(2.9,0.14);
\draw[rotate=60,-<](3.1,0)--(3.3,-0.14);
\draw[rotate=60,-<](3.1,0)--(3.3,0.14);
\draw[rotate=60](3.6,-0.1)--(3.4,-0.2);
\draw[rotate=60,->](3.6,-0.1)--(3.5,-0.14);
\draw[rotate=60](3.6,0.1)--(3.4,0.2);
\draw[rotate=60,->](3.6,0.1)--(3.5,0.14);
\draw[rotate=60](4,0.1)--(4.2,0.2);
\draw[rotate=60,-<](4,0.1)--(4.1,0.14);
\draw[rotate=60](4,-0.1)--(4.2,-0.2);
\draw[rotate=60,-<](4,-0.1)--(4.1,-0.14);
\draw[rotate=60](4.8,0.2)--(4.2,-0.2);
\draw[rotate=60](4.8,-0.2)--(4.2,0.2);
\draw[rotate=60,-<](4.5,-0)--(4.7,-0.14);
\draw[rotate=60,-<](4.5,-0)--(4.7,0.14);
\draw[rotate=60,->](4.5,-0)--(4.3,-0.14);
\draw[rotate=60,->](4.5,-0)--(4.3,0.14);
\draw[rotate=60](5,-0.1)--(4.8,-0.2);
\draw[rotate=60,->](5,-0.1)--(4.9,-0.15);
\draw[rotate=60](5,0.1)--(4.8,0.2);
\draw[rotate=60,->](5,0.1)--(4.9,0.15);
\draw[rotate=60](5.6,-0.2)--(5.4,-0.1);
\draw[rotate=60,-<](5.4,-0.1)--(5.5,-0.15);
\draw[rotate=60](5.6,0.2)--(5.4,0.1);
\draw[rotate=60,-<](5.4,0.1)--(5.5,0.15);
\draw[rotate=60](-3.4,0.2)--(-2.8,-0.2);
\draw[rotate=60](-3.4,-0.2)--(-2.8,0.2);
\draw[rotate=60,-<](-3.1,0)--(-2.9,-0.14);
\draw[rotate=60,-<](-3.1,0)--(-2.9,0.14);
\draw[rotate=60,->](-3.1,0)--(-3.3,-0.14);
\draw[rotate=60,->](-3.1,0)--(-3.3,0.14);
\draw[rotate=60](-3.6,-0.1)--(-3.4,-0.2);
\draw[rotate=60,-<](-3.6,-0.1)--(-3.5,-0.14);
\draw[rotate=60](-3.6,0.1)--(-3.4,0.2);
\draw[rotate=60,-<](-3.6,0.1)--(-3.5,0.14);
\draw[rotate=60](-4,0.1)--(-4.2,0.2);
\draw[rotate=60,->](-4,0.1)--(-4.1,0.14);
\draw[rotate=60](-4,-0.1)--(-4.2,-0.2);
\draw[rotate=60,->](-4,-0.1)--(-4.1,-0.14);
\draw[rotate=60](-4.8,0.2)--(-4.2,-0.2);
\draw[rotate=60](-4.8,-0.2)--(-4.2,0.2);
\draw[rotate=60,->](-4.5,-0)--(-4.7,-0.14);
\draw[rotate=60,->](-4.5,-0)--(-4.7,0.14);
\draw[rotate=60,-<](-4.5,-0)--(-4.3,-0.14);
\draw[rotate=60,-<](-4.5,-0)--(-4.3,0.14);
\draw[rotate=60](-5,-0.1)--(-4.8,-0.2);
\draw[rotate=60,-<](-5,-0.1)--(-4.9,-0.15);
\draw[rotate=60](-5,0.1)--(-4.8,0.2);
\draw[rotate=60,-<](-5,0.1)--(-4.9,0.15);
\draw[rotate=60](-5.6,-0.2)--(-5.4,-0.1);
\draw[rotate=60,->](-5.4,-0.1)--(-5.5,-0.15);
\draw[rotate=60](-5.6,0.2)--(-5.4,0.1);
\draw[rotate=60,->](-5.4,0.1)--(-5.5,0.15);
\draw[thick,red](0.8,0) circle (0.205);
\draw[thick,red](2.4,0) circle (0.205);
\draw[thick,red](-0.8,0) circle (0.205);
\draw[thick,red](-2.4,0) circle (0.205);
\draw[thick,red](0.4,0.69) circle (0.205);
\draw[thick,red](0.4,-0.69) circle (0.205);
\draw[thick,red](-0.4,0.69) circle (0.205);
\draw[thick,red](-0.4,-0.69) circle (0.205);
\draw[thick,red](1.2,2.08) circle (0.205);
\draw[thick,red](1.2,-2.08) circle (0.205);
\draw[thick,red](-1.2,2.08) circle (0.205);
\draw[thick,red](-1.2,-2.08) circle (0.205);
\draw[thick,red](1.9,3.29) circle (0.205);
\draw[thick,red](1.9,-3.29) circle (0.205);
\draw[thick,red](-1.9,3.29) circle (0.205);
\draw[thick,red](-1.9,-3.29) circle (0.205);
\draw[thick,red](2.6,4.5) circle (0.205);
\draw[thick,red](2.6,-4.5) circle (0.205);
\draw[thick,red](-2.6,4.5) circle (0.205);
\draw[thick,red](-2.6,-4.5) circle (0.205);
\draw[thick,red](3.8,0) circle (0.205);
\draw[thick,red](-3.8,0) circle (0.205);
\draw[thick,red](-5.2,0) circle (0.205);
\draw[thick,red](5.2,0) circle (0.205);
\draw[blue,dashed](0,0)--(3,5.196)node[above]{ $e^{\frac{ \pi i}{3}}$};	
\draw[teal,dashed](0,0)--(3,-5.196);
\draw[teal,dashed](0,0)--(-3,5.196)node[above]{ $e^{\frac{2 \pi i}{3}}$};
\draw[blue,dashed](0,0)--(-3,-5.196);
\draw[dashed] (0.8,0) arc (0:360:0.8);
\draw[dashed] (3.8,0) arc (0:360:3.8);
\draw[dashed] (2.4,0) arc (0:360:2.4);
\draw[dashed] (5.2,0) arc (0:360:5.2);
\coordinate (A) at (-5.2,0);
\fill (A) circle (1pt) node[below] {\scriptsize$\xi_8$};
\coordinate (b) at (-3.8,0);
\fill (b) circle (1pt) node[below] {\scriptsize$\xi_7$};
\coordinate (C) at (-0.8,0);
\fill (C) circle (1pt) node[below] {\scriptsize$\xi_5$};
\coordinate (d) at (-2.4,0);
\fill (d) circle (1pt) node[below] {\scriptsize$\xi_6$};
\coordinate (E) at (5.2,0);
\fill (E) circle (1pt) node[below] {\scriptsize$\xi_1$};
\coordinate (R) at (3.8,0);
\fill (R) circle (1pt) node[below] {\scriptsize$\xi_2$};
\coordinate (T) at (0.8,0);
\fill (T) circle (1pt) node[below] {\scriptsize$\xi_4$};
\coordinate (Y) at (2.4,0);
\fill (Y) circle (1pt) node[below] {\scriptsize$\xi_3$};
\coordinate (A1) at (0.4,0.69);
\fill[blue] (A1) circle (1pt) node[right] {\scriptsize$\omega^2\xi_5$};
\coordinate (A2) at (0.4,-0.69);
\fill[teal] (A2) circle (1pt) node[right] {\scriptsize$\omega\xi_5$};
\coordinate (A3) at (-0.4,0.69);
\fill[teal] (A3) circle (1pt) node[left] {\scriptsize$\omega\xi_4$};
\coordinate (A4) at (-0.4,-0.69);
\fill[blue] (A4) circle (1pt) node[left] {\scriptsize$\omega^2\xi_4$};
\coordinate (n1) at (1.9,3.29);
\fill[blue] (n1) circle (1pt) node[right] {\scriptsize$\omega^2\xi_7$};
\coordinate (m1) at (1.9,-3.29);
\fill[teal] (m1) circle (1pt) node[right] {\scriptsize$\omega\xi_7$};
\coordinate (j1) at (-1.9,3.29);
\fill[teal] (j1) circle (1pt) node[left] {\scriptsize$\omega\xi_2$};
\coordinate[blue] (k1) at (-1.9,-3.29);
\fill[blue] (k1) circle (1pt) node[left] {\scriptsize$\omega^2\xi_2$};
\coordinate (n2) at (1.2,2.08);
\fill[blue] (n2) circle (1pt) node[right] {\scriptsize$\omega^2\xi_6$};
\coordinate (m2) at (1.2,-2.08);
\fill[teal] (m2) circle (1pt) node[right] {\scriptsize$\omega\xi_6$};
\coordinate (j2) at (-1.2,2.08);
\fill[teal] (j2) circle (1pt) node[left] {\scriptsize$\omega\xi_3$};
\coordinate[blue] (k2) at (-1.2,-2.08);
\fill[blue] (k2) circle (1pt) node[left] {\scriptsize$\omega^2\xi_3$};
\coordinate (n23) at (2.6,4.5);
\fill[blue] (n23) circle (1pt) node[right] {\scriptsize$\omega^2\xi_8$};
\coordinate (m23) at (2.6,-4.5);
\fill[teal] (m23) circle (1pt) node[right] {\scriptsize$\omega\xi_8$};
\coordinate (j23) at (-2.6,4.5);
\fill[teal] (j23) circle (1pt) node[left] {\scriptsize$\omega\xi_1$};
\coordinate[blue] (k3) at (-2.6,-4.5);
\fill[blue] (k3) circle (1pt) node[left] {\scriptsize$\omega^2\xi_1$};
\coordinate (kap1) at (3.4,0);
\fill[orange] (kap1) circle (1pt) node[below] {\scriptsize$1$};
\coordinate (kap4) at (-3.4,0);
\fill[orange] (kap4) circle (1pt) node[below] {\scriptsize$-1$};
\coordinate (kap2) at (1.7,2.944);
\fill[orange] (kap2) circle (1pt) node[right] {\scriptsize$\varkappa_2$};
\coordinate (kap3) at (-1.7,2.944);
\fill[orange] (kap3) circle (1pt) node[right] {\scriptsize$\varkappa_3$};
\coordinate (kap6) at (1.7,-2.944);
\fill[orange] (kap6) circle (1pt) node[right] {\scriptsize$\varkappa_6$};
\coordinate (kap5) at (-1.7,-2.944);
\fill[orange] (kap5) circle (1pt) node[right] {\scriptsize$\varkappa_5$};
\draw[orange](3.3,-0.1)--(3.3,0.1)--(3.5,0.1)--(3.5,-0.1)--(3.3,-0.1);
\draw[rotate=60,orange](3.3,-0.1)--(3.3,0.1)--(3.5,0.1)--(3.5,-0.1)--(3.3,-0.1);
\draw[rotate=120,orange](3.3,-0.1)--(3.3,0.1)--(3.5,0.1)--(3.5,-0.1)--(3.3,-0.1);
\draw[rotate=180,orange](3.3,-0.1)--(3.3,0.1)--(3.5,0.1)--(3.5,-0.1)--(3.3,-0.1);
\draw[rotate=240,orange](3.3,-0.1)--(3.3,0.1)--(3.5,0.1)--(3.5,-0.1)--(3.3,-0.1);
\draw[rotate=300,orange](3.3,-0.1)--(3.3,0.1)--(3.5,0.1)--(3.5,-0.1)--(3.3,-0.1);
\end{tikzpicture}
}
\caption{\footnotesize   The jump contour $\Sigma^{E}$ for the $V^{E}$ when $\xi\in(-1/8,1)$. The red circles are $U(\xi)$, and orange  square are $\mathbb{B}_j$. }
	\label{figE}
\end{figure}
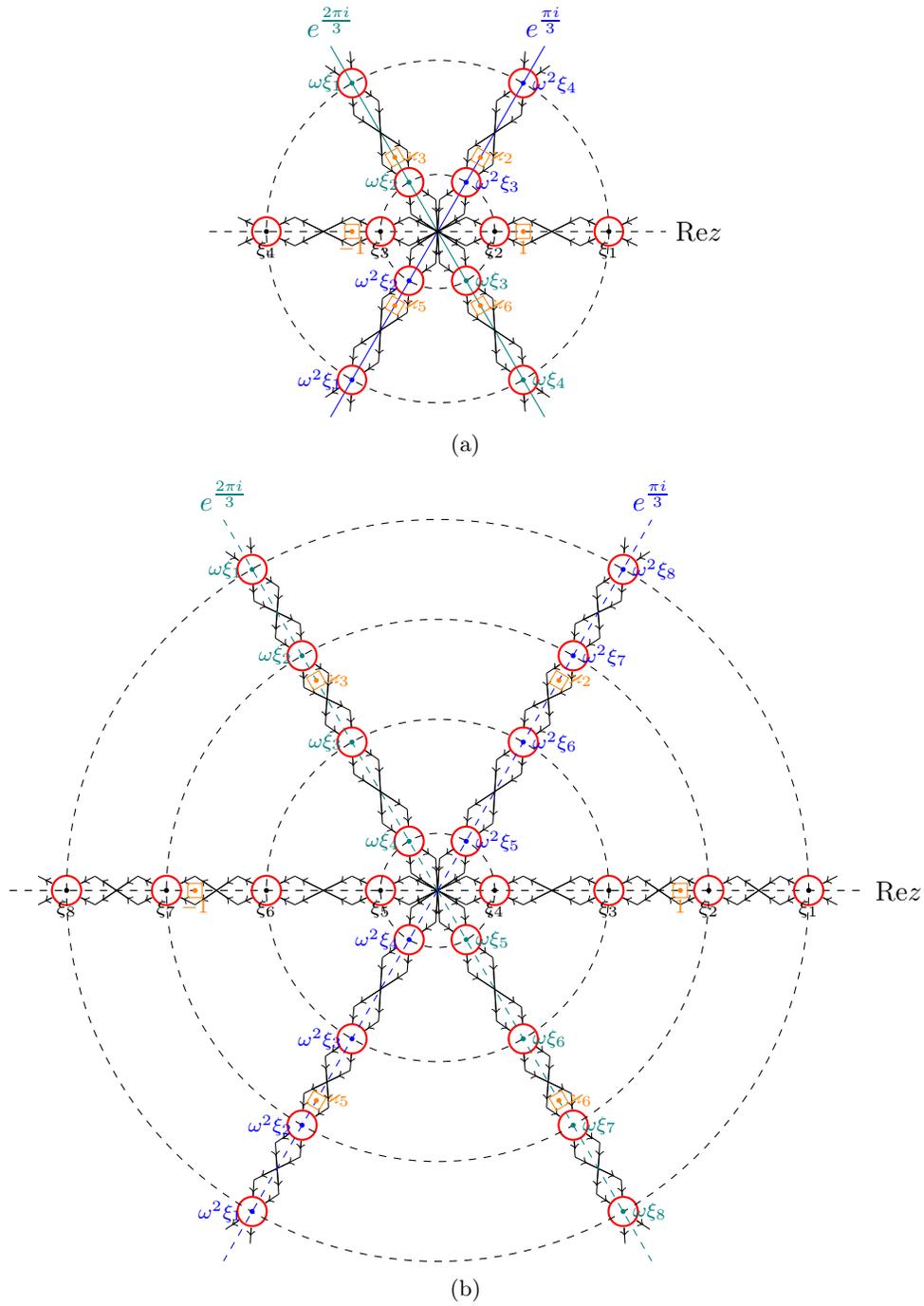
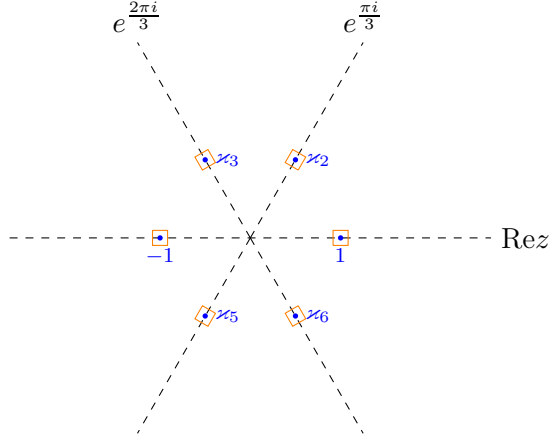
\begin{figure}[htp]
	\centering
		\begin{tikzpicture}
			\draw[dashed](-3.2,0)--(3.2,0)node[right]{ Re$z$};
			\draw[ dashed](0,0)--(1.5,2.6)node[above]{ $e^{\frac{ \pi i}{3}}$};
			\draw[ dashed](0,0)--(-1.5,2.6)node[above]{ $e^{\frac{2 \pi i}{3}}$};
			\draw[dashed](0,0)--(1.5,-2.6);
			\draw[ dashed](0,0)--(-1.5,-2.6);
			\coordinate (kap1) at (1.2,0);
			\fill[blue] (kap1) circle (1pt) node[below] {\scriptsize$1$};
			\draw[orange](1.1,-0.1)--(1.1,0.1)--(1.3,0.1)--(1.3,-0.1)--(1.1,-0.1);
			\coordinate (kap4) at (-1.2,0);
			\fill[blue] (kap4) circle (1pt) node[below] {\scriptsize$-1$};
			\draw[orange](-1.1,-0.1)--(-1.1,0.1)--(-1.3,0.1)--(-1.3,-0.1)--(-1.1,-0.1);
			\draw[rotate=60,orange](1.1,-0.1)--(1.1,0.1)--(1.3,0.1)--(1.3,-0.1)--(1.1,-0.1);		\draw[rotate=60,orange](-1.1,-0.1)--(-1.1,0.1)--(-1.3,0.1)--(-1.3,-0.1)--(-1.1,-0.1);
			\draw[rotate=120,orange](1.1,-0.1)--(1.1,0.1)--(1.3,0.1)--(1.3,-0.1)--(1.1,-0.1);		\draw[rotate=120,orange](-1.1,-0.1)--(-1.1,0.1)--(-1.3,0.1)--(-1.3,-0.1)--(-1.1,-0.1);
			\coordinate[orange] (kap2) at (0.6,1.039);
			\fill[blue] (kap2) circle (1pt) node[right]   {\scriptsize$\varkappa_2$};
			\coordinate[orange] (kap3) at (-0.6,1.039);
			\fill[blue] (kap3) circle (1pt) node[right] {\scriptsize$\varkappa_3$};
			\coordinate[orange] (kap6) at (0.6,-1.039);
			\fill[blue] (kap6) circle (1pt) node[right] {\scriptsize$\varkappa_6$};
			\coordinate[orange] (kap5)  at (-0.6,-1.039);
			\fill[blue] (kap5) circle (1pt) node[right] {\scriptsize$\varkappa_5$};
		\end{tikzpicture}	
	\caption{  The jump contour $\Sigma^{E}$ for the $V^{E}$ when $\xi\in(-\infty,-1/8)\cup(1,+\infty)$. The orange  square are $\mathbb{B}_j$. }
	\label{figE1}
\end{figure}
  We will show  that  the small norm  RH problem about  $E^{(2)}(z;\xi)$ is solvable for large $t$.
  By  {Proposition \ref{prov2}},  the jump matrix  $V^{E}$ admits   the following estimate
\begin{equation}
\parallel V^{E}(z)-I\parallel_p\lesssim\left\{\begin{array}{llll}
\exp\left\{-tK_p\right\},  & z\in \Sigma_{kj}\setminus U(\xi),\\[6pt]
\exp\left\{-tK_p'\right\},   & z\in \Sigma'_{kj}.
\end{array}\right. \label{VE-I}
\end{equation}
For $z\in \partial U(\xi)\cup\partial\mathbb{B}_j$,  $M^{(r)}(z)$ is bounded,  so   by  Proposition \ref{asymlo} and (\ref{asyMB}),
\begin{align}
&| V^{E}(z)-I|=   \big|M^{(r)}(z)^{-1}(M^{lo}(z)-I)M^{(r)}(z) \big| = \mathcal{O}(t^{-1/2}),\text{ }z\in\partial U(\xi)\nonumber\\
&| V^{E}(z)-I|=   \big|M^{(r)}(z)^{-1}(M^{B}_j(z)-I)M^{(r)}(z) \big| = \mathcal{O}(t^{-1}),\text{ }z\in\partial \mathbb{B}_j
.\label{VE}
\end{align}
Therefore,    the   existence and uniqueness  of  the RHP \ref{E2} can be shown  by using  a  small-norm RH problem \cite{RN9,RN10},  moreover  its solution   can be given by
\begin{equation}
E^{(2)}(z;\xi)=I+\frac{1}{2\pi i}\int_{\Sigma^{E}}\dfrac{\left( I+\varpi(s)\right) (V^{E}(s)-I)}{s-z}ds,\label{Ez}
\end{equation}
where the $\varpi\in L^\infty(\Sigma^{E})$ is the unique solution of following equation
\begin{equation}
(1-C_E)\varpi=C_E\left(I \right),
\end{equation}
where  $C_E$ is  a integral operator  defined by
\begin{equation}
C_E(f)(z)=\mathcal{P}^-\left( f(V^{E}(z) -I)\right) ,
\end{equation}
and  the $\mathcal{P}^-$ is the usual Cauchy projection operator on $\Sigma^{E}$.
By  (\ref{VE}),   we have
\begin{equation}
\parallel C_E\parallel\leq\parallel \mathcal{P}^-\parallel \parallel V^{E}(z)-I\parallel_2 \lesssim \mathcal{O}(t^{-1/2}),
\end{equation}
which implies that  $1-C_E$ is invertible  for   sufficiently large $t$.    So  $\varpi$  exists and is unique,
 moreover,
\begin{equation}
\parallel \varpi\parallel_{L^\infty(\Sigma^{E})}\lesssim\dfrac{\parallel C_E\parallel}{1-\parallel C_E\parallel}\lesssim t ^{-1/2},\label{normrho}
\end{equation}
which gives  the existence and uniqueness of $E^{(2)}$.
In order to reconstruct the solution $u(y,t)$ of (\ref{Novikov}), it is  necessary to consider the long time asymptotic behavior of $E^{(2)}(e^{\frac{i\pi}{6}})$. Note that when we estimate its  asymptotic behavior, from (\ref{Ez}) and (\ref{VE-I}) we only need to consider the calculation on $\partial U(\xi)$ because it  approaches zero exponentially on other boundary.
\begin{Proposition}
	As $z= e^{\frac{i\pi}{6}}$, we have
	\begin{align}
	E^{(2)}(e^{\frac{i\pi}{6}})=I+\frac{1}{2\pi i}\int_{\Sigma^{E}}\dfrac{\left( I+\varpi(s)\right) (V^{E}-I)}{s-e^{\frac{i\pi}{6}}}ds,
	\end{align}
	with long time asymptotic behavior
	\begin{equation}
	E^{(2)}(e^{\frac{i\pi}{6}})=I+t^{-1/2}H^{(0)}+\mathcal{O}(t^{-1+\rho}),\label{E0t}
	\end{equation}
where $H^{(0)}$  is explicitly  computed by
	\begin{align}
	H^{(0)}=&\sum_{ k=1 }^{p(\xi)}\frac{1}{2\pi i}\int_{\partial U_{\xi_k}}\dfrac{M^{(r)}(s)^{-1}A_k(\xi)M^{(r)}(s)}{(s-e^{\frac{i\pi}{6}})( s- \xi_k)}ds+\sum_{ k=1 }^{p(\xi)}\frac{\omega}{2\pi i}\int_{\partial U_{\xi_k}}\dfrac{M^{(r)}(s)^{-1}\Gamma_3\overline{A_k(\xi)}\Gamma_3M^{(r)}(s)}{(s-e^{\frac{i\pi}{6}})( s- \omega\xi_k)}ds\nonumber\\
		&+\sum_{ k=1 }^{p(\xi)}\frac{\omega^2}{2\pi i}\int_{\partial U_{\xi_k}}\dfrac{M^{(r)}(s)^{-1}\Gamma_2\overline{A_k(\xi)}\Gamma_2M^{(r)}(s)}{(s-e^{\frac{i\pi}{6}})( s- \omega^2\xi_k)}ds\nonumber\\
	=&\sum_{ k=1 }^{p(\xi)}\frac{1}{\xi_k-e^{\frac{i\pi}{6}}} M^{(r)}(\xi_k)^{-1}A_k(\xi)M^{(r)}(\xi_k)+\sum_{ k=1 }^{p(\xi)}\frac{\omega}{\omega\xi_k-e^{\frac{i\pi}{6}}} M^{(r)}(\xi_k)^{-1}\Gamma_3\overline{A_k(\xi)}\Gamma_3M^{(r)}(\xi_k)\nonumber\\
	&+\sum_{ k=1 }^{p(\xi)}\frac{\omega^2}{\omega^2\xi_k-e^{\frac{i\pi}{6}}} M^{(r)}(\xi_k)^{-1}\Gamma_2\overline{A_k(\xi)}\Gamma_2M^{(r)}(\xi_k).\label{H0}
	\end{align}
	The last equality  above  follows from a residue calculation.  Besides, for $j=1,...,6$
	\begin{align}
		E^{(2)}(\varkappa_j)=I+\frac{1}{2\pi i}\int_{\Sigma^{E}}\dfrac{\left( I+\varpi(s)\right) (V^{E}-I)}{s-\varkappa_j}ds,
	\end{align}
	with long time asymptotic behavior
	\begin{equation}
		E^{(2)}(\varkappa_j)=I+\mathcal{O}(t^{-1+\rho}).\label{Ekapj}
	\end{equation}
\end{Proposition}
In order to facilitate calculation, denote:
\begin{align}
	f_{2}=&\frac{\sum_{n=1}^3H^{(0)}_{3n}[M^{(r)}_\lozenge]_{n3}(e^{\pi i/6})}{[M^{(r)}_\lozenge]_{33}(e^{\pi i/6})}-\frac{\sum_{n=1}^3H^{(0)}_{1n}[M^{(r)}_\lozenge]_{n1}(e^{\pi i/6})}{[M^{(r)}_\lozenge]_{11}(e^{\pi i/6})},\label{f12}
\end{align}
and
\begin{align}
f_{1}=&\frac{1}{2}[\tilde{m}^{(r)}_\lozenge]_1\left( \frac{[M^{(r)}_\lozenge]_{33}(e^{\pi i/6})}{[M^{(r)}_\lozenge]_{11}(e^{\pi i/6})}\right)^{1/2} \left(T_1(e^{\pi i/6})T_3(e^{\pi i/6}) \right)^{-1/2}\nonumber\\
&\left(\frac{\sum_{j,n=1}^3H^{(0)}_{jn}[M^{(r)}_\lozenge]_{n1}(e^{\pi i/6})}{[\tilde{m}^{(r)}_\lozenge]_1}+\frac{1}{2}f_{2}  \right) \nonumber\\
&+\frac{1}{2}[\tilde{m}^{(r)}_\lozenge]_3\left( \frac{[M^{(r)}_\lozenge]_{33}(e^{\pi i/6})}{[M^{(r)}_\lozenge]_{11}(e^{\pi i/6})}\right)^{-1/2} \left(T_1(e^{\pi i/6})T_3(e^{\pi i/6}) \right)^{-1/2}\nonumber\\
&\left(\frac{\sum_{j,n=1}^3H^{(0)}_{jn}[M^{(r)}_\lozenge]_{n3}(e^{\pi i/6})}{[\tilde{m}^{(r)}_\lozenge]_3}-\frac{1}{2}f_{2}  \right)  .\label{f11}
\end{align}

Finally we consider the error between $E(z;\xi)$ and $E^{(2)}(z;\xi)$:
\begin{align}
	E^{(1)}(z;\xi)=E(z;\xi)E^{(2)}(z;\xi)^{-1},
\end{align}
which is a solution of  RH problem  only has singularities with:
\begin{align}
	\lim_{z\to \varkappa_j}E^{(1)}(z;\xi)=&\lim_{z\to \varkappa_j}M^R(z)M^B_j(\varkappa_l)^{-1}M^{(r)}(z)^{-1}E^{(2)}(\varkappa_j;\xi)^{-1}\nonumber\\
	\triangleq&\frac{E^{(1)}_{-2}(\xi)}{(z-\varkappa_j)^2}+\frac{E^{(1)}_{-1}(\xi)}{z-\varkappa_j}+\mathcal{O}(1),
\end{align}
which leads to
\begin{align}
	E^{(1)}(z;\xi)=I+\sum_{j=1 }^6\left( \frac{E^{(1),j}_{-2}(\xi)}{(z-\varkappa_j)^2}+\frac{E^{(1),j}_{-1}(\xi)}{z-\varkappa_j}\right) .
\end{align}
Here, from $M^R(z)$  and $M^{(r)}(z)$ both admit (\ref{MR1})-(\ref{MR3}), there has
\begin{align}
	E^{(1),j}_{-2}(\xi)=&\left(\begin{array}{ccc}
		\alpha^{(2)}_\pm &	\alpha^{(2)}_\pm & \beta^{(2)}_\pm \\
		-\alpha^{(2)}_\pm & -\alpha^{(2)}_\pm & -\beta^{(2)}_\pm\\
		0	&	0 & 0
	\end{array}\right)M^B_j(\varkappa_j)^{-1}\nonumber\\
&\left(\begin{array}{ccc}
	\tilde{\alpha}^{(2)}_\pm &	\tilde{\alpha}^{(2)}_\pm & \tilde{\beta}^{(2)}_\pm \\
	-\tilde{\alpha}^{(2)}_\pm & -\tilde{\alpha}^{(2)}_\pm & -\tilde{\beta}^{(2)}_\pm\\
	0	&	0 & 0
\end{array}\right)E^{(2)}(\varkappa_j;\xi)^{-1}.
\end{align}
As $t\to\infty$, (\ref{MBkap}) and (\ref{Ekapj}) lead to that
\begin{align}
		E^{(1),j}_{-2}(\xi)=&\left(\begin{array}{ccc}
		\alpha^{(2)}_\pm &	\alpha^{(2)}_\pm & \beta^{(2)}_\pm \\
		-\alpha^{(2)}_\pm & -\alpha^{(2)}_\pm & -\beta^{(2)}_\pm\\
		0	&	0 & 0
	\end{array}\right)\left(\begin{array}{ccc}
		\tilde{\alpha}^{(2)}_\pm &	\tilde{\alpha}^{(2)}_\pm & \tilde{\beta}^{(2)}_\pm \\
		-\tilde{\alpha}^{(2)}_\pm & -\tilde{\alpha}^{(2)}_\pm & -\tilde{\beta}^{(2)}_\pm\\
		0	&	0 & 0
	\end{array}\right)+\mathcal{O}(t^{-1+\rho})\nonumber\\
&=\mathcal{O}(t^{-1+\rho}).
\end{align}
Analogously, as $t\to\infty$, the  coefficient of $(z-\varkappa_j)^{-1}$ has
\begin{align}
	E^{(1),j}_{-1}(\xi)=\mathcal{O}(t^{-1+\rho}).
\end{align}
Summarizing above   results   gives
\begin{Proposition}\label{asyE}
	Taking $z= e^{\frac{i\pi}{6}}$, we have
	\begin{align}
		E(e^{\frac{i\pi}{6}};\xi)
		=&I+t^{-1/2}H^{(0)}+\mathcal{O}(t^{-1+\rho}),
	\end{align}
where $H^{(0)}$ is given by    (\ref{H0}).
\end{Proposition}

\section{Contribution from  $\bar\partial$-components }\label{sec8}
\quad Now we consider the   asymptotics  of $M^{(3)}$ of the  $\bar{\partial}$-problem 4,
whose solution   can be given by an integral equation
\begin{equation}
M^{(3)}(z)=I+\frac{1}{\pi}\iint_\mathbb{C}\dfrac{M^{(3)}(s)W^{(3)} (s)}{s-z}dm(s),\label{m3}
\end{equation}
where $m(s)$ is the Lebesgue measure on the $\mathbb{C}$.
Define the left Cauchy-Green integral  operator,
\begin{equation*}
fC_z(z)=\frac{1}{\pi}\iint_{\mathbb{C}}\dfrac{f(s)W^{(3)} (s)}{s-z}dm(s),
\end{equation*}
then  the above equation (\ref{m3})  can be rewritten as
\begin{equation}
\left(I-C_z \right) M^{(3)}(z)=I.\label{deM3}
\end{equation}
In the follow,  we   prove the existence of the  operator $\left(I-  {C}_z \right)^{-1}$
 in different  space-time regions
 $\xi\in(1,+\infty)\cup(-\infty,-1/8)$ and $\xi\in(-1/8,1)$.

\subsection{In space-time   regions   $  \xi <-1/8 $ and $   \xi >1 $ }\label{10.1}

\begin{lemma}\label{Cz}
	The norm of the integral operator $C_z$ admits the estimate
	\begin{equation}
	\parallel C_z\parallel_{L^\infty\to L^\infty}\lesssim  t ^{-1/2}, \ t\to\infty
	\end{equation}
\end{lemma}
\begin{proof}
	For any $f\in L^\infty$,
	\begin{align}
	\parallel fC_z \parallel_{L^\infty}&\leq \parallel f \parallel_{L^\infty}\frac{1}{\pi}\iint_C\dfrac{|W^{(3)} (s)|}{|z-s|}dm(s)\nonumber.
	\end{align}
	 As $W^{(3)} (s)$ is a sectorial function,  we just need to consider it on ever sector.  	
	 We only detail the calculation for matrix function  in the sector $\Omega_1$ as $\xi> 1$.
	
	Proposition \ref{asyE} and \ref{unim}  imply  the boundedness of $M^{(r)}(z)$ and $M^{(r)}(z)^{-1}$ for $z\in \bar{\Omega}$ except a small neighborhood of $\pm 1$, $\pm\omega$ and $\pm\omega^2$. Take $\Omega_1$ as an example, then we only need to consider the singularity of  $M^{(r)}(z)$ at $z=1$. However, (\ref{R1}) give that near $z=1$, $\bar{\partial}R_1 (s)\equiv1$. Hence, we can directly write that
\begin{align}
	\frac{1}{\pi}\iint_{\Omega_1}\dfrac{|W^{(3)} (s)|}{|z-s|}dm(s)\lesssim& \iint_{\Omega_1\setminus\mathbb{D}(1,\varepsilon)}\dfrac{|\bar{\partial}R_1 (s) e^{it\theta_{12}}|}{|z-s|}dm(s).
\end{align}
Referring  to (\ref{dbarRj}) in proposition \ref{proR}, the right   integral  can be divided to two part
\begin{align}
	\iint_{\Omega_1\setminus\mathbb{D}(1,\varepsilon)}\dfrac{|\bar{\partial}R_1 (s)|e^{-t\text{Im}\theta_{12}}}{|z-s|}dm(s)\lesssim I_1+I_2,	
\end{align}
with
\begin{align}
	&I_1=\iint_{\Omega_1\setminus\mathbb{D}(1,\varepsilon)}\dfrac{(|p_1' (s)|+|\mathcal{X}'(|s|)|)e^{-2t\text{Im}\theta}}{|z-s|}dm(s),\nonumber\\
	&I_2=\iint_{\Omega_1\setminus\mathbb{D}(1,\varepsilon)}\dfrac{|s|^{-1/2}e^{-2t\text{Im}\theta}}{|z-s|}dm(s).	\nonumber
\end{align}
 Furthermore, by Corollary \ref{Imtheta}, for $z\in\Omega_1\setminus\mathbb{D}(1,\varepsilon)$,
\begin{equation}
e^{-t\text{Im}\theta_{12}}\leq e^{-c(\xi)tv}.
\end{equation}
Therefore, via (\ref{s-z}),
	\begin{align}
	I_1&\leq \int_{0}^{+\infty}\int_{v}^{\infty}\dfrac{|p_1' (s)|+|\mathcal{X}'(|s|)|}{|z-s|}dudv\nonumber\\
	&\leq \parallel |s-z|^{-1}\parallel_{L^2(\mathbb{R}^+)} \left( \parallel p_1'\parallel_{L^2(\mathbb{R}^+)} +\parallel \mathcal{X}'\parallel_{L^2(\mathbb{R}^+)}\right)  e^{-c(\xi)tv}dv\nonumber\\
	&\lesssim \int_{0}^{+\infty}|v-y|^{-1/2} e^{-c(\xi)tv}dv\lesssim t^{-1/2}.
	\end{align}
Before we estimating the second term, we  consider for $p>2$,
\begin{align}
	\left( \int_{v}^{+\infty}|\sqrt{u^2+v^2}|^{-\frac{p}{2}}du\right) ^{{p}/{2}}&=\left( \int_{v}^{+\infty}|l|^{-\frac{p}{2}+1}u^{-1}dl\right) ^{{p}/{2}}\lesssim v^{-\frac{1}{2}+\frac{1}{p}}.
\end{align}
By Cauchy-Schwarz inequality,
\begin{align}
	I_2&\leq \int_{0}^{+\infty}\parallel |s-z|^{-1}\parallel_{L^q(\mathbb{R}^+)} \parallel |z|^{-1/2}\parallel_{L^p(\mathbb{R}^+)}e^{-c(\xi)tv}dv\nonumber\\
	&\lesssim\int_{0}^{+\infty}|v-y|^{1/q-1}v^{-\frac{1}{2}+\frac{1}{p}}e^{-c(\xi)tv}dv\nonumber\\
	&\lesssim\int_{0}^{+\infty}v^{-\frac{1}{2}}e^{-c(\xi)tv}dv\lesssim t^{-1/2}.
\end{align}
  So the proof is completed.
\end{proof}

Then from (\ref{deM3}), we immediately obtain the existence and uniqueness of $M^{(3)}(z)$ for $z\in\mathbb{C}$.
Take $z=e^{\frac{i\pi}{6}}$ in (\ref{m3}), then
\begin{equation}
	M^{(3)}(e^{\frac{i\pi}{6}})=	I+\frac{1}{\pi}\iint_\mathbb{C}\dfrac{M^{(3)}(s)W^{(3)} (s)}{s-e^{\frac{i\pi}{6}}}dm(s).
\end{equation}
To reconstruct the solution of (\ref{Novikov}), we need following proposition.
\begin{Proposition}\label{asyM3i}
	There exists a small positive constant $0<\rho<  {1}/{4} $ and a constant $T_1$, such that for all $t>T_1$,    $M^{(3)}(z)$    admits the following estimate
	\begin{align}
		\parallel M^{(3)}(e^{\frac{i\pi}{6}})-I\parallel=\parallel\frac{1}{\pi}\iint_\mathbb{C}\dfrac{M^{(3)}(s)W^{(3)} (s)}{s-e^{\frac{i\pi}{6}}}dm(s)\parallel\lesssim t^{-1+\rho}.\label{m3i}
	\end{align}
\end{Proposition}
\begin{proof}
 The proof  proceeds along the same steps as the proof of above Proposition. 	{Lemma \ref{Cz}}  and (\ref{deM3}) implies that for large $t$,   $\parallel M^{(3)}\parallel_\infty \lesssim1$. And for same reason, we only estimate the integral on sector $\Omega_1$ as $\xi>1$. Let $s=u+vi=le^{i\vartheta}$, then:
	 \begin{align}
	 	\frac{1}{\pi}\iint_{\Omega_1}\dfrac{|W^{(3)} (s)|}{|s-e^{\frac{i\pi}{6}}|}dm(s)\lesssim& \iint_{\Omega_1\setminus\mathbb{D}(1,\varepsilon)}\dfrac{|\bar{\partial}R_1 (s) e^{it\theta_{12}}|}{|s-e^{\frac{i\pi}{6}}|}dm(s)\label{Int2}.
	 \end{align}
	For the first term, we use another estimation (\ref{dbarRj2}):
	\begin{equation}
	\iint_{\Omega_1\setminus\mathbb{D}(1,\varepsilon)}\dfrac{|\bar{\partial}R_1 (s) e^{it\theta_{12}}|}{|s-e^{\frac{i\pi}{6}}|}dm(s)\lesssim I_3+I_4,
	\end{equation}
	with
	\begin{align}
		I_3=\iint_{\Omega_1}\dfrac{\left( |p_1' (s)|+\mathcal{X}'(\text{Re}z)\right) e^{-2t\text{Im}\theta}}{|e^{\frac{i\pi}{6}}-s|}dm(s),\hspace{0.5cm}
		I_4=\iint_{\Omega_1}\dfrac{|s|^{-1}e^{-2t\text{Im}\theta}}{|e^{\frac{i\pi}{6}}-s|}dm(s).	
	\end{align}	
	For $r\in H^{1,1}(\mathbb{R})$, $r'\in L^1(\mathbb{R})$, which together with $|p_1'|\lesssim |r'|$ implies $p_1'\in L^1(\mathbb{R})$. therefore,
	\begin{align}
		I_3\leq& \int_{0}^{+\infty} \left( \parallel p_1'\parallel_{L^1(\mathbb{R}^+)}+\parallel \mathcal{X}'\parallel_{L^1(\mathbb{R}^+)} \right) \sqrt{2} e^{-c(\xi)tv} dv\nonumber\\
		&\lesssim \int_{0}^{+\infty}  e^{-c(\xi)tv} dv\leq t^{-1}.
	\end{align}
The second  inequality  from $|e^{\frac{i\pi}{6}}-s|$ is bounded for $s\in\Omega_1$. And we recall the sufficiently  small  positive constant $\rho<\frac{1}{4}$ to bound $I_4$. And we partition it to two parts:
\begin{align}
I_4\leq\int_{0}^{\frac{1}{2}}\int_{\frac{v}{\tan\varphi}}^{+\infty}\dfrac{|s|^{-1}e^{-c(\xi)tv}}{|e^{\frac{i\pi}{6}}-s|}dudv+
\int_{\frac{1}{2}}^{+\infty}\int_{\frac{v}{\tan\varphi}}^{+\infty}\dfrac{|s|^{-1}e^{-c(\xi)tv}}{|e^{\frac{i\pi}{6}}-s|}dudv.
\end{align}
Noting that for $0<v<\frac{1}{4}$, we have
$$|s|^2=u^2+v^2<|s-e^{\frac{i\pi}{6}}|^2=(u- {\sqrt{3}}/{2})^2+(v- {1}/{2})^2,$$
while   $v>\frac{1}{4}$, $|s-e^{\frac{i\pi}{6}}|^2\lesssim|s|^2$.
therefore   the first integral becomes
\begin{align}
	&\int_{0}^{\frac{1}{4}}\int_{\frac{v}{\tan\varphi}}^{+\infty}\dfrac{|s|^{-1}e^{-c(\xi)tv}}{|e^{\frac{i\pi}{6}}-s|}dudv\nonumber\\
	&\leq\int_{0}^{\frac{1}{4}}\int_{\frac{v}{\tan\varphi}}^{+\infty} (u^2+v^2)^{-\frac{1}{2}-\frac{\rho}{2}}[( u- {\sqrt{3}}/{2})^2+(v-{1}/{2})^2)]^{-\frac{1}{2}+\frac{\rho}{2}}due^{-c(\xi)tv}dv\nonumber\\
	&\leq \int_{0}^{\frac{1}{4}}\left[ \int_{v}^{+\infty}\left(1+\left( {u}/{v} \right)^2  \right) ^{-\frac{1}{2}-\frac{\rho}{2}}v^{-\rho} d {u}/{v}\right]  (v- {1}/{2})^{-1+\rho}e^{-c(\xi)tv}dv\nonumber\\
	&\lesssim  \int_{0}^{\frac{1}{4}}v^{-\rho}e^{-c(\xi)tv}dv\lesssim t^{-1+\rho}.
\end{align}
The second integral can be bounded in same way
\begin{align}
	&\int_{\frac{1}{2}}^{+\infty}\int_{\frac{v}{\tan\varphi}}^{+\infty}\dfrac{|s|^{-1}e^{-c(\xi)tv}}{|e^{\frac{i\pi}{6}}-s|}dudv\nonumber\\
	&\lesssim \int_{\frac{1}{2}}^{+\infty}(v-{1}/{2})^{-\rho}e^{-c(\xi)tv}dv \lesssim e^{-\frac{c(\xi)}{2}t}.
\end{align}
This estimate  is strong enough to arrive at  the result (\ref{m3i}).
\end{proof}

\subsection{In space-time   region  $  -1/8 <\xi<  1 $  }

\begin{lemma}\label{Cz1}
The norm of the integral operator $C_z$ admits the estimate
	\begin{equation}
	\parallel C_z\parallel_{L^\infty\to L^\infty}\lesssim t^{-1/4}, \ \ t\to\infty.
	\end{equation}
\end{lemma}
\begin{proof}
	For any $f\in L^\infty$,
	\begin{align}
	\parallel fC_z \parallel_{L^\infty}&\leq \parallel f \parallel_{L^\infty}\frac{1}{\pi}\iint_\mathbb{C}\dfrac{|W^{(3)} (s)|}{|z-s|}dm(s)\nonumber.
	\end{align}
	Like in the case of $\xi\in(-\infty,-1/8)\cup(1,+\infty)$, (\ref{R11}) gives that  the singularity at $\pm 1,\ \pm\omega,\ \pm\omega^2$ of $M^{R}(z)$ will not have influence.
	
To avoid repetition, we omit the calculation of these regions. And in the proof of this lemma, we focus on the different regions of $\mathbb{C}$. So, we only detail the case for matrix functions having support in the sector $\Omega_{11}$ as $\xi\in (-1/8,0)$.
	
	Recall the definition of $W^{(3)} (s)=M^{R}(z)\bar{\partial}R^{(2)}(z)M^{R}(z)^{-1}$.  $W^{(3)} (s)\equiv0$ out of $\bar{\Omega}$.
	Proposition \ref{asyE}, \ref{unim} and \ref{asymlo}  implies  the boundedness of $M^{R}(z)$ and $M^{R}(z)^{-1}$ for $z\in \overline{\Omega_{11}}$, so
	\begin{equation}
	\frac{1}{\pi}\iint_{\Omega_{11}}\dfrac{|W^{(3)} (s)|}{|z-s|}dm(s)\lesssim \frac{1}{\pi}\iint_{\Omega_{11}}\dfrac{|\bar{\partial}R_{11} (s) e^{2it\theta}|}{|z-s|}dm(s).
	\end{equation}
	Referring  to (\ref{dbarRj3}) in proposition \ref{proR1}, and note that $|\mathcal{X}'(\text{Re}z)|=0$ in $\overline{\Omega_{11}}$, then the integral $\iint_{\Omega_{11}}\dfrac{|\bar{\partial}R_{11} (s)|}{|z-s|}dm(s)$	can be divided into two part:
	\begin{align}
	\iint_{\Omega_{11}}\dfrac{|\bar{\partial}R_{11} (s)|e^{-2t\text{Im}\theta}}{|z-s|}dm(s)\lesssim \hat{I}_1+\hat{I}_2,	
	\end{align}
	with
	\begin{align}
	&\hat{I}_1=\iint_{\Omega_{11}}\dfrac{|p_{11}'(|s|)|e^{-2t\text{Im}\theta}}{|z-s|}dm(s),\\
	&\hat{I}_2=\iint_{\Omega_{11}}\dfrac{|s-\xi_1|^{-1/2}e^{-2t\text{Im}\theta}}{|z-s|}dm(s).
	\end{align}
	Recall that $z=x+yi$, $s=\xi_1+u+vi$ with $x,y,u,v\in\mathbb{R}$, then lemma \ref{theta2} and Corollary \ref{theta2c} give that
	\begin{align}
	 \hat{I}_1\leq&\int_{0}^{R(\xi)}\int_{v}^{R(\xi)}\dfrac{|p_{11}'(|s|)|}{|z-s|}e^{-c^{(1)}(\xi)tvu}dudv+\int_{0}^{R(\xi)}\int_{R(\xi)}^{+\infty}\dfrac{|p_{11}'(|s|)|}{|z-s|}e^{-c^{(2)}(\xi)tv}dudv\nonumber\\
	&+\int_{R(\xi)}^{+\infty}\int_{v}^{+\infty}\dfrac{|p_{11}'(|s|)|}{|z-s|}e^{-c^{(2)}(\xi)tv}dudv.\label{38}
\end{align}
	For the first integral, if the  real part of $z$ is greater than 0 and less that $\int_{y}^{R(\xi)}$, it is easy to obtain:
\begin{align}
	&\int_{0}^{R(\xi)}\int_{v}^{R(\xi)}\dfrac{|p_{11}'(|s|)|}{|z-s|}e^{-c^{(1)}(\xi)tvu}dudv\nonumber\\
	&\leq\int_{0}^{R(\xi)}\parallel p_{11}'\parallel_2\parallel |z-s|^{-1}\parallel_2e^{-c^{(1)}(\xi)tv^2}dv\nonumber\\
	&\lesssim\int_{0}^{R(\xi)}|v-y|^{-1/2}e^{-c^{(1)}(\xi)tv^2}dv\nonumber\\
&= (\int_{0}^{y}+\int_{y}^{R(\xi)} )|v-y|^{-1/2}e^{-c^{(1)}(\xi)tv^2}dv  .
	\end{align}
And if $y<0$, we do not need to consider the  absolute value sign of $|v-y|$, namely, $\int_{0}^{y}$ part of the integral is zero. Similarly, if $y>R(\xi)$, $\int_{y}^{R(\xi)}$ part of the integral is zero.
We use the  inequality $e^{-z}\lesssim z^{-1/4}$ in the $\int_{0}^{y}$ part:
	\begin{align}
	\int_{0}^{y}(y-v)^{-1/2}e^{-c^{(1)}(\xi)tv^2}dv\lesssim \int_{0}^{y} (y-v)^{-1/2}v^{-1/2}dvt^{-1/4}\lesssim t^{-1/4}.
	\end{align}
	And for the $\int_{y}^{R(\xi)}$ part of the integral, we make the substitution $w=v-y:0\to +\infty$
	\begin{align}
	&\int_{y}^{R(\xi)}(v-y)^{-1/2}e^{-c^{(1)}(\xi)tv^2}dv=\int_{0}^{+\infty}w^{-1/2}e^{-c^{(1)}(\xi)ty(w+y)}dw\nonumber\\
	&=\int_{0}^{+\infty}w^{-1/2}e^{-c^{(1)}(\xi)tyw}dwe^{-c^{(1)}(\xi)ty^2}\lesssim e^{-c^{(1)}(\xi)ty^2}.
	\end{align}
Similarly, following the same step and using $e^{-z}\lesssim z^{-1/2}$, the second integral in (\ref{38}) has
\begin{align}
	&\int_{0}^{R(\xi)}\int_{R(\xi)}^{+\infty}\dfrac{|p_{11}'(|s|)|}{|z-s|}e^{-c^{(2)}(\xi)tv}dudv
	\leq\int_{0}^{R(\xi)}\parallel p_{11}'\parallel_2\parallel |z-s|^{-1}\parallel_2e^{-c^{(2)}(\xi)tv}dv\nonumber\\
	&\lesssim\int_{0}^{R(\xi)}|v-y|^{-1/2}e^{-c^{(2)}(\xi)tv}dv\lesssim t^{-1/2}.
\end{align}
And for the last term in (\ref{38}), similarly, we have
\begin{align}
	&\int_{R(\xi)}^{+\infty}\int_{v}^{R(\xi)}\dfrac{|p_{11}'(|s|)|}{|z-s|}e^{-c^{(2)}(\xi)tv}dudv\nonumber\\
	&\leq\int_{R(\xi)}^{+\infty}\parallel p_{11}'\parallel_2\parallel |z-s|^{-1}\parallel_2e^{-c^{(1)}(\xi)tv}dv\lesssim\int_{R(\xi)}^{+\infty}(v-R(\xi))^{-1/2}e^{-c^{(2)}(\xi)tv}dv\nonumber\\
	&=e^{-c^{(2)}(\xi)tR(\xi)}\int_{0}^{+\infty}(v-R(\xi))^{-1/2}e^{-c^{(2)}(\xi)t(v-R(\xi))}d(v-R(\xi))\nonumber\\
	&\lesssim e^{-c^{(2)}(\xi)tR(\xi)}.
\end{align}
	Like  the Lemma \ref{Cz}, we bound $\hat{I}_2$. For $p>2$, and $1/p+1/q=1$, analogously, we divide it into three parts:
	\begin{align}
		 \hat{I}_2\leq&\int_{0}^{R(\xi)}\int_{v}^{R(\xi)}\dfrac{|s-\xi_1|^{-1/2}}{|z-s|}e^{-c^{(1)}(\xi)tvu}dudv+\int_{0}^{R(\xi)}\int_{R(\xi)}^{+\infty}\dfrac{|s-\xi_1|^{-1/2}}{|z-s|}e^{-c^{(2)}(\xi)tv}dudv\nonumber\\
		&+\int_{R(\xi)}^{+\infty}\int_{v}^{+\infty}\dfrac{|s-\xi_1|^{-1/2}}{|z-s|}e^{-c^{(2)}(\xi)tv}dudv.\label{44}
	\end{align}
We assume $0<y<R(\xi)$, then the first term has:
	\begin{align}
	&\int_{0}^{R(\xi)}\int_{v}^{R(\xi)}\dfrac{|s-\xi_1|^{-1/2}}{|z-s|}e^{-c^{(1)}(\xi)tvu}dudv\nonumber\\
	&\leq \int_{0}^{R(\xi)} \parallel |s-\xi_1|^{-1/2}\parallel_p\parallel |z-s|^{-1}\parallel_q e^{-c^{(1)}(\xi)tv^2}dv\nonumber\\
	&\lesssim \int_{0}^{R(\xi)}v^{1/p-1/2}|y-v|^{1/q-1}e^{-c^{(1)}(\xi)tv^2}dv\nonumber\\
	&= (\int_{0}^{y}+\int_{y}^{R(\xi)} )v^{1/p-1/2}|y-v|^{1/q-1}e^{-c^{(1)}(\xi)tv^2}dv.
	\end{align}
If $y<0$,  $\int_{0}^{y}$ part of the integral is zero. And if $y>R(\xi)$, $\int_{y}^{R(\xi)}$ part of the integral is zero.	Analogously,
	\begin{align}
	&\int_{0}^{y}v^{1/p-1/2}|y-v|^{1/q-1}e^{-c^{(1)}(\xi)tv^2}dv\nonumber\\
	&\lesssim \int_{0}^{y}v^{1/p-1}(y-v)^{1/q-1} dv t^{-1/4}\lesssim t^{-1/4},
	\end{align}
	and
	\begin{align}
	&\int_{y}^{R(\xi)}v^{1/p-1/2}|y-v|^{1/q-1}e^{-c^{(1)}(\xi)tv^2}dv\nonumber\\
	&\leq \int_{y}^{+\infty}(v-y)^{-1/2} e^{-c^{(1)}(\xi)ty(v-y)}dve^{-c^{(1)}(\xi)ty^2}\lesssim e^{-c^{(1)}(\xi)ty^2}.
	\end{align}
	The other parts of (\ref{44}) has similar estimation like the case of $\hat{I}_1$. Then the result is confirmed.
\end{proof}

Then from (\ref{deM3}), we immediately arrive at  the existence and uniqueness of $M^{(3)}(z)$ for $z\in\mathbb{C}$.
Take $z=e^{\frac{i\pi}{6}}$ in (\ref{m3}), then
\begin{equation}
M^{(3)}(e^{\frac{\pi i}{6}})=I+\frac{1}{\pi}\iint_{\mathbb{C}}\dfrac{M^{(3)}(s)W^{(3)} (s)}{s-e^{\frac{\pi i}{6}}}dm(s).
\end{equation}
To reconstruct the solution of (\ref{Novikov}), we need following proposition.
\begin{Proposition}\label{asyM3i1}
	There exist constants $T_1$, such that for all $t>T_1$, the solution $M^{(3)}(z)$  of  $\bar{\partial}$-problem  admits the following estimation:
	\begin{align}
	\parallel M^{(3)}(e^{\frac{\pi i}{6}})-I\parallel=\parallel\frac{1}{\pi}\iint_{\mathbb{C}}\dfrac{M^{(3)}(s)W^{(3)} (s)}{s-e^{\frac{\pi i}{6}}}dm(s)\parallel\lesssim t^{-3/4}.\label{m3i1}
	\end{align}
\end{Proposition}
\begin{proof}
	 The proof proceeds along the same steps as the proof of  Proposition  \ref{asyM3i}. 	{Lemma \ref{Cz1}} and (\ref{deM3}) imply that for large $t$,   $\parallel M^{(3)}\parallel_\infty \lesssim1$. So we just need to bound $\iint_{\mathbb{C}}\dfrac{|W^{(3)} (s)|}{s-e^{\frac{\pi i}{6}}}dm(s)$.  And we only give the details on $\Omega_{11}$, the integral on other region can be obtained in the same way.  Referring  to (\ref{dbarRj3}) and (\ref{dbarRj4}) in proposition \ref{proR1}, and note that  $|\mathcal{X}'(\text{Re}z)|=0$ in $\overline{\Omega_{11}}$, this integral 	can be divided into four parts
	\begin{align}
	&\frac{1}{\pi}\iint_{\Omega_{11}}\dfrac{M^{(3)}(s)W^{(3)} (s)}{s-e^{\frac{\pi i}{6}}}dm(s)\lesssim \frac{1}{\pi}\iint_{\Omega_{11}}\dfrac{R_{11}(s)e^{-2t\text{Im}\theta}}{s-e^{\frac{\pi i}{6}}}dm(s)\nonumber\\ =&\frac{1}{\pi}\int_{0}^{R(\xi)}\int_{\frac{v}{\tan\varphi}}^{R(\xi)}\dfrac{\bar{\partial}R_{11}(s)e^{-2t\text{Im}\theta}}{s-e^{\frac{\pi i}{6}}}dudv+\int_{0}^{R(\xi)}\int_{R(\xi)}^{+\infty}\dfrac{\bar{\partial}R_{11}(s)e^{-2t\text{Im}\theta}}{s-e^{\frac{\pi i}{6}}}dudv\nonumber\\ &+\int_{R(\xi)}^{+\infty}\int_{\frac{v}{\tan\varphi}}^{+\infty}\dfrac{\bar{\partial}R_{11}(s)e^{-2t\text{Im}\theta}}{s-e^{\frac{\pi i}{6}}}dudv\lesssim\hat{I}_3+\hat{I}_{41}+\hat{I}_{42}+\hat{I}_{43}.
	\end{align}
	Here, near $\xi_1$ we use (\ref{dbarRj3}) and away from $\xi_1$ we use (\ref{dbarRj4}):
	\begin{align}
	&\hat{I}_3=\iint_{\Omega_{11}}\dfrac{|p_{11}'(s)|e^{-2t\text{Im}\theta}}{|e^{\frac{\pi i}{6}}-s|}dm(s),\nonumber\\
	&\hat{I}_{41}=\int_{0}^{R(\xi)}\int_{\frac{v}{\tan\varphi}}^{R(\xi)}|s-\xi_1|^{-1/2}|e^{\frac{\pi i}{6}}-s|^{-1}e^{-2t\text{Im}\theta}dudv,\nonumber\\
	&\hat{I}_{42}=\int_{0}^{R(\xi)}\int_{R(\xi)}^{+\infty}|s-\xi_1|^{-1/2}|e^{\frac{\pi i}{6}}-s|^{-1}e^{-2t\text{Im}\theta}dudv,\nonumber\\
	&\hat{I}_{43}=\int_{R(\xi)}^{+\infty}\int_{\frac{v}{\tan\varphi}}^{+\infty}|s-\xi_1|^{-1/2}|e^{\frac{\pi i}{6}}-s|^{-1}e^{-2t\text{Im}\theta}dudv.\nonumber	
	\end{align}
	For $\hat{I}_3$,  $|e^{\frac{\pi i}{6}}-s|^{-1}$ has nonzero  maximum. Together with  lemma \ref{theta2} and Corollary \ref{theta2c}, it has:
	\begin{align}
		\hat{I}_3\leq&\int_{0}^{R(\xi)}\int_{v}^{R(\xi)}|p_{11}'(|s|)|e^{-c^{(1)}(\xi)tvu}dudv+\int_{0}^{R(\xi)}\int_{R(\xi)}^{+\infty}|p_{11}'(|s|)|e^{-c^{(2)}(\xi)tv}dudv\nonumber\\
		&+\int_{R(\xi)}^{+\infty}\int_{v}^{+\infty}|p_{11}'(|s|)|e^{-c^{(2)}(\xi)tv}dudv.\label{56}
	\end{align}
For the first part in the right hand of above inequality
	\begin{align}
	&\int_{0}^{R(\xi)}\int_{v}^{R(\xi)}|p_{11}'(|s|)|e^{-c^{(1)}(\xi)tvu}dudv\leq \int_{0}^{R(\xi)} \parallel p_{11}'(s) \parallel_2\left( \int_v ^{+\infty}e^{-2c^{(1)}(\xi)tvu}du\right) ^{1/2}dv\nonumber\\
	&\lesssim t^{-1/2}\int_{0}^{R(\xi)}v ^{-1/2}e^{-2c^{(1)}(\xi)tv^2}dv=\mathcal{O} (t^{-3/4}). \nonumber
	\end{align}
	And the second term in (\ref{56}) has:
	\begin{align}
		&\int_{0}^{R(\xi)}\int_{R(\xi)}^{+\infty}|p_{11}'(|s|)|e^{-c^{(2)}(\xi)tv}dudv\leq\int_{0}^{R(\xi)} \parallel p_{11}'(s) \parallel_1e^{-c^{(2)}(\xi)tv} dv\nonumber\\
		&\lesssim \int_{0}^{R(\xi)}e^{-c^{(1)}(\xi)tv}dv=\mathcal{O} (t^{-3/2}).
	\end{align}
Similarly, the last integral in (\ref{56}) has
\begin{align}
	&\int_{R(\xi)}^{+\infty}\int_{v}^{+\infty}|p_{11}'(|s|)|e^{-c^{(2)}(\xi)tv}dudv\leq	\int_{R(\xi)}^{+\infty} \parallel p_{11}'(s) \parallel_1e^{-c^{(2)}(\xi)tv} dv\nonumber\\
	&\lesssim 	\int_{R(\xi)}^{+\infty}e^{-c^{(1)}(\xi)tv}dv=\mathcal{O} (t^{-3/2}e^{-c^{(1)}(\xi)tR(\xi)}).
\end{align}
	So $\hat{I}_3 \lesssim t^{-3/4}$. And for $\hat{I}_{41}$,
	lemma \ref{theta2} and Corollary \ref{theta2c} also give:
	\begin{align}
		\hat{I}_{41}\leq&\int_{0}^{R(\xi)}\int_{v}^{R(\xi)}|s-\xi_1|^{-1/2}e^{-c^{(1)}(\xi)tvu}dudv.\label{60}
	\end{align}	
	 Similarly we take $2<p<4$, and $1/p+1/q=1$, then the first integral in above  inequality has:
	\begin{align}
	&\int_{0}^{R(\xi)}\int_{v}^{R(\xi)}|s-\xi_1|^{-1/2}e^{-c^{(1)}(\xi)tvu}dudv\nonumber\\
	&\leq  \int_{0}^{R(\xi)} \parallel |s-\xi_1|^{-1/2} \parallel_p\left( \int_v ^{+\infty}e^{-qc^{(1)}(\xi)tvu}du\right) ^{1/q}dv\nonumber\\
	&= t^{-1/q} \int_{0}^{+\infty} v^{2/p-3/2}e^{-qc^{(1)}(\xi)tv^2}dv\lesssim t^{-3/4}.
	\end{align}
	 To bound $\hat{I}_{42}$, we use the fact that $|s-\xi_1|^{-1/2}|e^{\frac{\pi i}{6}}-s|^{-1}\lesssim u^{-3/2}$ in it integral domain:
	\begin{align}
		\hat{I}_{42}&\lesssim\int_{0}^{R(\xi)}\int_{R(\xi)}^{+\infty}u^{-3/2}e^{-c^{(2)}(\xi)tv}dudv\lesssim\int_{0}^{R(\xi)}e^{-c^{(2)}(\xi)tv}dv\lesssim t^{-1}.
	\end{align}
	Similarly,
	\begin{align}
		&\hat{I}_{43}\leq\int_{R(\xi)}^{+\infty}\int_{v}^{+\infty}u^{-3/2}e^{-c^{(2)}(\xi)tv}dudv\nonumber\\
		&\lesssim\int_{R(\xi)}^{+\infty}v^{-3/2}e^{-c^{(2)}(\xi)tv}dv\lesssim\int_{R(\xi)}^{+\infty}e^{-c^{(2)}(\xi)tv}dv\lesssim t^{-1}e^{-c^{(2)}(\xi)tR(\xi)}.\nonumber
	\end{align}
	The result is confirmed.
\end{proof}

\section{Long-time asymptotic  behavior  }\label{sec9}

\quad Now we begin to construct the long time asymptotics of the Novikov equation (\ref{Novikov}).
 Inverting the sequence of transformations (\ref{transm1}), (\ref{transm2}), (\ref{transm3}) and (\ref{transMr}), we have
\begin{align}
M(z)=&M^{(3)}(z)E(z;\xi)M^{(r)}(z)R^{(2)}(z)^{-1}T(z)^{-\sigma_3}.
\end{align}
To  reconstruct the solution $u(x,t)$ by using (\ref{recons u}),   we take $z=e^{\frac{\pi i}{6}}$. In this case,  $ R^{(2)}(z)=I$. Further using   Propositions \ref{proT}  and  \ref{asyM3i},  we can obtain that
\begin{align}
	M(e^{\frac{\pi i}{6}})=&M^{(3)}(e^{\frac{\pi i}{6}}) E(e^{\frac{\pi i}{6}};\xi)M^{(r)}_\lozenge(e^{\frac{\pi i}{6}})
	T(e^{\frac{\pi i}{6}})^{-1}.
\end{align}

\noindent$\blacktriangleright$ \ For $\xi\in(-\infty,-1/8)\cup(1,+\infty)$, we have
\begin{align}
	M(e^{\frac{\pi i}{6}})=& M^{(r)}_\lozenge(e^{\frac{\pi i}{6}})T(e^{\frac{\pi i}{6}})^{-1}
	+\mathcal{O}(t^{-1+\rho}),
\end{align}
Substitute above estimats into  (\ref{recons u}) and (\ref{recons x}), and obtain
\begin{align}
	&u(x,t)=u(y(x,t),t)\nonumber\\
	=&\frac{1}{2}\tilde{m}_1(y,t)\left(\frac{M_{33}(e^{\frac{i\pi}{6}};y,t)}{M_{11}(e^{\frac{i\pi}{6}};y,t)} \right)^{1/2}+ \frac{1}{2}\tilde{m}_3(y,t)\left(\frac{M_{33}(e^{\frac{i\pi}{6}};y,t)}{M_{11}(e^{\frac{i\pi}{6}};y,t)} \right)^{-1/2}-1
	\nonumber\\
	=&u^r(y(x,t);\tilde{\mathcal{D}}_\lozenge)\left( T_1(e^{\frac{\pi i}{6}})T_3(e^{\frac{\pi i}{6}})\right) ^{-1/2}+\left( T_1(e^{\frac{\pi i}{6}})T_3(e^{\frac{\pi i}{6}})\right) ^{-1/2}-1+\mathcal{O}(t^{-1+\rho}), \nonumber
\end{align}
where
\begin{align}
	&x(y,t)=y+c_+(x,t)=y+\frac{1}{2} \ln\frac{M_{33}(e^{\frac{i\pi}{6}};y,t)}{M_{11}(e^{\frac{i\pi}{6}};y,t)}\nonumber\\
	=&y+\frac{1}{2} \ln\left(\frac{[M^{(r)}_\lozenge]_{33}(e^{\frac{i\pi}{6}};y,t)}{[M^{(r)}_\lozenge]_{11}(e^{\frac{i\pi}{6}};y,t)} \right)+\frac{1}{2} \ln\left(\frac{T_1(e^{\frac{\pi i}{6}})}{T_3(e^{\frac{\pi i}{6}})}\right)+\mathcal{O}(t^{-1+\rho}) \nonumber\\
	=& c_+^r(x,t;\tilde{\mathcal{D}}_\lozenge)+\frac{1}{2} \ln\left[ T_1(e^{\frac{\pi i}{6}})T_3(e^{\frac{\pi i}{6}})\right]+\mathcal{O}(t^{-1+\rho}) ,\label{resultx}
\end{align}
where $u^r(x,t;\tilde{\mathcal{D}}_\lozenge)$, $[\tilde{m}^{(r)}_\lozenge]_j$ and $c_+^r(x,t;\tilde{\mathcal{D}}_\lozenge)$ shown in Corollary \ref{sol}.

\noindent$\blacktriangleright$ \ For  $\xi\in(-1/8,1)$, Propositions \ref{proT},  \ref{asyE},  \ref{asyM3i1} and Corollary \ref{sol} give that
\begin{align}
M(z)=& M^{(3)}(e^{\frac{\pi i}{6}})E(e^{\frac{\pi i}{6}})  M^{(r)}_\lozenge(e^{\frac{\pi i}{6}}) T(e^{\frac{\pi i}{6}})^{-1}\nonumber\\
=&\left(I+H^{(0)}t^{-1/2} \right) M^{(r)}_\lozenge  (e^{\frac{\pi i}{6}})
T(e^{\frac{\pi i}{6}})^{-\sigma_3}+\mathcal{O}(t^{-3/4}).
\end{align}
Similarly, substituting above estimation into  (\ref{recons u}) and (\ref{recons x}), then we arrive at:
\begin{align}
	&u(x,t)=u(y(x,t),t)\nonumber\\
	=&\frac{1}{2}\tilde{m}_1(y,t)\left(\frac{M_{33}(e^{\frac{i\pi}{6}};y,t)}{M_{11}(e^{\frac{i\pi}{6}};y,t)} \right)^{1/2}+ \frac{1}{2}\tilde{m}_3(y,t)\left(\frac{M_{33}(e^{\frac{i\pi}{6}};y,t)}{M_{11}(e^{\frac{i\pi}{6}};y,t)} \right)^{-1/2}-1
	\nonumber\\
	=&u^r(y(x,t);\tilde{\mathcal{D}}_\lozenge)\left[ T_1(e^{\frac{\pi i}{6}})T_3(e^{\frac{\pi i}{6}})\right]^{-1/2}+\left[ T_1(e^{\frac{\pi i}{6}})T_3(e^{\frac{\pi i}{6}})\right] ^{-1/2}-1\nonumber\\
	&+f_1(y(x,t),t)t^{-1/2}+\mathcal{O}(t^{-3/4}),\label{resultu1}
\end{align}
and
\begin{align}
	&x(y,t)=y+c_+(x,t) = y+\frac{1}{2} \ln\frac{M_{33}(e^{\frac{i\pi}{6}};y,t)}{M_{11}(e^{\frac{i\pi}{6}};y,t)}\nonumber\\
	=&y+\frac{1}{2} \ln\left[ {T_1(e^{\frac{\pi i}{6}})} T_3(e^{\frac{\pi i}{6}})^{-1}\right] \nonumber\\
	&+\frac{1}{2} \ln \frac{[M^{(r)}_\lozenge]_{33}(e^{\frac{i\pi}{6}};y,t)+\sum_{n=1}^3H^{(0)}_{3n}[M^{(r)}_\lozenge]_{n3}(e^{\frac{i\pi}{6}};y,t)t^{-1/2}
+\mathcal{O}(t^{-3/4})}{[M^{(r)}_\lozenge]_{11}(e^{\frac{i\pi}{6}};y,t)+\sum_{n=1}^3H^{(0)}_{1n}[M^{(r)}_\lozenge]_{n1}(e^{\frac{i\pi}{6}};y,t)t^{-1/2}+\mathcal{O}(t^{-3/4})}  \nonumber\\
	=& c_+^r(x,t;\tilde{\mathcal{D}}_\lozenge)+\frac{1}{2} \ln\left[ {T_1(e^{\frac{\pi i}{6}})} T_3(e^{\frac{\pi i}{6}})^{-1}\right]+f_2(y(x,t),t)t^{-1/2}+\mathcal{O}(t^{-3/4}) .\label{resultx1}
\end{align}

Summering obtained results above,  we achieve our  main results

\begin{theorem}\label{last}   Let $u(x,t)$ be the solution for  the Novikov equation  (\ref{Novikov})  with generic initial data   $u_0(x)\in \mathcal{S}(\mathbb{R})$ and
 its associated with  scatting data $\left\lbrace  r(z),\left\lbrace \zeta_n,c_n\right\rbrace^{6N_0}_{n=1}\right\rbrace$. Let $\xi=\frac{y}{t}$  and
  $q^r_\lozenge(x,t)$ be the $\mathcal{N}(\lozenge)$-soliton solution corresponding to   scattering data
$\tilde{\mathcal{D}_\lozenge}=\left\lbrace  0,\left\lbrace \zeta_n,C_nT^2(\zeta_n)\right\rbrace_{n\in\lozenge}\right\rbrace$ shown in Corollary \ref{sol}.  Then there exists a large constant $T_1=T_1(\xi)$ such that for all $T_1<t\to\infty$, we have
\begin{itemize}
\item[{\rm (i)}]  In space-time   regions   $  \xi <-1/8 $ and $   \xi >1 $,
\begin{align}
	u(x,t)=& u^r(y(x,t);\tilde{\mathcal{D}}_\lozenge)\left[ T_1(e^{\frac{\pi i}{6}})T_3(e^{\frac{\pi i}{6}})\right] ^{-1/2}\nonumber\\
	&+\left[ T_1(e^{\frac{\pi i}{6}})T_3(e^{\frac{\pi i}{6}})\right] ^{-1/2}-1+\mathcal{O}(t^{-1+\rho}), \ 0<\rho< {1}/{4},
\end{align}
where
\begin{align}
	x(y,t)= c_+^r(x,t;\tilde{\mathcal{D}}_\lozenge)+\frac{1}{2} \ln \left[ {T_1(e^{\frac{\pi i}{6}})} T_3(e^{\frac{\pi i}{6}})^{-1}\right]  +\mathcal{O}(t^{-1+\rho}),\nonumber
\end{align}
and   $u^r(x,t;\tilde{\mathcal{D}}_\lozenge)$, $c_+^r(x,t;\tilde{\mathcal{D}}_\lozenge)$ are given by Prop. \ref{proT} and Corol. \ref{sol}  respectively.

\item[{\rm (ii)}] In space-time   region  $  -1/8 <\xi<  1 $,
\begin{align}
u(x,t)=& u^r(y(x,t);\tilde{\mathcal{D}}_\lozenge)\left[ T_1(e^{\frac{\pi i}{6}})T_3(e^{\frac{\pi i}{6}})\right] ^{-1/2}\nonumber\\
&+\left[ T_1(e^{\frac{\pi i}{6}})T_3(e^{\frac{\pi i}{6}})\right] ^{-1/2}-1+f_1t^{-1/2}+\mathcal{O}(t^{-3/4}),
\end{align}
where
\begin{align}
x(y,t)&=  c_+^r(x,t;\tilde{\mathcal{D}}_\lozenge)+\frac{1}{2} \ln \left[ {T_1(e^{\frac{\pi i}{6}})} T_3(e^{\frac{\pi i}{6}})^{-1}\right]+f_2t^{-1/2}+\mathcal{O}(t^{-3/4}), \nonumber
\end{align}
and  $u^r(x,t;\tilde{\mathcal{D}}_\lozenge)$, $c_+^r(x,t;\tilde{\mathcal{D}}_\lozenge)$  are given by Prop. \ref{proT}, Corol. \ref{sol},  respectively;
and $f_{1}$ and $f_{2}$   are given by (\ref{f11}) and (\ref{f12}),  respectively.
\end{itemize}
\end{theorem}
Our results also show that the poles on curve soliton solutions of Novikov  equation has dominant contribution to the solution as $t\to\infty$.\vspace{6mm}

\noindent\textbf{Acknowledgements}

This work is supported by  the National Science
Foundation of China (Grant No. 11671095,  51879045).

\hspace*{\parindent}
\\

\end{document}